\documentclass[letterpaper,12pt]{article}
\usepackage{amsmath}
\usepackage{amssymb}
\usepackage{amsfonts}
\usepackage{setspace}
\usepackage{bbm}
\usepackage{color}
\usepackage{enumerate}
\usepackage{mathrsfs}
\usepackage[longnamesfirst]{natbib}
\usepackage[bottom]{footmisc}
\usepackage{graphicx}
\usepackage{subcaption}
\usepackage{hyperref}
\usepackage[shortlabels]{enumitem}
\usepackage[usenames,dvipsnames,svgnames,table]{xcolor}
\usepackage{multirow}
\usepackage{textgreek}
\usepackage[margin=1.15in]{geometry}

\usepackage{tikz}
\usetikzlibrary{topaths,calc}
\usepackage{enumitem}

\newtheorem{theorem}{Theorem}[section]
\newtheorem{lemma}{Lemma}[section]
\newtheorem{proposition}{Proposition}[section]
\newtheorem{assumption}{Assumption}

\newtheorem{definition}{Definition}[section]
\newtheorem{remark}{Remark}[section]
\newtheorem{example}{Example}[section]
\newenvironment{proof}[1][Proof]{\noindent \textbf{#1.} }{\  \rule{0.5em}{0.5em}}
\newtheorem{assumptionalt}{Assumption}[assumption]

\newenvironment{assumptionp}[1]{
  \renewcommand\theassumptionalt{#1}
  \assumptionalt
}{\endassumptionalt}

\newcommand{\mb}[1]{\mathbb{#1}}
\newcommand{\wh}[1]{\widehat{#1}}

\newcommand{\mr}[1]{\mathrm{#1}}

\newcommand{\mc}[1]{\mathcal{#1}}

\newcommand{\ind}{1\!\mathrm{l}}
\newcommand{\ul}[1]{\underline{#1}}
\newcommand{\ol}[1]{\overline{#1}}

\hypersetup{
     colorlinks   = true,
     citecolor    = Blue,
     urlcolor     = Black,
     linkcolor    = Blue
}

\makeatletter
\renewcommand\paragraph{\@startsection{paragraph}{4}{\z@}%
                                    {0pt \@plus1ex \@minus.2ex}%
                                    {-1em}%
                                    {\normalfont\normalsize\bfseries}}
\makeatother

\usepackage{float}

\newcommand{\E}{\mathbb{E}}
\graphicspath{{pic/}}
\DeclareGraphicsExtensions{.pdf,.jpg,.png}
\newcommand{\bpm}{\begin{pmatrix}}
\newcommand{\epm}{\end{pmatrix}}

\usepackage{bibunits}
\usepackage[longnamesfirst]{natbib}

\begin{document}

\defaultbibliography{sm}
\defaultbibliographystyle{chicago}
\begin{bibunit}
\shortcites{BenTal2013}

\author{%
{Timothy Christensen\thanks{%
 Department of Economics, New York University. \texttt{timothy.christensen@nyu.edu}}
 \quad
 Benjamin Connault\thanks{%
 Latour Trading LLC. \texttt{benjamin.connault@gmail.com}}
 }
}

\title{%
Counterfactual Sensitivity and Robustness\footnote{
We are very grateful to a co-editor and four anonymous referees, R.~Allen, S.~Bonhomme, G.~Chamberlain, A.~Galvao, B.~Kaplowitz, J.~Lazarev, M.~Masten, K.~Menzel, D.~Miller, M.~Mogstad, F.~Molinari, J.~Nesbit, A.~Poirier, J.~Porter, T.~Sargent, A.~Torgovitsky, and numerous seminar and conference participants for helpful comments and suggestions. Maximilian Huber provided outstanding research assistance. Support from the National Science Foundation via grant SES-1919034 is gratefully acknowledged.}
}

\date{May 25, 2022}

\maketitle
\thispagestyle{empty}

\begin{abstract}  
\singlespacing
\noindent 
We propose a framework for analyzing the sensitivity of counterfactuals to parametric assumptions about the distribution of latent variables in structural models. In particular, we derive bounds on counterfactuals as the distribution of latent variables spans nonparametric neighborhoods of a given parametric specification while other ``structural'' features of the model are maintained. Our approach recasts the infinite-dimensional problem of optimizing the counterfactual with respect to the distribution of latent variables (subject to model constraints) as a finite-dimensional convex program. We also develop an MPEC version of our method to further simplify computation in models with endogenous parameters (e.g., value functions) defined by equilibrium constraints. We propose plug-in estimators of the bounds and two methods for inference. We also show that our bounds converge to the sharp nonparametric bounds on counterfactuals as the neighborhood size becomes large. To illustrate the broad applicability of our procedure, we present empirical applications to matching models with transferable utility and dynamic discrete choice models.

\medskip 


\medskip


\end{abstract} 

\clearpage
\setcounter{page}{1}

\pagenumbering{arabic}

\newpage

\section{Introduction}

Researchers frequently make parametric assumptions about the distribution of latent variables in structural models. These assumptions are typically made for computational convenience\footnote{Examples include the conventional Gumbel (or type-I extreme value) assumption in discrete choice models following \cite{McFadden1974}, dynamic discrete choice models following \cite{Rust}, and matching models with transferable utility following \cite{Dagsvik} and \cite{ChooSiow}. Models of static or dynamic discrete games often impose parametric assumptions about the distribution of payoff shocks---see, e.g., \cite{Berry1992}, \cite{AM2007}, \cite{BBL}, and \cite{CilibertoTamer}.} or because simulation-based methods are used for estimation. In many models, such as those we consider in this paper, the distribution of latent variables is not nonparametrically identified. This raises the possibility that model parameters and the outcomes of policy experiments, or \emph{counterfactuals}, may be only partially identified when parametric assumptions are relaxed. That is, different distributions may fit the data equally well in-sample, but may yield different values of the counterfactual. It is therefore natural to question whether counterfactuals are sensitive or robust to researchers' parametric assumptions, especially when evaluating the credibility of structural modeling exercises. 

This paper proposes a framework for analyzing the sensitivity of counterfactuals to parametric assumptions about the distribution of latent variables in a class of structural models. In particular, we derive bounds on counterfactuals as the distribution of latent variables spans nonparametric neighborhoods of a given parametric specification while other ``structural'' features of the model are maintained. This approach is in the spirit of global sensitivity analysis advocated by \cite{Leamer1985} (see also \cite{Tamer2015}). Global sensitivity analyses are important in this context: many structural models are nonlinear so policy interventions can have different effects at different points in the parameter space. But a major difficulty with implementing global sensitivity analyses is tractability. A more tractable alternative are local sensitivity analyses, which are based on small perturbations around a chosen specification. Because local approaches rely on linearization, they may fail to correctly characterize the range of counterfactuals predicted by a nonlinear model when the distribution  differs nontrivially from the researcher's chosen parametric specification. 

Our main insight is to borrow from the robustness literature in economics pioneered by \citeauthor{HS2001} (\citeyear{HS2001}, \citeyear{HS2008}) to simplify computation using convex programming.\footnote{Our approach is also related to the field of \emph{distributionally robust optimization} in operations research. See, e.g., \cite{Shapiro2017}, \cite{DuchiNamkoong}, and references therein.} Following this literature, we define neighborhoods around the researcher's parametric specification using statistical divergence (e.g., Kullback--Leibler divergence), with the option to add certain shape restrictions as appropriate. For tractability, we restrict our attention to models that may be written as a finite number of moment (in)equalities, where the expectation is with respect to the distribution of latent variables. While restrictive, this class accommodates many important models of static and dynamic discrete choice, discrete games, and matching. 

To describe our procedure, consider the problem of minimizing or maximizing the counterfactual at a fixed value of structural parameters by varying the distribution of latent variables over a  neighborhood, subject to the model's (in)equality restrictions. We use duality to recast this infinite-dimensional optimization problem as a finite-dimensional convex program. The value of this \emph{inner} program is treated as a criterion function, which is optimized in an \emph{outer} optimization with respect to structural parameters. Importantly, the dimension of the inner problem is independent of the neighborhood size, making our procedure tractable over both small and large neighborhoods.
To further simplify computation, we develop an MPEC version of our procedure for models featuring endogenous parameters (e.g., value functions) defined by equilibrium constraints. We show that this implementation can produce significant computational gains for dynamic discrete choice models in particular.

Our approach is conceptually different from nonparametric partial identification analyses which derive bounds on counterfactuals under minimal distributional assumptions. But as we show, bounds computed using our procedure converge to the (sharp) nonparametric bounds in the limit as the neighborhood size becomes large. 
Aside from sensitivity analyses, our methods may therefore be used to approximate nonparametric bounds by taking the neighborhood size to be large but finite.

For estimation and inference, we propose simple plug-in estimators of the bounds and establish their consistency. We also propose and theoretically justify two methods for inference: a computationally simple but conservative projection procedure and a relatively more efficient  bootstrap procedure. 

We illustrate our procedures with two empirical applications. The first revisits the ``marital college premium'' estimates reported in \cite{CSW}, which relied on an i.i.d. Gumbel (type-I extreme value) assumption for the distribution of individuals' idiosyncratic marital preferences (see also \cite{ChooSiow}). The second empirical application performs a counterfactual welfare analysis in the canonical dynamic discrete choice model of \cite{Rust}. 

\medskip

\paragraph{Related literature.} 
Our approach has connections with global prior sensitivity in Bayesian analysis \citep{ChamberlainLeamer,Leamer1982,Berger1984}, most notably \cite{GKU} and \cite{Ho} who consider sets of priors constrained by Kullback--Leibler divergence relative to a default prior. 

Motivated by questions of sensitivity, 
\cite{CTT} study inference in semiparametric likelihood models using sieve approximations for the infinite-dimensional nuisance parameter (the distribution of latent variables in our setting). For the class of moment-based models we consider, our approach instead eliminates the infinite-dimensional nuisance parameter via a convex program of fixed dimension. 

Several other works have used convex duality to characterize identified sets in models with latent variables. Most closely related are \cite{EGH} and \cite{Schennach}.\footnote{Works using other notions of ``duality'' to construct identified sets include \cite{BMM}, \cite{GalichonHenry2011}, \cite{ChesherRosen2017}, and \cite{Li}.} The problem we study is different, both because of its focus on counterfactuals, rather than structural parameters, and because the optimization is performed over a neighborhood, rather than over all distributions. As a consequence, our estimation and inference methods are also quite different.

\cite{Torgovitsky2019QE} uses linear programming to characterize sharp identified sets in latent variable models defined by quantile restrictions. Within this class, his approach is more computationally convenient than ours for characterizing identified sets. Several important moments or counterfactuals cannot be expressed as quantile restrictions, such as social surplus in discrete choice models and Bellman equations in dynamic discrete choice models. Our approach is compatible with these moments and counterfactuals, thereby allowing the user to characterize identified sets in broader classes of model as well as to perform sensitivity analyses.

There is also a literature deriving nonparametric bounds in specific latent variable  models. Examples include \cite{Manski2007,Manski2014},  \cite{AllenRehbeck}, \cite{TTY}, \cite{Laffers}, \cite{Torgovitsky}, and \cite{GualdaniSinha}. Most closely related is \cite{NoretsTang}, who construct identified sets of counterfactual conditional choice probabilities (CCPs) in dynamic binary choice models. Their approach is specific to counterfactual CCPs and to dynamic binary choice models. Our approach allows for a wider range of counterfactual (e.g., welfare), shape restrictions, and multinomial choice, in addition to performing sensitivity analyses.\footnote{\cite{KSS} and  \cite{KKLS} consider the converse problem, in which flow payoffs are nonparametric (as they can be in our setting) but the distribution of latent payoff shocks is known.}

Finally, our work is complementary to the recent literature on local sensitivity---see, e.g., \cite{KOE}, \cite{AGS2017,AGS2018}, \cite{AK}, \cite{BW}, and \cite{Mukhin}. Much of this literature is concerned with local misspecification of moment conditions, which is different from the setting we consider.

\medskip

\paragraph{Outline.} Section~\ref{s:procedure} introduces our procedure, estimators of the bounds, and shows our approach recovers nonparametric bounds  as the neighborhood size becomes large. Section~\ref{s:implementation} discusses practical aspects and implementation details. Section~\ref{s:delta} gives guidance for interpreting the neighborhood size. Empirical applications are presented in Section~\ref{s:examples}. Section~\ref{s:asymptotics} discusses estimation and inference. The online appendix presents extensions of our methodology, connections with local sensitivity analyses, additional empirical results, and proofs of our main results. A secondary online appendix presents background material on Orlicz classes and supplemental proofs.

\section{Procedure}\label{s:procedure}

We begin in Section~\ref{s:setup} by describing the class of models to which our procedure may be applied. Section~\ref{s:method} describes our approach, Section~\ref{s:duality} shows how duality is used to simplify the bounds, and Section~\ref{s:estimators} introduces our estimators of the bounds. Section~\ref{s:sharp} shows our bounds converge to the sharp nonparametric bounds as the neighborhood size becomes large.

\subsection{Setup}\label{s:setup}

We consider a class of models that link a structural parameter $\theta \in \Theta \subset \mb R^{d_\theta}$, a  vector of targeted moments $P_0 \in \mathcal P \subseteq \mathbb R^{d_P}$, and possibly an auxiliary parameter $\gamma_0 \in \Gamma$ (a metric space) via the moment restrictions
\begin{subequations}
\begin{align}
 \mb E^F[ g_1(U,\theta,\gamma_0)] & \leq P_{10} ,  \label{e:mod:1} \\
 \mb E^F[ g_2(U,\theta,\gamma_0)] & = P_{20} ,  \label{e:mod:2} \\
 \mb E^F[ g_3(U,\theta,\gamma_0)] & \leq 0 , \label{e:mod:3} \\
 \mb E^F[ g_4(U,\theta,\gamma_0)] & = 0 , \label{e:mod:4} 
\end{align}
\label{e:mod}%
\end{subequations}
where $g_1,\ldots,g_4$ are vectors of moment functions, $P_0 = (P_{10},P_{20})$ is partitioned conformably, and $\mb E^F$ denotes expectation with respect to a vector of latent variables $U \sim F$. We assume that the researcher has consistent estimators $(\hat P,\hat \gamma)$ of $(P_0,\gamma_0)$. We also assume that the researcher is interested in a (scalar) counterfactual of the form
\begin{equation} \label{e:kappa}
 \kappa = \mb E^F[k(U,\theta,\gamma_0)]  \,.
\end{equation}
This setup accommodates counterfactuals that do not depend explicitly on $U$, in which case (\ref{e:kappa}) reduces to $\kappa = k(\theta,\gamma_0)$. Note that $\kappa$ will still depend on the distribution of $U$ through $\theta$, whose values are disciplined by the moment conditions (\ref{e:mod}).

Several models and counterfactuals of interest fall into this framework. We review three examples before proceeding.

\begin{example}[Discrete choice and consumer welfare]\label{ex:choice-welfare} \normalfont
Suppose an individual derives utility $h_j(X,\theta) + U_{j}$ from choice $j \in \mc J_0 := \{0, 1, \ldots, J\}$, where $X \in \mc X$ are observed covariates and $U = (U_{j})_{j \in \mc J_0}$ is latent (to the econometrician). We assume, as typical, that $U$ is drawn independently across individuals from a continuous distribution $F$. The probability that an individual with characteristics $x$ chooses $j$ is
\begin{equation} \label{eq:ccp-dc}
 p(j|x) = \mb P_F \left( h_j(x,\theta) + U_{j} = \textstyle \max_{j' \in \mc J_0} \left( h_{j'}(x,\theta) + U_{j'} \right) \right)\,,
\end{equation}
where $\mb P_F$ denotes probabilities when $U \sim F$.  In empirical work, $\theta$ is typically estimated using a criterion that fits the model-implied choice probabilities (\ref{eq:ccp-dc}) to probabilities observed in the data.  Welfare analyses are often based on the social surplus \citep{McFadden1978}
\[
 W(x) = \mb E^F\left[\textstyle \max_{j \in \mc J_0} \left(  h_{j}(x,\theta) + U_{j} \right) \right] ,
\]
which is the average utility consumers with characteristics $x$ derive from the choice problem. A related welfare measure is the change in surplus $\Delta W(x_a,x_b) = W(x_a) - W(x_b)$ associated with a shift from $x_b$ to $x_a$. In practice, it is common to assume the $U_j$ are i.i.d. Gumbel (type-I extreme value), as this yields closed-form expressions for choice probabilities and the welfare measures $W(x)$ and $\Delta W(x_a,x_b)$.

Our approach may be used to perform a sensitivity analysis of $W(x)$ and $\Delta W(x_a,x_b)$ to parametric assumptions about $F$ when $\mathcal X$ is finite. 
A leading example is matching models with finitely many agent types---see Section~\ref{s:csw} and references therein. 
Understanding the sensitivity of $W(x)$ and $\Delta W(x_a,x_b)$ to $F$ is important in this case because $W(x)$ and $\Delta W(x_a,x_b)$ are not nonparametrically identified.\footnote{See, e.g., \cite{BerryHaile2010,BerryHaile2014} and \cite{AllenRehbeck} for nonparametric identification of utilities and welfare measures in discrete choice models when characteristics have continuous support.} 

In our notation, $g_2$ collects indicator functions representing the choice probabilities (\ref{eq:ccp-dc}) across covariates $x \in \mc X$ and choices $j \in \mc J := \{1,\ldots,J\}$ ($j = 0$ is redundant):
\[
 g_2(U,\theta) = \left( \ind \left\{  \textstyle h_j(x,\theta) + U_{j} = \max_{j' \in \mc J_0} \left( h_{j'}(x,\theta) + U_{j'} \right)  \right\} \right)_{(j,x) \in \mc J \times \mc X} 
\]
and $P_{20} = (\Pr(j|x))_{(j,x) \in \mc J \times \mathcal X}$ is the vector of true choice probabilities. There are no $g_1$, $g_3$, $g_4$, or $\gamma$ in this model. Finally, $k(U,\theta) = \max_{j \in \mc J_0} \left(  h_{j}(x,\theta) + U_{j} \right)$ for $W(x)$ and $k(U,\theta) = \max_{j \in \mc J_0} \left(  h_{j}(x_a,\theta) + U_{j} \right) - \max_{j \in \mc J_0} \left(  h_{j}(x_b,\theta) + U_{j} \right)$ for $\Delta W(x_a,x_b)$.
$\hfill \square$
\end{example}

\begin{example}[Discrete games]\label{ex:game} \normalfont
Following \cite{BresnahanReiss,BresnahanReiss1991}, \cite{Berry1992}, and \cite{Tamer2003}, consider the complete-information game  in Table~\ref{table:simplegame:regressor}.

\begin{table}[h]
\begin{center}
  \begin{tabular}{*{4}{c}}
    \multicolumn{2}{c}{} & \multicolumn{2}{c}{Firm $2$\phantom{00000000}}\\
    \multicolumn{1}{c}{} &  & $0$  & $1$ \\[4pt]
    \multirow{2}*{Firm $1$}  & $0\phantom{000}$ & $(0,0)$ & $(0,\beta_2'x  + U_2)$ \\
    & $1\phantom{000}$ & $(\beta_1' x + U_1 ,0)$ & $(\beta_1' x - \Delta_1 + U_1,\beta_2 ' x - \Delta_2 + U_2)$ \\
  \end{tabular}
  \end{center}
\caption{\label{table:simplegame:regressor} Payoff matrix for (Firm 1, Firm 2) when $X = x$.}
\end{table}

Here $U = (U_1,U_2)$ is the latent (to the econometrician) component of firms' profits, which is independent of covariates $X$. Suppose that the solution concept is restricted to equilibria in pure strategies. The econometrician may estimate the probabilities of the potential market structures $(0,0)$, $(0,1)$, $(1,0)$, $(1,1)$ (conditional on $X$) from data on a large number of markets. 
As the model is incomplete---there are values of $U$ for which there are multiple equilibria---moment inequality methods are typically used in empirical work to avoid restricting the equilibrium selection mechanism. However, strong parametric assumptions are often made about the distribution of $U$ (typically bivariate Normal) to derive the model-implied probabilities for different market structures; see, e.g., \cite{Berry1992}, \cite{CilibertoTamer}, \cite{BMM}, and \cite{KlineTamer2016}. It therefore seems natural to also question the sensitivity of counterfactuals to parametric assumptions for $U$.

This model falls into our setup when the regressors $X$ have finite support $\mc X$.\footnote{Continuous regressors are often discretized in empirical applications; see, e.g.,  \cite{CilibertoTamer}, \cite{Grieco}, \cite{KlineTamer2016}, and \cite{CCT}.} In our notation, $g_1$ collects the moment inequalities that bound the probabilities of $(0,1)$ and $(1,0)$ across $x \in \mc X$, with $P_{10}$ denoting the corresponding true probabilities. The inequalities are typically expressed as upper bounds on the probabilities of $(0,1)$ and $(1,0)$; we flip the sign to be compatible with (\ref{e:mod:1}):
\[
\begin{aligned}
 g_1(U,\theta) & = 
 \left[  \begin{array}{c}
  \left( - \ind\{ U_1 \geq -\beta_1'x ; U_2 \leq \Delta_2-\beta_2'x\} \right)_{x \in \mc X} \\
  \left( - \ind\{ U_1 \leq \Delta_1-\beta_1'x ; U_2 \geq -\beta_2'x \} \right)_{x \in \mc X} \end{array}  \right] ,
 &
 P_{10} & = 
 \left[ \begin{array}{c}
  \left( - \Pr((1,0)|X = x) \right)_{x \in \mc X} \\
  \left( - \Pr((0,1)|X = x) \right)_{x \in \mc X} \end{array} \right] ,
\end{aligned}
\]
where $\theta = (\Delta_1,\Delta_2,\beta_1,\beta_2)$. Similarly, $g_2$ and $P_{20}$ collect the moment conditions and probabilities for outcomes $(0,0)$ and $(1,1)$, which are always realized as the result of unique equilibria:
\begin{align*}
 g_2(U,\theta) & = 
 \left[ \begin{array}{c}
  \left(  \ind\{ U_1 \leq -\beta_1'x; \, U_2 \leq -\beta_2'x\} \right)_{x \in \mc X} \\
  \left(  \ind\{ U_1 \geq \Delta_1 - \beta_1'x; \, U_2 \geq \Delta_2 -\beta_2'x \} \right)_{x \in \mc X} \end{array} \right] ,
 &
 P_{20} & = 
 \left[ \begin{array}{c}
  \left( \Pr((0,0)|X = x) \right)_{x \in \mc X} \\
  \left( \Pr((1,1)|X = x) \right)_{x \in \mc X} \end{array} \right] .
\end{align*}
There is no $g_3$, $g_4$, or $\gamma$ in this model. \cite{CilibertoTamer} compute upper bounds on the probability of entrants under a counterfactual payoff shift, say $\tau(\theta)$. The function $k(U,\theta) = \ind\{ U_1 \geq \tau(\theta) - \beta_1'x\}$ corresponds to the upper bound on the probability of firm 1 entering when $X = x$ under this counterfactual. 
$\hfill \square$
\end{example}

\begin{example}[Dynamic discrete choice]\label{ex:DDC} \normalfont
Consider a canonical dynamic discrete choice (DDC) model following \cite{Rust}. The decision maker solves
\begin{equation} \label{e:rust-emax}
 V(s) = \mathbb E^{F}\left[ \max_{d \in \mc D_0} \left( \pi_{d,s}(\theta_\pi) + U_d + \beta E[V(s')| d, s] \right) \right] ,
\end{equation}
where $s \in \mc S$ is a Markov state variable, $\mc D_0 = \{0,1,\ldots,D\}$ is the set of actions, $\pi_{d,s}$ is the flow payoff for action $d$ in state $s$ which is parameterized by $\theta_\pi$, $U_d$ is a latent payoff shock, $\beta \in (0,1)$ is a discount parameter, and $E[\,\cdot\,|d, s]$ denotes expectation with respect to the future state $s'$. The distribution $F$ of $U = (U_d)_{d \in \mc D_0}$ is typically assumed to be continuous and independent of $s$. The CCP of action $d$ in state $s$ is
\begin{equation} \label{e:rust-ccp}
 p(d|s) = \mathbb P_F \left( \pi_{d,s}(\theta_\pi) + U_d + \beta E[V(s')| d, s] = \max_{d' \in \mc D_0} \left( \pi_{d',s}(\theta_\pi) + U_{d'} + \beta E[V(s')| d', s] \right)  \right) ,
\end{equation}
where $\mathbb P_F$ denotes probabilities when $U \sim F$. 

It is standard to assume the $U_d$ are i.i.d. Gumbel, as this yields closed-form expressions for the expectation in (\ref{e:rust-emax}) and multinomial-logit expressions for the CCPs (\ref{e:rust-ccp}). Parameters $\theta_\pi$ or $(\theta_\pi,\beta)$ are typically estimated using a criterion function that fits the model-implied CCPs (\ref{e:rust-ccp}) to probabilities observed in the data. Counterfactuals are then computed by solving (\ref{e:rust-emax}) under alternative laws of motion, flow payoffs, or other interventions. 

When $\mc S$ is finite, model parameters, counterfactual CCPs, and counterfactual welfare measures are typically not identified without parametric restrictions on $F$. Our procedure may be used perform a sensitivity analysis of counterfactuals to parametric assumptions on $F$ as follows. Let $\theta = (\theta_\pi,v,\tilde v)$ or $\theta = (\theta_\pi,\beta,v,\tilde v)$, where $v = (V(s))_{s \in \mathcal S}$ and $\tilde v = (\tilde V(s))_{s \in \mathcal S}$ collect the baseline and counterfactual value functions across $s \in \mc S$. Also let $\gamma = (M_d)_{d \in \mc D_0}$ collect the transition matrices for $s$, $g_2$ collect indicator functions for the CCPs  (\ref{e:rust-ccp}) across states $s \in \mc S$ and choices $d \in \mc D := \{1,\ldots,D\}$ ($d = 0$ is redundant):
\[
 g_2(U,\theta,\gamma) = \left( \ind \left\{ \pi_{d,s}(\theta_\pi) + U_d + \beta M_{d,s} v = \max_{d' \in \mc D_0} \left( \pi_{d',s}(\theta_\pi) + U_{d'} + \beta M_{d',s} v \right) \right\} \right)_{(d,s) \in \mc D \times \mathcal S} 
\]
with $M_{d,s}$ denoting the $s$th row of $M_d$, and $P_{20} = (\Pr(d|s))_{(d,s) \in \mc D \times \mathcal S}$ collect the corresponding true CCPs. Finally, $g_4$ collects moment functions representing (\ref{e:rust-emax}) in the baseline model and under the counterfactual:
\begin{equation}
 g_4(U, \theta, \gamma) =  \left[ \begin{array}{c}  ( \max_{d \in \mc D_0} \{  \pi_{d,s}(\theta_\pi) + U_d + \beta M_{d,s} v  \}  - v_{s} )_{s \in \mathcal S}\\[8pt]
  ( \max_{d \in {\mc D}_0} \{ \tilde \pi_{d,s}(\theta_\pi) + U_d + \tilde \beta \tilde M_{d,s} \tilde v  \}  - \tilde v_{s} )_{s \in \mathcal S} \end{array} \right] , \label{eq:ex.ddc.fp}
\end{equation}
where $v_s = V(s)$, $\tilde v_s = \tilde V(s)$, and $\tilde \pi$, $\tilde \beta$, $\tilde M_d$ denote counterfactual flow payoffs, discount factor, and law of motion.\footnote{If $\mb E^F[\max_{d \in \mc D_0}U_d]$ is finite, then $v \mapsto (\mb E^F[ \max_{d \in \mc D_0} \{  \pi_{d,s}(\theta_\pi) + U_d + \beta M_{d,s} v  \} ])_{s \in \mc S}$ is a $\ell^\infty$-contraction of modulus $\beta$ on $\mb R^{|\mc S|}$. Hence, there is a unique $(v,\tilde v)$ solving $\E^F[ g_4(U,\theta,\gamma)] = 0$ at any fixed $(\theta_\pi, \beta, \tilde \beta, F)$. The solution $(v,\tilde v)$ must collect the solutions to (\ref{e:rust-emax}) in the baseline model and counterfactual across states: $v = (V(s))_{s \in \mathcal S}$ and $\tilde v = (\tilde V(s))_{s \in \mathcal S}$. It follows that $F$ satisfies $\E^F[g_4(U,\theta,\gamma)] = 0$ at $\theta = (\theta_\pi,\beta,v,\tilde v)$ if and only if $(v,\tilde v)$ corresponds to the value functions $V$ and $\tilde V$ under $F$.} We recommend including the location normalizations $\E^F[U_d] = 0$ for $d \in \mc D_0$ in $g_4$ for interpretability. We also recommend including scale normalizations in $g_4$ so that $\E^F[\max_{d \in \mc D_0} U_d]$ is finite. For instance, in Section~\ref{s:rust} we normalize $\E^F[U_d^2]$ for all $d \in \mc D_0$.

Counterfactual CCPs can be computed using
\[
 k(U, \theta, \gamma) = \ind \left\{   \tilde \pi_{d,s}(\theta_\pi) + U_d + \tilde \beta \tilde M_{d,s} \tilde v  = \max_{d' \in {\mc D}_0} \left( \tilde \pi_{d',s}(\theta_\pi) + U_{d'} + \tilde \beta \tilde M_{d',s} \tilde v \right) \right\} .
\]
Change in average welfare corresponds to $k(\theta,\gamma) =  w' (\tilde v - v)$ for a weight vector $w$. 
\hfill $\square$
\end{example}

\medskip

\begin{remark} 
We allow for conditional moments models with $\mb E[g_1(U,X,\theta,\gamma)|X=x] \leq P_{10}(x)$ (and similarly for (\ref{e:mod:2})-(\ref{e:mod:4})) if $U$ is independent of $X$ and $X$ takes values in a finite set $\mc X$. Moment functions are then stacked across $x \in \mc X$ to form $g_1$, $g_2$, $g_3$, and $g_4$ (see Examples~\ref{ex:choice-welfare}-\ref{ex:DDC}). 
Appendix~\ref{ax:extension} discusses extensions to conditional moment models where the distribution of $U$ may vary with the value of (discrete) covariates, and to non-separable models with discrete covariates. Models with continuous covariates fall outside the scope of our procedure.
\end{remark}

\begin{remark} 
Our setup relies on the counterfactual being expressible as (\ref{e:kappa}). If $k$ is vector-valued, our procedure can be applied to compute the support function\footnote{A closed convex set is determined by its support function---see \citeauthor{Rockafellar} (\citeyear{Rockafellar}, Section 13).} of the identified set of counterfactuals: set $k^\tau(U,\theta,\gamma) = \tau'k(U,\theta,\gamma)$ for a conformable unit vector $\tau$ and replace (\ref{e:kappa}) with $\kappa^\tau = \E^F[k^\tau(U,\theta,\gamma_0)]$. Our setup excludes counterfactuals that are infinite-dimensional, such as the distribution of the number of firms in a market. 
\end{remark}

\begin{remark}
The distribution $F$ is not nonparametrically identified in any of the above examples or, more generally, in the class of models (\ref{e:mod}) when the support of $U$ contains many more points than there are moment conditions (e.g., when $U$ is continuously distributed). 
\end{remark}

In common practice, a seemingly reasonable or computationally convenient distribution, say $F_*$, is assumed by the researcher and maintained throughout the analysis (e.g., bivariate Normal in Example~\ref{ex:game} and i.i.d. Gumbel in Examples~\ref{ex:choice-welfare} and \ref{ex:DDC}). 
Given $F_*$ and estimates $\hat P = (\hat P_1,\hat P_2)$ of $P_0$ and $\hat \gamma$ of $\gamma_0$, the researcher computes an estimate $\hat \theta$ of $\theta$ using a criterion function based on the moment conditions
\begin{equation}
\begin{aligned}
 \mb E^{F_*}[ g_1(U,\theta,\hat \gamma)] & \leq \hat P_{1} \,, &  \mb E^{F_*}[ g_2(U,\theta,\hat \gamma)] & = \hat P_{2} \,, \\
 \mb E^{F_*}[ g_3(U,\theta,\hat \gamma)] & \leq 0 \,,  &
 \mb E^{F_*}[ g_4(U,\theta,\hat \gamma)] & = 0 \,. 
\end{aligned} \label{e:sample}
\end{equation}
Finally, the researcher estimates the counterfactual using $\hat \kappa = \mb E^{F_*}[k(U,\hat \theta,\hat \gamma)]$. If $k$ does not depend on $U$, then the estimated counterfactual is simply $\hat \kappa = k(\hat \theta,\hat \gamma)$. In this case $\hat \kappa$ will still depend {implicitly} on $F_*$ through $\hat \theta$.\footnote{While this discussion has assumed point identification of $\theta$ and $\kappa$ for sake of exposition, our methods allow structural parameters and counterfactuals to be partially identified.}

The researcher's chosen specification $F_*$ is used both for estimation of $ \theta$ and again when computing the counterfactual. A natural question is: to what extent does the counterfactual depend on the choice of distribution? The main contribution of this paper is to provide a tractable econometric framework for answering this question.

\subsection{Our Approach}\label{s:method}

As a sensitivity analysis, we shall relax the researcher's parametric assumption and allow $F$ to vary over nonparametric neighborhoods $\mc N_\delta$ of $F_*$, where $\delta$ is a measure of neighborhood ``size''. When we do so, there may be multiple pairs $(\theta,F) \in \Theta \times \mc N_\delta$ that satisfy (\ref{e:mod}) but which yield different values of the counterfactual. Our objects of interest are the smallest and largest values of the counterfactual over all such $(\theta,F)$ pairs:
\begin{align}
 \ul \kappa_\delta  &  = \inf_{\theta \in \Theta,F \in \mc N_\delta} \mb E^F[k(U,\theta,\gamma_0)] \quad \text{ subject to (\ref{e:mod})} , \label{e:kappa_lower}\\
 \ol \kappa_\delta  & = \sup_{\theta \in \Theta,F \in \mc N_\delta} \mb E^F[k(U,\theta,\gamma_0)] \quad \text{ subject to (\ref{e:mod})} .  \label{e:kappa_upper}
\end{align}
By focusing on $\ul \kappa_\delta$ and $\ol \kappa_\delta$, our approach naturally accommodates models with partially-identified structural parameters and counterfactuals. Our approach also sidesteps having to compute the identified set of structural parameters. 

The optimization problems (\ref{e:kappa_lower}) and (\ref{e:kappa_upper}) are made tractable by a convenient choice of $\mc N_\delta$. Following \cite{HS2001} and \cite{MMR}, we consider neighborhoods constrained by $\phi$-divergence \citep{Csiszar1975}:
\begin{equation} \label{e:nbhd}
\begin{aligned}
 \mc N_\delta & = \{ F \in \mc F : D_\phi(F\|F_*) \leq \delta \} \,, \\ 
 D_\phi(F\|F_*) & = \left[ \begin{array}{ll} \displaystyle \int \phi \left( \frac{\mr d F}{\mr d F_*} \right)\, \mr d  F_* & \mbox{if $F \ll F_*$},\\
 +\infty & \mbox{otherwise}, \end{array} \right.
 \end{aligned}
\end{equation}
where $\mc F$ denotes all probability measures on the support\footnote{That is, $\mc U$ is the set of all values that $U$ could conceivably take according to the model, which is possibly larger that the support of the measure $F_*$.} $\mc U$ of $U$ and $F \ll F_*$ denotes absolute continuity of $F$ with respect to $F_*$.  The convex function $\phi : [0,\infty) \to \mb R_+ \cup \{+\infty\}$ penalizes deviations of $F$ from $F_*$. For example, $\phi(x) = x \log x - x + 1$ corresponds to Kullback--Leibler (KL) divergence, $\phi(x) = \frac{1}{2}(x - 1)^2$ corresponds to Pearson $\chi^2$ divergence, and 
\[
 \phi(x) = \frac{x^p - 1 - p(x - 1)}{p(p-1)}\,, \quad (p > 1)\,,
\]
corresponds to $L^p$ divergence. If $F_*$ has positive (Lebesgue) density, then the absolute continuity condition merely rules out $F$ with mass points.

\begin{remark}
Normalizations and other shape restrictions may be added by augmenting the moment functions $g_1,\ldots,g_4$. Examples include: (i) location normalizations, e.g. $\mb E^F[U] =0$ or $\mb E^F[\ind \{U_i \leq 0\} - 0.5] = 0$ for each element $U_i$ of $U$; (ii) scale normalizations, e.g. $\mb E^F[U_i^2] =1$; (iii) covariance normalizations, e.g.  $\mb E^F[UU'] = I$; and (iv)  smoothness restrictions, e.g. $\mb E^F[\ind\{U_i \leq a_{k+1}\} - \ind\{U_i \leq a_{k}\}] \leq C$ for $a_1 < \ldots < a_K$ and a positive constant $C$.
\end{remark}

\begin{remark}
Appendix~\ref{ax:exchangeable} shows that shape restrictions including symmetry, exchangeability, and, more generally, invariance under a finite group of transforms, are also easy to impose.
\end{remark}

\subsection{Dual Formulation}\label{s:duality}

We use convex duality to simplify computation of $\ul \kappa_\delta$ and $\ol \kappa_\delta$. We start by noting $\ul \kappa_\delta$ and $\ol \kappa_\delta$ may be written as the solution to two profiled optimization problems:
\begin{align*}
 \ul \kappa_\delta & = \inf_{\theta \in \Theta} \ul K_\delta(\theta;\gamma_0,P_0)  \,, &
 \ol \kappa_\delta & = \sup_{\theta \in \Theta} \ol K_\delta(\theta;\gamma_0,P_0) \,,
\end{align*}
where the criterion functions $\ul K_\delta(\theta;\gamma_0,P_0) $ and $\ol K_\delta(\theta;\gamma_0,P_0) $ are, respectively, the infimum and supremum of $\mb E^F[k(U,\theta,\gamma_0)]$ with respect to $F \in \mc N_\delta$ subject to the moment conditions  (\ref{e:mod}). In what follows, it is helpful to define the criterion functions at a generic $(\gamma,P)$. To do so, we say that the moment conditions (\ref{e:mod}) hold ``at $(\theta,\gamma,P)$'' if they hold when $\gamma_0$ is replaced by $\gamma$ and $P_0$ is replaced by $P$. Then
\begin{align}
 \ul K_\delta(\theta;\gamma,P) & = \inf_{F \in \mc N_\delta} \mb E^F[k(U,\theta,\gamma)] \quad \mbox{subject to (\ref{e:mod}) holding at $(\theta,\gamma,P)$} \,, \label{e:crit_l}\\
 \ol K_\delta(\theta;\gamma,P) & = \sup_{F \in \mc N_\delta} \mb E^F[k(U,\theta,\gamma)] \quad \mbox{subject to (\ref{e:mod}) holding at $(\theta,\gamma,P)$} \,, \label{e:crit_u} 
\end{align}
with the understanding that $ \ul K_\delta(\theta;\gamma,P) = +\infty$ and $ \ol K_\delta(\theta;\gamma,P) = -\infty$ if there does not exist a distribution in $\mc N_\delta$ for which the moment conditions (\ref{e:mod}) hold at $(\theta,\gamma,P)$. 

We first impose some mild regularity conditions on $F_*$, $\phi$, and the moment functions to justify the dual formulation. Similar conditions are used in generalized empirical likelihood estimation (see, e.g., \cite{KR}). Let $\Phi_0$ denote the set of all $\phi : [0,\infty) \to \mb R \cup \{+\infty\}$ such that $\phi$ is continuously differentiable on $(0,+\infty)$ and strictly convex, with $\phi(1) = \phi'(1) = 0$, $\phi(0) < +\infty$, $\lim_{x \downarrow 0} \phi'(x) < 0$, $\lim_{x \to +\infty} \phi(x)/x = +\infty$, $\lim_{x \to +\infty} \phi'(x) > 0$, and $\lim_{x \to +\infty} x \phi'(x)/\phi(x) < + \infty$. The functions inducing KL, $\chi^2$, and $L^p$ divergence all belong to $\Phi_0$.

Let $\phi^\star(x) = \sup_{t\geq 0 : \phi(t) < +\infty} (tx - \phi(t))$ denote the convex conjugate of $\phi \in \Phi_0$ and let $\psi(x) = \phi^\star(x) - x$.
Define $\mathcal E = \{ f : \mathcal U \to \mathbb R \mbox{ for which } \E^{F_*}[\psi(c|f(U)|)] < \infty \mbox{ for all } c > 0\}$. The class $\mc E$ is an Orlicz class of functions (see Appendix~\ref{ax:Orlicz} for details). For example,
\begin{align*}
 \mc E &  = \{f : \mc U \to \mb R : \mb E^{F_*}[e^{c|f(U)|}] < \infty \mbox{ for all } c > 0 \} & & \mbox{for KL divergence,} \\
 \mc E & = \{f : \mc U \to \mb R : \mb E^{F_*}[f(U)^2] < \infty\} & & \mbox{for $\chi^2$ divergence, and} \\
 \mc E & = \{f : \mc U \to \mb R : \mb E^{F_*}[|f(U)|^q] < \infty\} & & \mbox{for $L^p$ divergence ($p^{-1} +q^{-1} = 1$).}
\end{align*}
Let $g = (g_1,g_2,g_3,g_4)$ denote the vector formed by stacking each of the moment functions from (\ref{e:mod:1})--(\ref{e:mod:4}). Our key regularity condition is the following:

\begin{assumptionp}{\textPhi}\label{a:phi}
(i)\hskip\labelsep  $\phi \in \Phi_0$.
\begin{enumerate}[topsep=-18pt,itemsep=0pt,parsep=0pt,partopsep=0pt]
\item[(ii)] $k(\,\cdot\,,\theta,\gamma)$ and each entry of $g(\,\cdot\,,\theta,\gamma)$ belong to $\mc E$ for each $\theta \in \Theta$ and $\gamma \in \Gamma$.
\end{enumerate}
\end{assumptionp}

For KL divergence, the class $\mc E$ contains of bounded functions (e.g., indicator functions) and functions that are additively separable in $U$ provided $F_*$ has tails that decay faster than exponentially (e.g., Gaussian but not Gumbel). Assumption~\ref{a:phi} therefore fails for KL divergence in Examples~\ref{ex:choice-welfare} and \ref{ex:DDC}, but holds for $\chi^2$ or $L^p$ divergence as these only require finite second or $q$th moments, respectively. 

Let $d = \sum_{i=1}^4 d_i$ where $d_i$ is the dimension of $g_i$, let $\Lambda = \mb R^{d_1}_+ \times \mb R^{d_2} \times \mb R^{d_3}_+ \times \mb R^{d_4}$, and let $\lambda_{12}$ denote the first $d_1 + d_2$ elements of $\lambda$. A derivation of the following criterion functions is presented in Appendix~\ref{sec:dual_K}.

\begin{proposition} \label{prop:criterion} 
Suppose that Assumption~\ref{a:phi} holds. Then the criterion functions (\ref{e:crit_l}) and (\ref{e:crit_u}) may be restated as
\begin{align}
 \ul K_\delta(\theta;\gamma,P) & = \sup_{\eta > 0, \zeta \in \mb R, \lambda \in \Lambda} -\eta \E^{F_*}\left[ {\textstyle \phi^\star \left(\frac{k(U,\theta,\gamma) + \zeta + \lambda' g(U,\theta,\gamma) }{-\eta} \right) } \right] - \eta \delta - \zeta - \lambda_{12}'P \label{e:dual:1} \,,\\
 \ol K_\delta(\theta;\gamma,P) & = \inf_{\eta > 0, \zeta \in \mb R, \lambda \in \Lambda}  \phantom{-}\eta \E^{F_*}\left[ {\textstyle \phi^\star \left( \frac{k(U,\theta,\gamma) - \zeta- \lambda' g(U,\theta,\gamma)}{\eta} \right) } \right] + \eta \delta + \zeta + \lambda_{12}'P \,. \label{e:dual:2} 
\end{align}
Moreover, the value of (\ref{e:dual:1}) is $+\infty$ (equivalently, the value of (\ref{e:dual:2}) is $-\infty$) if and only if there is no distribution in $\mc N_\delta$ under which (\ref{e:mod}) holds at $(\theta,\gamma,P)$.
\end{proposition}

\begin{remark}
Problems (\ref{e:dual:1}) and (\ref{e:dual:2}) are convex in $(\eta, \zeta , \lambda)$. The parameter $\eta$  is the Lagrange multiplier for the constraint $D_\phi(F\|F_*) \leq \delta$. Similarly, $\lambda$ collects the Lagrange multipliers for the moment (in)equalities (\ref{e:mod:1})--(\ref{e:mod:4}). These multipliers are non-negative if they correspond to inequality restrictions and unconstrained otherwise. Finally, $\zeta$ is the Lagrange multiplier for the constraint $\int \mathrm d F = 1$, which ensures that the optimization is over probability measures.
\end{remark}

Problems (\ref{e:dual:1}) and (\ref{e:dual:2}) simplify in some special cases. For KL neighborhoods, $\phi^\star(x) = e^x - 1$ and the multiplier $\zeta$ has a closed-form solution, leading to 
\begin{align*}
 \ul K_\delta (\theta;\gamma,P) &  = \sup_{\eta > 0, \lambda \in \Lambda} -\eta \log \E^{F_*}\left[ e^{- (k(U,\theta,\gamma)+ \lambda' g(U,\theta,\gamma) )/\eta} \right]  - \eta \delta - \lambda_{12}'P \,,  \\ 
 \ol K_\delta (\theta;\gamma,P) & = \inf_{\eta > 0,\lambda \in \Lambda} \eta \log \E^{F_*}\left[ e^{(k(U,\theta,\gamma)- \lambda' g(U,\theta,\gamma) )/\eta} \right]  + \eta \delta + \lambda_{12}'P \,. 
\end{align*}
Another special case is when $k(u,\theta,\gamma)$ does not depend on $u$. To analyze this case, consider
\begin{equation} \label{e:p0_dist}
 \Delta(\theta; \gamma, P)  := \inf_{F} D_\phi(F\|F_*) \quad \mbox{subject to (\ref{e:mod}) holding at $(\theta,\gamma,P)$} .
\end{equation}
The value $\Delta(\theta; \gamma, P)$ is the minimum $\phi$-divergence between $F_*$ and a distribution $F$ for which the moment conditions hold at $(\theta,\gamma,P)$. We show in Proposition~\ref{prop:dual_dist} that $\Delta(\theta; \gamma, P)$ has an equivalent dual formulation:
\begin{equation} \label{e:qual:dual}
 \Delta(\theta;\gamma,P) = \sup_{\zeta \in \mb R,\lambda \in \Lambda} - \E^{F_*}\Big[ \phi^\star(-\zeta - \lambda' g(U,\theta,\gamma))  \Big] - \zeta - \lambda_{12}'P\,. 
\end{equation}
For KL divergence, $\zeta$ may be solved for in closed-form and problem (\ref{e:qual:dual}) simplifies to
\[
 \Delta(\theta;\gamma,P) = \sup_{\lambda \in \Lambda} \; - \log \E^{F_*}\Big[ e^{ - \lambda' g(U,\theta,\gamma)}  \Big] - \lambda_{12}'P\,.
\]
When $k$ does not depend on $u$, by a change of variables\footnote{Substitute $\eta \zeta - k(\theta,\gamma)$ in place of $\zeta$ in (\ref{e:dual:1}) and  $\eta \zeta + k(\theta,\gamma)$ in place of $\zeta$ in (\ref{e:dual:2}), then substitute $\eta \lambda$ in place of $\lambda$ in both (\ref{e:dual:1}) and  (\ref{e:dual:2}).} we may then restate problems (\ref{e:dual:1}) and (\ref{e:dual:2}) as
\begin{equation} \label{e:k_implicit}
 \ul K_\delta (\theta; \gamma, P) = \left[ \begin{array}{l}
 k(\theta,\gamma) \\[2pt]
 + \infty 
 \end{array} \right. ,
 \quad 
 \ol K_\delta (\theta; \gamma, P) = \left[ \begin{array}{ll}
 k(\theta, \gamma) & \mbox{if $\Delta(\theta; \gamma, P) \leq \delta$,} \\[2pt]
 - \infty & \mbox{if $\Delta(\theta; \gamma, P) > \delta$.}
 \end{array} \right.
\end{equation}

An important feature of our approach is that the optimization problems (\ref{e:dual:1}), (\ref{e:dual:2}), and (\ref{e:qual:dual}) are convex and their dimension does not increase with $\delta$. This feature is not shared by other seemingly natural approaches to flexibly model $F$, such as mixtures or other finite-dimensional sieves. As we show in Section~\ref{s:sharp}, our procedure may be used to approximate sharp nonparametric bounds on counterfactuals by taking $\delta$ to be large but finite.

\subsection{Estimation}\label{s:estimators}

We now propose simple estimators of the bounds $\ul \kappa_\delta$ and $\ol \kappa_\delta$ based on ``plugging in'' consistent estimators $(\hat P,\hat \gamma)$ of $(P_0,\gamma_0)$. Estimators $\hat{\ul \kappa}_\delta$ and $\hat{\ol \kappa}_\delta$ are computed by optimizing criterion functions with respect to $\theta$:
\begin{equation*}
\begin{aligned}
 \hat{\ul \kappa}_\delta & = \inf_{\theta \in \Theta} \hat{\ul K}_\delta(\theta) \,, \quad & \quad 
 \hat{\ol \kappa}_\delta & = \sup_{\theta \in \Theta} \hat{\ol K}_\delta(\theta)\,,
\end{aligned}
\end{equation*}
where
\[
 \hat{\ul K}_\delta(\theta) = \left[ \begin{array}{l}
 \ul K_\delta (\theta; \hat \gamma, \hat P) \\[2pt]
 + \infty 
 \end{array} \right. ,
 \quad 
 \hat{\ol K}_\delta(\theta) = \left[ \begin{array}{ll}
 \ol K_\delta (\theta;\hat \gamma,\hat P) & \mbox{if $\Delta(\theta;\hat \gamma,\hat P) < \delta$,} \\[2pt]
 - \infty & \mbox{if $\Delta(\theta;\hat \gamma,\hat P) \geq \delta$,}
 \end{array} \right.
\]
and $\ul K_\delta (\theta; \hat \gamma, \hat P)$, $\ol K_\delta (\theta; \hat \gamma, \hat P)$, and $\Delta(\theta;\hat \gamma,\hat P)$ are the criterion functions (\ref{e:dual:1}), (\ref{e:dual:2}), and (\ref{e:qual:dual}) evaluated at $(\hat \gamma, \hat P)$. If $k(u,\theta,\gamma) = k(\theta,\gamma)$, then we simply have
\[
 \hat{\ul K}_\delta(\theta) = \left[ \begin{array}{l}
 k(\theta,\hat \gamma)\\
 + \infty 
 \end{array} \right., \quad
 \hat{\ol K}_\delta(\theta) = \left[ \begin{array}{ll}
 k(\theta,\hat \gamma) & \mbox{if $\Delta(\theta;\hat \gamma,\hat P) < \delta$,} \\
  - \infty & \mbox{if $\Delta(\theta;\hat \gamma,\hat P) \geq \delta$.}
 \end{array} \right.
\]
In Section~\ref{s:consistency} we establish consistency of $\hat{\ul \kappa}_\delta$ and $\hat{\ol \kappa}_\delta$ and derive their asymptotic distribution.

\subsection{Nonparametric Bounds on Counterfactuals}\label{s:sharp}

We define the (nonparametric) identified set of counterfactuals as
\[
 \mc K = \left\{ \mb E^{F}[k(U,\theta,\gamma_0)] : \mbox{(\ref{e:mod}) holds for some $\theta \in \Theta$ and $F \in \mc F_\theta$} \right\} ,
\]
where $\mc F_\theta = \{ F \in \mc F : \mb E^F[ g(U,\theta,\gamma_0)] \mbox{ is finite and } F \ll \mu \}$ denotes all distributions on $\mc U$ that are absolutely continuous with respect to a $\sigma$-finite dominating measure $\mu$ and for which the moments in (\ref{e:mod}) are finite at $\theta$. We impose existence of a density with respect to $\mu$ as it is often a structural assumption used, e.g., to avoid ties in CCPs or to establish existence of equilibria. The main result of this section shows that $\ul \kappa_\delta$ and $\ol \kappa_\delta$ approach the sharp nonparametric bounds $\inf\mc K$ and $\sup \mc K$ as $\delta$ becomes large.

We first introduce some additional regularity conditions. Say $k$ is ``$\mu$-essentially bounded'' if $|k(\cdot,\theta,\gamma_0)|$ has finite $\mu$-essential supremum\footnote{The $\mu$-essential supremum of a function $f$ is denoted $\mu\text{-}\mr{ess}\sup f$ and is the smallest value $c$ for which $\mu(\{u : f(u) > c\}) = 0$. The $\mu$-essential infimum, denoted $\mu\text{-}\mr{ess}\inf$, is defined analogously.}
for each $\theta \in \Theta$.  This holds trivially if $k$ is bounded (e.g., counterfactual CCPs in Examples~\ref{ex:game} and \ref{ex:DDC} and change in average welfare in Example~\ref{ex:DDC}). Models with unbounded $k$ may be reparameterized (as a proof device) by setting $\tilde \theta = (\theta,\kappa)$, appending $k(U,\theta,\gamma_0) - \kappa$ as an element of $g_4$, and setting $k(U,\tilde \theta,\gamma_0) = \kappa$. 

We also require a constraint qualification condition. This is a sufficient condition for establishing equivalence of ``nonparametric'' primal and dual problems in Appendix~\ref{ax:sharp}, which is an intermediate step in the proof of the following result. Let $0_{d_i}$ denote a $d_i \times 1$ vector of zeros, $\mc C = \mb R^{d_1}_+ \times \{0_{d_2}\} \times \mb R^{d_3}_+ \times \{0_{d_4}\}$, $\mc G(\theta,\gamma) = \{ \E^{F} [ g(U,\theta,\gamma) ] :  F \in \mc N_\infty \}$ where $\mc N_\infty = \{F : D_\phi(F \|F_*) < \infty\}$, and $\vec P = (P,0_{d_3+d_4})$. For $A,B \subseteq \mb R^d$, we let $\mr{ri}(A)$ denote the relative interior of $A$ and $A + B = \{a + b : a \in A, b \in B\}$.

\begin{definition}\label{cond:s}
\emph{Condition S} holds at $(\theta,\gamma,P)$ if $\vec P \in \mr{ri}(\mc G(\theta,\gamma) + \mc C) $.
\end{definition}

Using relative interior instead of interior allows for moment functions that are collinear at some $\theta$ (i.e., some moments are redundant). To give some intuition, consider moment equality models. Condition S requires that (\ref{e:mod}) holds at $(\theta, \gamma, P)$ under some $F \in \mc N_\delta$ that is ``interior'' to $\mc N_\infty$, in the sense that one can perturb the (non-redundant) moments in any direction by perturbing $F$. For moment inequality models, Condition S also requires that there is $F \in \mc N_\infty$ under which all moment inequalities hold strictly at $(\theta, \gamma, P)$.

Let $\Theta_I = \{ \theta \in \Theta:$ (\ref{e:mod}) holds for some $F \in \mc F_\theta\}$ denote the (nonparametric) identified set for $\theta$. 
Define the ``nonparametric'' objective function
\begin{equation}
 \ul K_{np}(\theta;\gamma,P)  = \inf_{F \in \mc F_\theta} \mb E^F[k(U,\theta,\gamma)] \quad \mbox{subject to (\ref{e:mod}) holding at $(\theta,\gamma,P)$} \,, \label{e:crit_np_l}
\end{equation}
with the understanding that $\ul K_{np}(\theta;\gamma,P) = +\infty$ if the infimum runs over an empty set. Let $\ol K_{np}(\theta;\gamma,P)$ denote the analogous supremum. Evidently, 
\[
 \inf \mc K = \inf_{\theta \in \Theta} \ul K_{np}(\theta;\gamma_0,P_0)  \quad \mbox{and} \quad \sup \mc K = \sup_{\theta \in \Theta} \ol K_{np}(\theta;\gamma_0,P_0)\,.
\]

\begin{definition}
$\Theta_I$ is \emph{$S$-regular} if for all $\epsilon > 0$ there exist $\ul \theta, \ol \theta \in \Theta_I$ such that Condition S holds at $(\ul \theta, \gamma_0, P_0)$ and $(\ol \theta, \gamma_0, P_0)$, $\ul K_{np}(\ul \theta; \gamma_0, P_0) < \inf \mc K + \epsilon$, and $\ol K_{np}(\ol \theta; \gamma_0, P_0) > \sup \mc K - \epsilon$.
\end{definition}
 Intuitively, S-regularity requires that the values the counterfactual takes at ``boundary'' points of $\Theta_I$ (i.e., at which Condition S fails) are not materially more extreme than values it can take at points ``inside'' $\Theta_I$ (i.e., at which Condition S holds). This condition can be verified under more primitive continuity conditions on $k$ and $g$. A sufficient (but not necessary) condition for S-regularity is that Condition S holds at $(\theta,\gamma_0,P_0)$ for all $\theta \in \Theta_I$.

\begin{theorem}\label{t:sharp}
Suppose that Assumption~\ref{a:phi} holds, $k$ is $\mu$-essentially bounded, $\Theta_I$ is $S$-regular, and $\mu$ and $F_*$ are mutually absolutely continuous. Then
\begin{align*}
 \lim_{\delta \to \infty} \ul \kappa_\delta & = \inf \mc K \,, &
 \lim_{\delta \to \infty} \ol \kappa_\delta & = \sup \mc K \,.
\end{align*}
\end{theorem}

Theorem~\ref{t:sharp} shows that our procedure can be used to approximate the sharp nonparametric bounds $\inf \mc K$ and $\sup \mc K$ by setting  $\delta$ to be large but finite.
If $\mu$ is Lebesgue measure---which it often is in applications---then the mutual absolute continuity condition in Theorem~\ref{t:sharp} is satisfied whenever $F_*$ has  strictly positive density over $\mc U$. 

\begin{remark} 
Appendix~\ref{ax:sharp} presents the dual forms of $\ul K_{np}$ and $\ol K_{np}$. Unlike $\ul K_\delta$ and $\ol K_\delta$, the duals of $\ul K_{np}$ and $\ol K_{np}$ are  min-max and max-min problems which involve an inner optimization over $u$. These problems may be computationally challenging, especially when $u$ is multivariate. Comparing Proposition~\ref{prop:criterion} with the duals in Appendix~\ref{ax:sharp}, we see that setting $\delta < \infty$ replaces a ``hard-max'' (an optimization over $u$) with a ``soft-max'' (a convex expectation).  In this respect, adding the constraint $F \in \mc N_\delta$ may be viewed as a regularization of the nonparametric objective functions, similar to the use of entropic penalization to regularize objective functions in optimal transport problems---see, e.g., \cite{Cuturi2013}. Smaller values of $\delta$ impose a stronger regularization.
\end{remark}

Theorem~\ref{t:sharp} is silent on the issue of how large $\delta$ needs to be so that $\ul \kappa_\delta$ and $\ol \kappa_\delta$ are close to the nonparametric bounds. While this is model- and counterfactual-specific, the following toy example suggests that relatively small values of $\delta$ may suffice in some problems where the counterfactual is a choice probability.

\begin{example}\label{ex:simple} \normalfont
Consider the problem 
\[
 \ol \kappa_\delta = \sup_{\theta \in \mb R, F \in \mc N_\delta} \E^F[ \ind\{U \leq \theta \}] \quad \mbox{subject to} \quad \E^F[U - \theta] = 0,
\]
where $\mc N_\delta$ is defined by KL divergence and $F_*$ is the $N(0,1)$ distribution. When $F = F_*$, the only solution to $\E^F[U - \theta] = 0$ is $\theta = 0$. Therefore, the value of the counterfactual under $F_*$ is $\E^{F_*}[ \ind\{U \leq 0 \}] = \frac{1}{2}$ whereas $\sup \mc K = 1$. In Appendix~\ref{ax:simple}, we derive the large-$\delta$ approximation $\ol \kappa_\delta = 1 - 2 \pi e^{-2\delta - 1}(1 + o(1))$. By symmetry, $\ul \kappa_\delta = 2 \pi e^{-2\delta - 1}(1 + o(1))$ and $\inf \mc K = 0$. Therefore, in this example, $\ul \kappa_\delta$ and $\ol \kappa_\delta$ converge rapidly to $\inf \mc K$ and $\sup \mc K$ as $\delta$ increases. \hfill $\square$
\end{example}

More generally, suppose  the dual problems (\ref{e:dual:1}) and (\ref{e:dual:2}) have unique solutions $\ul \eta$ and $\ol \eta$ for $\eta$, where the optimization is performed over $\eta \geq 0$.\footnote{Optimizing over $\eta \geq 0$ rather than $\eta > 0$ does not affect the optimal value---see Proposition~\ref{prop:dual}.}  Under appropriate regularity conditions (see, e.g., \cite{MilgromSegal}), it follows that 
\[
 \frac{\partial \ul K_\delta(\theta;\gamma,P)}{\partial \delta} = -\ul \eta, \quad \quad \frac{\partial \ol K_\delta(\theta;\gamma,P)}{\partial \delta} = \ol \eta.
\]
One can therefore infer from $\ul \eta$ and $\ol \eta$ the extent to which, if at all, the bounds at any fixed $\theta$ would widen further if $\delta$ was increased.

\section{Practical Considerations}\label{s:implementation}

We now discuss practical details for implementing our procedure. Section~\ref{s:computation} discusses computational methods, Section~\ref{s:mpec} presents our MPEC approach, and Section~\ref{s:overid}~discusses methods for dealing with over-identified models.

\subsection{Computation}\label{s:computation}

There are three aspects to computation: (i) computing the expectations with respect to $F_*$ in the objective functions, (ii) solving the inner optimization problems over Lagrange multipliers, and (iii) solving the outer optimization problems over $\theta$.

The expectations in the objective functions (\ref{e:dual:1}), (\ref{e:dual:2}), and (\ref{e:qual:dual}) are available in closed form for certain settings,\footnote{An earlier draft derived closed-form expressions for a discrete game of complete information with Gaussian payoff shocks and KL neighborhoods---see \url{https://arxiv.org/abs/1904.00989v2}.} in which case the dimension of $u$ does not play a role in the computational complexity of our procedure.
Otherwise, the expectations will need to be computed numerically. If so, the dimension of $u$ will play a role in terms of determining how many quadrature points or Monte Carlo draws are needed to control numerical approximation error. In the empirical applications we used a randomized quasi-Monte Carlo approach based on scrambled Halton sequences as in \cite{Owen2017}.

The inner optimization with respect to Lagrange multipliers can be solved rapidly: it is convex and gradients and Hessians are available in closed-form. The envelope theorem can be used to derive gradients for the outer optimization when $k$ and  $g$  are  differentiable in $\theta$.\footnote{In practice, we smoothed any non-smooth moments  and used automatic differentiation to compute derivatives with respect to $\theta$ if these were not easily available analytically.} Our procedures were all implemented in Julia with the inner and outer optimizations solved using Knitro. A general-purpose implementation of our methods in Julia is provided in the supplemental material.

As with parameter estimation in nonlinear structural models, the outer optimization with respect to $\theta$ is typically non-convex. In applications, we used an iterative multi-start procedure in an attempt to converge to global optima. Computation times are reported in the applications below.

\subsection{MPEC Approach}\label{s:mpec}

We now describe and formally justify an MPEC version of our procedure in the spirit of \cite{SuJudd}. This approach simplifies computation in models with endogenous parameters defined by equilibrium conditions (e.g., value functions defined by Bellman equations), resulting in significant computational gains for DDC models in particular.

Suppose $\theta = (\theta_s, \theta_e)$ and $g_4 = (g_{4s}, g_{4e})$ where $\theta_s$ are ``deep'' structural parameters and $\theta_e$ are ``endogenous'' parameters that are defined implicitly by $g_{4e}$. That is, for any $(\theta_s,\gamma,F)$, the parameter $\theta_e = \theta_e(\theta_s,\gamma,F)$ solves
\[
 \mb E^F[ g_{4e}( U,( \theta_s, \theta_e), \gamma ) ] = 0 \,.
\]
For instance, in Example~\ref{ex:DDC} we have $\theta_s = \theta_\pi$ or $(\theta_\pi,\beta)$, while $\theta_e = (v,\tilde v)$ collects the value functions in the baseline model and counterfactual, and $g_{4e}$ collects the functions representing the corresponding Bellman equations, as in display (\ref{eq:ex.ddc.fp}). Although our procedure can be implemented as described in Section~\ref{s:procedure}, that implementation does not make use of the fact that $\theta_e$ is defined implicitly by $g_{4e}$. 

To leverage this structure, consider the subset of moments conditions excluding $g_{4e}$:
\begin{equation}
\label{e:mod-e}
\begin{aligned}
 \mb E^F[ g_1(U,\theta,\gamma_0)] & \leq P_{10} , & 
 \mb E^F[ g_2(U,\theta,\gamma_0)] & = P_{20} ,   \\
 \mb E^F[ g_3(U,\theta,\gamma_0)] & \leq 0 , & 
 \mb E^F[ g_{4s}(U,\theta,\gamma_0)] & = 0 , &  
\end{aligned}
\end{equation}
and define criterion functions using these only:
\begin{align}
  \ul K^s_{\delta}(\theta;\gamma,P) & = \inf_{F \in \mc N_\delta} \mb E^F[k(U,\theta,\gamma)] \quad \mbox{ subject to (\ref{e:mod-e}) holding at $(\theta,\gamma,P)$} \,,\label{e:crit_l:mpec} \\
  \ol K^s_{\delta}(\theta;\gamma,P) & = \sup_{F \in \mc N_\delta} \mb E^F[k(U,\theta,\gamma)] \quad \mbox{ subject to (\ref{e:mod-e}) holding at $(\theta,\gamma,P)$} \,. \label{e:crit_u:mpec} 
\end{align}
Under the conditions of Proposition~\ref{prop:criterion}, these criterion functions may be restated as
\begin{align}
 \ul K_{\delta}^s(\theta;\gamma,P) & = \sup_{\eta > 0, \zeta \in \mb R, \lambda \in \Lambda_s} -\eta \E^{F_*}\left[ {\textstyle \phi^\star \left( \frac{k(U,\theta,\gamma) + \zeta + \lambda' g_s(U,\theta,\gamma) }{-\eta}\right) } \right] - \eta \delta - \zeta - \lambda_{12}'P \label{e:dual:1:mpec}\,,  \\
 \ol K_{\delta}^s(\theta;\gamma,P) & = \inf_{\eta > 0, \zeta \in \mb R, \lambda \in \Lambda_s} \phantom{-} \eta \E^{F_*}\left[ {\textstyle \phi^\star \left( \frac{k(U,\theta,\gamma) - \zeta- \lambda' g_s(U,\theta,\gamma) }{\eta}\right) } \right] + \eta \delta + \zeta + \lambda_{12}'P \,, \label{e:dual:2:mpec}
\end{align}
with $g_s = (g_1,g_2,g_3,g_{4s})$ and $\Lambda_s = \mb R^{d_1}_+ \times \mb R^{d_2} \times \mb R^{d_3}_+ \times \mb R^{d_{4s}}$ with $d_{4s} = \dim(g_{4s})$. Problems (\ref{e:dual:1:mpec}) and (\ref{e:dual:2:mpec}) simplify analogously to (\ref{e:k_implicit}) when $k$ does not depend on $u$, with the minimum divergence problem $\Delta$ defined using $g_s$ in place of $g$.

In our MPEC approach, the criterion functions (\ref{e:dual:1:mpec}) and (\ref{e:dual:2:mpec}) are optimized with respect to $\theta$, with the remaining moment conditions involving $g_{4e}$ appended as constraints. Importantly, these constraints are evaluated under the ``least favorable'' distributions $\ul F_{\delta,\theta}$ and $\ol F_{\delta,\theta}$ that solve problems (\ref{e:crit_l:mpec}) and (\ref{e:crit_u:mpec}), respectively. The following proposition formally justifies this approach.

\begin{proposition} \label{p:mpec}
Suppose that Assumption~\ref{a:phi} holds. Then the problems
\[
 \inf_{\theta \in \Theta} \ul K_\delta(\theta;\gamma,P)
\]
and
\[
 \inf_{\theta \in \Theta} \ul K_{\delta}^s(\theta;\gamma,P) \;\mbox{ subject to } \; \E^{\ul F_{\delta,\theta}}[g_{4e}(U,\theta,\gamma)] = 0
\]
have the same value. An analogous result holds for the upper bound.
\end{proposition}

To implement our MPEC approach, note that the expectations in the constraints may be expressed in terms of changes of measure. Let $\ul m_{\delta,\theta} = {\mr d \ul F_{\delta,\theta}}/{\mr d F_*}$ and $\ol m_{\delta,\theta} = {\mr d \ol F_{\delta,\theta}}/{\mr d F_*}$ so that
\[
\begin{aligned}
 \E^{\ul F_{\delta,\theta}}[\, \cdot \,] & = \E^{F_*}[ \ul m_{\delta,\theta}(U) \, \, \cdot \,] \,, &
 \E^{\ol F_{\delta,\theta}}[\, \cdot \,] & = \E^{F_*}[ \ol m_{\delta,\theta}(U) \, \, \cdot \,] \,.
\end{aligned}
\]
If $k$ depends on $u$, then we construct $\ul m_{\delta,\theta}$ and $\ol m_{\delta,\theta}$ from solutions to (\ref{e:dual:1:mpec}) and (\ref{e:dual:2:mpec}), say $(\ul \eta, \ul \zeta, \ul \lambda)$ and $(\ol \eta, \ol \zeta, \ol \lambda)$ (these solutions exist under the regularity conditions below). If $\ul \eta > 0$, then the distribution solving (\ref{e:crit_l:mpec}) is unique and is induced by the change of measure
\begin{equation} \label{e:m_lower}
 \ul m_{\delta,\theta}(u) = \dot \phi^\star\left( \frac{k(u,\theta,\gamma) + \ul \zeta + \ul \lambda' g_s(u,\theta,\gamma)}{ -\ul \eta} \right)  ,
\end{equation}
where $\dot \phi^\star(x) = \frac{d \phi^\star(x)}{dx}$. The function $\ol m_{\delta,\theta}(u)$ is constructed similarly, replacing $(\ul \eta, \ul \zeta, \ul \lambda)$ in (\ref{e:m_lower}) by $(-\ol \eta, - \ol \zeta, - \ol \lambda)$. For KL divergence the change of measure simplifies to
\[
  \ul m_{\delta,\theta}(u) = \frac{e^{(k(u,\theta,\gamma) + \ul \lambda' g_s(u,\theta,\gamma))/ -\ul \eta}}{\mb E^{F_*}\left[e^{(k(u,\theta,\gamma) + \ul \lambda' g_s(u,\theta,\gamma))/ -\ul \eta}\right]} \,,
\]
and similarly for $\ul m_{\delta,\theta}(u)$.

If $\ul \eta = 0$, then there may be multiple minimizing distributions. As shown in the proof of Proposition~\ref{prop:dual_F}, each such distribution must be supported on  
\[
 \ul A_{\delta,\theta} := \{ u : k(u,\theta,\gamma) + \ul \lambda' g_s(u,\theta,\gamma) = F_*\text{-}\mr{ess} \inf ( k(\cdot,\theta,\gamma) + \ul \lambda' g_s(\cdot,\theta,\gamma))\}\,.
\]
Note $F_*(\ul A_{\delta,\theta}) > 0$ is required for $\ul \eta = 0$ to be a solution. Otherwise, any distribution supported on $\ul A_{\delta,\theta}$ is not absolutely continuous with respect to $F_*$ and is therefore not in $\mc N_\delta$. If $\ul \eta = 0$ and $F_*(\ul A_{\delta,\theta}) > 0$, then we construct $\ul m_{\delta,\theta}$  by restricting $F_*$ to $\ul A_{\delta , \theta}$ and rescaling:
\[
 \ul m_{\delta,\theta}(u) = \ind\{u \in \ul A_{\delta,\theta}\}/F_*(\ul A_{\delta, \theta}).
\]
The function $\ol m_{\delta, \theta}(u)$ is constructed analogously, replacing $\ul \lambda$ with $-\ol \lambda$ and the set $\ul A_{\delta,\theta}$ with $\ol A_{\delta, \theta} = \{ u : k(u,\theta,\gamma) - \ol \lambda' g_s(u,\theta,\gamma) = F_*\text{-}\mr{ess} \sup (k(\cdot,\theta,\gamma) - \ol \lambda' g_s(\cdot,\theta,\gamma))\}$.

If $k$ does not depend on $u$, then $\ul m_{\delta,\theta}$ and $\ol m_{\delta,\theta}$ are constructed from solutions to a version of problem (\ref{e:qual:dual}) with $g_s$ in place of $g$. Under the regularity conditions below, this program has a solution, say $(\ul \zeta, \ul \lambda)$. In this case, we define
\begin{equation} \label{e:m_delta}
 \ul m_{\delta,\theta}(u) = \ol m_{\delta,\theta}(u) = \dot \phi^\star\left( - \ul \zeta - \ul \lambda' g_s(u,\theta,\gamma) \right) \,.
\end{equation}
For KL divergence the change of measure simplifies to
\[
 \ul m_{\delta,\theta}(u) = \ol m_{\delta,\theta}(u) = \frac{e^{- \ul \lambda' g_s(u,\theta,\gamma)}}{\mb E^{F_*}\left[e^{- \ul \lambda' g_s(u,\theta,\gamma)}\right]} \,.
\]

\begin{proposition}\label{prop:dual_F}
Suppose that Assumption~\ref{a:phi} holds, Condition S holds at $(\theta,\gamma,P)$, and there exists a distribution $F$ with $D(F\|F_*) < \delta$ under which (\ref{e:mod-e}) holds at $(\theta,\gamma,P)$. Then the distributions $\ul F_{\delta,\theta}$ and $\ol F_{\delta,\theta}$ induced by $\ul m_{\delta,\theta}$ and $\ol m_{\delta,\theta}$ solve (\ref{e:crit_l:mpec}) and (\ref{e:crit_u:mpec}), respectively. 
\end{proposition}

\paragraph{Example.} We consider a numerical example for the DDC model of \cite{Rust}  based on the parameterization in Section 5.4 of \cite{NoretsTang}. The counterfactual they consider is a hypothetical change in the law of motion of the state. We follow these papers and use state-space of dimension 90. As $|\mc S| = 90$ and $\mc D_0 = \{0,1\}$, there are 90 functions in $g_2$ representing the observed CCPs. There are another 180 functions  in $g_{4e}$ representing the Bellman equations in the baseline model and counterfactual across states. We also impose the normalization $\E^F[U_d] = 0$ for $d = 0,1$. Hence, $g_{4s}(U,\theta,\gamma) = (U_0,U_1)$. Our MPEC approach has 92 moments in the inner optimization (90 for CCPs and two mean-zero normalizations on the shocks) with the remaining 180 moments representing the Bellman equations appended as constraints. The full approach uses all 272 moments in the inner optimization. 

Table~\ref{tab:mpec} reports computation times for the inner optimization problems (\ref{e:dual:2}) and (\ref{e:dual:2:mpec}) (denoted $\ol K_\delta$) for maximizing the counterfactual CCP in the highest mileage state.\footnote{The times in Table~\ref{tab:mpec} are based on initializing the solver at $\eta = 1$, $\zeta = 0$, and $\lambda = 0$. When embedded in the outer optimization over $\theta$, computation times for the inner problem are reduced significantly by using a warm start that initializes at the solution to the inner problem at the previous value of $\theta$.} We also report times for solving the minimum divergence problem (\ref{e:qual:dual}) (denoted $\Delta$) using the full set of moment functions $g$ and its MPEC analogue using $g_s$. Neighborhoods are constrained by a hybrid of KL and $\chi^2$ divergence as in the empirical applications---see Section~\ref{s:examples}. As can be seen, the inner optimization problems are solved at least 20 times faster for the MPEC implementation, with the relative efficiency increasing in $\delta$.

\begin{table}
\begin{center}
\caption{\label{tab:mpec}Computation times (in seconds) for the inner problems}
\begin{tabular}{lcccc} \hline \hline
 \multicolumn{1}{c}{Implementation} & \multicolumn{4}{c}{Objective} \\\cline{2-5} \\[-10pt]
 &  $\ol K_{0.01}$ & $\ol K_{0.10}$ & $\ol K_{1.00}$ & $\Delta$ \\[2pt]
 MPEC (92 moments) & 0.207 & \phantom{1}0.232 & \phantom{1}0.256 & 0.108 \\[2pt]
 Full (272 moments) & 4.317 & 12.978 & 43.699 & 3.365 \\[2pt] \hline \\[-10pt]
\end{tabular}
\parbox{\textwidth}{\small \emph{Note:} Expectations are computed using 50,000 scrambled Halton draws. Computations are performed in Julia v1.6.4 and Knitro v12.4.0 on a 2.7GHz MacBook Pro with 16GB memory.} 
\end{center}
\end{table}

\subsection{Over-identification}\label{s:overid}

In over-identified models (i.e., where the number of moment conditions $d$ exceeds the dimension $d_\theta$ of $\theta$), there might not exist $\theta \in \Theta$ for which the sample moment conditions (\ref{e:sample}) hold under $F_*$. We propose two methods for handling over-identified models.

First, one may compute the smallest value of $\delta$ for which there exists $F \in \mc N_\delta$ consistent with the sample moment conditions (\ref{e:sample}) by solving the optimization problem
\[
 \hat \delta = \inf_{\theta \in \Theta} \Delta(\theta;\hat \gamma,\hat P).
\]
The interval $[\hat{\ul \kappa}_\delta,\hat{\ol \kappa}_\delta]$ will be nonempty for $\delta > \hat \delta$. If the model is correctly specified under $F_*$,\footnote{Neither our theoretical results developed in Section~\ref{s:procedure} nor the estimation and inference results in Section~\ref{s:asymptotics} require correct specification of the model under $F_*$.} then $\hat \delta$ will converge in probability to zero under the conditions of Theorem~\ref{t:c-const}. In this case, the interval $[\hat{\ul \kappa}_\delta,\hat{\ol \kappa}_\delta]$ will be nonempty with probability approaching one for each fixed $\delta > 0$.

It is also possible that $\hat \delta = +\infty$ in correctly specified but over-identified models when $\hat P$ is incompatible with certain model restrictions. For instance, CCPs are often estimated nonparametrically using empirical choice frequencies. If some choices aren't observed in the data, then the estimated CCPs will be zero even though model-implied CCPs are strictly positive. 

This issue can be circumvented in models defined by equality restrictions only (hence $P_0 \equiv P_{20}$) using the following two-step approach. First, compute a preliminary estimator~$\tilde \theta$ of $\theta$ based on (\ref{e:sample}). Then, set $\hat P = \mb E^{F_*}[g_2(U,\tilde \theta, \hat \gamma)]$. This second-step estimator $\hat P$ is compatible with the model by construction, thereby ensuring that the interval $[\hat{\ul \kappa}_\delta, \hat{\ol \kappa}_\delta]$ is nonempty for each $\delta > 0$. The estimator $\hat P$ will be consistent and asymptotically normal under mild regularity conditions provided the model is correctly specified under $F_*$, so the consistency and inference results developed in Section~\ref{s:asymptotics} will also apply.

\section{Interpreting the Neighborhood Size}\label{s:delta}

This section presents some theoretical results and practical methods to help interpret the neighborhood size $\delta$. Sections~\ref{s:invariance} and \ref{s:different_phi} discuss properties of $\phi$-divergences and their implications for interpreting $\delta$. Section~\ref{s:lfd} shows how to construct the ``least favorable'' distributions that minimize or maximize the counterfactual. Section~\ref{s:lens} gives a practical, model-based metric for interpreting $\delta$.

\subsection{Invariance}\label{s:invariance}

A defining property of $\phi$-divergences are their invariance to invertible transformations. That is, if $T$ is an invertible transformation and $G$ and $G_*$ denote the distributions of $T(U)$ when $U \sim F$ and $U \sim F_*$, respectively, then $D_{\phi}(F\|F_*) = D_{\phi}(G\|G_*)$.\footnote{See, e.g., \cite{LieseVajda}. A more direct statement is in \cite{QiaoMinematsu}.} An important consequence of invariance is that $\delta$ has the same interpretation under a change in units. For instance, if one researcher writes a model in terms of dollars with $U \sim F_*$ and another researcher uses thousands of dollars with $U \sim G_*$ for $G_*(u) = F_*(10^{-3} u)$, then $F$ is in $\mc N_\delta$ if and only if its rescaled counterpart $G$ is in a $\delta$-neighborhood of $G_*$. A second consequence is that neighborhood size is invariant under invertible location and scale transformations of $F_*$ (e.g., $N(\mu,\Sigma)$ versus $N(0,I)$).

\subsection{Least Favorable Distributions}\label{s:lfd}

A useful feature of our approach is that the ``least favorable'' distributions (LFDs) that attain the smallest or largest values of the counterfactual may easily be recovered. To help interpret $\delta$, one may plot the LFDs and compute other quantities of interest (e.g., correlations or welfare measures) under them. 

Section~\ref{s:mpec} describes how to construct LFDs when our MPEC approach is used. LFDs for our full (i.e., non-MPEC) approach are a special case with $g_4 = g_{4s}$. To briefly summarize, consider the LFD $\ul F_{\delta,\theta}$ solving the minimization problem (\ref{e:crit_l}). First suppose that $k$ depends on $u$. Let $(\ul \eta, \ul \zeta, \ul \lambda)$ solve problem (\ref{e:dual:1}). If $\ul \eta > 0$, then $\ul F_{\delta,\theta}$ is unique and its change-of-measure $\ul m_{\delta,\theta} = \mr d\ul F_{\delta,\theta}/\mr d F_*$ is given by 
\begin{equation} \label{e:m:general}
 \ul m_{\delta,\theta}(u) = \dot \phi^\star\left( \frac{k(u,\theta,\gamma) + \ul \zeta + \ul \lambda' g(u,\theta,\gamma)}{ -\ul \eta} \right) .
\end{equation}
The LFD $\ol F_{\delta, \theta}$ solving the maximization problem (\ref{e:crit_u}) is constructed similarly, replacing $(\ul \eta, \ul \zeta, \ul \lambda)$ in (\ref{e:m:general}) with $(-\ol \eta, - \ol \zeta, - \ol \lambda)$, where $(\ol \eta, \ol \zeta, \ol \lambda)$ solves (\ref{e:dual:2}). If $\ul \eta = 0$ or $\ol \eta = 0$, then there may exist multiple distributions solving (\ref{e:crit_l}) and  (\ref{e:crit_u}) at $\theta$. LFDs in this case are constructed analogously to the method described in Section~\ref{s:mpec}.
Note that $\ul \eta = 0$ or $\ol \eta = 0$ is unlikely if $k$ and/or elements of $g$ are unbounded in $u$---see the discussion in Section~\ref{s:mpec}. 
If $k$ does not depend on $u$, then we set
\begin{equation} \label{e:m:delta}
 \ul m_{\delta,\theta}(u) = \ol m_{\delta,\theta}(u) = \dot \phi^\star \left( - \ul \zeta - \ul \lambda' g(u,\theta,\gamma) \right)
\end{equation}
where $(\ul \zeta, \ul \lambda)$ solves (\ref{e:qual:dual}). While there may exist multiple distributions solving  (\ref{e:crit_l}) and  (\ref{e:crit_u}) in this case, the distribution induced by (\ref{e:m:delta}) has smallest $\phi$-divergence relative to $F_*$ among all such distributions.

\subsection{Viewing Neighborhood Size through the Lens of the Model}\label{s:lens}

Another method for interpreting $\delta$ is based on measuring the variation in the moments at the distributions solving (\ref{e:kappa_lower}) and (\ref{e:kappa_upper}) relative to their values under $F_*$.

Consider the sets of minimizing and maximizing values of $\theta$ at which $\ul \kappa_\delta$ and $\ol \kappa_\delta$ are attained, say $\ul \Theta_\delta$ and $\ol \Theta_\delta$. These are nonempty under the regularity conditions in Section~\ref{s:asymptotics}. While the moment conditions (\ref{e:mod}) hold at any $\theta \in \ul \Theta_\delta \cup \ol \Theta_\delta$ under the corresponding LFD, they will typically not hold at $\theta$ under $F_*$. We therefore define
\begin{multline*}
 size(\delta)  = \sup_{\theta \in \ul \Theta_\delta \cup \ol \Theta_\delta} \max   \Big\{ \big\| \left(\mb E^{F_*}[g_1(U,\theta,\gamma_0)] - P_{10}\right)_+ \big\|_\infty , \big\| \mb E^{F_*}[g_1(U,\theta,\gamma_0)] - P_{20} \big\|_\infty , \\
 \big\| \left(\mb E^{F_*}[g_3(U,\theta,\gamma_0)] \right)_+ \big\|_\infty , \big\| \mb E^{F_*}[g_4(U,\theta,\gamma_0)] \big\|_\infty \Big\} \,,
\end{multline*}
where $(v)_+ = ( \max\{v_i,0\})_{i=1}^d$ for a vector $v \in \mb R^d$. The quantity $size(\delta)$ is the maximum degree to which the moments at $\theta  \in \ul \Theta_\delta \cup \ol \Theta_\delta$ violate  (\ref{e:mod}) under $F_*$.

 This measure is informative about the extent to which the distortions to $F_*$ required to attain the smallest and largest values of the counterfactual over $\mc N_\delta$ are reflected in (\ref{e:mod}). Small values of $size(\delta)$ indicate that the LFDs supporting $\ul \kappa_\delta$ and $\ol \kappa_\delta$ distort $F_*$ in a way that moves the counterfactual but barely moves the moments. Conversely, large values of $size(\delta)$ indicate that distortions required to increase or decrease the counterfactual also have a material impact on the moments. In practice, this measure can be computed by replacing $(P_0,\gamma_0)$ by estimators $(\hat P,\hat \gamma)$ and $\ul \Theta_\delta$ and $\ol \Theta_\delta$ by the minimizers and maximizers of the sample criterions or by the estimators of $\ul \Theta_\delta$ and $\ol \Theta_\delta$ introduced in Section~\ref{s:bootstrap}.

\subsection{Relating Different Divergences}\label{s:different_phi}

It is well known that $\phi$-divergences are equivalent over local neighborhoods (see, e.g., Theorem 4.1 of \cite{CsiszarShields}). However, $\ul \kappa_\delta$ and $\ol \kappa_\delta$ may depend on the choice of $\phi$ when $\delta$ is not arbitrarily small. Bounds induced by different $\phi$ functions may be related as follows. Let $\mc N_{\delta,1}$ and $\mc N_{\delta,2}$ denote $\delta$-neighborhoods induced by $\phi_1$ and $\phi_2$, respectively. The quantity
\[
 \bar a = \sup_{x \geq 0, x \neq 1} \frac{\phi_1(x)}{\phi_2(x)} 
\]
is a measure of relative neighborhood size: if $\bar a < \infty$ then $\mc N_{\delta,2} \subseteq \mc N_{\bar a \delta,1}$ for each $\delta > 0$, as shown formally in the proof of Proposition~\ref{prop:order} below. 
For instance, when comparing KL divergence ($\phi_1(x) = x \log x - x + 1$) and $\chi^2$ divergence ($\phi_2(x) = \frac{1}{2}(x-1)^2$)  we obtain $\bar a = 2$. Therefore, $\delta$-neighborhoods under $\chi^2$ divergence  are contained in $2\delta$-neighborhoods under KL divergence. Interchanging $\phi_1$ and $\phi_2$ produces $\bar a = +\infty$, which reflects the fact that KL divergence is weaker than $\chi^2$ divergence. 

Let $\ul \kappa_{\delta,1}$ and $\ul \kappa_{\delta,2}$ denote the smallest counterfactual from display (\ref{e:kappa_lower}) over $\mc N_{\delta,1}$ and $\mc N_{\delta,2}$, respectively. Define $\ol \kappa_{\delta,1}$ and $\ol \kappa_{\delta,2}$ analogously.

\begin{proposition}\label{prop:order}
Suppose that  Assumption~\ref{a:phi} holds for both $\phi_1$ and $\phi_2$ and $\bar a$ is finite. Then $[\ul \kappa_{\delta,2}, \ol \kappa_{\delta,2}] \subseteq [\ul \kappa_{\bar a \delta, 1}, \ol \kappa_{\bar a \delta, 1}] $ for each $\delta > 0$.
\end{proposition}

It follows from Proposition~\ref{prop:order} that bounds that are wide under $\phi_2$ must necessarily be wide under $\phi_1$. Similarly, narrow bounds under $\phi_1$ must also be narrow under $\phi_2$. Note also that the inclusion in Proposition~\ref{prop:order} holds for any counterfactual.

\section{Empirical Applications}\label{s:examples}

\subsection{Marital College Premium}\label{s:csw}

\cite{CSW}, henceforth CSW, study the evolution of marital returns to education using a frictionless matching model with transferable utility \citep{ChooSiow}. Within this framework, the ``marital college premium'' is the additional expected utility that an individual would derive from the marriage market if they had a (counterfactually) higher level of education. CSW find that marital college premiums for women in the United States increased significantly across cohorts from the mid to late 20th century, particularly for the more highly educated. 

As is conventional following \cite{Dagsvik} and \cite{ChooSiow}, CSW assume latent variables representing individuals' idiosyncratic marital preferences are i.i.d. Gumbel. The marital college premium is only partially identified when the distribution of these latent variables is not specified. We therefore perform a sensitivity analysis of CSW's estimates to departures from this conventional parametric assumption. 

Our analysis makes several findings. 
First, it seems impossible to draw conclusions about whether marital college premiums have increased or decreased over time under small nonparametric relaxations of the i.i.d. Gumbel assumption. Interestingly, premiums have narrow nonparametric bounds at {fixed} parameter values, but a slight relaxation of the i.i.d. Gumbel assumption allows for significant variation in parameters which, in turn, produces uninformatively wide bounds. As parameters are just-identified under any fixed distribution of shocks \citep{GalichonSalanie}, further restrictions on parameters or shape restrictions on the distribution are required to tighten the bounds. We show that imposing exchangeability can tighten the bounds significantly.

\medskip

\paragraph{Model and Benchmark Estimates.}

Agents are male or female and one of $J$ types (education levels). 
A type-$a$ male receives utility $\varepsilon_{a0}$ if he chooses to be unmatched and $z_{ab} + \varepsilon_{ab}$ if he matches with a type-$b$ female. Similarly, a type-$b$ female receives utility $e_{0b}$ if she chooses to be unmatched and $t_{ab} + e_{ab}$ if she matches with a type-$a$ male. The parameters $(z_{ab},t_{ab})_{a,b=1}^J$ represent the common deterministic component of marital preferences. The latent shocks $(\varepsilon_{a0},\ldots,\varepsilon_{aJ})$ and $(e_{0b},\ldots,e_{Jb})$ represent individuals' idiosyncratic  marital preferences. Shocks are i.i.d. across individuals and have mean zero. The type $b$ to $b'$ marital education premium for females is the difference in expected marital utility between types $b$ and $b'$:
\begin{equation} \label{eq:kappa_matching}
 \kappa = \mathbb E^F \left[ \max_{a = 0,\ldots,J} \Big( t_{ab'} + e_{ab'} \Big) \right] - \mathbb E^F \left[ \max_{a = 0,\ldots,J} \Big( t_{ab} + e_{ab} \Big)  \right] ,
\end{equation}
where $F$ denotes the distribution of $(e_{0b},\ldots,e_{Jb'})$ and $t_{0b} = t_{0b'} = 0$.

CSW use data from the American Community Survey. They form 28 cohorts indexed by female birth year from 1941 (cohort 1) to 1968 (cohort 28), each of which is treated as an independent marriage market. We focus on CSW's estimates for whites. There are $J = 5$ types: ``high-school dropouts'', ``high-school graduates'', ``some college'', ``college graduate'', and ``college-plus''. 
We center our analysis on the ``some college'' to ``college graduate'' premium, though we obtained qualitatively similar results (not reported) for the ``college graduate'' to ``college-plus'' premium. Figure~\ref{fig:csw_3} presents estimates and 95\% confidence sets (CSs) for the premium  under the i.i.d. Gumbel assumption (cf. Figure 21  in CSW) based on CSW's replication files.

\medskip

\paragraph{Implementation.}

The model reduces to a standard individual-level discrete choice problem for each type (see CSW's Propositions 1 and 2). We assume that the distribution of females' preference shocks does not depend on their type, so we drop the $b$ subscript and consider a single random vector $U = (e_0,\ldots,e_J)$. We allow the distribution $F$ of $U$ to vary across cohorts and implement our procedures cohort-by-cohort.\footnote{In view of the just-identification results of \cite{GalichonSalanie}, we would obtain the same bounds if $F$ was homogeneous across cohorts. Allowing for heterogeneity in own-type would result in wider bounds.}

Under any fixed $F$, a cohort's parameters $(t_{ab})_{a=1}^J$ are just-identified from the marriage probabilities for that cohort's type-$b$ women \citep{GalichonSalanie}. We therefore impose only the moment conditions involving the parameters $\theta = (t_{ab},t_{ab'})_{a = 1}^J$ appearing in (\ref{eq:kappa_matching}), as the remaining parameters can be chosen to fit the remaining marriage probabilities under the resulting least-favorable distribution.
We form $g_2$ to explain the type $b$ and $b'$ marriage probabilities for women in a given cohort:
\[
 g_2(U,\theta) =  \left[ \begin{array}{c}
  (\ind \{ t_{ab} + e_{a}  = \max_{a' = 0,\ldots,J} ( t_{a'b} + e_{a'} )  \})_{a=1}^J  \\
  (\ind \{ t_{ab'} + e_{a}  = \max_{a' = 0,\ldots,J} ( t_{a'b'} + e_{a'} )  \})_{a=1}^J \end{array} \right] 
\]
and form $\hat P_2$ using CSW's estimates of the corresponding type-$b$ and $b'$ marriage probabilities. We set  $g_4(U,\theta) = (e_j,e_j^2 - \pi^2/6)_{j=0}^{J}$ so that shocks have mean zero and the same variance as the Gumbel distribution. The scale normalization also ensures that the nonparametric bounds on the premium are finite at any fixed $\theta$. As $J = 5$, there are 22 moments (10 for marriage probabilities and 12 location/scale normalizations), and $\theta$ has dimension $10$. 

We consider a second implementation which imposes invariance of $F$ under rotations and reflections of potential spouse types, so that the model-implied marriage probabilities depend on $\theta$ but not the labeling of potential spouse types (though they may depend on their ordering).\footnote{Allowing dependence on the ordering of types seems desirable here as types correspond to education levels, which are naturally ordered.} Formally, this shape restriction corresponds to dihedral exchangeability (see Appendix~\ref{ax:exchangeable}); we refer to it simply as ``exchangeability''.
Under this shape restriction, $F$ must satisfy the 22 moment conditions under all 12 rotations and reflections of the elements of $U$. This implementation therefore imposes a total of $264$ moment conditions. Rather than including all 264 moments separately, it suffices to form $g_2$ and $g_4$ by taking the averages of the 22 moments across the 12 permutations (see Appendix~\ref{ax:exchangeable}). Both implementations therefore have inner optimization problems of the same dimension.

Computations are performed as described in Section~\ref{s:computation}. The first implementation uses 50,000 scrambled Halton draws to compute the expectations. The second uses 10,000 draws which are concatenated over the 12 permutations (see Remark~\ref{rmk:concatenate}), for a total of 120,000 draws. Computation times are reported in Appendix~\ref{ax:csw}. CSs for $\ul \kappa_\delta$ and $\ol \kappa_\delta$ are computed using the bootstrap procedure in Section~\ref{s:bootstrap}. Appendix~\ref{ax:csw} discusses bootstrap details and presents projection CSs using the method from Section~\ref{s:projection}.

We define neighborhoods using a hybrid of KL and $\chi^2$ divergence:
\[
 \phi(x) = \left[ \begin{array}{ll}
 x \log x - x + 1 & \mbox{if $x \leq e$,} \\
 {\textstyle \frac{1}{2e}(x-e)^2 + (x-e) + 1} & \mbox{if $x > e$.}
 \end{array} \right.
\]
We use this divergence because Assumption~\ref{a:phi}(ii) fails for KL divergence, whereas hybrid divergence only requires finite second moments for Assumption~\ref{a:phi}(ii). The LFDs under hybrid divergence are also everywhere positive, which is not guaranteed under $\chi^2$ divergence. We repeated our analysis with neighborhoods constrained by $\chi^2$ and $L^4$ divergences as robustness checks. Overall, our findings are not sensitive to $\phi$ (see Appendix~\ref{ax:csw} for a discussion).

\medskip

\paragraph{Findings.}

\begin{figure}[t]
\begin{center}
\begin{subfigure}[t]{\textwidth}
\begin{center}
\caption*{Smaller neighborhoods}
\vskip -2pt
\begin{subfigure}[t]{0.49\textwidth}
\begin{center}
\caption{Without exchangeability} \label{fig:csw_3_short}
\vskip -6pt
\includegraphics[width=\linewidth]{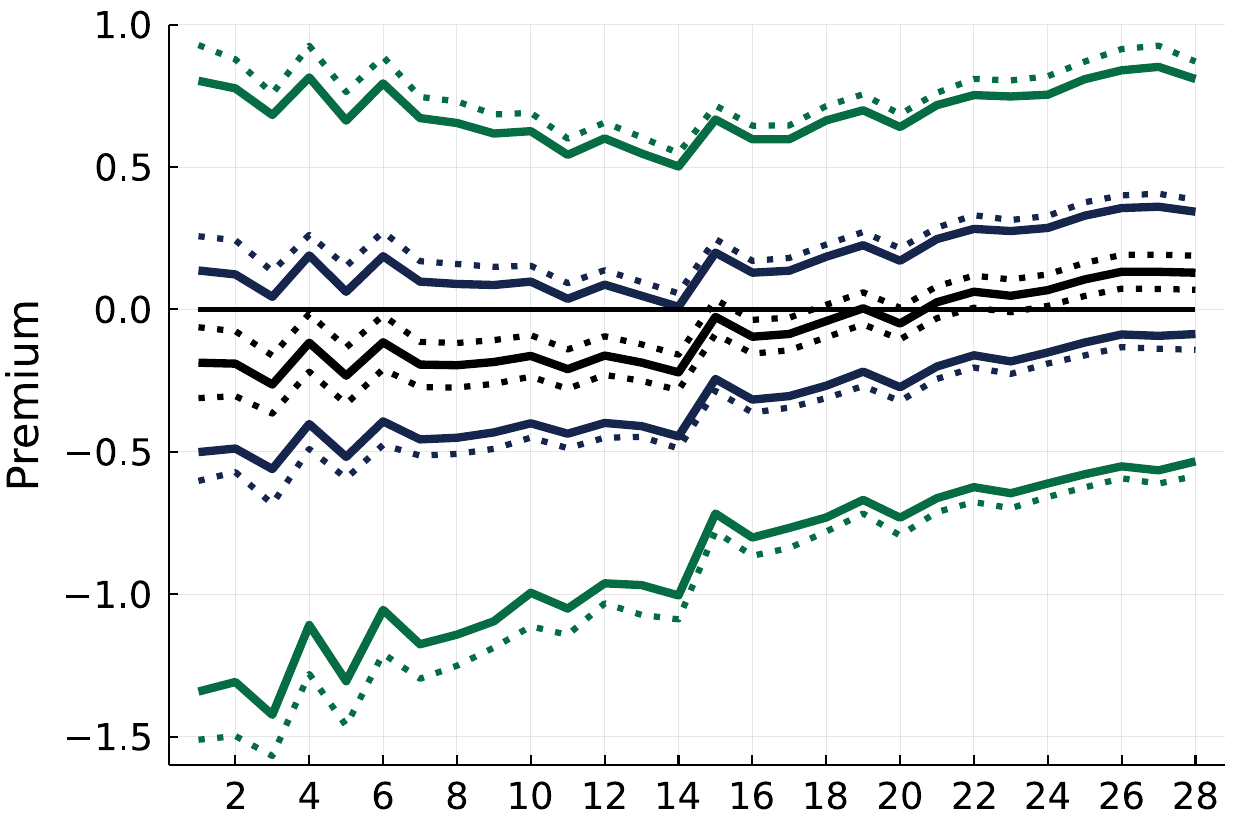}
\end{center}
\end{subfigure}
\begin{subfigure}[t]{0.49\textwidth}
\begin{center}
\caption{With exchangeability\phantom{out}} \label{fig:csw_3_short_exch}
\vskip -6pt
\includegraphics[width=\linewidth]{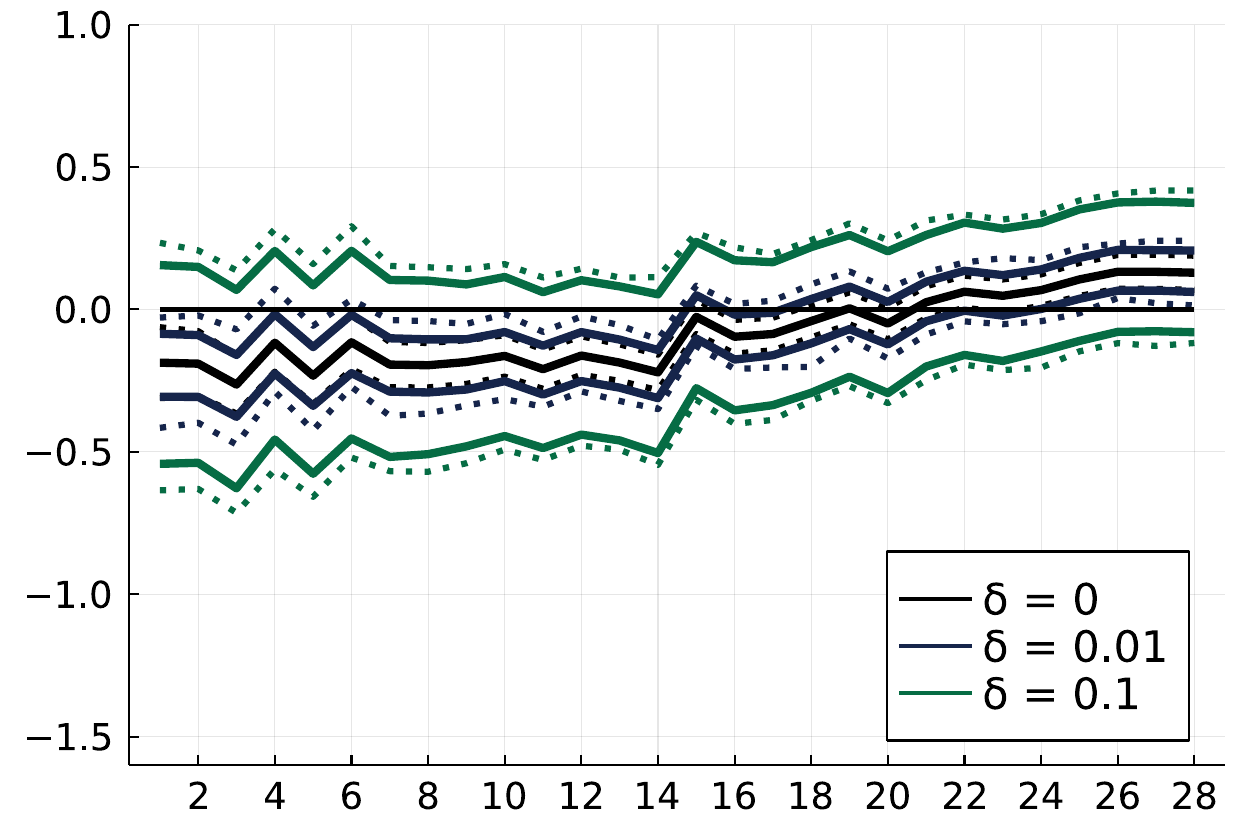}
\end{center}
\end{subfigure}
\end{center}
\end{subfigure}

\vskip 10pt

\begin{subfigure}[t]{\textwidth}
\begin{center}
\caption*{Larger neighborhoods}
\vskip -2pt
\begin{subfigure}[t]{0.49\textwidth}
\begin{center}
\caption{Without exchangeability} \label{fig:csw_3_full}
\vskip -6pt
\includegraphics[width=\linewidth]{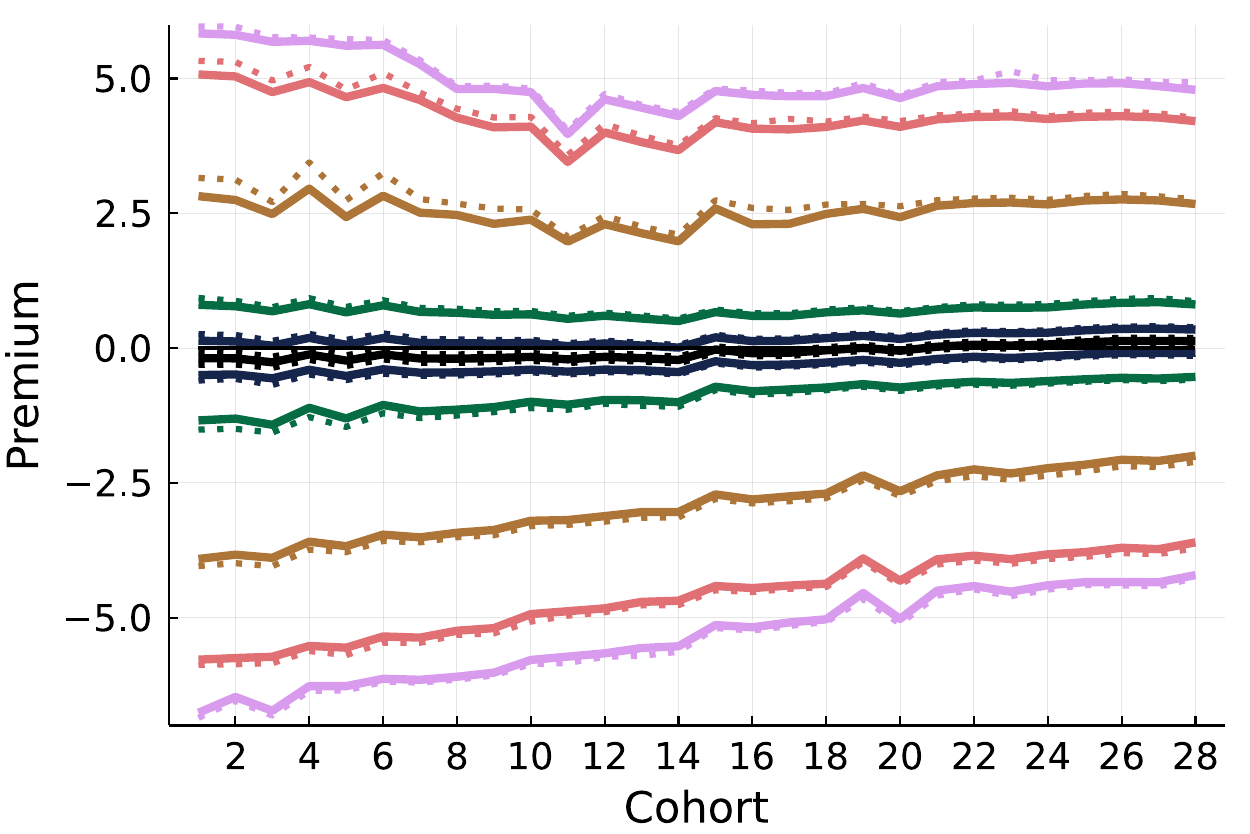}
\end{center}
\end{subfigure}
\begin{subfigure}[t]{0.49\textwidth}
\begin{center}
\caption{With exchangeability\phantom{out}} \label{fig:csw_3_full_exch}
\vskip -6pt
\includegraphics[width=\linewidth]{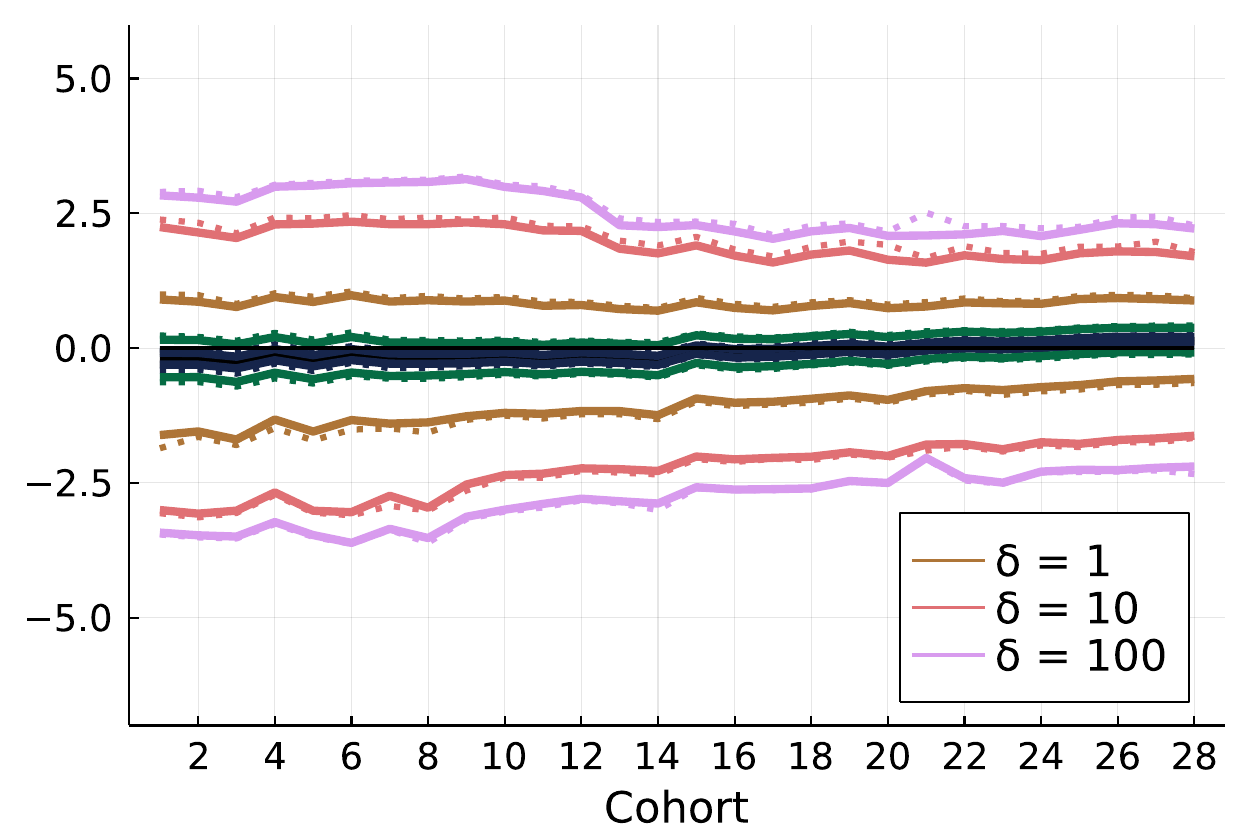}
\end{center}
\end{subfigure}
\end{center}
\end{subfigure}

\caption{\label{fig:csw_3}Sensitivity analysis of the ``some college'' to ``college graduate'' premium across cohorts. \emph{Note:} Solid lines are estimates, dotted lines are (cohort-wise) 95\% CSs. CSW's estimates and CSs correspond to $\delta = 0$.}
\vskip -15pt
\end{center}
\end{figure}

Figure~\ref{fig:csw_3} presents a sensitivity analysis of the ``some college'' to ``college graduate'' premium. Cohort-wise estimates and CSs for $\ul \kappa_{\delta}$ and $\ol \kappa_{\delta}$ are presented, beginning at $\delta = 0.01$ and increasing $\delta$ by factors of 10 up to $\delta = 100$. Even with $\delta = 0.01$, estimates of $\ul \kappa_\delta$ and $\ol \kappa_\delta$ lie uniformly below and above zero across cohorts without exchangeability (see Figure~\ref{fig:csw_3_short}).  
Imposing exchangeability can tighten the bounds, with the bounds for $\delta = 0.01$ significantly negative in early cohorts and significantly positive in later cohorts (see Figure~\ref{fig:csw_3_short_exch}). But the $\delta = 0.1$ bounds with exchangeability again contain zero across all cohorts. Bounds for larger $\delta$ presented in Figures~\ref{fig:csw_3_full} and \ref{fig:csw_3_full_exch} are uninformatively wide.

\begin{figure}
\begin{center}
\begin{subfigure}[t]{\textwidth}
\begin{center}
\caption{Without exchangeability} \label{fig:csw_3_lfd}
\vskip -2pt
\begin{subfigure}[t]{0.43\textwidth}
\begin{center}
\caption*{Type 0 (Unmatched)} \label{fig:qq_1_0}
\vskip -6pt
\includegraphics[width=\linewidth]{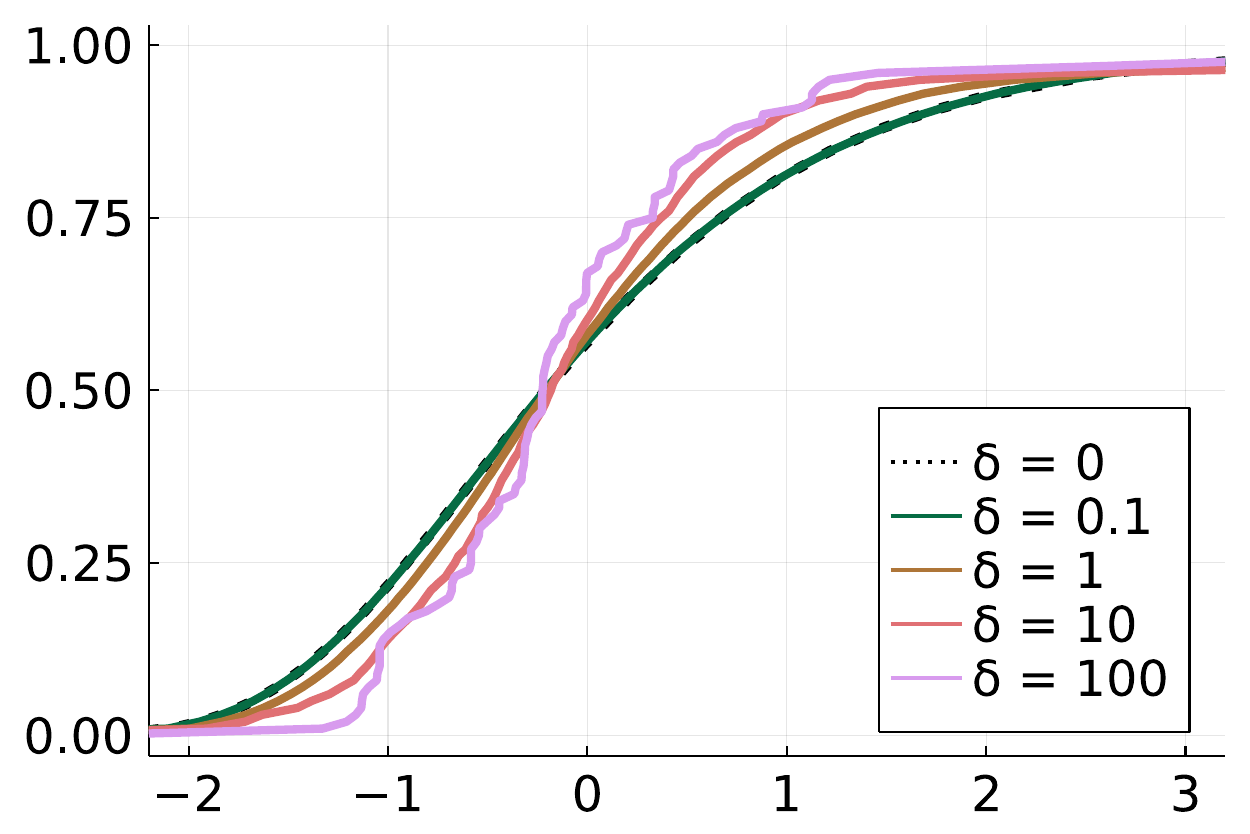} 
\end{center}
\end{subfigure}
\begin{subfigure}[t]{0.43\textwidth}
\begin{center}
\caption*{Type 1 (High-school dropout)} \label{fig:qq_1_1}
\vskip -6pt
\includegraphics[width=\linewidth]{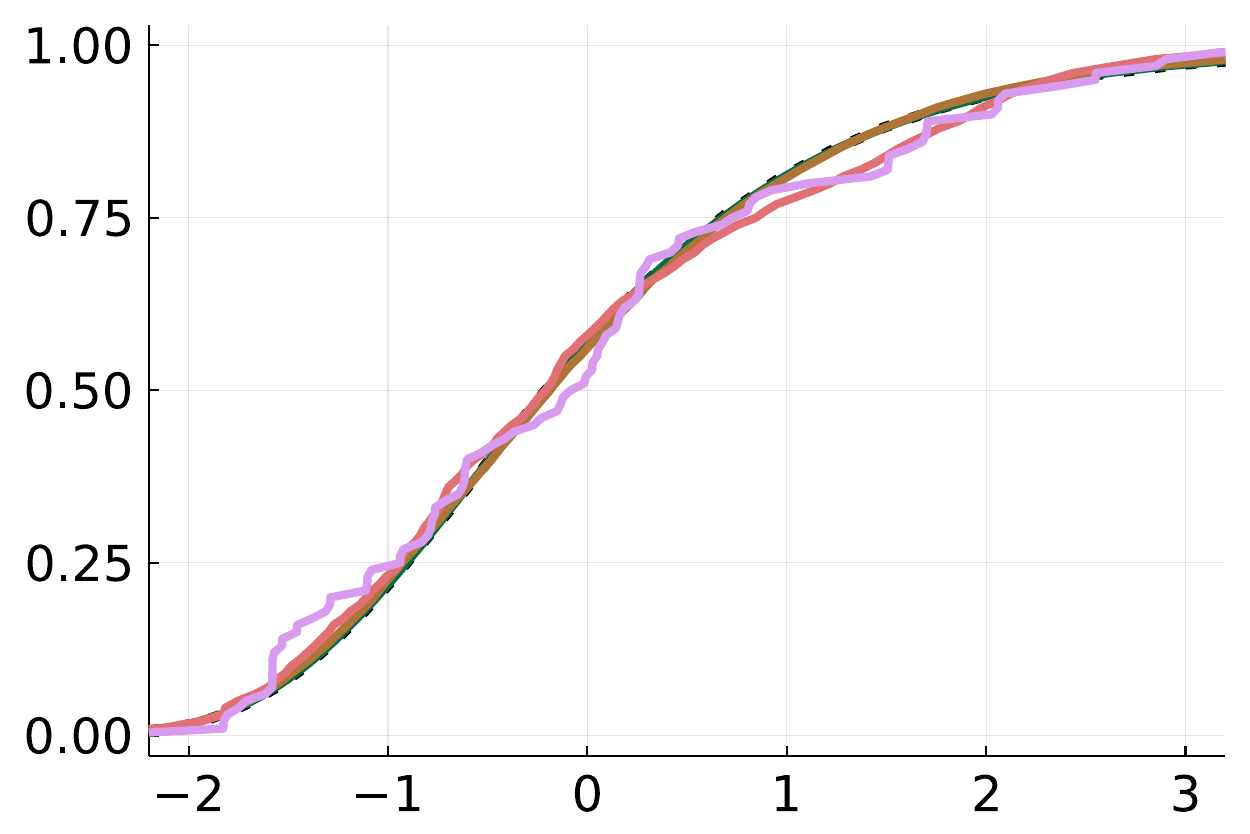} 
\end{center}
\end{subfigure}
\begin{subfigure}[t]{0.43\textwidth}
\begin{center}
\caption*{Type 2 (High-school graduate)} \label{fig:qq_1_2}
\vskip -6pt
\includegraphics[width=\linewidth]{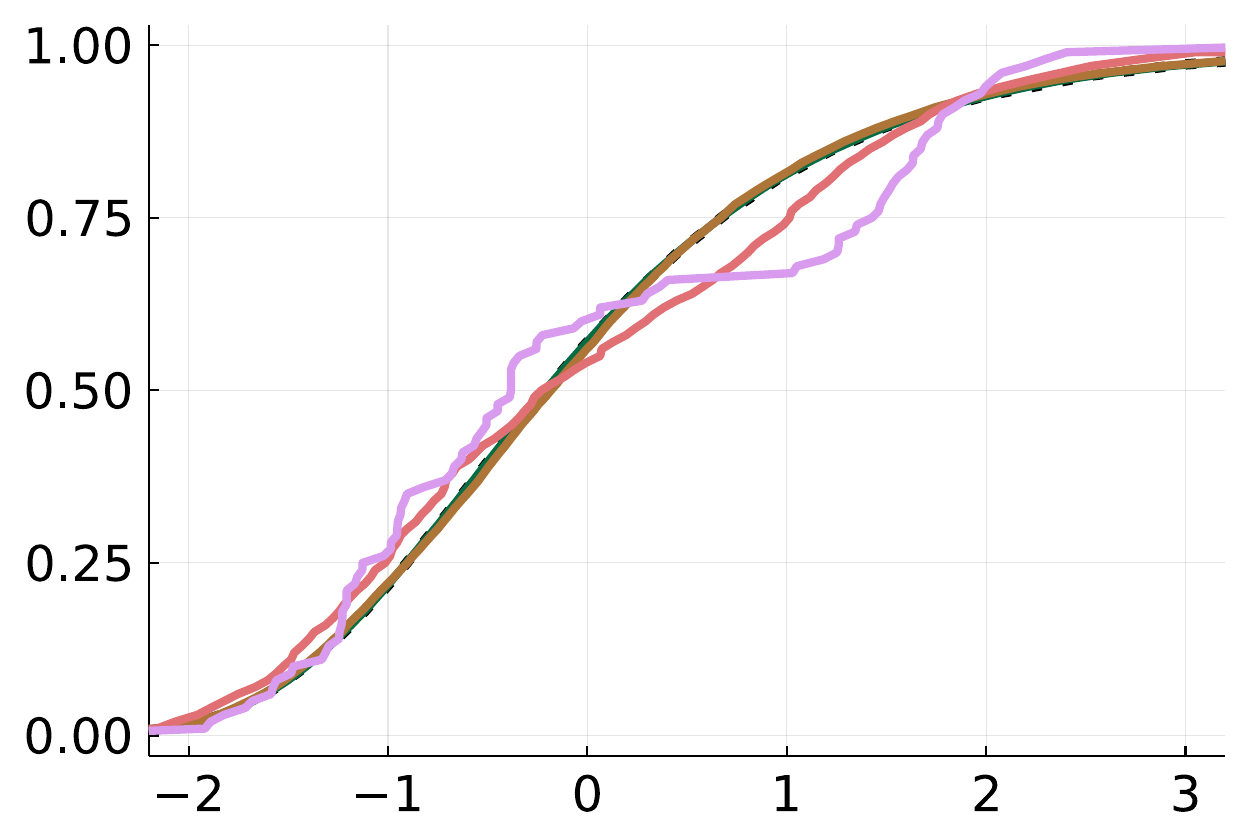} 
\end{center}
\end{subfigure}
\begin{subfigure}[t]{0.43\textwidth}
\begin{center}
\caption*{Type 3 (Some college)} \label{fig:qq_1_3}
\vskip -6pt
\includegraphics[width=\linewidth]{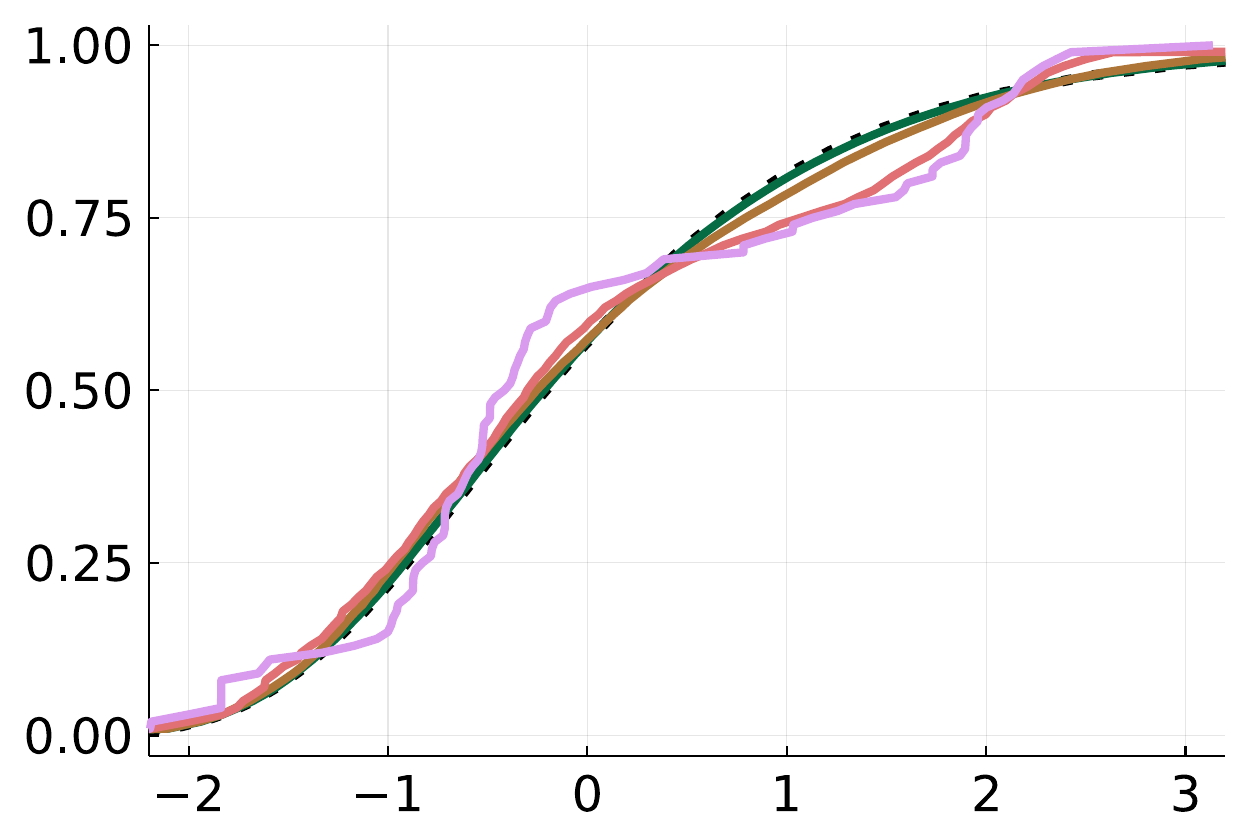} 
\end{center}
\end{subfigure}
\begin{subfigure}[t]{0.43\textwidth}
\begin{center}
\caption*{Type 4 (College graduate)} \label{fig:qq_1_4}
\vskip -6pt
\includegraphics[width=\linewidth]{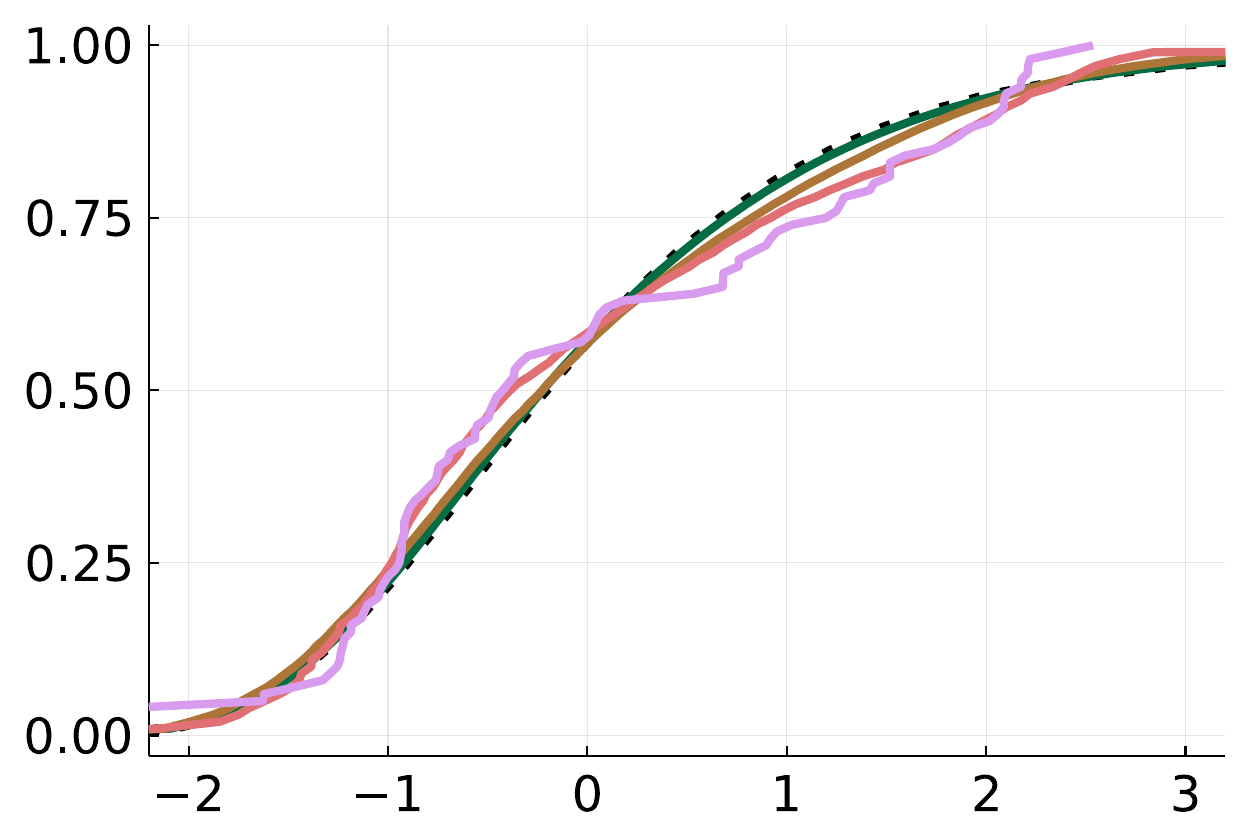} 
\end{center}
\end{subfigure}
\begin{subfigure}[t]{0.43\textwidth}
\begin{center}
\caption*{Type 5 (College-plus)} \label{fig:qq_1_5}
\vskip -6pt
\includegraphics[width=\linewidth]{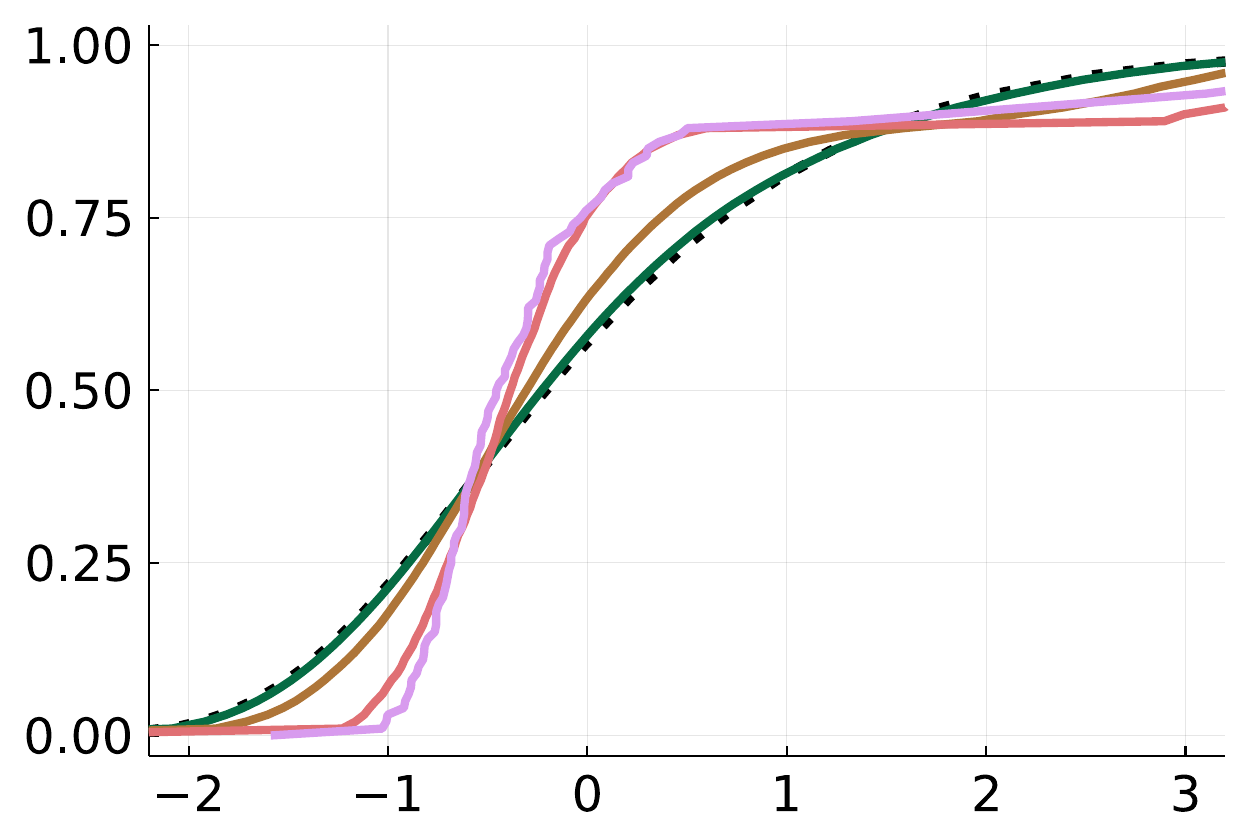} 
\end{center}
\end{subfigure}
\end{center}
\end{subfigure}

\vskip 10pt

\begin{subfigure}[t]{\textwidth}
\begin{center}
\caption{With exchangeability (all types)} \label{fig:csw_3_lfd_exch}
\vskip -2pt
\begin{subfigure}[t]{0.45\textwidth}
\begin{center}
\includegraphics[width=\linewidth]{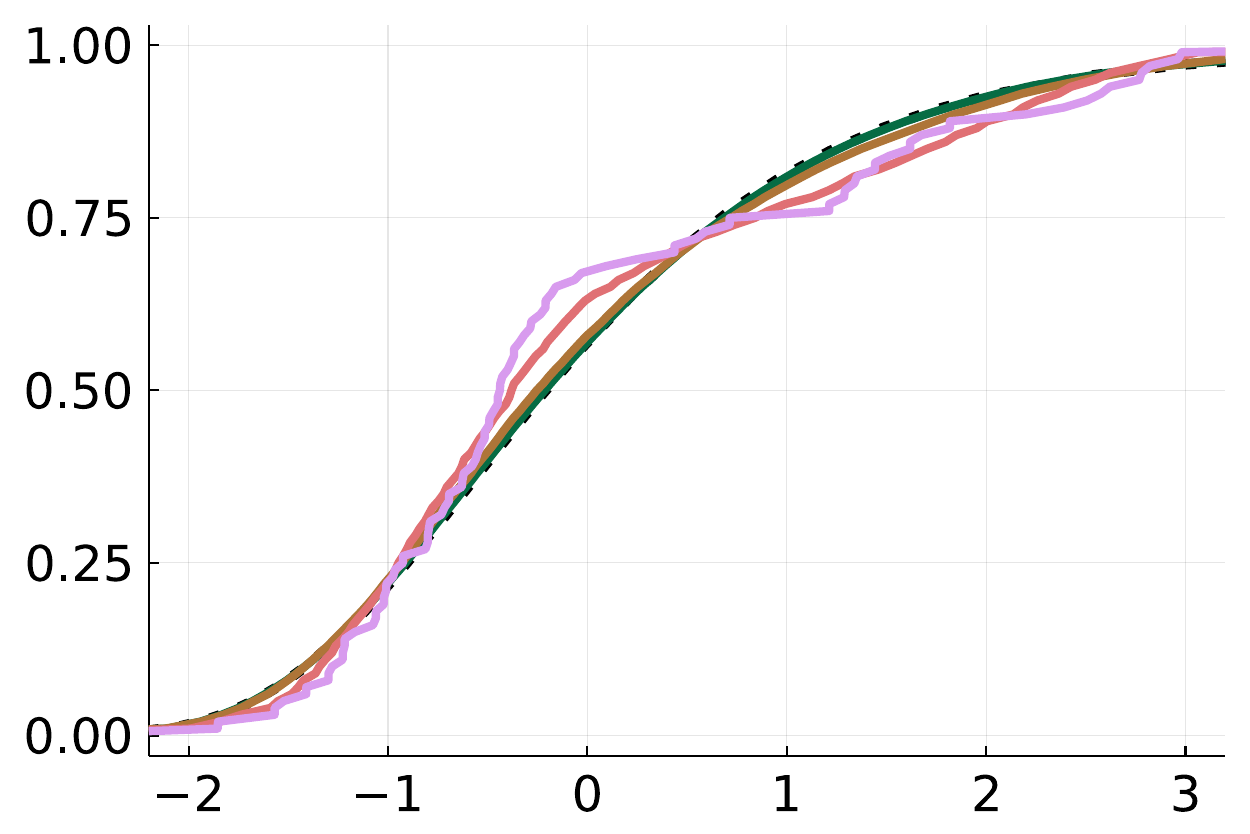} 
\label{fig:qq_2}
\end{center}
\end{subfigure}
\end{center}
\end{subfigure}

\vskip -20pt

\caption{\label{fig:csw_3_qq}Marginal CDFs for the LFDs maximizing the ``some college'' to ``college graduate'' premium in cohort 1 across potential spouse types. 
}
\end{center}
\vskip -14pt
\end{figure}

\begin{table}[t]
\begin{center}
\caption{\label{tab:csw_3}Metrics for interpreting $\delta$}
{
\begin{tabular}{ccccccccc} \hline \hline
 & & \multicolumn{3}{c}{Without exchangeability} & & \multicolumn{3}{c}{With exchangeability} \\
  \cline{3-5} \cline{7-9} \\[-10pt]
$\delta$ & & $\rho_{\max}$, $\ul \kappa_\delta$ & $\rho_{\max}$, $\ol \kappa_\delta$ & $size$ & & $\rho_{\max}$, $\ul \kappa_\delta$ & $\rho_{\max}$, $\ol \kappa_\delta$ & $size$ \\[2pt] 
0.01 & & -0.015 & -0.014 & 0.010 & & -0.022 & \phantom{-}0.013 & 0.006 \\
0.10 & & -0.071 & -0.073 & 0.038 & & -0.061 & \phantom{-}0.054 & 0.023 \\
1 & & -0.247 & -0.197 & 0.112 & & -0.139 & \phantom{-}0.115 & 0.099 \\
10 & & -0.502 & -0.496 & 0.242 & & -0.204 & \phantom{-}0.236 & 0.176 \\
100 & & -0.620 & -0.576 & 0.266 & & -0.266 & \phantom{-}0.284 & 0.178 \\ \hline \\[-10pt]
\end{tabular}
}
\parbox{\textwidth}{\small \emph{Note:} Averages across cohorts of the largest element of the correlation matrix for $U$ under the LFDs at which the estimated lower bounds ($\rho_{\max}$, $\ul \kappa_\delta$) and upper bounds ($\rho_{\max}$, $\ol \kappa_\delta$) are attained, and our $size$ measure from Section~\ref{s:lens}. Each is computed at the parameter values at which the estimated upper and lower bounds are attained.} 
\end{center}
\vskip -15pt
\end{table}

To understand better what is meant by ``small'' and ``large'' neighborhoods, Figure~\ref{fig:csw_3_qq} plots marginal CDFs for the LFDs under which the upper bounds for cohort 1 are attained. Similar LFDs (not reported) were obtained for other cohorts and the lower bounds. Without exchangeability, the LFDs with $\delta = 0.1$ are almost identical to Gumbel (plots with $\delta = 0.01$ are indistinguishable from Gumbel). LFDs appear close to Gumbel across most potential spouse types with $\delta = 1$, while for $\delta = 10$ and $\delta = 100$ the LFDs have kinks and indicate shifts in mass from the center of the distribution to the tails. 

Under exchangeability (Figure~\ref{fig:csw_3_lfd_exch}), the marginal distribution of shocks is independent of potential spouse type. In this case the LFDs for $\delta = 1$ or smaller  are virtually indistinguishable from Gumbel. LFDs with $\delta = 10$ and $\delta = 100$ are also less kinked than Figure~\ref{fig:csw_3_lfd} because distortions are spread more evenly across potential spouse types.

We also computed the largest correlation of shocks under the LFDs at which the bounds are attained and our $size$ measure from Section~\ref{s:lens}. As these quantities are stable across cohorts, we present their averages across cohorts in Table~\ref{tab:csw_3}. Shocks are independent when $\delta = 0$ and only very weakly correlated for small $\delta$, while for large $\delta$ some shocks are strongly negatively correlated. The maximal correlations under exchangeability are smaller, especially for large $\delta$. Turning to the $size$ measure, the LFDs for $\delta = 0.01$ without exchangeability shift the model-implied marriage probabilities by 0.01 (on average, across cohorts) from their values under the i.i.d. Gumbel assumption. LFDs for $\delta = 10$ and $\delta = 100$ shift  marriage probabilities around 0.25  (on average, across cohorts). Imposing exchangeability reduces the $size$ measure by around 25\% because model parameters do not vary as much under this shape restriction.

In view of the small-$\delta$ bounds in Figure~\ref{fig:csw_3}, the LFDs in Figure~\ref{fig:csw_3_qq}, and the metrics in Table~\ref{tab:csw_3}, it seems impossible to draw conclusions about how the sign of the premium has changed over time under slight nonparametric relaxations of the i.i.d. Gumbel assumption. To help understand why, Figure~\ref{fig:csw_3_inner_both} plots bounds where $F$ is allowed to vary but $\theta$ is held fixed at CSW's estimates. These ``fixed-$\theta$''  bounds for $\delta = 10$ and $\delta = 100$ are almost identical, and are roughly the same width as the $\delta = 0.01$ bounds in Figure~\ref{fig:csw_3}. The width of the bounds in Figure~\ref{fig:csw_3} therefore seems largely attributable to the additional variation in $\theta$ that is permitted when parametric assumptions for $F$ are relaxed.

Overall, our findings are complementary to \cite{GualdaniSinha} who perform a nonparametric reanalysis of CSW using the PIES methodology of \cite{Torgovitsky2019QE}. Although they do not derive nonparametric bounds on the marital education premium itself, only terms that contribute to it, they also find no evidence of an increase in premiums across cohorts.

\subsection{Welfare Analysis in a Rust Model}\label{s:rust}

Our second empirical illustration is a sensitivity analysis for welfare counterfactuals in the DDC model of \cite{Rust}.

\medskip

\paragraph{Model and Benchmark Estimates.} 

We focus on the specification in Table IX of \cite{Rust} where maintenance costs are linear in the state (i.e., mileage). In the notation of Example~\ref{ex:DDC}, $|\mc S| = 90$, $\beta = 0.9999$, and $\theta_\pi = (RC, MC)$ where $RC$ is the replacement cost and $MC$ is a maintenance cost parameter. Our counterfactual of interest is the change in average welfare arising from a 10\% reduction in maintenance costs. Hence, $\pi_{1,s}(\theta_\pi) = \tilde \pi_{1,s}(\theta_\pi) = -RC$ and  $\pi_{0,s}(\theta_\pi) = -0.001 MC \times s$ (baseline) and $\tilde \pi_{0,s}(\theta_\pi) = 0.9 \pi_{0,s}(\theta_\pi)$ (counterfactual). The counterfactual function is $k(\theta, \gamma) = w' (\tilde v - v)$ where $w$ is the stationary distribution of the state in the baseline model. 

Under the i.i.d. Gumbel assumption, the estimated counterfactual at the maximum likelihood estimate (MLE) of $\theta_\pi$ is 73.07 and its 95\% CS is [48.25,101.31].\footnote{We construct this CS by simulation. We draw $\hat \theta^*_\pi \sim N(\hat \theta_\pi,\hat \Sigma)$ where $\hat \theta_\pi$ is the MLE and $\hat \Sigma$ is an estimate of the inverse information matrix. For each $\hat \theta^*_\pi$ draw, we compute the baseline and counterfactual value functions $v^*$ and $\tilde v^*$, and hence the counterfactual $\hat \kappa^* = w'( \tilde v^* - v^*)$.} Note the counterfactual is point-identified under the i.i.d. Gumbel assumption because $\theta_\pi$ is point-identified.

\medskip

\paragraph{Implementation.}

We estimate CCPs using Rust's Group 4 data. Nonparametric estimates of the 90 CCPs are zero in many states, so we proceed as in Section~\ref{s:overid} and take the model-implied CCPs at the MLE of $\theta_\pi$ (under the i.i.d. Gumbel assumption) as our estimate $\hat P_2$. We drop moment conditions for CCPs in states where the replacement probability is less than $0.001$ to avoid numerical instabilities induced by including near-degenerate moments. This reduces the dimension of $g_2$ to 71. We normalize $F$ so that shocks have mean zero and the same variance as the Gumbel distribution by appending $\E^F[U_d] = 0$ and $\E^F[U_d^2 - \pi^2/6] = 0$, for $d = 0,1$, to $g_4$. In total, there are 255 moments (71 for CCPs, 180 for Bellman equations, and 4 location/scale normalizations) and $\theta = (\theta_\pi, v, \tilde v)$ has dimension 182. 

We implement our methods as described in Section~\ref{s:mpec}. The inner optimization uses 75 moments (71 for CCPs and 4 for normalizations), with the remaining 180 moments appended as constraints in the outer optimization. We define neighborhoods using hybrid divergence from Section~\ref{s:csw} so that Assumption~\ref{a:phi}(ii) holds. Similar results are obtained with $\chi^2$ and $L^4$ neighborhoods (see Appendix~\ref{ax:rust}). Expectations are computed using 50,000 scrambled Halton draws---see Appendix~\ref{ax:rust} for computation times. We compute 95\% CSs for $\ul \kappa_\delta$ and $\ol \kappa_\delta$ using the bootstrap procedure from Section~\ref{s:bootstrap} and projection procedure from Section~\ref{s:projection}. See Appendix~\ref{ax:rust} for details.

\medskip

\paragraph{Findings.}

\begin{figure}
\begin{center}
\includegraphics[width=0.6\textwidth]{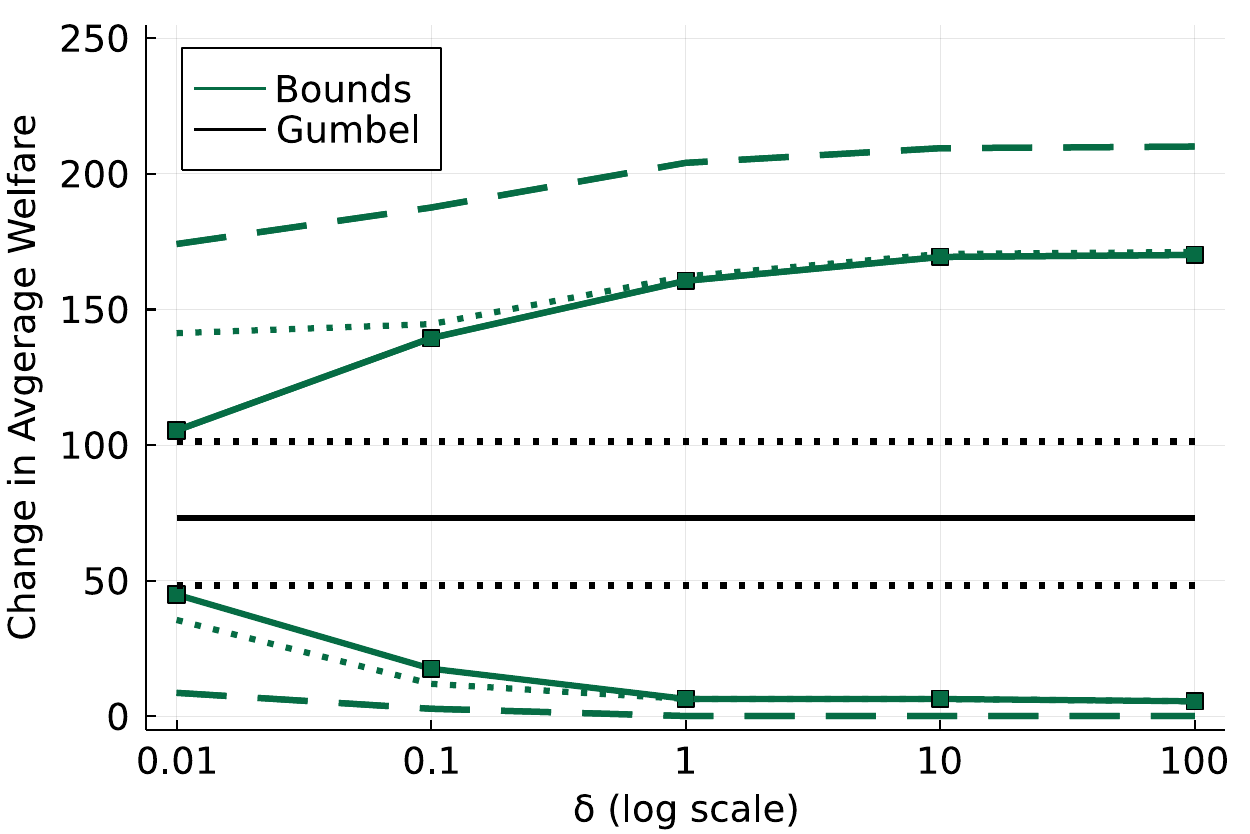} 
\caption{\label{fig:rust_bounds} Sensitivity analysis for change in average welfare under a 10\% maintenance cost subsidy. \emph{Note:} Solid lines are estimates, dotted lines are bootstrap CSs, dashed lines are projection CSs.}
\end{center}
\vskip -14pt
\end{figure}

Estimates and CSs for $\ul \kappa_\delta$ and $\ol \kappa_\delta$ are plotted in Figure~\ref{fig:rust_bounds} for values of $\delta$ from $0.01$ to $100$.\footnote{The width of the bootstrap CSs relative to the bounds reduces as $\delta$ gets large. We re-estimated our bounds using several different draws of bootstrapped CCPs in place of $\hat P_2$ and obtained bounds that spanned a range similar to the bootstrap CSs for small $\delta$, but which for many draws converged to values close to our estimates of the bounds for large $\delta$. This corroborates the behavior of our bootstrap CSs. We conjecture that other features of the model are potentially more important than the numerical values of the CCPs in determining nonparametric bounds on the welfare counterfactual. } As can be seen, the bounds expand rapidly under slight relaxations of the i.i.d. Gumbel assumption then stabilize around $\delta = 1$, where the lower bound is 6.45 and the upper bound of 160.5 represents approximately 220\% of the value under the i.i.d. Gumbel assumption.

\begin{figure}[t]
\begin{center}
\begin{subfigure}{0.43\textwidth}
\begin{center}
\caption*{Lower bound} \label{fig:rust_ldf_lower}
\includegraphics[width=\linewidth]{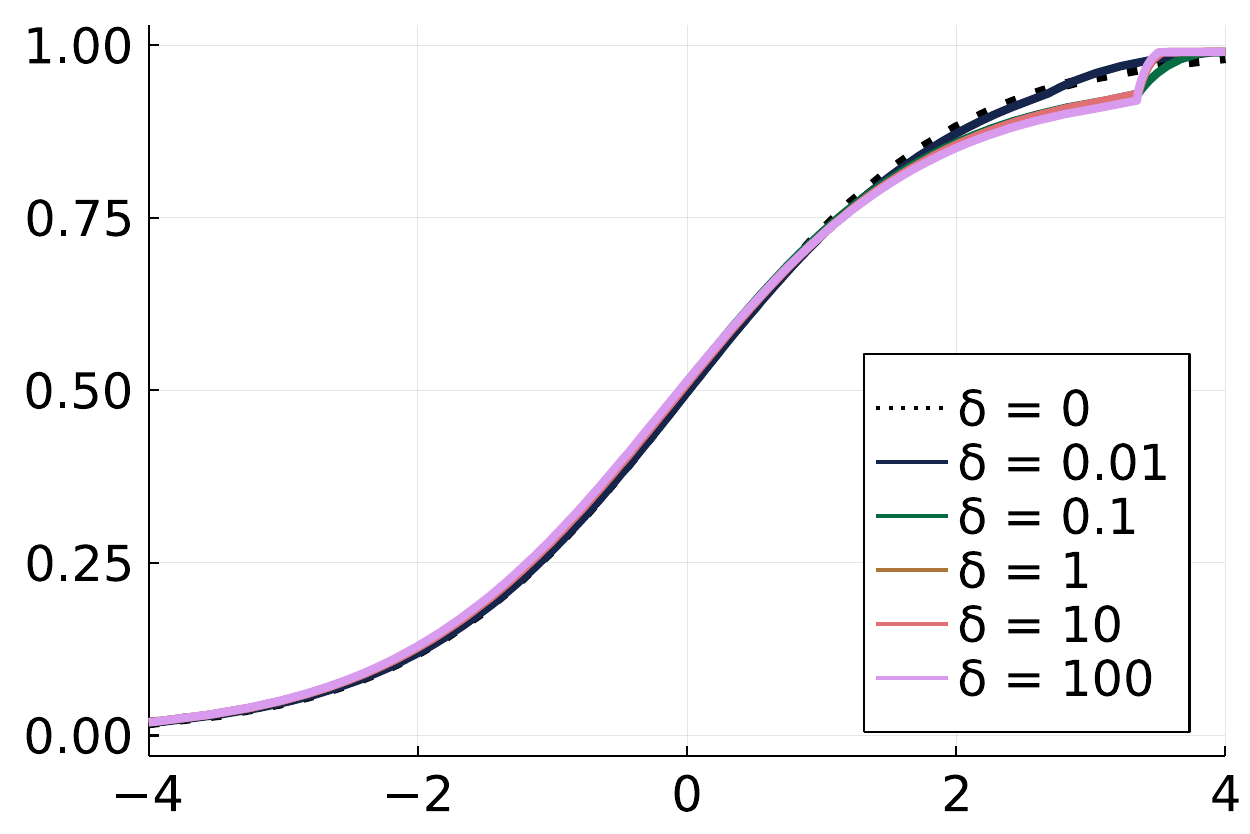} 
\end{center}
\end{subfigure}
\begin{subfigure}{0.43\textwidth}
\begin{center}
\caption*{Upper bound} \label{fig:rust_lfd_upper}
\includegraphics[width=\linewidth]{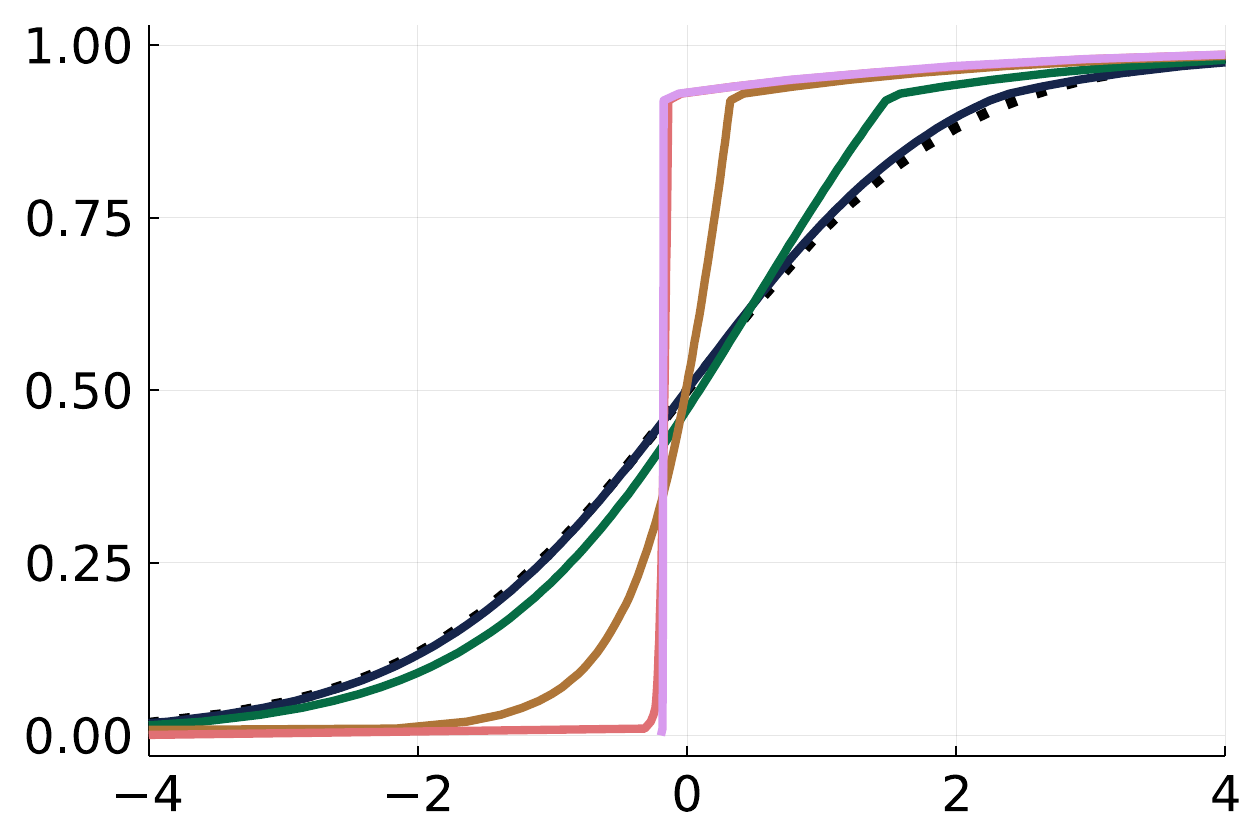} 
\end{center}
\end{subfigure}

\caption{\label{fig:rust_lfd} CDFs of $U_1 - U_0$ under the LFDs at which the estimated lower and upper bounds on the welfare counterfactual are attained.}
\end{center}
\end{figure}

To interpret $\delta$, in Figure~\ref{fig:rust_lfd} we plot the CDFs of $U_1 - U_0$ under the LFDs at which the estimated bounds $\hat{\ul \kappa}_\delta$ and $\hat{\ol \kappa}_\delta$ are attained. LFDs were computed as described in Section~\ref{s:lfd} using the construction (\ref{e:m:delta}). The distributions appear very close to logistic (their distribution when $\delta = 0$) for $\delta = 0.01$. Therefore, we see that large differences in welfare counterfactuals can arise under very slight departures from the i.i.d. Gumbel assumption. LFDs for the upper bound shift increasing amounts of mass to the center of the distribution of $U_1 - U_0$ as $\delta$ increases. LFDs corresponding to the lower bound are relatively less distorted, but have increasing amounts of mass shifted into the right tail. These are similar for $\delta = 0.1$ through $\delta = 100$ because the estimated lower bound stabilizes for smaller values of $\delta$ than the upper bound  (cf. Figure~\ref{fig:rust_bounds}).

\begin{table}[t]
\begin{center}
\caption{\label{tab:rust}Metrics for interpreting $\delta$}
{
\begin{tabular}{ccccccccccc} \hline \hline
 & & \multicolumn{4}{c}{Lower bound} & & \multicolumn{4}{c}{Upper bound} \\
  \cline{3-6} \cline{8-11} \\[-10pt]
$\delta$ & & $corr$ & $size$ & $RC$ & $MC$ & & $corr$ & $size$ & $RC$ & $MC$ \\[2pt] 
0 & & \phantom{-}0.000 & 0.000 & 10.208 & 2.294  & & \phantom{-}0.000 & 0.000 & 10.208 & 2.294 \\
0.01 & & \phantom{-}0.036 & 0.010 & \phantom{1}7.357 & 1.411 & & -0.027 & 0.016 & 13.390 & 3.307 \\
0.1 & & -0.058 & 0.039 & \phantom{1}5.186 & 0.553 & & \phantom{-}0.149 & 0.109 & 16.134 & 4.374 \\
1 & & -0.045 & 0.039 & \phantom{1}4.023 & 0.203 & & \phantom{-}0.616 & 0.346 & 17.166 & 5.038 \\
10 & & -0.040 & 0.039 & \phantom{1}4.022 & 0.202 & & \phantom{-}0.765 & 0.461 & 17.595 & 5.331 \\
100 & & -0.063 & 0.039 & \phantom{1}3.931 & 0.176 & & \phantom{-}0.764 & 0.469 & 17.626 & 5.365 \\ \hline \\[-10pt]
\end{tabular}
}
\parbox{\textwidth}{\small \emph{Note:} Correlation of $U_0$ and $U_1$ under the LFD at which the estimated lower and upper bounds are attained ($corr$), our $size$ measure from Section~\ref{s:lens}, and replacement and maintenance cost parameters at which the estimated lower and upper bounds are attained.} 
\end{center}
\end{table}

Table~\ref{tab:rust} lists other metrics to help interpret the neighborhood size. The first is the correlation of $U_0$ and $U_1$ under the LFDs at which $\hat{\ul \kappa}_\delta$ and $\hat{\ol \kappa}_\delta$ are attained. These are very small for $\delta = 0.01$ and remain small under the LFDs for $\hat{\ul \kappa}_\delta$ as $\delta$ increases, while $U_0$ and $U_1$ are strongly positively correlated under the LFDs for $\hat{\ol \kappa}_\delta$, especially for larger $\delta$ values. Given the asymmetry in distortions between the lower and upper values, we compute our $size$ measure separately for both. We measure distortions to the moments corresponding to the CCPs as these are most directly interpretable within the context of the model. We see that the LFDs for $\delta = 0.01$ are distorting $F_*$ in a manner that shifts the model-implied CCPs by at most 0.016. By contrast, the LFDs for $\delta = 10$ and $\delta = 100$ shift the model-implied CCPs from their values under the i.i.d. Gumbel assumption by at most $0.04$ for $\hat{\ul \kappa}_\delta$ and $0.47$ for $\hat{\ol \kappa}_\delta$. 

The parameters at which $\hat{\ul \kappa}_\delta$ and $\hat{\ol \kappa}_\delta$ are attained are also revealing about neighborhood size. Table~\ref{tab:rust} presents MLEs of $MC$ and $RC$, which are similar to the values reported in Table IX of \cite{Rust}. We see from Table~\ref{tab:rust} that  $\hat{\ul \kappa}_\delta$ and $\hat{\ol \kappa}_\delta$  are attained at very different parameter values, with much smaller cost parameters for the lower bound and larger parameters for the upper bound, even for $\delta = 0.01$. Intuitively, a smaller $MC$ means that the saving from the subsidy---which is proportional---must be small. Correspondingly, a low $RC$ is needed to help the model to fit the observed CCPs at the smaller $MC$. While it is known that payoff parameters are not identified without parametric assumptions on $F$, it is perhaps surprising that these parameters vary by so much under slight relaxations of the i.i.d. Gumbel assumption. For instance, with $\delta = 0.01$ the lower bound is attained with cost parameters $RC = 7.357$ and $MC = 1.411$ while the upper bound is attained with cost parameters that are roughly double these values.

\section{Estimation and Inference}\label{s:asymptotics}

We begin in Section~\ref{s:consistency} by establishing consistency and the asymptotic distribution of the estimators $\hat{\ul \kappa}_\delta$ and $\hat{\ol \kappa}_\delta$ from Section~\ref{s:estimators}. We then present a bootstrap-based inference method in Section~\ref{s:bootstrap} and a projection-based inference method in Section~\ref{s:projection}.

\subsection{Large-sample Properties of Plug-in Estimators}\label{s:consistency}

We first introduce some regularity conditions. 
Recall the space $\mc E$ from Assumption~\ref{a:phi}. We equip $\mc E$ with the Orlicz norm (see Appendix~\ref{ax:Orlicz})
\[
 \|f\|_\psi = \inf_{c > 0} \frac{1}{c} \left( 1 + \E^{F_*}[\psi(c|f(U)|)] \right).
\]
This norm is equivalent to the $L^2(F_*)$ norm for $\chi^2$ and hybrid divergence and equivalent to the $L^q(F_*)$ norm for $L^p$ divergence ($p^{-1} + q^{-1} = 1$), while for KL divergence it is stronger than any $L^p(F_*)$ norm with $p < \infty$ but weaker than the sup-norm. Say that a class of functions $\{f_a : a \in \mc A\} \subset \mc E$ indexed by a metric space $\mc A$ is \emph{$\mc E$-continuous in $a$} if $a' \to a$ in $\mc A$ implies $\|f_a - f_{a'}\|_\psi \to 0$. We also require a slightly stronger notion of constraint qualification than Condition S from Section~\ref{s:sharp}.

\begin{definition}\label{cond:sprime}
\emph{Condition S'} holds at $(\theta,\gamma,P)$ if $\vec P \in \mr{int}(\mc G(\theta,\gamma) + \mc C) $.
\end{definition}

Condition S' replaces ``relative interior'' in Condition S with ``interior''. 
Finally, recall $\Delta(\theta;\gamma,P)$ from (\ref{e:qual:dual}) and let $\Theta_\delta(\gamma,P) = \{\theta \in \Theta : \Delta(\theta;\gamma,P) < \delta\}$.

\begin{assumptionp}{M}\label{a:m}
(i)\hskip\labelsep $k(\cdot;\theta,\gamma)$ and each entry of $g(\cdot;\theta,\gamma)$ are $\mc E$-continuous in $(\theta,\gamma)$;
\begin{enumerate}[topsep=-18pt,itemsep=0pt,parsep=0pt,partopsep=0pt]
\item[(ii)] $(\theta,\gamma) \mapsto \mb E^{F_*}[\phi^\star(a_1 + a_2 k(U,\theta,\gamma) + a_3' g(U,\theta,\gamma))]$ is continuous for each $(a_1,a_2,a_3) \in \mb R \times \mb R \times \mb R^{d}$;
\item[(iii)] $\Theta_\delta(\gamma_0,P_0)$ is nonempty and Condition S' holds at $(\theta,\gamma_0,P_0)$ for each $\theta \in \Theta_\delta(\gamma_0,P_0)$;
\item[(iv)] $\mr{cl}(\Theta_\delta(\gamma_0,P_0)) \supseteq \{ \theta \in \Theta : \Delta(\theta;\gamma_0,P_0) \leq \delta\}$;
\item[(v)] $\Theta$ is a compact subset of $\mb R^{d_\theta}$.
\end{enumerate}  
\end{assumptionp}

Parts (i) and (ii) of Assumption~\ref{a:m} are continuity conditions. If $k$ and $g$ consist of indicator functions, then these conditions hold provided the probabilities of the events under $F_*$ are continuous in $(\theta,\gamma)$. In models without $\gamma$, these conditions simply require continuity in $\theta$. 

There are two parts to Assumption~\ref{a:m}(iii). The nonemptyness condition holds when the model is correctly specified under $F_*$ or, more generally, when there is at least one $F \in \mc N_\delta$ that satisfies (\ref{e:mod}) for some $\theta$. The second part is a constraint qualification. This condition requires that for each $\theta \in \Theta_\delta(\gamma_0,P_0)$, there is a distribution $F$ under which (\ref{e:mod}) holds at $(\theta,\gamma_0,P_0)$ that is ``interior'' to $\mc N_\infty$, in the sense that one can perturb the moments at $(\theta,\gamma_0,P_0)$ in all directions by perturbing $F$. Condition S' also requires that there is $F \in \mc N_\infty$ under which any inequality restrictions at $(\theta, \gamma_0, P_0)$ hold strictly. Note, however, that we do not require that this $F$ belongs to $\mc N_\delta$, only to $\mc N_\infty$. We therefore do not view this condition as overly restrictive. We also conjecture it could be relaxed using a notion similar to $S$-regularity from Section~\ref{s:sharp}.  

Assumption~\ref{a:m}(iv) is made for convenience and can be relaxed; this condition simply ensures that there do not exist values of $\theta$ at which $\Delta(\theta;\gamma_0,P_0) = \delta$ that are separated from $\Theta_\delta(\gamma_0,P_0)$.
Assumption~\ref{a:m}(v) is standard and can be relaxed.

\begin{theorem}\label{t:c-const}
Suppose that Assumptions~\ref{a:phi} and \ref{a:m} hold and $(\hat \gamma,\hat P) \to_p (\gamma_0,P_0)$ or, if there is no auxiliary parameter, $\hat P \to_p P_0$. Then $\hat{\ul \kappa}_\delta \to_p \ul \kappa_\delta$ and $\hat{\ol \kappa}_\delta \to_p \ol \kappa_\delta$.
\end{theorem}

To derive the asymptotic distribution of the estimators, we assume $\gamma_0$ is known and suppress dependence of all quantities on $\gamma$ for the remainder of this section. This entails no loss of generality for models without $\gamma$, such as Examples~\ref{ex:choice-welfare} and \ref{ex:game} and the application in Section~\ref{s:csw}. In DDC models this presumes the law of motion of the state is known. The asymptotic distribution therefore reflects only sampling uncertainty from the estimated CCPs, which is the case for confidence sets reported when laws of motion are first estimated ``offline''. Extending our approach to accommodate sampling variation in $\hat \gamma$ in a tractable manner appears to require exploiting application-specific model structure, which we defer to future work.

Define
\begin{equation}
\begin{aligned}
 \ul b_\delta(P) & = \inf_{\theta \in \Theta_\delta(P)} \ul K_\delta(\theta;P) \,, & & & 
 \ol b_\delta(P) & = \sup_{\theta \in \Theta_\delta(P)} \ol K_\delta(\theta;P) \,. \label{e:kappa-functions}
\end{aligned}
\end{equation}
In this notation, $\ul \kappa_\delta = \ul b_\delta(P_0)$ and $\ol \kappa_\delta = \ol b_\delta(P_0)$ (see Lemma~\ref{lem:endpoints}) and $\hat{\ul \kappa}_\delta = \ul b_\delta(\hat P)$ and $\hat{\ol \kappa}_\delta = \ol b_\delta(\hat P)$. 
We derive the asymptotic distribution of $\hat{\ul \kappa}_\delta$ and $\hat{\ol \kappa}_\delta$ by showing $\ul b_\delta$ and $\ol b_\delta$ are directionally differentiable and applying a suitable delta method. Say $f : \mb R^{d_1+d_2} \to \mb R$ is (Hadamard) directionally differentiable at $P_0$ if there is a continuous map $df_{P_0}[\,\cdot\, ] : \mb R^{d_1+d_2} \to \mb R$ such that
\[
  \lim_{n \to \infty} t_n^{-1} \left( f(P_0+t_n h_n)-f(P_0) \right) = df_{P_0}[h]
\]
for all sequences $t_n \downarrow 0$ and $h_n \to h$ \cite[p. 480]{Shapiro1990}. If $df_{P_0}[h]$ is linear in $h$ then $f$ is (fully) differentiable at $P_0$. 
We introduce some additional notation used to define the directional derivatives of $\ul b_\delta$ and $\ol b_\delta$. Let 
\begin{align*}
 \ul \Xi_\delta(\theta;P) & = \mr{argsup}_{\eta \geq 0, \zeta \in \mb R, \lambda \in \Lambda} -\E^{F_*}\Big[ (\eta \phi)^\star( -k(U,\theta) - \zeta - \lambda' g(U,\theta)) \Big] - \eta \delta - \zeta - \lambda_{12}' P \,, 
\end{align*}
where $(\eta\phi)^\star$ denotes the convex conjugate of $x \mapsto \eta \cdot \phi(x)$, and let $\ol \Xi_\delta(\theta;P)$ denote the analogous arginf for the minimization problem corresponding to the upper bound. Recall that $\ul \lambda_{12} = (\ul \lambda_1,\ul \lambda_2)$ collects the first $d_1 + d_2$ elements of $\ul \lambda$. Let
\[
 \ul \Lambda_\delta(\theta;P) = \{ ( \lambda_1, \lambda_2) : ( \eta, \zeta, \lambda_1,  \lambda_2,  \lambda_3,  \lambda_4) \in \ul \Xi_\delta(\theta;P)\} 
\]
denote the projection of $\ul \Xi_\delta(\theta;P)$ for $\ul \lambda_{12}$. We let $\ol \Lambda_\delta(\theta;P)$ denoting the analogous projection of $\ol \Xi_\delta(\theta;P)$. Finally, let 
\[
\begin{aligned}
 \ul \Theta_\delta(P_0) & = \mr{arg}\min_{\theta \in \Theta} \ul K_\delta(\theta;P_0)\,, & 
 \ol \Theta_\delta(P_0) & = \mr{arg}\max_{\theta \in \Theta} \ol K_\delta(\theta;P_0)\,. 
\end{aligned}
\]
The sets $\ul \Theta_\delta(P_0)$ and $\ol \Theta_\delta(P_0)$ are nonempty and compact under Assumptions~\ref{a:phi} and \ref{a:m}. 

The following regularity conditions are presented for the general case where $k$ depends on $u$. It may be possible to weaken some of these regularity conditions in the special case in which $k$ does not depend on $u$.

\begin{assumptionp}{\ref{a:m} (continued)}
(vi)\hskip\labelsep $\ul \Theta_\delta(P_0) \subseteq \Theta_\delta(P_0)$ and $\ol \Theta_\delta(P_0) \subseteq \Theta_\delta(P_0)$;
\begin{enumerate}[topsep=-18pt,itemsep=0pt,parsep=0pt,partopsep=0pt]
\item[(vii)] $\theta \mapsto \ul \Lambda_\delta(\theta;P_0)$ and $\theta \mapsto \ol \Lambda_\delta(\theta;P_0)$ are lower hemicontinuous at each $\theta \in \ul \Theta_\delta(P_0)$ and $\theta \in \ol \Theta_\delta(P_0)$, respectively.
\end{enumerate}
\end{assumptionp}

\begin{theorem}\label{t:asydist}
Suppose that Assumptions~\ref{a:phi} and \ref{a:m} hold. Then $\ul b_\delta$ and $\ol b_\delta$ are directionally differentiable at $P_0$, with
\begin{align*}
 d\ul b_{\delta,P_0}[h] & = \min_{\theta \in \ul \Theta_\delta(P_0)} \max_{\ul \lambda_{12} \in \ul \Lambda_\delta(\theta;P_0)} -\ul \lambda_{12}' h \,, & 
 d\ol b_{\delta,P_0}[h] & = \max_{\theta \in \ol \Theta_\delta(P_0)} \min_{\ol \lambda_{12} \in \ol \Lambda_\delta(\theta;P_0)}  \ol \lambda_{12}' h \,.
\end{align*}
Moreover, if $\sqrt n (\hat P - P_0) \to_d Z \sim N(0,\Sigma)$ with $\Sigma$ finite, then
\[
 \sqrt n \left( \left( \begin{array}{c}
 \hat{\ol \kappa}_\delta  \\
 \hat{\ul \kappa}_\delta  \end{array} \right)
 - \left( \begin{array}{c}
 \ol \kappa_\delta \\
  \ul \kappa_\delta 
 \end{array} \right) \right)
 \to_d 
 \left( \begin{array}{c}
 d\ul b_{\delta,P_0}[Z]  \\
 d\ol b_{\delta,P_0}[Z]
 \end{array} \right) .
\]
\end{theorem}

The asymptotic distribution presented in Theorem~\ref{t:asydist} is non-Gaussian. In the special case in which $\cup_{\theta \in \ul \Theta_\delta(P_0)} \ul \Lambda_\delta( \theta;P_0) = \{\ul \lambda_{12}\}$, the asymptotic distribution of $\hat{\ul \kappa}_\delta$ simplifies to $N(0,\ul \lambda_{12}'\Sigma \ul \lambda_{12})$. An analogous simplification holds for $\hat{\ol \kappa}_\delta$ when $\cup_{\theta \in \ol \Theta_\delta(P_0)} \ol \Lambda_\delta( \theta;P_0)$ is a singleton.

\subsection{Inference Procedure 1: Bootstrap}\label{s:bootstrap}

Our first inference procedure specializes the general approach of \cite{FangSantos} for inference on directionally differentiable functions to the present setting. Define
\begin{equation*}
\begin{aligned}
 \wh{d\ul b}_{\delta,P_0}[h] & = \inf_{\theta \in \hat{\ul \Theta}_{\delta,n}} \sup_{\ul \lambda_{12} \in \ul \Lambda_\delta(\theta;\hat P)} -\ul \lambda_{12}' h \,, & & & 
 \wh{d\ol b}_{\delta,P_0}[h] & = \sup_{\theta \in \hat{\ol \Theta}_{\delta,n}} \inf_{\ol \lambda_{12} \in \ol \Lambda_\delta(\theta;\hat P)}  \ol \lambda_{12}' h \,, & 
\end{aligned}
\end{equation*}
where 
\begin{equation*}\begin{aligned}
 \hat{\ul \Theta}_{\delta,n} & = \{\theta \in \Theta_\delta(\hat P) : \ul K_\delta(\theta; \hat P) \leq \hat{\ul \kappa}_\delta + \hat \nu \sqrt{\log n/n}\} \,, \mbox{ and}\\
 \hat{\ol \Theta}_{\delta,n} & = \{\theta \in \Theta_\delta(\hat P) : \ol K_\delta(\theta; \hat P) \geq \hat{\ol \kappa}_\delta - \hat \nu \sqrt{\log n/n}\} \,,
\end{aligned}
\end{equation*}
with $\hat \nu$ a (possibly random) positive scalar tuning parameter for which $\hat \nu \to_p \nu > 0$. Any such $\hat \nu$ results in a confidence set with asymptotically correct coverage. 
We give some practical guidance for choosing $\hat \nu$ below.

Let $\hat P^*$ denote a bootstrapped version of $\hat P$. In practice any bootstrap can be used provided it satisfies mild consistency conditions. In the empirical application in Section~\ref{s:csw} we simply draw $\hat P^* \sim N(\hat P,\hat \Sigma/n)$ where $\hat \Sigma$ is a consistent estimator of $\Sigma$. Let 
\begin{align*}
 \hat {\ul c}_{\alpha} & = \mbox{ $\alpha$-quantile of } \wh{d\ul b}_{\delta,P_0}[\sqrt n(\hat P^* - \hat P)] \,, &
 \hat {\ol c}_{\alpha} & = \mbox{ $\alpha$-quantile of } \wh{d\ol b}_{\delta,P_0}[\sqrt n(\hat P^* - \hat P)] \,,
\end{align*}
where the quantiles are computed by resampling $\hat P^*$ (conditional on the data). 
Lower, upper, and two-sided $100(1-\alpha)$\% CSs for $\ul \kappa_\delta$ and $\ol \kappa_\delta$ are 
\begin{align*}
 & CS_{\delta,L}^{1-\alpha} = \left[ \hat{\ul \kappa}_\delta - { \frac{\hat {\ul c}_{1-\alpha}}{\sqrt n} }, +\infty\right) ,  \\
 & CS_{\delta,U}^{1-\alpha} = \left(-\infty, \hat{\ol \kappa}_\delta - { \frac{\hat {\ol c}_{\alpha}}{\sqrt n}  } \right] , & 
 & CS_\delta^{1-\alpha} = \left[ \hat{\ul \kappa}_\delta - { \frac{\hat {\ul c}_{1-\alpha/2}}{\sqrt n}}, \hat{\ol \kappa}_\delta - { \frac{\hat {\ol c}_{\alpha/2}}{\sqrt n}  } \right] .
\end{align*}
We require a slight strengthening of Assumption~\ref{a:m}(vii) to establish validity of the procedure. As before, regularity conditions are presented for the general case where $k$ depends on $u$. It may be possible to weaken these conditions when $k$ does not depend on $u$.

\begin{assumptionp}{\ref{a:m} (continued)}
(vii')\hskip\labelsep $(\theta,P) \mapsto \ul \Lambda_\delta(\theta;P)$ and $(\theta,P) \mapsto \ol \Lambda_\delta(\theta;P)$ are lower hemicontinuous at $(\theta,P_0)$ for each $\theta \in \ul \Theta_\delta(P_0)$ and $\theta \in \ol \Theta_\delta(P_0)$, respectively.
\end{assumptionp}

\begin{theorem}\label{t:ci}
Suppose that Assumptions~\ref{a:phi} and \ref{a:m}(i)--(vi),(vii') hold, $\sqrt n (\hat P - P_0) \to_d Z \sim N(0,\Sigma)$ with $\Sigma$ finite, and $\hat P^*$ satisfies Assumption 3 of \cite{FangSantos}. Then the distribution of $\wh{d\ul b}_{\delta,P_0}[\sqrt n(\hat P^* - \hat P)]$ and $\wh{d\ol b}_{\delta,P_0}[\sqrt n(\hat P^* - \hat P)]$ (conditional on the data) is consistent for the asymptotic distribution derived in Theorem~\ref{t:asydist}. Moreover, if the CDFs of $d\ul b_{\delta,P_0}[Z]$ and $d\ol b_{\delta,P_0}[Z]$ are continuous and increasing at their $\alpha/2$, $\alpha$, $1-\alpha$, and $1-\alpha/2$ quantiles, then
\begin{align*}
 \lim_{n \to \infty} \Pr( \ul \kappa_\delta \in CS_{\delta,L}^{1-\alpha}) & = 1-\alpha \,, \\ 
 \lim_{n \to \infty} \Pr( \ol \kappa_\delta \in CS_{\delta,U}^{1-\alpha}) & = 1-\alpha \,, &
 \liminf_{n \to \infty} \Pr( [\ul \kappa_\delta,\ol \kappa_\delta] \subseteq CS_\delta^{1-\alpha}) & \geq 1-\alpha \,.
\end{align*}
\end{theorem}

Any $\hat \nu$ that satisfies $\hat\nu \to_p \nu > 0$ results in asymptotically valid CSs. In view of the functional forms of $\wh{d\ul b}_{\delta,P_0}[\,\cdot\,]$ and $\wh{d\ul b}_{\delta,P_0}[\,\cdot\,]$, smaller $\hat \nu$ produce (weakly) wider CSs. In the CSW application, we set $\hat \nu$ equal to the minimum diagonal element of the covariance matrix of the moments evaluated at $(\hat \theta, \hat \gamma,\hat P)$ under $F_*$, where $\hat \theta$ is computed under $F_*$. We chose this quantity as it is related to the convexity of the inner problem for small $\delta$. In practice, this resulted in $\hat \nu$ between 0.001 and 0.01. We recommend setting $\hat  \nu$ to be of a similarly small magnitude, then performing a sensitivity analysis to check that critical values aren't too dependent on $\hat \nu$. 
Setting $\hat \nu = 0$ and replacing $\hat{\ul \Theta}_{\delta,n}$ and $\hat{\ol \Theta}_{\delta,n}$ by $\{\hat{\ul \theta}_\delta\}$ and $\{\hat{\ol \theta}_\delta\}$ where $\hat{\ul \theta}_\delta$ and $\hat{\ol \theta}_\delta$ minimize and maximize the sample criterions is also valid, but may be conservative.

\subsection{Inference Procedure 2: Projection}\label{s:projection}

This second approach is computationally simple but possibly conservative.\footnote{We are grateful to a referee for suggesting this approach.} Suppose we have random vectors $\hat P_{1,U}^{1-\alpha}$, $\hat P_{2,U}^{1-\alpha}$, and $\hat P_{2,L}^{1-\alpha}$ that form a $100(1-\alpha)$\% rectangular CS for $P_0$:
\begin{equation} \label{e:cvg:P}
 \liminf_{n \to \infty} \Pr\left(  P_{10} \leq \hat P_{1,U}^{1-\alpha}\,,\, \hat P_{2,L}^{1-\alpha} \leq P_{20} \leq \hat P_{2,U}^{1-\alpha} \right) \geq 1-\alpha \,,
\end{equation}
where the inequalities should be understood to hold element-wise (we discuss how to construct a rectangular CS for $P_0$ below). 

The idea behind this approach is to replace any moment conditions involving $P$ by inequalities constructed from the rectangular CS. Define the criterion functions
\[
 \hat{\ul K}_{\delta,1-\alpha}(\theta) = \left[ \begin{array}{l}
 \ul K_{\delta,cs} (\theta;  \hat P_{1-\alpha}) \\[2pt]
  + \infty 
 \end{array} \right. ,
 \quad 
 \hat{\ol K}_{\delta,1-\alpha}(\theta) = \left[ \begin{array}{ll}
 \ol K_{\delta,cs} (\theta; \hat P_{1-\alpha}) & \mbox{if $\Delta_{cs}(\theta;\hat P_{1-\alpha}) < \delta$,} \\[2pt]
 - \infty & \mbox{if $\Delta_{cs}(\theta;\hat P_{1-\alpha}) \geq \delta$,}
 \end{array} \right.
\]
where $\ul K_{\delta,cs}$, $\ol K_{\delta,cs}$, and $\Delta_{cs}$ are versions of (\ref{e:dual:1}), (\ref{e:dual:2}), and (\ref{e:qual:dual}) formed using 
\begin{equation} \label{e:mod-relaxed}
\begin{aligned} 
 \mb E^{F}[ g_1(U,\theta)] & \leq \hat P_{1,U}^{1-\alpha} \,, &  
 \mb E^{F}[ g_2(U,\theta)] & \leq \hat P_{2,U}^{1-\alpha} \,, &  
 \mb E^{F}[ -g_2(U,\theta)] & \leq -\hat P_{2,L}^{1-\alpha} \,,
\end{aligned}
\end{equation}
as well as (\ref{e:mod:3}) and (\ref{e:mod:4}). 
In these criterions, $\Lambda$ is replaced by $\Lambda_{cs} = \mb R_+^{d_1 + 2 d_2 + d_3} \times \mb R^{d_4}$, $g$ is replaced by $g_{cs} = (g_1, g_2, -g_2, g_3, g_4)$, $P$ is replaced by $\hat P_{1-\alpha} = (\hat P_{1,U}^{1-\alpha},\hat P_{2,U}^{1-\alpha},-\hat P_{2,L}^{1-\alpha})$, and $\lambda_{12}$ denotes the first $d_1 + 2d_2$ elements of $\lambda$. 

Critical values are computed by optimizing the criterions $\hat{\ul K}_{\delta,1-\alpha}$ and $\hat{\ol K}_{\delta,1-\alpha}$ with respect to $\theta$:
\begin{equation*} 
\begin{aligned} 
 \hat{\ul \kappa}_{\delta,1-\alpha} & = \inf_{\theta \in \Theta} \hat{\ul K}_{\delta,1-\alpha}(\theta) \,, &
 \hat{\ol \kappa}_{\delta,1-\alpha} & = \sup_{\theta \in \Theta} \hat{\ol K}_{\delta,1-\alpha}(\theta)\,.
\end{aligned}
\end{equation*}
Lower, upper, and two-sided $100(1-\alpha)$\% CSs for $\ul \kappa_\delta$ and $\ol \kappa_\delta$ are then given by
\begin{align*}
 & CS_{\delta,L}^{1-\alpha} = \left[ \hat{\ul \kappa}_{\delta,1-\alpha}, +\infty \right) \,, & 
 & CS_{\delta,U}^{1-\alpha} = \left(-\infty, \hat{\ol \kappa}_{\delta,1-\alpha} \right] \,, &
 & CS_\delta^{1-\alpha} = \left[ \hat{\ul \kappa}_{\delta,1-\alpha}, \hat{\ol \kappa}_{\delta,1-\alpha} \right] \,.
\end{align*}

\begin{theorem}\label{t:ci:2}
Suppose that Assumptions~\ref{a:phi} and \ref{a:m}(i),(iii)--(v) hold and $\hat P_{1-\alpha}$ satisfies (\ref{e:cvg:P}). Then
\begin{align*}
 \liminf_{n \to \infty} \Pr( \ul \kappa_\delta \in CS_{\delta,L}^{1-\alpha}) & \geq 1-\alpha \,, \\ 
 \liminf_{n \to \infty} \Pr( \ol \kappa_\delta \in CS_{\delta,U}^{1-\alpha}) & \geq 1-\alpha \,, &
 \liminf_{n \to \infty} \Pr( [\ul \kappa_\delta,\ol \kappa_\delta] \subseteq CS_\delta^{1-\alpha}) & \geq 1-\alpha\,.
\end{align*}
\end{theorem}

To construct a rectangular CS for $P_0$ satisfying (\ref{e:cvg:P}), suppose $\sqrt n(\hat P - P_0) \to_d N(0,\Sigma)$ and we have a consistent estimator $\hat \Sigma$ of $\Sigma$. Let $\hat \sigma$ denote the vector formed by taking the square root of each diagonal entry of $\hat \Sigma$. Partition $\hat \sigma$ conformably as $\hat \sigma = (\hat \sigma_{(1)},\hat \sigma_{(2)})$ and set 
\begin{align*}
 \hat P_{1,L}^{1-\alpha} & = \hat P_1 + n^{-1/2} \hat c_{1-\alpha,1} \hat \sigma_{(1)} \,, & 
 \hat P_{2,L}^{1-\alpha} & = \hat P_2 - n^{-1/2} \hat c_{1-\alpha,2} \hat \sigma_{(2)} \,, & 
 \hat P_{2,U}^{1-\alpha} & = \hat P_2 + n^{-1/2} \hat c_{1-\alpha,2} \hat \sigma_{(2)} \,,
\end{align*}
where the (scalar) critical values $\hat c_{1-\alpha,1}$ and $\hat c_{1-\alpha,2}$ solve 
\[
 \Pr\left(\max_{1 \leq i \leq d_1} Z_i/\hat \sigma_i \leq \hat c_{1-\alpha,1},\max_{d_1+1 \leq i \leq d_2} |Z_i/\hat \sigma_i| \leq \hat c_{1-\alpha,2}\right) = 1-\alpha  \,, \quad Z \sim N(0, \hat \Sigma)\,.
\]
If $d_2 = 0$, then $\hat c_{1-\alpha,1}$ is the $(1-\alpha)$-quantile of $\max_{1 \leq i \leq d_1} Z_i/\hat \sigma_i$; similarly, if $d_1 = 0$, then $\hat c_{2,1-\alpha}$ is the $(1-\alpha)$-quantile of $\max_{1 \leq i \leq d_2} |Z_i/\hat \sigma_i|$.

\section{Conclusion} \label{s:conc}

This paper introduced a framework for analyzing the sensitivity of counterfactuals to parametric assumptions about the distribution of latent variables in structural models. In particular, we derived bounds on the set of counterfactuals obtained as the distribution of latent variables spans nonparametric neighborhoods of a given parametric specification while other ``structural'' model features are maintained. We illustrated our procedure with empirical applications to matching models and dynamic discrete choice.

\let\oldbibliography\thebibliography
\renewcommand{\thebibliography}[1]{\oldbibliography{#1}
\setlength{\itemsep}{0pt}}

\putbib

\end{bibunit}

\begin{bibunit}

\appendix

\newpage
\clearpage
\pagenumbering{arabic}\renewcommand{\thepage}{\arabic{page}}

\begin{center}
{\Large Online Appendix to ``Counterfactual Sensitivity and Robustness''}

\vskip 24pt

{\large Timothy Christensen \quad \quad Benjamin Connault}

\end{center}

\vskip 8pt

This supplement presents extensions of our methodology in Appendix~\ref{ax:extension}, additional results on nonparametric bounds on counterfactuals in Appendix~\ref{ax:sharp}, connections with local approaches to sensitivity analysis in Appendix~\ref{ax:local}, additional details on the empirical applications in Appendix~\ref{ax:empirical}, and proofs of results from the main text in Appendix~\ref{ax:proofs}.

\section{Extensions} \label{ax:extension}

This appendix presents three extensions of our methodology. Proofs of all results in this appendix are deferred to Appendix~\ref{ax:extension_proofs}.

\subsection{Group Invariance} \label{ax:exchangeable}

In certain settings it can be attractive to impose shape restrictions on $F$ such as symmetry, exchangeability, or, more generally, invariance to a finite group of transforms. For instance, imposing exchangeability of $F$ in discrete choice modeling ensures that alternatives' choice probabilities depend on their deterministic components of utility but not their labeling. These shape restrictions can be easily imposed whenever $F_*$ is invariant.

Formally, let $J$ denote the number of elements of $U$ and let $\Pi$ be a finite commutative group of transforms on $\mb R^J$---see, e.g., Section 1.4 of \cite{LehmannCasella}. We say that a distribution $F$ of $U$ is \emph{$\Pi$-invariant} if $\varpi U \sim F$ for all $\varpi \in \Pi$. 

\medskip

\paragraph{Example A.1 (Symmetry)} Central symmetry corresponds to $\Pi = \{I, -I\}$ for $I$ the identity matrix. Sign symmetry corresponds to taking $\Pi$ to be the collection of all $2^J$ diagonal matrices with $\pm 1$ in each diagonal entry. \hfill $\square$

\medskip

\paragraph{Example A.2 (Exchangeability)} Let $\Pi_J$ denote the group of all $J!$ permutation matrices of dimension $J$. Full exchangeability (permutation invariance) corresponds to $\Pi = \Pi_J$. Cyclic exchangeability (rotation invariance) corresponds to $\Pi = \Pi^{c}_J$ where $\Pi^c_J$ is the collection of all $J$ cyclic permutation matrices of dimension $J$ ($\Pi^c_J = \Pi_J$ when $J = 2$ and is a strict subset otherwise). When $J \geq 3$, dihedral exchangeability (rotation and reflection invariance) corresponds to taking $\Pi$ to be the set of all $2J$ permutation matrices representing rotations and reflections of $\{1,\ldots,J\}$. These types of exchangeability ensure the elements of $U$ are identically distributed, but they have different implications for the joint distribution of the elements of $U$. For instance, the distribution of $(U_i,U_j)$ for $i \neq j$ depends on $i-j$ and $|i-j|$ under cyclic and dihedral exchangeability, but is independent of $(i,j)$ under full exchangeability. \hfill $\square$

\medskip

Let $\mc N_\delta^\Pi = \{F \in \mc N_\delta : F$ is $\Pi$-invariant$\}$. We are interested in 
\begin{equation}
 \ul \kappa_{\delta}^\Pi : = \inf_{\theta \in \Theta, F \in \mc N_\delta^\Pi} \mb E^{F}[k(U,\theta,\gamma_0)] \quad \text{ subject to (\ref{e:mod})} , \label{e:kappa_lower-ex}
\end{equation}
and $\ol \kappa_{\delta}^\Pi$  defined as the analogous supremum. One may write $\ul \kappa_{\delta}^\Pi$ and $\ol \kappa_{\delta}^\Pi$ as the value of two optimization problems in which criterion functions $\ul K_\delta^\Pi(\theta;\gamma_0,P_0)$ and $\ol K_\delta^\Pi(\theta;\gamma_0,P_0)$ are optimized with respect to $\theta$. For a generic $(\theta,\gamma,P)$, define
\begin{equation}
 \ul K_\delta^\Pi(\theta;\gamma,P)  = \inf_{F \in \mc N_\delta^\Pi} \mb E^F[k(U,\theta,\gamma)] \quad \mbox{subject to (\ref{e:mod}) holding at $(\theta,\gamma,P)$} \,, \label{e:crit_l_ex}
\end{equation}
and define $\ol K_\delta^\Pi(\theta;\gamma,P)$ as the analogous supremum. These criterions have dual representations as finite-dimensional convex programs when $F_*$ is $\Pi$-invariant. Define
\[
\begin{aligned}
 k^\Pi(U,\theta,\gamma) & = \frac{1}{|\Pi|} \sum_{\varpi \in \Pi} k(\varpi U,\theta,\gamma) \,, & 
 g^\Pi_j(U,\theta,\gamma) & = \frac{1}{|\Pi|} \sum_{\varpi \in \Pi} g_j(\varpi U,\theta,\gamma) \,, \;\; j = 1,2,3,4, 
\end{aligned}
\]
where $|\Pi|$ denotes the cardinality of $\Pi$, and let $g^\Pi = (g_1^\Pi,g_2^\Pi,g_3^\Pi,g_4^\Pi)$.

\begin{proposition} \label{prop:criterion-exch} 
Suppose that Assumption~\ref{a:phi} holds and $F_*$ is $\Pi$-invariant. Then 
\begin{align}
 \ul K_\delta^\Pi(\theta;\gamma,P) & = \sup_{\eta > 0, \zeta \in \mb R, \lambda \in \Lambda} -\eta \E^{F_*}\left[ {\textstyle \phi^\star \left(\frac{k^\Pi(U,\theta,\gamma) + \zeta + \lambda' g^\Pi(U,\theta,\gamma) }{-\eta} \right) } \right] - \eta \delta - \zeta - \lambda_{12}'P \label{e:dual:1_exch} \,,\\
 \ol K_\delta^\Pi(\theta;\gamma,P) & = \inf_{\eta > 0, \zeta \in \mb R, \lambda \in \Lambda} \phantom{-} \eta \E^{F_*}\left[ {\textstyle \phi^\star \left( \frac{k^\Pi(U,\theta,\gamma) - \zeta- \lambda' g^\Pi(U,\theta,\gamma)}{\eta} \right) } \right] + \eta \delta + \zeta + \lambda_{12}'P \,. \label{e:dual:2_exch} 
\end{align}
Moreover, the value of problem (\ref{e:dual:1_exch}) is $+\infty$ (equivalently, the value of problem (\ref{e:dual:2_exch}) is $-\infty$) if and only if there is no distribution in $\mc N_\delta^\Pi$ under which (\ref{e:mod}) holds at $(\theta,\gamma,P)$.
\end{proposition}

\begin{remark}\label{rmk:exchangeable_moments}  
If $F$ is $\Pi$-invariant and satisfies (\ref{e:mod}), then it must also satisfy (\ref{e:mod}) under all $|\Pi|$ transformations of the elements of $U$. Therefore, in effect there are a total of $|\Pi| \times d$ moment conditions imposed in the inner optimization, namely
\begin{equation} \label{eq:mom-exchangeable}
 \begin{aligned}
 \mb E^F[ g_1(\varpi U,\theta,\gamma_0)] & \leq P_{10} , & 
 \mb E^F[ g_2(\varpi U,\theta,\gamma_0)] & = P_{20} ,   \\
 \mb E^F[ g_3(\varpi U,\theta,\gamma_0)] & \leq 0 , & 
 \mb E^F[ g_4(\varpi U,\theta,\gamma_0)] & = 0 , &  
\end{aligned} \quad \mbox{for all } \varpi \in \Pi\,.
\end{equation}
In principle one could form a criterion by including all $|\Pi| \times d$ moments. By $\Pi$-invariance of $F_*$ and convexity of the objective, the multipliers on the moments $g(\varpi U, \theta,\gamma)$ will be identical across all $\varpi \in \Pi$. It therefore suffices to form the criterion using only the $d$ averaged moments $g^\Pi$ rather than the full set of $|\Pi| \times d$ moments, thereby reducing the dimension of the inner optimization by a factor of $|\Pi|$.
\end{remark}

\begin{remark}\label{rmk:concatenate}
When Monte Carlo integration is used to compute expectations, taking a sample from $F_*$ and then concatenating the sample across each of its $|\Pi|$ transformations ensures the empirical distribution of the random draws is $\Pi$-invariant.
\end{remark}

\subsection{Conditional Moment Models}

Consider the conditional moment model
\begin{equation}
\label{e:mod-conditional}
\begin{aligned}
 \mb E^F[ g_1(U,X,\theta,\gamma_0)|X = x] & \leq P_{10,x} , & 
 \mb E^F[ g_2(U,X,\theta,\gamma_0)|X = x] & = P_{20,x} ,   \\
 \mb E^F[ g_3(U,X,\theta,\gamma_0)|X = x] & \leq 0 , & 
 \mb E^F[ g_4(U,X,\theta,\gamma_0)|X = x] & = 0 , &  
\end{aligned}
\mbox{for all $x \in \mc X$}
\end{equation}
where $\mathcal X$ is a finite set, and a counterfactual\footnote{Note $\kappa$ can be the expected value at a particular $x_0$ if $k(U,x,\theta,\gamma_0) = 0$ for $x \neq x_0$. More generally, $\kappa$ can be a weighted average by incorporating the weighting into the definition of $k(u,x,\theta,\gamma_0)$.}
\begin{equation} \label{e:kappa-conditional}
 \kappa = \sum_{x \in \mc X} \mb E^F[k(U,X,\theta,\gamma_0)|X = x] \,.
\end{equation}
Suppose the researcher assumes $U|X = x \sim F_*$ for each $x$. We wish to relax this assumption and allow each conditional distribution of $U$ given $X = x$, say $F_x$, to vary in a neighborhood $\mc N_{\delta_x}$ of $F_*$. In doing so, we are allowing the conditional distributions $F_x$ to vary with $x$, and therefore relaxing independence of $U$ and $X$.\footnote{The case with $U$ independent of $X$ is subsumed in (\ref{e:mod}) by stacking the moment functions and reduced-form parameters by values of the conditioning variable, as in Examples~\ref{ex:choice-welfare}--\ref{ex:DDC}.} 

We assume each $\mc N_\delta$ is defined by the same $\phi$ to simplify the exposition, but we allow the neighborhood size to vary with $x$.  Let $\boldsymbol \delta = (\delta_x)_{x \in \mc X}$. We are interested in
\begin{equation}
 \ul \kappa_{\boldsymbol \delta}   : = \inf_{\theta \in \Theta,(F_x \in \mc N_{\delta_x})_{x \in \mc X}} \sum_{x \in \mc X} \mb E^{F_x}[k(U,x,\theta,\gamma_0)] \quad \text{ subject to (\ref{e:mod-conditional})} , \label{e:kappa_lower-conditional}
\end{equation}
and $\ol \kappa_{\boldsymbol \delta}$ defined as the analogous supremum. One may write $\ul \kappa_{\boldsymbol \delta}$ and $\ol \kappa_{\boldsymbol \delta}$ as the value of two optimization problems where $\ul K_{\boldsymbol \delta}(\theta;\gamma_0,P_0)$ and $\ol K_{\boldsymbol \delta}(\theta;\gamma_0,P_0)$ are optimized with respect to $\theta$. Let $P = (P_x)_{x \in \mc X}$ where $P_{x} = (P_{1,x},P_{2,x})$ is partitioned conformably with $g_1$ and $g_2$. For a generic $(\theta,\gamma,P)$, define
\[
 \ul K_{\boldsymbol \delta}(\theta;\gamma,P) = \inf_{(F_x \in \mc N_{\delta_x})_{x \in \mc X}} \sum_{x \in \mc X} \mb E^{F_x}[k(U,x,\theta,\gamma_0)]  \quad \mbox{subject to (\ref{e:mod-conditional}) holding at $(\theta,\gamma,P)$} ,
\]
and define $\ol K_{\boldsymbol \delta}(\theta;\gamma,P)$ as the analogous supremum. These criterion functions have dual forms analogous to Proposition~\ref{prop:criterion}. Let $g(\cdot,x,\theta,\gamma) = (g_1(\cdot,x,\theta,\gamma),\ldots,g_4(\cdot,x,\theta,\gamma))$. Recall $d = \sum_{i=1}^4 d_i$ where $d_i$ is the dimension of $g_i$, and  $\Lambda = \mb R^{d_1}_+ \times \mb R^{d_2} \times \mb R^{d_3}_+ \times \mb R^{d_4}$. Let $\lambda_{12,x}$ denote the first $d_1 + d_2$ elements of $\lambda_x \in \Lambda$.

\begin{assumptionp}{\textPhi-conditional}\label{a:phi-conditional}
(i)\hskip\labelsep  $\phi \in \Phi_0$.
\begin{enumerate}[topsep=-18pt,itemsep=0pt,parsep=0pt,partopsep=0pt]
\item[(ii)] $k(\,\cdot\,, x, \theta,\gamma)$ and each entry of $g(\,\cdot\,, x, \theta,\gamma)$ belong to $\mc E$ for each $(\theta,\gamma,x) \in \Theta \times \Gamma\times \mc X$.
\end{enumerate}
\end{assumptionp}

\begin{proposition} \label{prop:criterion-conditional} 
Suppose that Assumption~\ref{a:phi-conditional} holds. Then 
\begin{align}
 & \ul K_{\boldsymbol \delta}(\theta;\gamma,P) \label{e:dual:1-conditional} \\
 & = \sup_{(\eta_x > 0, \zeta_x \in \mb R, \lambda_x \in \Lambda)_{x \in \mc X}} \sum_{x \in \mc X} \left( -\eta_x \E^{F_*}\left[ {\textstyle \phi^\star \left(\frac{k(U,x,\theta,\gamma) + \zeta_x + \lambda_x' g(U,x,\theta,\gamma) }{-\eta_x} \right) } \right] - \eta_x \delta_x - \zeta_x - \lambda_{12,x}'P_x \right) \notag ,\\
 & \ol K_{\boldsymbol \delta}(\theta;\gamma,P) \label{e:dual:2-conditional} \\
 & = \inf_{(\eta_x > 0, \zeta_x \in \mb R, \lambda_x \in \Lambda)_{x \in \mc X}} \sum_{x \in \mc X} \left( \eta_x \E^{F_*}\left[ {\textstyle \phi^\star \left( \frac{k(U,x,\theta,\gamma) - \zeta_x- \lambda_x' g(U,x,\theta,\gamma)}{\eta_x} \right) } \right] + \eta_x \delta_x + \zeta_x + \lambda_{12,x}'P_x \right) . \notag 
\end{align}
Moreover, the value of (\ref{e:dual:1-conditional}) is $+\infty$ (equivalently, the value of (\ref{e:dual:2-conditional}) is $-\infty$) if and only if for some $x \in \mc X$ there is no distribution in $\mc N_{\delta_x}$ under which (\ref{e:mod-conditional}) holds at $(\theta,\gamma,P)$.
\end{proposition}

As before, estimators $\hat{\ul \kappa}_{\boldsymbol \delta}$ and $\hat{\ol \kappa}_{\boldsymbol \delta}$ of $\ul \kappa_{\boldsymbol \delta}$ and $\ol \kappa_{\boldsymbol \delta}$ are computed by optimizing sample criterions with respect to $\theta$. Let $\hat P = (\hat P_x)_{x \in \mc X}$. The sample criterions are
\[
 \hat{\ul K}_{\boldsymbol \delta}(\theta) = \left[ \begin{array}{l}
 \ul K_{\boldsymbol \delta} (\theta; \hat \gamma, \hat P) \\[2pt]
 + \infty 
 \end{array} , \right. 
 \quad 
 \hat{\ol K}_{\boldsymbol \delta}(\theta) = \left[ \begin{array}{ll}
 \ol K_{\boldsymbol \delta} (\theta;\hat \gamma,\hat P) & \mbox{if $\Delta_x(\theta;\hat \gamma,\hat P_x) < \delta_x$ for each $x \in \mc X$,} \\[2pt]
 - \infty & \mbox{if $\Delta_x(\theta;\hat \gamma,\hat P_x) \geq \delta_x$ for some $x \in \mc X$,}
 \end{array} \right.
\]
where $\ol K_{\boldsymbol \delta} (\theta; \hat \gamma, \hat P)$ and $\ol K_{\boldsymbol \delta} (\theta; \hat \gamma, \hat P)$ denote the programs in Proposition~\ref{prop:criterion-conditional} evaluated at $(\hat \gamma, \hat P)$, and 
\[
 \Delta_x(\theta;\hat \gamma, \hat P_x) = \sup_{\zeta_x \in \mb R,\lambda_x \in \Lambda} - \E^{F_*}\Big[ \phi^\star(-\zeta_x - \lambda_x' g(U,x,\theta,\hat \gamma))  \Big] - \zeta_x - \lambda_{12,x}'\hat P_x .
\]

\subsection{Non-separable Models}

Consider the model
\begin{equation}
\label{e:mod-nonsep-0}
\begin{aligned}
 \mb E^H[\tilde g_1(U,X,\theta,\tilde \gamma_0)] & \leq P_{10} , & 
 \mb E^H[\tilde  g_2(U,X,\theta,\tilde \gamma_0)] & = P_{20} ,   \\
 \mb E^H[\tilde  g_3(U,X,\theta,\tilde \gamma_0)] & \leq 0 , & 
 \mb E^H[\tilde  g_4(U,X,\theta,\tilde \gamma_0)] & = 0 , &  
\end{aligned}
\end{equation}
and counterfactual 
\begin{equation} \label{e:kappa-nonsep-0}
 \kappa = \mb E^H[\tilde k(U,X,\theta,\tilde \gamma_0)] \,,
\end{equation}
where the expectation is with respect to the distribution $H$ of $(U,X)$ and $X$ takes values in a finite set $\mc X$. 
Suppose the researcher assumes $U|X = x \sim F_*$ for each $x$. We wish to relax this assumption and allow the conditional distribution of $U$ given $X = x$, say $F_x$, to vary in a neighborhood $\mc N_{\delta_x}$ of $F_*$.  

Write $H(u,x) = q_{0,x} \cdot F_x(u)$ where $q_{0,x} = \Pr(X = x)$. The vector $q_0 = (q_{0,x})_{x \in \mc X}$ can be consistently estimated from data on $X$. Let $\gamma_0 = (\tilde \gamma_0, q_0)$. Define $g_1(U,x,\theta,\gamma_0) = q_{0,x} \cdot \tilde g_1(U,x,\theta,\tilde \gamma_0)$ and similarly for $g_2$, $g_3$, $g_4$, and $k$. The model (\ref{e:mod-nonsep-0}) and counterfactual (\ref{e:kappa-nonsep-0}) can then be written as
\begin{equation}
\label{e:mod-nonsep}
\begin{aligned}
 \sum_x \mb E^{F_x}[g_1(U,x,\theta,\gamma_0)] & \leq P_{10} , & 
 \sum_x \mb E^{F_x}[g_2(U,x,\theta,\gamma_0)] & = P_{20} ,   \\
 \sum_x \mb E^{F_x}[g_3(U,x,\theta,\gamma_0)] & \leq 0 , & 
 \sum_x \mb E^{F_x}[g_4(U,x,\theta,\gamma_0)] & = 0 ,  &  
\end{aligned}
\end{equation}
and $\kappa = \sum_x \mb E^{F_x}[k(U,x,\theta,\gamma_0)]$. 
We again assume each $\mc N_\delta$ is defined by the same $\phi$ function, but allow the neighborhood size to vary with $x$. Let $\boldsymbol \delta = (\delta_x)_{x \in \mc X}$. We are interested in 
\[
 \ul \kappa_{\boldsymbol \delta}   : = \inf_{\theta \in \Theta,(F_x \in \mc N_{\delta_x})_{x \in \mc X}} \sum_{x} \mb E^{F_x}[k(U,x,\theta,\gamma_0)] \quad \text{ subject to (\ref{e:mod-nonsep})} , 
\]
and $\ol \kappa_{\boldsymbol \delta}$  defined as the analogous supremum. One may write $\ul \kappa_{\boldsymbol \delta}$ and $\ol \kappa_{\boldsymbol \delta}$ as the value of two optimization problems where criterion functions $\ul K_{\boldsymbol \delta}(\theta;\gamma_0,P_0)$ and $\ol K_{\boldsymbol \delta}(\theta;\gamma_0,P_0)$ are optimized with respect to $\theta$. For a generic $(\theta,\gamma,P)$, define
\[
 \ul K_{\boldsymbol \delta}(\theta;\gamma,P) = \inf_{(F_x \in \mc N_{\delta_x})_{x \in \mc X}} \sum_{x} \mb E^{F_x}[k(U,x,\theta,\gamma_0)]  \quad \mbox{s.t. (\ref{e:mod-nonsep}) holding at $(\theta,\gamma,P)$} \,, 
\]
and define $\ol K_{\boldsymbol \delta}(\theta;\gamma,P)$ as the analogous supremum. These criterion functions have dual forms analogous to Proposition~\ref{prop:criterion}. Let $g(\cdot,x,\theta,\gamma) = (g_1(\cdot,x,\theta,\gamma),\ldots,g_4(\cdot,x,\theta,\gamma))$. The remaining notation the same as Proposition~\ref{prop:criterion}.

\begin{proposition} \label{prop:criterion-nonsep} 
Suppose that Assumption~\ref{a:phi-conditional} holds. Then 
\begin{align}
 & \ul K_{\boldsymbol \delta}(\theta;\gamma,P) \label{e:dual:1-nonsep} \\
 & = \sup_{(\eta_x > 0, \zeta_x \in \mb R)_{x \in \mc X},\lambda \in \Lambda} \sum_x \left( -\eta_x \E^{F_*}\left[ {\textstyle \phi^\star \left(\frac{k(U,x,\theta,\gamma) + \zeta_x + \lambda' g(U,x,\theta,\gamma) }{-\eta_x} \right) } \right] - \eta_x \delta_x - \zeta_x - \lambda_{12}'P \right) \notag \,,\\
 & \ol K_{\boldsymbol \delta}(\theta;\gamma,P) \label{e:dual:2-nonsep}  \\
 & = \inf_{(\eta_x > 0, \zeta_x \in \mb R)_{x \in \mc X}, \lambda \in \Lambda} \sum_x \left( \eta_x \E^{F_*}\left[ {\textstyle \phi^\star \left( \frac{k(U,x,\theta,\gamma) - \zeta_x- \lambda' g(U,x,\theta,\gamma)}{\eta_x} \right) } \right] + \eta_x \delta_x + \zeta_x + \lambda_{12}'P \right) \,. \notag
\end{align}
Moreover, the value of (\ref{e:dual:1-nonsep}) is $+\infty$ (equivalently, the value of (\ref{e:dual:2-nonsep}) is $-\infty$) if and only if there is no  $H(u,x) = q_{0,x} \cdot F_x(u)$ with $F_x \in \mc N_{\delta_x}$ under which (\ref{e:mod-nonsep-0}) holds at $(\theta,\gamma,P)$.
\end{proposition}

As before, estimators $\hat{\ul \kappa}_{\boldsymbol \delta}$ and $\hat{\ol \kappa}_{\boldsymbol \delta}$ of $\ul \kappa_{\boldsymbol \delta}$ and $\ol \kappa_{\boldsymbol \delta}$ are computed by optimizing sample criterion functions with respect to $\theta$. The sample criterion functions are
\[
 \hat{\ul K}_{\boldsymbol \delta}(\theta) = \left[ \begin{array}{l}
 \ul K_{\boldsymbol \delta} (\theta; \hat \gamma, \hat P) \\[2pt]
 + \infty 
 \end{array} \right. ,
 \quad 
 \hat{\ol K}_{\boldsymbol \delta}(\theta) = \left[ \begin{array}{ll}
 \ol K_{\boldsymbol \delta} (\theta;\hat \gamma,\hat P) & \mbox{if $\Delta_{nonsep}(\theta;\hat \gamma,\hat P) < 0$} \\[2pt]
 - \infty & \mbox{if $\Delta_{nonsep}(\theta;\hat \gamma,\hat P) \geq 0$,}
 \end{array} \right.
\]
where $\ol K_{\boldsymbol \delta} (\theta; \hat \gamma, \hat P)$ and $\ol K_{\boldsymbol \delta} (\theta; \hat \gamma, \hat P)$ denote the programs in Proposition~\ref{prop:criterion-nonsep} evaluated at $(\hat \gamma, \hat P)$ with $\hat \gamma = (\hat{\tilde \gamma}, \hat q)$ for estimators $\hat{\tilde \gamma}$ of $\tilde \gamma$ and $\hat q$ of $q_0$, and
\begin{multline*}
  \Delta_{nonsep}(\theta;\gamma,P) \\
  = \sup_{\substack{(\eta_x \geq 0 ,\zeta_x \in \mb R)_{x \in \mc X},\lambda \in \Lambda \\ \sum_{x \in \mc X} \eta_x \leq 1}} \left(- \sum_{x \in \mc X} \E^{F_*}\Big[ (\eta_x \phi)^\star(-\zeta_x - \lambda_x' g(U,x,\theta,\gamma))  \Big] - \eta_x \delta_x - \zeta_x \right) - \lambda_{12}'P \,.
\end{multline*}
By similar arguments to Appendix~\ref{ax:Delta}, $\Delta_{nonsep}(\theta;\gamma,P)$ may be shown to be the dual of
\[
 \inf_{ t \in \mb R,(F_x)_{x \in \mc X}} t \quad \mbox{s.t.} \quad D_\phi(F_x\|F_*) \leq \delta_x + t \mbox{ for each $x \in \mc X$ and (\ref{e:mod-nonsep}) holding at $(\theta,\gamma,P)$.}
\]
Therefore, if there exists $F_x$ with $D_\phi(F_x\|F_*) < \delta_x$ for each $x$ such that (\ref{e:mod-nonsep}) holds at $(\theta,\gamma,P)$, then $\Delta_{nonsep}(\theta;\gamma,P) < 0$.

\section{Additional Results on Nonparametric Bounds}\label{ax:sharp}

This appendix presents further details to supplement Section~\ref{s:sharp}. Proofs of all results in this appendix are presented in Appendix~\ref{ax:sharp_proofs}. Our first result concerns the behavior of $\ul \kappa_\delta$ and $\ol \kappa_\delta$ as the neighborhood size $\delta$ becomes large. 
Recall $\mc N_\infty = \{ F : D_\phi(F \|F_*) < \infty\}$. Let
\[
 \mc K_\infty = \{ \mb E^{F}[k(U,\theta,\gamma_0)] : \mbox{(\ref{e:mod}) holds at $(\theta,\gamma_0,P_0)$ for some $\theta \in \Theta$, $F \in \mc N_\infty$}\} \,.
\]

\begin{lemma} \label{lem:limits}
Suppose that Assumption~\ref{a:phi} holds. Then 
\begin{align*}
 \lim_{\delta \to \infty} \ul \kappa_\delta & = \inf \mc K_\infty \,, &
 \lim_{\delta \to \infty} \ol \kappa_\delta & = \sup \mc K_\infty \,.
\end{align*}
\end{lemma}

Next, we characterize bounds on $\mc K_\infty$ using profiled optimization problems and derive their dual forms. Define
\begin{equation}
 \ul K_\infty(\theta;\gamma_0, P_0) = \inf_{F \in \mc N_\infty} \mb E^F[k(U,\theta,\gamma_0)] \quad \mbox{subject to (\ref{e:mod}) holding at $(\theta,F)$} \,, \label{e:crit_inf_l}
\end{equation}
and let $\ol K_\infty(\theta;\gamma_0, P_0)$ denote the analogous supremum. By definition, we have
\[
\begin{aligned}
 \inf \mc K_\infty & = \inf_{\theta \in \Theta} \ul K_\infty(\theta;\gamma_0,P_0) \,, & & & 
 \sup \mc K_\infty & = \sup_{\theta \in \Theta} \ol K_\infty(\theta;\gamma_0,P_0) \,.
\end{aligned}
\]
Let $F_*\text{-}\mr{ess}\inf$ and $F_*\text{-}\mr{ess}\sup$ denote the $F_*$-essential infimum and supremum, respectively.

\begin{lemma}\label{lem:dual:infty}
Suppose that Assumption~\ref{a:phi} holds and Condition S holds at $(\theta,\gamma,P)$. Then
\begin{align*}
 \ul K_\infty (\theta;\gamma,P) & = \sup_{\lambda \in \Lambda : F_*\text{-}\mr{ess}\inf (  k(\cdot,\theta,\gamma) + \lambda'g(\cdot,\theta,\gamma) ) > - \infty} \left( F_*\text{-}\mr{ess}\inf (  k(\cdot,\theta,\gamma) + \lambda'g(\cdot,\theta,\gamma) ) - \lambda_{12}'P \right) , \\
 \ol K_\infty (\theta;\gamma,P) & = \inf_{\lambda \in \Lambda : F_*\text{-}\mr{ess}\sup (k(\cdot,\theta,\gamma) - \lambda' g(\cdot,\theta,\gamma)) < +\infty} \left(  F_*\text{-}\mr{ess}\sup (k(\cdot,\theta,\gamma) - \lambda' g(\cdot,\theta,\gamma)) + \lambda_{12}'P \right) .
\end{align*}
\end{lemma}

We now derive analogous dual representations for the criterion functions $\ul K_{np}$ and $\ol K_{np}$ from Section~\ref{s:sharp} (see display (\ref{e:crit_np_l})).
We require a slightly different constraint qualification:

\begin{definition}\label{cond:snp}
\emph{Condition S$_{np}$} holds at $(\theta,\gamma,P)$ if $\vec P \in \mr{ri}(\{ \E^{F} [ g(U,\theta,\gamma) ] :  F \in \mc F_\theta \}  + \mc C) $.
\end{definition}

If $F_*$ and $\mu$ are mutually absolutely continuous, then Condition S$_{np}$ is equivalent to Condition S from Section~\ref{s:sharp}  (see Lemma~\ref{lem:cq_np}).

\begin{lemma}\label{lem:dual:sharp}
Suppose that Condition S$_{np}$ holds at $(\theta,\gamma,P)$ and $k$ is $\mu$-essentially bounded. Then
\begin{align*}
 \ul K_{np} (\theta;\gamma,P) & = \sup_{\lambda \in \Lambda : \mu\text{-}\mr{ess}\inf (  k(\cdot,\theta,\gamma) + \lambda'g(\cdot,\theta,\gamma) ) > - \infty} \left( \mu\text{-}\mr{ess}\inf (  k(\cdot,\theta,\gamma) + \lambda'g(\cdot,\theta,\gamma) ) - \lambda_{12}'P \right) , \\
 \ol K_{np} (\theta;\gamma,P) & = \inf_{\lambda \in \Lambda : \mu\text{-}\mr{ess}\sup (k(\cdot,\theta,\gamma) - \lambda' g(\cdot,\theta,\gamma)) < +\infty} \left(  \mu\text{-}\mr{ess}\sup (k(\cdot,\theta,\gamma) - \lambda' g(\cdot,\theta,\gamma)) + \lambda_{12}'P \right) .
\end{align*}
\end{lemma}

\section{Local Sensitivity}\label{ax:local}

In this appendix, we first introduce a measure of local sensitivity of the counterfactual with respect to $F$. We then contrast our approach with recent work on local sensitivity.

\subsection{Measure of Local Sensitivity}

Our \emph{measure of local sensitivity} of the counterfactual $\kappa$ with respect to $F$ at $F_*$ is
\[
 s =  \lim_{\delta \downarrow 0} \frac{( \ol \kappa_\delta - \ul \kappa_\delta)^2}{4\delta} \,.
\]
If $s$ is finite, then under the regularity conditions below
\[
\begin{aligned}
 \ul \kappa_\delta & = \kappa_* - \sqrt{ \delta s} + o( \sqrt \delta) \,, &
 \ol \kappa_\delta & = \kappa_* + \sqrt{ \delta s} + o( \sqrt \delta)  \,, & \mbox{as $\delta \downarrow 0$,}
\end{aligned}
\]
 where $\kappa_* = \mb E^{F_*}[k(U,\theta_*,\gamma_0)]$ and $\theta_*$ solves (\ref{e:mod}) under $F_*$. 

To draw connections with the local sensitivity literature, we restrict attention to moment equality models and impose (standard) regularity conditions. These conditions  allow us to characterize $s$ very tractably via an influence function representation, which leads to a simple estimator $\hat s$ of $s$. Assume that under $F_*$ the moment conditions  (\ref{e:mod:2}) and (\ref{e:mod:4})  point identify a structural parameter $\theta_* \in \mr{int}(\Theta)$, where we again assume $\Theta$ is compact. With some abuse of notation, let
\[
 g(u,\theta,\gamma,P_2) = \left[ \begin{array}{c}
 g_2(u,\theta,\gamma)-P_2 \\
  g_4(u,\theta,\gamma) \end{array} \right] \,,
\]
$g_*(u) = g(u,\theta_*,\gamma_0,P_{20})$, and $k_*(u) = k(u,\theta_*,\gamma_0)$. Let $\mb E^{F_*}[g(U,\theta,\gamma_0,P_{20})]$ and $\mb E^{F_*}[k(U,\theta,\gamma_0)]$ be continuously differentiable with respect to $\theta$ at $\theta_*$, $G = \left. \frac{\partial}{\partial \theta'} \mb E^{F_*}[g(U,\theta,\gamma_0,P_{20})] \right|_{\theta = \theta_*}$ have full rank, $V =  \mb E^{F_*}[g_*(U)g_*(U)']$ be finite and positive definite, $\mb E^{F_*}[k(U,\theta_*,\gamma_0)^2]$ be finite, and $k(\cdot,\theta,\gamma_0)$ and $g(\cdot,\theta,\gamma_0,P_{20})$ be $L^2(F_*)$-continuous in $\theta$ at $\theta_*$. 

Define the \emph{influence function} of the counterfactual $\kappa$ with respect to $F$ at $F_*$ as
\[
 \iota(u) = \mb M k_*(u) - J' (G'V^{-1}G)^{-1}G'V^{-1} g_*(u)  \,,
\]
where $\mb M k_*(u)  = k_*(u) - \kappa_* - \mb E^{F_*}[ k_*(U)g_*(U)'](V^{-1} - V^{-1}G(G'V^{-1}G)^{-1}G'V^{-1})g_*(u)$ and $J =  \left. \frac{\partial}{\partial \theta} \mb E^{F_*}[k(U,\theta,\gamma_0)] \right|_{\theta = \theta_*}$. 
The following theorem relates $s$ and $\iota$. We restrict attention to neighborhoods characterized by $\chi^2$ divergence. Other $\phi$-divergences are locally equivalent to $\chi^2$ divergence, so this restriction entails no great loss of generality.\footnote{See Theorem 4.1 of \cite{CsiszarShields}. The quantity $2 \mb E^{F_*}[\iota(U)^2]$ should be rescaled by a factor of $\phi''(1)$ for other $\phi$ divergences.}

\begin{theorem}\label{t:local}
Suppose that the above GMM-type regularity conditions hold and neighborhoods are defined using $\chi^2$ divergence. Then $s = 2 \mb E^{F_*}[\iota(U)^2]$.
\end{theorem}

The proof of Theorem~\ref{t:local} is presented in Appendix~\ref{ax:local_proofs}.
In addition to reporting an estimated counterfactual $\hat \kappa = \E^{F*}[k(U,\hat \theta, \hat \gamma)]$, researchers could also report an estimate of its local sensitivity to $F$:
\[
 \hat s = 2 \mb E^{F_*}[(\hat k(U)-\hat \kappa)^2] + 2 \hat Q'\hat V \hat Q - 4 \mb E^{F_*}[\hat g(U)(\hat k(U)-\hat \kappa)]'\hat Q \,,
\]
where $\hat k(u) = k(u,\hat \theta,\hat \gamma)$, $\hat g(u) = g(u,\hat \theta,\hat \gamma,\hat P_2)$, $\hat V = \mb E^{F_*}[\hat g(U)\hat g(U)']$, and 
\[
 \hat Q' = \mb E^{F_*}[ \hat k(U)\hat g(U)'](\hat V^{-1} - \hat V^{-1}\hat G(\hat G'\hat V^{-1}\hat G)^{-1}\hat G'\hat V^{-1}) + \hat J' (\hat G'\hat V^{-1}\hat G)^{-1}\hat G'\hat V^{-1} \,,
\]
with $\hat G =  \frac{\partial}{\partial \theta'} \mb E^{F_*}[g(U, \theta,\hat \gamma,\hat P_2)] |_{\theta = \hat \theta}$ and $\hat J =  \frac{\partial}{\partial \theta} \mb E^{F_*}[k(U, \theta,\hat \gamma)] |_{\theta = \hat \theta}$. Lemma~\ref{lem:local} in Appendix~\ref{ax:local_proofs} shows $\hat s$ is consistent. Bounds on counterfactuals as $F$ varies over small neighborhoods of $F_*$ can then be estimated using $\hat \kappa \pm \sqrt{\delta \hat s}$.

\subsection{Comparison with Other Approaches}

We now compare our approach with \cite{AGS2017,AGS2018}, henceforth AGS, and \cite{BW}, henceforth BW. To simplify the comparison, we consider models characterized by moments of the form (\ref{e:mod:2}) with $d_2 \geq d_\theta$ and in which there is no $\gamma$.

AGS consider a setting in which the moments (\ref{e:mod:2}) are locally misspecified:
\begin{equation}\label{e:local-misspec}
 \E^{F_*}[g_2(U,\theta_*)] = P_{20} + n^{-1/2}c\,.
\end{equation}
Suppose a researcher has a first-stage estimator $\hat P_2$, computes an estimator $\hat \theta$ by minimizing
\[
 (\E^{F_*}[g_2(U,\theta)] - \hat P_2)' \hat W (\E^{F_*}[g_2(U,\theta)] - \hat P_2),
\]
then estimates the counterfactual using $\hat \kappa = \E^{F_*}[k(U,\hat \theta)]$. AGS's measure of \emph{sensitivity} of $\hat \kappa$ to $\hat P_2$ is $J'(G'WG)^{-1}G'W$, where $W$ is the probability limit of $\hat W$. The first-order asymptotic bias of $\hat \kappa$ due to local misspecification is therefore $J'(G'WG)^{-1}G'Wc$. AGS's measure of \emph{informativeness} of $\hat P_2$ for $\hat \kappa$ is $1$, meaning that all sampling variation in $\hat \kappa$ is explained by sampling variation in $\hat P_2$. Our measure $s$ instead characterizes ``specification variation'' in $\kappa$ as the researcher varies $F$ subject to the moment condition (\ref{e:mod:2}). 

BW consider estimation of a target parameter ($\kappa$ in our context) using a reference model $\mc M_R = \{(\theta,F) \in \Theta \times \{F_*\} \}$ when the true $(\theta_0,F_0)$ possibly belongs to a larger model $\mc M_L =  \{(\theta,F) \in \Theta \times \mc N_\delta \}$ with $\delta \downarrow 0$ as the sample size $n$ increases so that $n \delta \to \tau \geq 0$. 
BW seek estimators of $\kappa$ under $\mc M_R$ that minimize worst-case asymptotic bias or MSE over $\mc M_L$. Consider the one-step estimator
\[
 \hat \kappa = \E^{F_*}[k(U,\hat \theta)] + a'(\E^{F_*}[g_2(U,\hat \theta)] - \hat P_2) \,,
\]
where $\hat \theta$ is a $\sqrt n$-consistent estimator of $\theta_*$ and $a \in \mb R^{d_2}$ satisfies $J' = -a' G$ so that $\hat \kappa$ does not depend asymptotically on $\hat \theta$. The true counterfactual is $\kappa_0 = \E^{F_0}[k(U,\theta_0)]$ where $(\theta_0,F_0) \in \mc M_L$ satisfies $\E^{F_0}[g_2(U,\theta_0)] = P_{20}$. 
If $\mc M_R$ is correctly specified so that $\E^{F_*}[g_2(U,\theta_*)] = P_{20}$, then for any $a$ the worst-case asymptotic bias of the one-step estimator is 
\[
 \lim_{n \to \infty} \sup_{(\theta_0,F_0) \in \mc M_L:  \E^{F_0}[g_2(U,\theta_0)] = P_{20}} |\sqrt n(\kappa_* - \kappa_0))| = \sqrt{\tau s} \,,
\]
where $s$ is our measure of local sensitivity. 

 If we allow for local misspecification of $\mc M_R$, so that $\E^{F_*}[g_2(U,\theta_*)] \neq P_{20}$, then the worst-case asymptotic bias of the one-step estimator is
\[
 \lim_{n \to \infty} \sup_{(\theta_0,F_0) \in \mc M_L:  \E^{F_0}[g_2(U,\theta_0)] = P_{20}} \left|\sqrt n\left(\kappa_* - \kappa_0 + a'\left( \E^{F_*}[g_2(U,\theta_*)] - P_{20} \right)\right)\right| = \sqrt{\tau s_a} \,,
\]
where $s_a$ is our local sensitivity measure with $k$ replaced by $k + a'g_2$.

\section{Additional Details for the Empirical Applications}\label{ax:empirical}

\subsection{Marital College Premium}\label{ax:csw}

\paragraph*{Bootstrap Details.}
CSs reported Section~\ref{s:csw} with $\delta > 0$ are computed using the bootstrap procedure from Section~\ref{s:bootstrap}. To implement the bootstrap, we take 1,000 independent draws of $\hat P_2^* \sim N(\hat P_2,\hat \Sigma)$ where $\hat \Sigma$ is CSW's estimate of the covariance matrix of $\hat P_2$. We compute $\hat P_2$ and $\hat \Sigma$ based on CSW's replication files.

\medskip

\paragraph*{Fixed-$\theta$ Bounds.}

Figure~\ref{fig:csw_3_inner} plots lower and upper bounds on the ``some college'' to ``college graduate'' premium across cohorts when $\theta$ is fixed at CSW's  estimates (computed under $F_*$) but $F$ is allowed to vary. These bounds for large $\delta$ contain zero across each cohort,  and are approximately the same width as the bounds with $\delta = 0.01$ reported in Figure~\ref{fig:csw_3_short} where both $\theta$ and $F$ are allowed to vary. Imposing exchangeability (Figure~\ref{fig:csw_3_inner_exch}) is seen to tighten the bounds substantially, producing bounds that span negative values only for early cohorts and positive values only in the latest few cohorts.

\begin{figure}
\makebox[\textwidth]{
\begin{subfigure}{.49\textwidth}
  \centering
  \caption{Without exchangeability} \label{fig:csw_3_inner}
  \includegraphics[width=\linewidth]{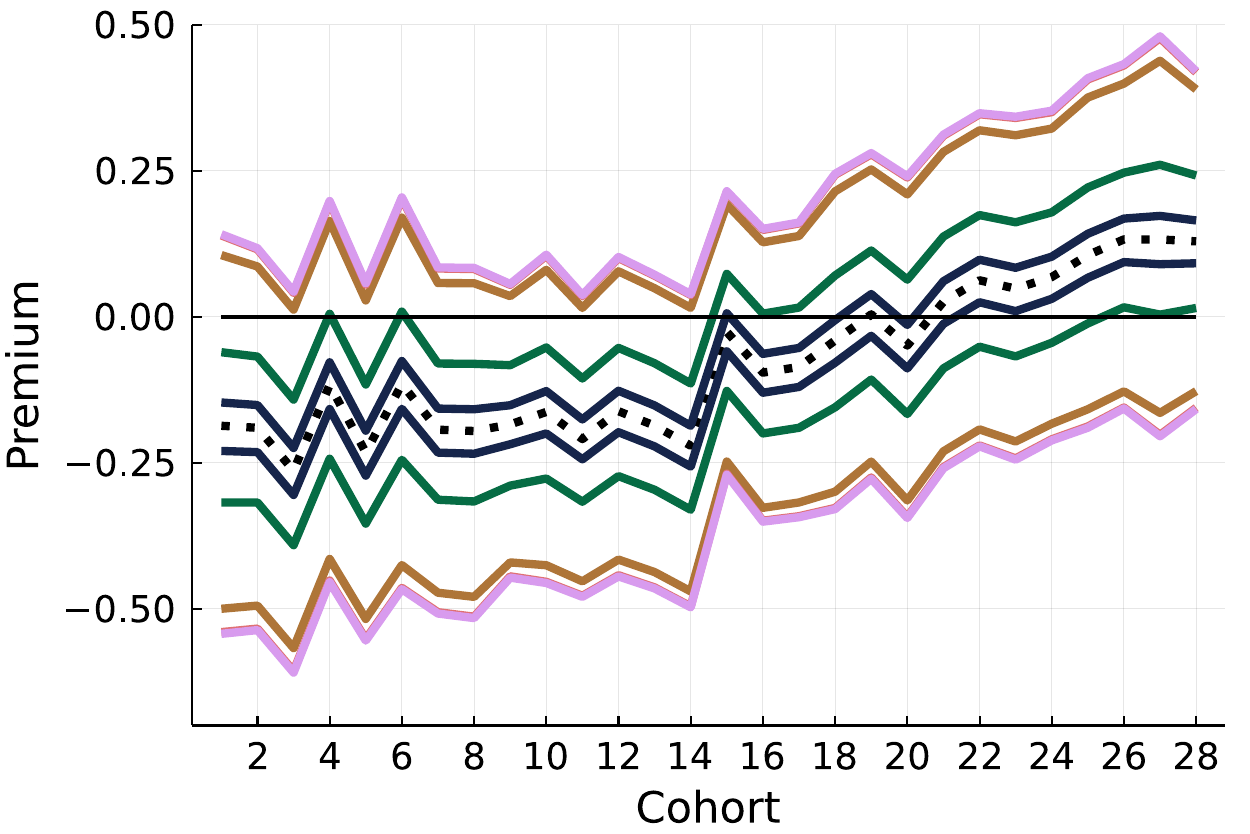}
  \end{subfigure}%
\begin{subfigure}{.49\textwidth}
  \centering
  \caption{With exchangeability} \label{fig:csw_3_inner_exch}
  \includegraphics[width=\linewidth]{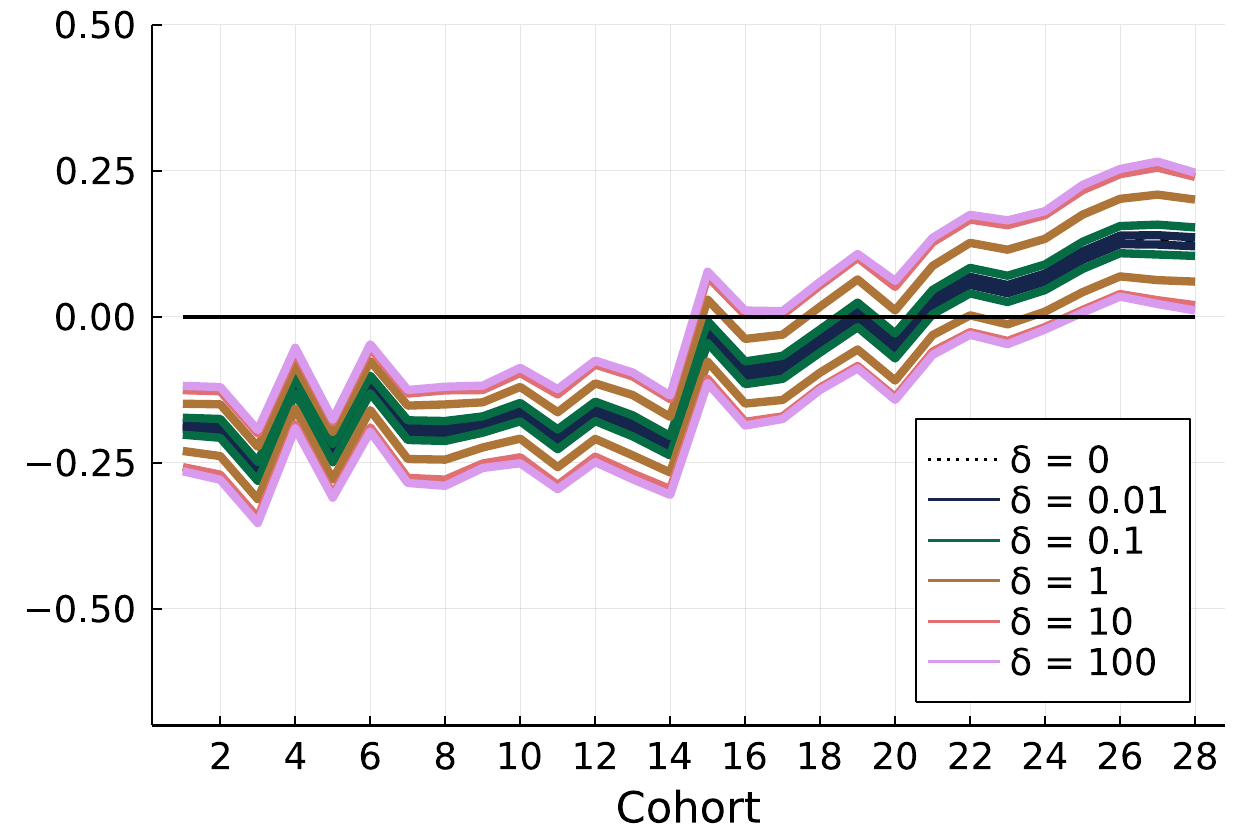}
\end{subfigure}}
\centering

\caption{\label{fig:csw_3_inner_both}Fixed-$\theta$ bounds on the ``some college'' to ``college graduate'' premium when structural parameters are held fixed at CSW's estimates.}

\end{figure}

\medskip

\paragraph*{Projection CSs.}

Figure~\ref{fig:csw_3_proj} reports projection CSs  computed using the procedure in Section~\ref{s:projection}. We formed 95\% rectangular CSs for each cohort's $P_{20}$ as described in Section~\ref{s:projection} using CSW's estimates for $\hat P_2$ and their asymptotic variance estimates for $\hat \Sigma$. These CSs are significantly wider than the bootstrap CSs reported in Figure~\ref{fig:csw_3}. Some conservativeness is to be expected, as these CSs project a 95\% CS for $P_{20}$ down to one dimension. The relative inefficiency is especially pronounced for the earlier cohorts. Note also from Figure~\ref{fig:csw_3_proj_exch} that the projection CSs with $\delta = 0.01$  span zero across each cohort, whereas the bootstrap CSs with $\delta = 0.01$ in Figure~\ref{fig:csw_3_short_exch} contain negative values only in some early cohorts and positive values only in later cohorts.

\begin{figure}[t]
\makebox[\textwidth]{
\begin{subfigure}{.49\textwidth}
  \centering
  \caption{Without exchangeability} \label{fig:csw_3_proj}
  \includegraphics[width=\linewidth]{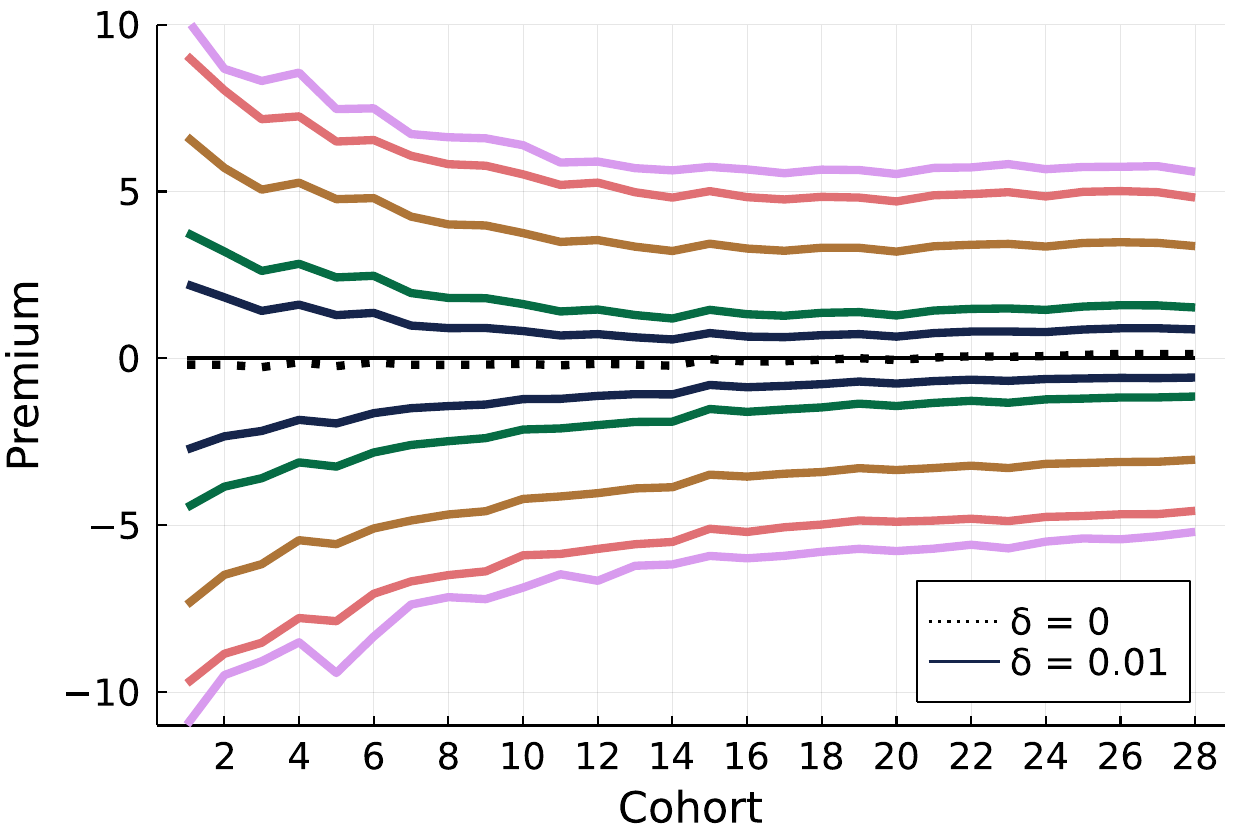}
  \end{subfigure}%
\begin{subfigure}{.49\textwidth}
  \centering
  \caption{With exchangeability} \label{fig:csw_3_proj_exch}
  \includegraphics[width=\linewidth]{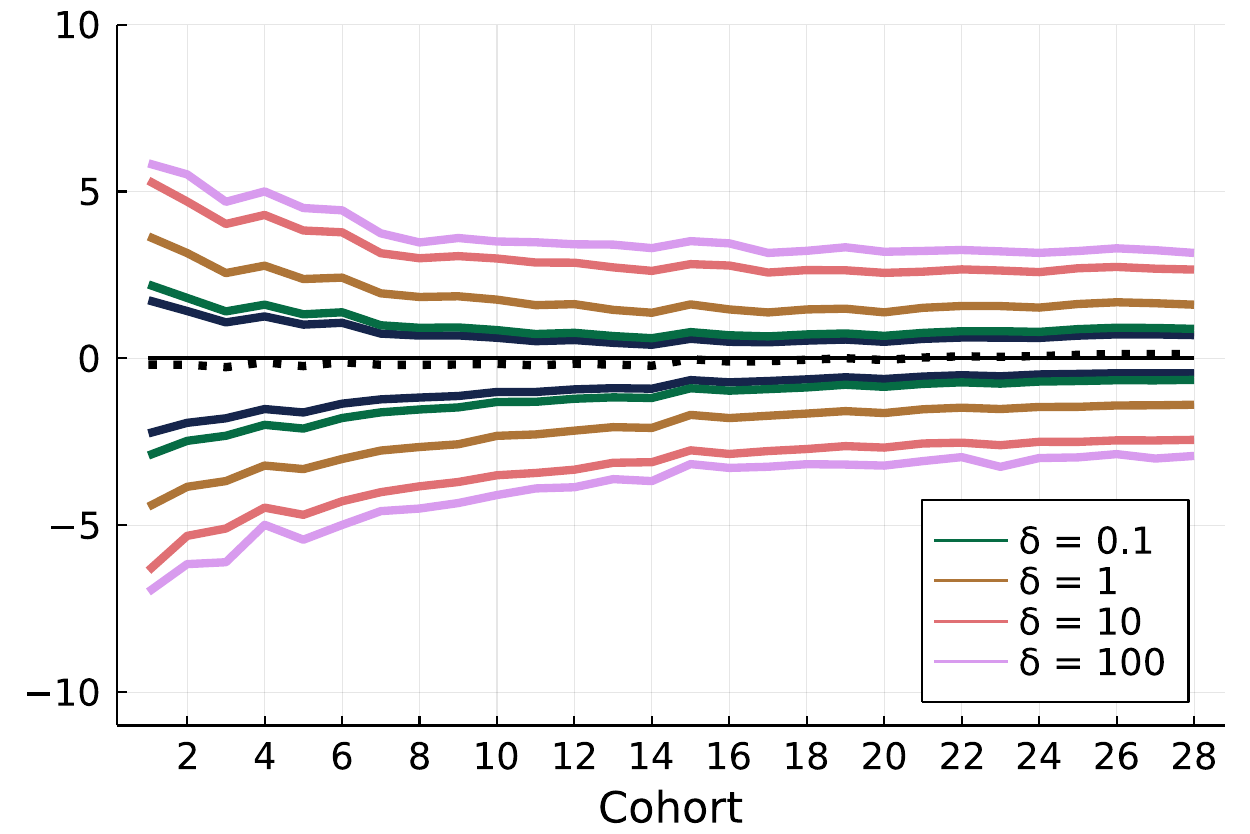}
\end{subfigure}}
\centering

\caption{\label{fig:csw_3_proj}Projection 95\% CSs for bounds on the ``some college'' to ``college graduate'' premium across cohorts. }

\end{figure}

\medskip

\paragraph*{Computation Times.}

Table~\ref{tab:csw_time} reports times for solving the inner problem for maximizing the premium in cohort 1. This optimization problem defines the criterion function $\ol K_\delta(\theta;\hat P)$. As times vary with $\theta$, we report times at CSW's estimates. Times increase somewhat with $\delta$, but are all under 0.6 seconds. The outer optimization times varied with cohort, $\delta$, and implementation but were typically solved in at most a few minutes (often under 90 seconds).

\begin{table}
\begin{center}
\caption{\label{tab:csw_time}Computation times for the inner problem in the matching application}
\begin{tabular}{cccccc} \hline \hline
 \multicolumn{1}{c}{Implementation} & \multicolumn{5}{c}{$\delta$} \\\cline{2-6} \\[-10pt]
 &  0.01 & 0.1 & 1 & 10 & 100 \\[2pt]
 Without exchangeability & 0.074 & 0.056 & 0.076 & 0.579 & 0.188 \\[2pt]
 With exchangeability & 0.146 & 0.184 & 0.350 & 0.311 & 0.488 \\[2pt] \hline \\[-10pt]
\end{tabular}
\parbox{\textwidth}{\small \emph{Note:} Times (in seconds) for solving the inner optimization problem for maximizing the premium in cohort 1 at CSW's parameter estimate $\theta$. We use 50,000 Monte Carlo draws without exchangeability and 120,000 draws with exchangeability. All computations are performed in Julia version 1.6.4 and Knitro 12.4.0 on a 2.7GHz MacBook Pro with 16GB memory.} 
\end{center}
\vskip -14pt
\end{table}

\medskip

\paragraph*{Sensitivity to $\phi$.}

Using $\chi^2$ and $L^4$ divergences produced near identical bounds for $\delta = 0.01$ and $0.1$. The $\chi^2$ bounds with $\delta = 1$ and $10$   were at most 10\% narrower than the hybrid bounds. The  $L^4$ bounds were 60\%-70\% of the width of the hybrid bounds for $\delta = 1$, $10$, and $100$ across cohorts ($L^4$ divergence is stronger than $\chi^2$ and hybrid divergence). The shapes of the sets were also similar to those reported for hybrid divergence. Overall, these results show that the conclusions we draw from our analysis are not sensitive to the choice of $\phi$.

\subsection{Welfare Analysis in a Rust Model}\label{ax:rust}

\paragraph*{Bootstrap Details.}
Bootstrap CSs reported Section~\ref{s:rust} with $\delta > 0$ are computed using the procedure from Section~\ref{s:bootstrap}. We take 1,000 independent draws of $\hat \theta_\pi^* \sim N(\hat \theta_\pi, \hat \Sigma)$ where $\hat \theta_\pi$ is the MLE of $(RC, MC)$ under the i.i.d. Gumbel assumption and $\hat \Sigma$ is an estimate of the inverse information matrix. We then set $\hat P_2^*$ to be the model-implied CCPs at $\hat \theta_\pi^*$ under the i.i.d. Gumbel assumption. 

As $k$ depends only implicitly on $u$ through $\theta$, we compute $\hat{\ul \kappa}_\delta$ and $\hat{\ol \kappa}_\delta$ using the criterion functions in display (\ref{e:k_implicit}), which is more computationally efficient than using criterions (\ref{e:dual:1}) and (\ref{e:dual:2}). The $\lambda$ multipliers on the minimum divergence problem $\Delta(\theta; P)$ in (\ref{e:k_implicit}) differ from $\lambda$ in criterions (\ref{e:dual:1}) and (\ref{e:dual:2}) by the factor $\eta$ (see the discussion in Section~\ref{s:duality}). As our bootstrap methods are derived based on criterions (\ref{e:dual:1}) and (\ref{e:dual:2}), when implementing the bootstrap we rescale the multiplier $\lambda$ solving (\ref{e:qual:dual}) by the multiplier $\eta$ on the constraint $\Delta(\theta;\hat P) \leq \delta$ in the outer optimization.\footnote{This rescaling is also justified as follows. Let $\ol b_\delta(P) =  \sup_{\theta \in \Theta : \Delta(\theta;P) \leq \delta} k(\theta)$ and note $\ol \kappa_\delta = \ol b_\delta(P_0)$ and $\hat{\ol \kappa}_\delta = \ol b_\delta(\hat P)$. By similar arguments to Corollary 5 of \cite{MilgromSegal}, one may deduce that the directional derivative of $\ol b_\delta(P)$ at $P_0$ involves multiplying the directional derivative of $P \mapsto \Delta(\theta;P)$ at $P_0$ by the multiplier for $\Delta(\theta;P) \leq \delta$. The directional derivative of $P \mapsto \Delta(\theta;P)$ at $P_0$ may be shown to be
\[
 \lim_{n \to \infty} t_n^{-1} \left( \Delta(\theta;P_0+t_n h_n)-\Delta(\theta;P_0) \right) = \sup_{\lambda_{12} \in \ul \Lambda(\theta;P_0)} - \lambda_{12}' h,
\]
where $\ul \Lambda(\theta;P_0)$ is constructed analogously to $\ul \Lambda_\delta(\theta;P)$ in Section~\ref{s:bootstrap} using the set of multipliers that solve the minimum divergence problem (\ref{e:qual:dual}).} As $\eta$ and $\lambda$ are computed separately in the outer and inner optimizations, respectively, it is computationally most convenient to implement our bootstrap CSs with $\hat \nu = 0$. As discussed in Section~\ref{s:bootstrap}, this choice is valid but possibly conservative. Despite this potentially conservative choice, the bootstrap CSs are not materially wider than the bootstrap CSs under the i.i.d. Gumbel assumption.

To construct the projection CSs, we form a 95\% rectangular CS for $P_{20}$ as described in Section~\ref{s:projection}. For each draw of $\hat \theta_\pi^*$ we compute the model-implied CCPs $\hat P_2^*$ under the i.i.d. Gumbel assumption. We construct $t$-statistics for each CCP by centering $\hat P_2^*$ at $\hat P_2$ and studentizing by its standard deviation across draws. For each draw we compute the maximum of the absolute value of the $t$-statistics. We then take the critical value $\hat c_{2,1-\alpha}$ to be the $1-\alpha$ quantile of the maximum statistic across draws.

\medskip

\begin{table}
\begin{center}
\caption{\label{tab:rust_time}Computation times for the inner problem in the DDC application}
\begin{tabular}{cccccc} \hline \hline
  & \multicolumn{5}{c}{$\delta$} \\\cline{2-6} \\[-10pt]
 &  0.01 & 0.1 & 1 & 10 & 100 \\[2pt]
 Lower bound & 0.124 & 0.144 & 0.164 & 0.285 & 0.265 \\[2pt]
 Upper bound & 0.101 & 0.119 & 0.142 & 0.266 & 1.039 \\[2pt] \hline \\[-10pt]
\end{tabular}
\parbox{\textwidth}{\small \emph{Note:} Times (in seconds) for solving the inner optimization problem at the parameter values at which  $\hat{\ul \kappa}_\delta$ and $\hat{\ol \kappa}_\delta$ are attained. All computations are performed in Julia version 1.6.4 and Knitro 12.4.0 on a 2.7GHz MacBook Pro with 16GB memory.} 
\end{center}
\vskip-14pt
\end{table}

\paragraph*{Computation times.}

Table~\ref{tab:rust_time} reports computation times for the inner optimization for evaluating the criterion functions $\hat{\ul K}_\delta(\theta;\hat \gamma, \hat P)$ and $\hat{\ol K}_\delta(\theta;\hat \gamma, \hat P)$ at the parameter values at which $\hat{\ul \kappa}_\delta$ and $\hat{\ol \kappa}_\delta$ are attained. The computation times correspond to solving the minimum divergence problem $\Delta(\theta;\hat \gamma, \hat P)$ because $k$  does not depend on $u$ (cf. display (\ref{e:k_implicit})). The outer optimizations were typically solved in a few minutes in a 8-core environment with 64GB memory.

\medskip

\paragraph*{Sensitivity to $\phi$.}

Bounds with $\chi^2$-divergence  were between 4\% narrower and 1\% wider than the bounds for hybrid divergence for all values of $\delta$. Repeating the analysis with $L^4$-divergence, which is stronger than $\chi^2$ and hybrid divergence, produced bounds that were 10-30\% narrower than the hybrid divergence bounds up to $\delta = 1$ and at most 5\% narrower than the hybrid divergence bounds for larger values of $\delta$. As with the matching application, these results again show that the conclusions we draw from our analysis are not sensitive to the choice of $\phi$ function.

\section{Proofs of Main Results}\label{ax:proofs}

Throughout the proofs, we abbreviate upper-semicontinuous and upper-semicontinuity to u.s.c. and lower-semicontinuous and lower-semicontinuity to l.s.c.

\subsection{Proofs for Section~\ref{s:procedure}}\label{s:procedure_proofs}

\begin{proof}[Proof of Proposition~\ref{prop:criterion}]
Immediate from Proposition~\ref{prop:dual} in Appendix~\ref{sec:dual_K}.
\end{proof}

\medskip

Recall Condition S from Definition~\ref{cond:s} and Condition S$_{np}$ from Definition~\ref{cond:snp}.

\begin{lemma}\label{lem:cq_np}
Suppose that Assumption~\ref{a:phi} holds and $\mu$ and $F_*$ are mutually absolutely continuous. Then Condition S holds at $(\theta,\gamma,P)$ if and only if Condition S$_{np}$ holds at $(\theta,\gamma,P)$.
\end{lemma}

\begin{proof}[Proof of Lemma~\ref{lem:cq_np}]
In view of H\"older's inequality for Orlicz classes (see (\ref{eq:holder})), Assumption~\ref{a:phi} implies $\mc N_\infty = \{F : D_\phi(F\|F_*) < \infty\} \subseteq \mc F_\theta$. Therefore, 
\[
 \mc G_\infty := \{\mb E^F[g(U,\theta,\gamma)] : F \in \mc N_\infty\} \subseteq \{\mb E^F[g(U,\theta,\gamma)] : F \in \mc F_\theta\} =: \mc G_\theta\,.
\]
By Corollary 6.6.2 of \cite{Rockafellar}, it suffices to show  $\mr{ri}(\mc G_\infty) = \mr{ri}(\mc G_\theta)$. As $\mr{ri}(\mc G_\infty) \subseteq \mc G_\infty \subseteq \mc G_\theta$, it suffices to show $\mc G_\theta \subseteq \mr{cl}(\mc G_\infty)$ \cite[Remark 2.1.9]{HU-L}. For any $x \in \mc G_\theta$, we have $x = \mb E^F[g(U,\theta,\gamma)]$ for some $F \in \mc F_\theta$. As $F \ll \mu$ and $F_*$ and $\mu$ are mutually absolutely continuous, $F$ has a density, say $m$, with respect to $F_*$. For each $n \geq 1$, let $m (u) \wedge n = \min\{m(u), n\}$ and define
\[
 m_n(u) = \frac{m(u) \wedge n}{\int (m (u) \wedge n) \, \mr d F_*(u)} \,.
\]
Each $F_n$ defined by $\mr d F_n = m_n \mr d F_*$ belongs to $\mc N_\infty$. It follows that $\E^{F_n}[g(U,\theta,\gamma)]  \in \mc G_\infty$. By monotone convergence, we have $\mb E^{F_n}[g(U,\theta,\gamma)] \to x$. Therefore, $x \in \mr{cl}(\mc G_\infty)$.
\end{proof}

\begin{proof}[Proof of Theorem~\ref{t:sharp}]
We prove only the result for $\inf \mc K$; the result for $\sup \mc K$ follows similarly. 
Note
\[
 \inf \mc K = \inf_{\theta \in \Theta} \ul K_{np}(\theta;\gamma_0,P_0) = \inf_{\theta \in \Theta_I} \ul K_{np}(\theta;\gamma_0,P_0) \,,
\] 
where the first equality is by definition and the second equality holds because, if $\theta \not \in \Theta_I$, then there does not exist a distribution $F \in \mc F_\theta$ under which the moment conditions hold at $(\theta,\gamma_0,P_0)$ and consequently $\ul K_{np}(\theta;\gamma_0,P_0) = +\infty$. If $\theta \not \in \Theta_I$, then there does not exist $F \in \mc N_\infty$ under which the moment conditions hold at $(\theta,\gamma_0,P_0)$ either because $\mc N_\infty \subseteq \mc F_\theta$ for all $\theta$ under Assumption~\ref{a:phi}. Therefore, $\ul K_\infty(\theta;\gamma_0,P_0) = +\infty$ in that case too. We therefore have
\[
 \inf \mc K_\infty = \inf_{\theta \in \Theta_I} \ul K_\infty(\theta;\gamma_0,P_0) \,.
\]

In view of Lemma~\ref{lem:limits}, it suffices to show $\inf \mc K = \inf \mc K_\infty$. Note that $\inf \mc K \leq \mc \inf \mc K_\infty$ holds by virtue of the inclusion $\mc N_\infty \subseteq \mc F_\theta$ for all $\theta$. For the reverse inequality, choose any $\epsilon > 0$. By $S$-regularity of $\Theta_I$, there exists $\ul \theta \in \Theta_I$ for which Condition S holds at $(\ul \theta, \gamma_0, P_0)$ and for which $\ul K_{np}(\ul \theta; \gamma_0, P_0) \leq \inf \mc K + \epsilon$. As Condition S holds at $(\ul \theta,\gamma_0,P_0)$ and $\mu \ll F_* \ll \mu$, Lemma~\ref{lem:cq_np} implies that Condition S$_{np}$ must also hold at $(\ul \theta,\gamma_0,P_0)$. Moreover, the $\mu$-essential infimum and $F_*$-essential infimum of any function are equal because $\mu \ll F_* \ll \mu$. Therefore by Lemmas~\ref{lem:dual:infty} and \ref{lem:dual:sharp}, we have  $\ul K_\infty(\ul \theta;\gamma_0,P_0) = \ul K_{np}(\ul \theta;\gamma_0,P_0)$. It follows by definition of $\inf \mc K_\infty$ that $\inf \mc K_\infty \leq \ul K_\infty(\ul \theta; \gamma_0, P_0) = \ul K_{np}(\ul \theta;\gamma_0,P_0) \leq \inf \mc K + \epsilon$. Therefore, $\inf \mc K_\infty \leq \inf \mc K$.
\end{proof}

\subsection{Proofs for Section~\ref{s:implementation}}

\begin{proof}[Proof of Proposition~\ref{p:mpec}]
We prove the result for $\ul K_\delta$; the proof for $\ol K_\delta$ follows similarly.
Consider
\begin{align}
 v^A & :=  \inf_{\theta \in \Theta,F \in \mc N_\delta} \mb E^F[k(U,\theta,\gamma)] \quad \mbox{subject to (\ref{e:mod}) holding at $(\theta,\gamma,P)$} \,, \tag{Program A} \\
 v^B & := \inf_{\theta \in \Theta} \mb E^{\ul F_{\delta,\theta}}[k(U,\theta,\gamma)]  \quad  \mbox{subject to }   \mb E^{\ul F_{\delta,\theta}}[g_{4e}(U,\theta,\gamma)] = 0  \,, \tag{Program B}
\end{align}
where $\ul F_{\delta,\theta}$ solves 
\[
 \inf_{F \in \mc N_\delta} \mb E^F[k(U,\theta,\gamma)] \quad \mbox{subject to (\ref{e:mod-e}) holding at $(\theta,\gamma,P)$}\,,
\]
and $v^B = +\infty$ if there is no solution to this problem. Program A is the approach described in Section~\ref{s:procedure} whereas Program B is equivalent to our MPEC implementation.

The inequality $v^A \leq v^B$ is trivial if $v^B = +\infty$. If $v^B$ is finite, for any $\varepsilon > 0$ there exists $\theta^B_\varepsilon \in \Theta$ for which $\mb E^{\ul F_{\delta,\theta^B_\varepsilon}}[k(U,\theta^B_\varepsilon,\gamma)] \leq v^B + \varepsilon$ and $\mb E^{\ul F_{\delta,\theta^B_\varepsilon}}[g_{4e}(U,\theta^B_\varepsilon,\gamma)] = 0$ where $\ul F_{\delta,\theta^B_\varepsilon}$ is well defined by Lemma~\ref{lem:v_properties}(ii). As $(\theta^B_\varepsilon,\ul F_{\delta,\theta^B_\varepsilon})$ are feasible for Program A, we have $v^A \leq v^B + \varepsilon$. As $\varepsilon$ is arbitrary, we have $v^A \leq v^B$ whenever $v^B > -\infty$. 

A similar argument applies when $v^B = -\infty$: for any $n \in \mb N$ there exists $\theta^B_n \in \Theta$ for which $\mb E^{\ul F_{\delta,\theta^B_n}}[k(U,\theta^B_\varepsilon,\gamma)]  \leq -n$ and $\mb E^{\ul F_{\delta,\theta^B_n}}[g_{4e}(U,\theta^B_\varepsilon,\gamma)]  = 0$, where the distribution $\ul F_{\delta,\theta^B_n}$ is well defined by Lemma~\ref{lem:v_properties}(ii). As $(\theta^B_n,\ul F_{\delta,\theta^B_n})$ are feasible for Program A, we have $v^A \leq -n$. Hence, $v^A = v^B = -\infty$.

Note $v^B \leq v^A$ holds trivially if $v^A = +\infty$. If $v^A$ is finite, rewrite Program B as
\[
 \inf_{\kappa \in \mb R,\theta \in \Theta} \kappa  \quad \mbox{subject to } \mb E^{\ul F_{\delta,\theta,\kappa}}[g_{4e}(U,\theta,\gamma)] = 0\,,
\]
where $\ul F_{\delta,\theta,\kappa}$ solves the feasibility program
\begin{equation} \label{e:v_b_feasible}
 \inf_{F \in \mc N_\delta} 0 \quad \mbox{subject to (\ref{e:mod-e}) and $\E^F[k(U,\theta,\gamma)] = \kappa$ holding at $(\theta,\gamma,P)$.}
\end{equation}
For any $\varepsilon > 0$ there exists $\theta^A_\varepsilon \in \Theta$ and $F^A_\varepsilon \in \mc N_\delta$ such that the constraints in Program A are satisfied, i.e. $\mb E^{F^A_\varepsilon}[ g_1(U,\theta^A_\varepsilon,\gamma)] \leq P_{1}$, $\ldots$, $\mb E^{F^A_\varepsilon}[ g_4(U,\theta^A_\varepsilon,\gamma)] = 0$, and 
\[
 \mb E^{F^A_\varepsilon}[ k(U,\theta^A_\varepsilon,\gamma)] \leq v^A + \varepsilon \,.
\]
Then $\ul F^A_\varepsilon$ solves the feasibility program (\ref{e:v_b_feasible}) with $\theta = \theta^A_\varepsilon$ and $\kappa = \kappa^A_\varepsilon := \mb E^{F^A_\varepsilon}[ k(U,\theta^A_\varepsilon,\gamma)]$. Note that $\mb E^{F^A_\varepsilon}[g_{4e}(U,\theta^A_\varepsilon,\gamma)] = 0$ also holds by construction. Therefore, $(\kappa^A_\varepsilon,\theta^A_\varepsilon)$ are feasible for the augmented form of Program B. It follows that $v^B \leq \kappa^A_\varepsilon \leq v^A + \varepsilon$ holds for each $\varepsilon > 0$. As $\varepsilon > 0$ is arbitrary, we have  $v^B \leq v^A$ whenever $v^A > -\infty$. 

A similar argument applies if $v^A = -\infty$: for any $n \in \mb N$, we may choose $\theta^A_n \in \Theta$ and $F^A_n \in \mc N_\delta$ such that the constraints in Program A are satisfied and $\mb E^{F^A_n}[ k(U,\theta^A_n,\gamma)] \leq -n$. It follows that $v^B \leq -n$. Hence, $v^B = v^A = -\infty$.
\end{proof}

\begin{proof}[Proof of Proposition~\ref{prop:dual_F}]
We prove the result for $\ul F_{\delta, \theta}$, the result for $\ol F_{\delta, \theta}$ follows similarly. We drop dependence of $k$ and $g$ on $(\theta,\gamma)$ to simplify notation in what follows. 

First, suppose $k$ depends on $u$. The dual formulation is justified by Proposition~\ref{prop:criterion}. A dual solution $(\ul \eta, \ul \zeta, \ul \lambda)$ exists by Proposition~\ref{prop:dual}(iii). 

Suppose $\ul \eta > 0$. We wish to show that the change of measure $\ul m_{\delta,\theta}(u) = \dot \phi^\star(-\ul \eta^{-1} (k(u) + \ul \zeta + \ul \lambda' g_s(u) ))$ induces a distribution that solves the primal problem (\ref{e:crit_l:mpec}) at $\theta$. Differentiability of the objective function in $(\eta, \zeta, \lambda)$ is guaranteed by Assumption~\ref{a:phi}. Also note that Assumption~\ref{a:phi}(i) ensures $\dot \phi^\star \geq 0$. The first-order condition (FOC) for $\ul \zeta$ is
\[
 0 = \E^{F_*}\left[ \dot \phi^\star(-\ul \eta^{-1} (k(U) + \ul \zeta + \ul \lambda' g_s(U) ))\right] - 1
\]
which implies $\E^{F_*}[\ul m_{\delta,\theta}] = 1$ and hence that $\ul F_{\delta,\theta}$ is a probability measure. The FOC for $\ul \lambda$ is 
\begin{align*}
 0 & \geq \E^{F_*}\left[ \dot \phi^\star(-\ul \eta^{-1} (k(U) + \ul \zeta + \ul \lambda' g_s(U) )) g_1(U)\right] - P_1 \,, \\
 0 & = \E^{F_*}\left[ \dot \phi^\star(-\ul \eta^{-1} (k(U) + \ul \zeta + \ul \lambda' g_s(U) )) g_2(U)\right] - P_2 \,, \\
 0 & \geq \E^{F_*}\left[ \dot \phi^\star(-\ul \eta^{-1} (k(U) + \ul \zeta + \ul \lambda' g_s(U) )) g_3(U)\right] \,, \\
 0 & = \E^{F_*}\left[ \dot \phi^\star(-\ul \eta^{-1} (k(U) + \ul \zeta + \ul \lambda' g_s(U) )) g_{4s}(U)\right] \,,
\end{align*}
hence (\ref{e:mod:1})--(\ref{e:mod:3}) and $\E^F[g_{4s}(U,\theta,\gamma)] = 0$ hold at $(\theta,\gamma,P)$ under $\ul F_{\delta,\theta}$. The FOC for $\ul \eta > 0$ is
\begin{align*}
 0 & =  \E^{F_*}\left[ \dot \phi^\star(-\ul \eta^{-1} (k(U) + \ul \zeta + \ul \lambda' g_s(U) ))(-\ul \eta^{-1} (k(U) + \ul \zeta + \ul \lambda' g_s(U) ))  \right]  \\
 & \quad  - \E^{F_*}\left[ \phi^\star(-\ul \eta^{-1} (k(U) + \ul \zeta + \ul \lambda' g_s(U) )) \right] -  \delta \,.
\end{align*}
By Assumption~\ref{a:phi}(i), we may write the convex conjugate $\phi^{\star\star}$ of $\phi^\star$ using its Legendre transform:
\[
 \phi^{\star \star}(x^\star) = x^\star (\dot \phi^\star)^{-1}(x^\star) - \phi^\star( (\dot \phi^\star)^{-1}(x^\star) )
\]
for any $x^\star$ in the range of $\dot \phi^\star$ \cite[Theorem 26.4]{Rockafellar}. Setting $x^\star = \dot \phi^\star(x)$ and noting that $\phi^{\star \star} = \phi$ holds by the Fenchel--Moreau theorem, we obtain
\[
 \phi(\dot \phi^\star(x)) = x \dot \phi^\star(x) - \phi^\star( x) \,.
\]
It follows that we may rewrite the FOC for $\ul \eta$ as $\delta = \mb E^{F^*} \left[ \phi( \ul m_{\delta, \theta}(U) ) \right]$ and so $\ul F_{\delta, \theta} \in \mc N_\delta$.

Now suppose $\ul \eta = 0$. Here we wish to show that $\ul m_{\delta,\theta}(u) = \ind\{u \in \ul A_{\delta,\theta}\}/F_*(\ul A_{\delta, \theta})$ induces a distribution that solves the primal problem (\ref{e:crit_l:mpec}) at $\theta$. As the neighborhood constraint $F \in \mc N_\delta$ is not binding, the value of the objective must be the same as the optimal value when $\delta  =\infty$. In view of Lemma~\ref{lem:dual:infty}, the value is $F_*\text{-}\mr{ess}\inf (  k(\cdot) + \ul \lambda'g_s(\cdot) ) - \ul \lambda_{12}'P$. We can write problem (\ref{e:dual:1:mpec}) as a nested optimization:
\[
 \sup_{ \lambda \in \Lambda_s} \left( \sup_{\eta > 0, \zeta \in \mb R} -\eta \E^{F_*}\left[ {\textstyle \phi^\star \left( \frac{k(U) + \zeta + \lambda' g_s(U) }{-\eta}\right) } \right] - \eta \delta - \zeta - \lambda_{12}'P \right) .
\]
At $\lambda = \ul \lambda$, the inner problem is the dual of $\inf_{F \in \mc N_\delta} \E^F[k(U) + \ul \lambda' g_s(U) - \ul \lambda_{12}'P]$. As $\ul \eta = 0$, the constraint $F \in \mc N_\delta$ is not binding and so the minimizing distribution must be supported on $\ul A_{\delta,\theta}$. Finally, by convexity of $\phi$, the distribution induced by $\ul m_{\delta,\theta}$ minimizes $D_\phi(\,\cdot\,\|F_*)$ among all distributions with support $\ul A_{\delta, \theta}$.

Now suppose $k$ does not depend on $u$. By Proposition~\ref{prop:dual_dist}, the primal and dual values of (\ref{e:p0_dist}) are equal and a dual solution exists. By similar arguments to above, $\E^{F_*}[\ul m_{\delta,\theta}(U)] = 1$, and (\ref{e:mod:1})--(\ref{e:mod:3}) and $\E^F[g_{4s}(U,\theta,\gamma)] = 0$ hold at $(\theta,\gamma,P)$ under $\ul F_{\delta,\theta}$. Finally, as there exists $F \in \mc N_\delta$ under which the moment conditions (\ref{e:mod:1})--(\ref{e:mod:3}) and $\E^F[g_{4s}(U,\theta,\gamma)] = 0$ hold at $(\theta,\gamma,P)$,  we must have $D(\ul F_{\delta,\theta}\|F_*) \leq D(F\|F_*) \leq \delta$, as required.
\end{proof}

\subsection{Proofs for Section~\ref{s:delta}}

\begin{proof}[Proof of Proposition~\ref{prop:order}]
As $\phi_1(x) \leq \bar a \phi_2(x)$ for all $x > 0$, we have $D_{\phi_1}(F\|F_*) \leq \bar a D_{\phi_2}(F\|F_*)$. Hence, $\mc N_{\delta,2} \subseteq \mc N_{\bar a \delta,1}$ for each $\delta > 0$. The result follows from this inclusion, noting that $\ul \kappa_{\bar a \delta, 1}$ and $\ul \kappa_{\bar a \delta, 1}$ are both finite because Assumption~\ref{a:phi} holds for $\phi_1$.
\end{proof}

\subsection{Proofs for Section~\ref{s:asymptotics}}

We first present some preliminary lemmas. 

\begin{lemma}\label{lem:cgt}
Suppose that Assumptions~\ref{a:phi} and \ref{a:m}(i),(v) hold. Let $\{(F_n,\theta_n,\gamma_n,P_n)\} \subseteq \mc N_\delta \times \Theta \times \Gamma \times \mc P$ with $(\gamma_n,P_n) \to (\tilde \gamma,\tilde P) \in \Gamma \times \mc P$ and with (\ref{e:mod}) holding under $F_n$ at $(\theta_n,\gamma_n,P_n)$. Then: there exists a convergent subsequence $(F_{n_l},\theta_{n_l},\gamma_{n_l},P_{n_l}) \to (\tilde F,\tilde \theta,\tilde \gamma,\tilde P) \in \mc N_\delta \times \Theta \times \Gamma \times \mc P$ along which $\lim_{l \to \infty} \mb E^{F_{n_l}}[k(U,\theta_{n_l},\gamma_{n_l})] = \mb E^{\tilde F}[k(U,\tilde \theta,\tilde \gamma)]$ and similarly for each entry of $g_1,\ldots,g_4$, and (\ref{e:mod}) holds under $\tilde F$ at $(\tilde \theta, \tilde \gamma, \tilde P)$.
\end{lemma}

\begin{proof}[Proof of Lemma~\ref{lem:cgt}]
Let $m_n = \frac{\mr d F_n}{\mr d F_*}$. By Assumption~\ref{a:m}(v), $\{\theta_n\}$ has a convergent subsequence $\{\theta_{n_l}\}$. As $\{m_{n_l}\}$ is $\|\cdot\|_\phi$-norm bounded (Lemma~\ref{lem:orlicz}(ii)), taking a further subsequence if necessary we may assume $\{m_{n_l}\}$ is $\mc E$-weakly convergent to $\tilde m \in \mc L$ (see Appendix~\ref{ax:Orlicz}).  By the triangle inequality, the H\"older inequality (\ref{eq:holder}),  $\mc E$-weak convergence, and Assumption~\ref{a:m}(i), we have
\begin{align*}
 & \left| \E^{F_{n_l}}[ m_{n_l}(U) k(U,\theta_{n_l},\gamma_{n_l})] - \E^{F_*}[ \tilde m(U)  k(U,\tilde  \theta,\tilde \gamma)] \right| \\
 &  \quad \leq |\E^{F_*}[ ( m_{n_l}(U) - \tilde m(U) ) k(U,\tilde \theta,\tilde \gamma)]| + \|m_{n_l}\|_\phi \|k(\,\cdot\,,\theta_{n_l},\gamma_{n_l})-k(\,\cdot\,,\tilde \theta,\tilde \gamma)\|_{\psi} \to 0 .
\end{align*}
It follows by similar arguments that
\begin{align*}
 \E^{F_*}[\tilde  m(U)] & = 1 \,, &
 \E^{F_*}[\tilde  m(U)g_1(U,\tilde  \theta,\tilde \gamma)] & \leq \tilde P_1 \,, &
 \E^{F_*}[\tilde  m(U)g_2(U,\tilde  \theta,\tilde \gamma)] & = \tilde P_2 \,, \\
 & & 
 \E^{F_*}[\tilde  m(U)g_3(U,\tilde \theta,\tilde \gamma)] & \leq 0 \,, &
 \E^{F_*}[\tilde  m(U)g_4(U,\tilde \theta,\tilde \gamma)] & = 0  \,.
\end{align*}
Finally, By Lemma~\ref{lem:orlicz}(i), we have $\delta \geq \liminf_{l \to \infty} \E^{F_*}[ \phi(m_{n_l}(U))] \geq \E^{F_*}[ \phi(\tilde m(U))]$.
\end{proof}

\begin{lemma}\label{lem:endpoints}
Suppose that Assumptions~\ref{a:phi} and \ref{a:m}(i),(iii)--(v) hold. Then $\ul \kappa_\delta$ and $\ol \kappa_\delta$ are finite, and
\[
\begin{aligned}
 \ul \kappa_\delta & = \inf_{\theta \in \Theta_\delta(\gamma_0,P_0)} \ul K_\delta(\theta;\gamma_0, P_0) \,, & 
 \ol \kappa_\delta & = \sup_{\theta \in \Theta_\delta(\gamma_0,P_0)} \ol K_\delta(\theta;\gamma_0, P_0) \,.
\end{aligned}
\]
\end{lemma}

\begin{proof}[Proof of Lemma~\ref{lem:endpoints}]
We prove the result only for $\ul \kappa_\delta$; the result for $\ol \kappa_\delta$ follows similarly.

Finiteness of $\ul \kappa_\delta$ follows by Assumptions~\ref{a:phi} and \ref{a:m}(i)(v) and the H\"older inequality (\ref{eq:holder}). To simplify notation, we suppress dependence of $\Theta_\delta(\gamma_0,P_0)$ on $(\gamma_0,P_0)$ in what follows. Suppose there is $\ul \theta \not\in \Theta_\delta$ with $\ul K_\delta(\ul \theta ;\gamma_0, P_0) < \inf_{\theta \in \Theta_\delta} \ul K_\delta(\theta;\gamma_0, P_0) $. Then there must exist $ F_{\ul \theta} \in \mc N_\delta$ satisfying (\ref{e:mod}) at $(\ul \theta,\gamma_0,P_0)$. As $\Delta(\ul \theta;\gamma_0, P_0) = \delta$, it follows by convexity of $\phi$ that $ F_{\ul \theta}$ must be unique. Therefore
\begin{equation} 
  \E^{F_{\ul \theta}}[k(U,\ul \theta,\gamma_0)] = \ul K_\delta(\ul \theta;\gamma_0, P_0) < \inf_{\theta \in \Theta_\delta} \ul K_\delta(\theta;\gamma_0, P_0) \leq \inf_{\theta \in \Theta_\delta} \E^{F_\theta}[ k(U,\theta_0,\gamma_0)]  \,, \label{e:ktheta0}
\end{equation}
where, for each $\theta \in \Theta_\delta$, the distribution $F_\theta$ solves $\inf_{F} D_\phi(F\|F_*)$ subject to (\ref{e:mod}). Existence of $F_\theta$ follows by similar arguments to the proof of Lemma~\ref{lem:v_properties}(ii); its uniqueness follows by strict convexity of $\phi$. 

Choose $\{\theta_n\} \subset \Theta_\delta$ with $\theta_n \to \ul \theta$ (we may choose such a sequence by Assumption~\ref{a:m}(iv)). By Lemma~\ref{lem:cgt}, there is a subsequence $\{(\theta_{n_l},F_{\theta_{n_l}},\gamma_0,P_0)\}$ with $(\theta_{n_l},F_{\theta_{n_l}}) \to (\ul \theta,\ul F)$ for some $\ul F \in \mc N_\delta$ for which (\ref{e:mod}) holds under $\ul F$ at $(\ul \theta,\gamma_0,P_0)$. It follows by uniqueness of $F_{\ul \theta}$ that $\ul F = F_{\ul \theta}$. By Lemma~\ref{lem:cgt}, we therefore have 
\[
 \inf_{\theta \in \Theta_\delta} \E^{F_\theta}[k(U,\theta,\gamma_0)] \leq \lim_{l \to \infty} \E^{F_{\theta_{n_l}}}[  k(U,\theta_{n_l},\gamma_0)]  = \E^{F_{\ul \theta}}[   k(U,\ul \theta,\gamma_0)]  \,,
\]
which contradicts (\ref{e:ktheta0}). 
\end{proof}

\medskip

Define
\begin{equation*}
\begin{aligned}
 \ul b_\delta(\gamma,P) & = \inf_{\theta \in \Theta_\delta(\gamma,P)} \ul K_\delta(\theta;\gamma,P) \,, & 
 \ol b_\delta(\gamma,P) & = \inf_{\theta \in \Theta_\delta(\gamma,P)} \ol K_\delta(\theta;\gamma,P) \,.
\end{aligned}
\end{equation*}

\begin{lemma}\label{lem:c-cts}
Suppose that Assumptions~\ref{a:phi} and \ref{a:m}(i)--(v) hold. Then $\ul b_\delta(\gamma,P)$ and $\ol b_\delta(\gamma,P)$ are continuous at $(\gamma_0,P_0)$.
\end{lemma}

\begin{proof}[Proof of Lemma~\ref{lem:c-cts}]
We prove the result only for $\ul b_\delta$; the result for $\ol b_\delta$ follows similarly.

Fix $\varepsilon > 0$. By Lemma~\ref{lem:endpoints}, we may choose $\theta_\varepsilon \in \Theta_\delta(\gamma_0,P_0)$ such that $\ul K_\delta(\theta_\varepsilon;\gamma_0,P_0) < \ul b_\delta(\gamma_0,P_0) + \varepsilon$. By Lemma~\ref{lem:phi-cts} and Assumption~\ref{a:m}(ii) we have $\Delta(\theta_\varepsilon;\gamma,P) < \delta$ on a neighborhood $N$ of $(\gamma_0,P_0)$. Moreover, by Lemma~\ref{lem:mult:unif}(i) and Assumption~\ref{a:m}(i)--(iii) we have 
\[
 \ul K_\delta(\theta_\varepsilon;\gamma,P)  < \ul K_\delta(\theta_\varepsilon;\gamma_0,P_0) + \varepsilon
\]
on a neighborhood $N'$ of $(\gamma_0,P_0)$. On $N \cap N'$ we therefore have
\[
  \ul b_\delta(\gamma,P) 
  \leq \ul K_\delta(\theta_\varepsilon; \gamma,P) 
  < \ul K_\delta(\theta_\varepsilon;\gamma_0,P_0) + \varepsilon
  < \ul b_\delta(\gamma_0,P_0) +2  \varepsilon \,,
\]
establishing u.s.c. of  $\ul b_\delta(\gamma,P)$ at $(\gamma_0,P_0)$.

To establish l.s.c., suppose there is $\varepsilon > 0$ and $(\gamma_n,P_n) \to (\gamma_0,P_0)$ along which
\begin{equation}\label{e:c-cts:1}
 \ul b_\delta(\gamma_n, P_n) \leq \ul b_\delta(\gamma_0,P_0) - 2\varepsilon \,.
\end{equation}
Note $\Theta_\delta(\gamma_n,P_n)$ is nonempty for $n$ sufficiently large by Lemma~\ref{lem:phi-cts} and Assumption~\ref{a:m}(ii)(iii). 
For each $n$ sufficiently large, choose $\theta_n \in \Theta_\delta(\gamma_n,P_n)$ and $F_n \in \mc N_\delta$ for which 
\begin{equation}\label{e:c-cts:2}
 \mb E^{F_n}[k(U,\theta_n,\gamma_n)]  < \ul b_\delta(\gamma_n, P_n) + \varepsilon\,.
\end{equation}
By Lemma~\ref{lem:cgt} there is a subsequence $(F_{n_l},\theta_{n_l},\gamma_{n_l},P_{n_l}) \to (\ul F,\ul \theta, \gamma_0, P_0)$ for some $\ul F \in \mc N_\delta$ and $\ul \theta \in \Theta$, such that (\ref{e:mod}) holds under $\ul F$ at $(\ul \theta,\gamma_0,P_0)$, and for which
\[
 \lim_{l \to \infty} \E^{F_{n_l}}[ k(U,\theta_{n_l},\gamma_{n_l})] = \E^{\ul F}[  k(U,\ul \theta,\gamma_0)] \geq \ul K_\delta(\ul \theta;\gamma_0, P_0)\,.
\]
In view of (\ref{e:c-cts:1}) and (\ref{e:c-cts:2}) and Lemma~\ref{lem:endpoints}, this implies $\ul K_\delta(\ul \theta;\gamma_0, P_0) \leq \ul b_\delta(\gamma_0,P_0) - \varepsilon = \ul \kappa_\delta - \varepsilon$, contradicting the definition of $\ul \kappa_\delta$.
\end{proof}

\begin{proof}[Proof of Theorem~\ref{t:c-const}]
Note that $\ul \kappa_\delta = \ul b_\delta(\gamma_0,P_0)$ and $\ol \kappa_\delta = \ol b_\delta(\gamma_0,P_0)$ by Lemma~\ref{lem:endpoints} and $\hat{\ul \kappa}_\delta = \ul b_\delta(\hat \gamma,\hat P)$ and $\hat{\ol \kappa}_\delta = \ol b_\delta(\hat \gamma,\hat P)$ by definition. The result now follows by Lemma~\ref{lem:c-cts} and Slutsky's theorem.
\end{proof}

\begin{lemma}\label{lem:dual_stable}
Suppose that Assumptions~\ref{a:phi} and \ref{a:m}(i),(ii) hold, Condition S' holds at $(\theta,\gamma,P)$, and $\Delta(\theta;\gamma,P) < \delta$. Then there is a neighborhood $N$ of $(\theta,\gamma,P)$ such that Condition S' holds at $(\tilde \theta, \tilde \gamma, \tilde P)$ and $\Delta(\tilde \theta;\tilde \gamma, \tilde P) < \delta$ for all $(\tilde \theta,\tilde \gamma,\tilde P) \in N$.
\end{lemma}

\begin{proof}[Proof of Lemma~\ref{lem:dual_stable}]
By Lemma~\ref{lem:attain-2}, Condition S' holds at all $(\tilde \theta,\tilde \gamma,\tilde P)$ in a neighborhood $N'$ of $(\theta,\gamma,P)$. Moreover, $\Delta(\tilde \theta; \tilde \gamma, \tilde P) < \delta$ holds at all $(\tilde \theta,\tilde \gamma,\tilde P)$ in a neighborhood $N''$ of $(\theta,\gamma,P)$ by Lemma~\ref{lem:phi-cts}. Set $N = N' \cap N''$. 
\end{proof}

\medskip

In the remainder of this subsection we drop dependence of all quantities on $\gamma$.

\medskip

\begin{proof}[Proof of Theorem~\ref{t:asydist}]
We prove the result only for $\ul b_\delta$; the result for $\ol b_\delta$ follows similarly.

\underline{Step 1:} We first show $\ul \Theta_\delta(P_0)$ is nonempty and compact. 
For nonemptiness, choose $\{\theta_n\}$ such that $\ul K_\delta(\theta_n;P_0) \downarrow \ul \kappa_\delta$. Let $F_n$ solve the primal problem for $\theta_n$. By Lemma~\ref{lem:cgt}, there is a subsequence $(F_{n_l},\theta_{n_l}) \to (\ul F,\ul \theta)$ with $\ul F \in \mc N_\delta$ and $\ul \theta \in \Theta$ such that (\ref{e:mod}) holds under $\ul F$ at $(\ul \theta,P_0)$ and for which
\[
 \ul \kappa_\delta = \lim_{l \to \infty} \mb E^{F_{n_l}}[k(U,\theta_{n_l})] = \mb E^{\ul F}[k(U,\ul \theta)] \,.
\]
Therefore, $\Theta_\delta(P_0)$ is nonempty. We may deduce by similar arguments that $\ul \Theta_\delta(P_0)$ is closed. Compactness now follows by Assumption~\ref{a:m}(v).

\underline{Step 2:} We now prove directional differentiability. Let $P_n = P_0 + t_n h_n$ with $t_n \downarrow 0$ and $h_n \to h$. Choose $\ul \theta \in \ul \Theta_\delta(P_0)$. By Lemma~\ref{lem:dual_stable} and Assumption~\ref{a:m}(iii)(vi), Condition S' holds at $(\ul \theta,P_n)$ and $\Delta(\ul \theta;P_n) < \delta$ for $n$ sufficiently large, so by Proposition~\ref{prop:dual}(iv) the set $\ul \Lambda_\delta(\ul \theta; P_n)$ is nonempty and compact for $n$ sufficiently large. It now follows by definition of the objective (\ref{e:dual:1}) that 
\[
 \ul b_\delta(P_n) - \ul b_\delta(P_0) 
  \leq \ul K_\delta(\ul \theta; P_n) - \ul K_\delta(\ul \theta; P_0)  \\
  \leq t_n \times  -\ul \lambda_{12}'h_n \,,
\]
for all $\ul \lambda_{12} \in \ul \Lambda_\delta(\ul \theta; P_n)$. Finally, by Lemma~\ref{lem:mult:unif}(ii) we obtain
\[
 \limsup_{n \to \infty} \frac{\ul b_\delta(P_n) - \ul b_\delta(P_0)}{t_n} \leq \max_{\ul \lambda_{12} \in \ul \Lambda_\delta(\ul \theta;P_0)} -\ul \lambda_{12}'h\,.
\]
Taking the infimum of both sides over $\ul \theta \in \ul \Theta_\delta$ yields
\begin{equation}\label{e:dd-ub}
 \limsup_{n \to \infty} \frac{\ul b_\delta(P_n) - \ul b_\delta(P_0)}{t_n} \leq \inf_{\theta \in \ul \Theta_\delta} \max_{\ul \lambda_{12} \in \ul \Lambda_\delta( \theta;P_0)} -\ul \lambda_{12}'h\,.
\end{equation}

For the lower bound, choose $\theta_n \in \Theta_\delta(P_n)$ with $\ul K_\delta(\theta_n;P_n) \leq \ul b_\delta(P_n) + t_n^2$ for all $n$ sufficiently large. Take a subsequence $\{\theta_{n_l}\}$. By Assumption~\ref{a:m}(v) (taking a further subsequence if necessary), we have $\theta_{n_l} \to \ul \theta \in \Theta$. By similar arguments to step 1, we may in fact deduce that $\ul \theta \in \ul \Theta_\delta$. Reasoning as above, for $l$ sufficiently large we have
\begin{equation*}
\begin{aligned}
 \ul b_\delta(P_{n_l}) - \ul b_\delta(P_0) 
 & \geq \ul K_\delta(\theta_{n_l}; P_{n_l}) - \ul K_\delta(\theta_{n_l}; P_0) - t_{n_l}^2 
 \geq t_{n_l}  \times -\ul \lambda_{12}'h_{n_l} - t_{n_l}^2 \,,
\end{aligned}
\end{equation*}
where the final inequality holds for any $\ul \lambda_{12} \in \ul \Lambda_\delta(\theta_{n_l};P_0)$. By Assumption~\ref{a:m}(vii), we may choose $\ul \lambda_{12,n_l} \in \ul \Lambda_\delta(\theta_{n_l};P_0)$  for which $-\ul \lambda_{12,n_l} ' h \to \max_{\ul \lambda_{12} \in \ul \Lambda_\delta(\ul \theta;P_0)} - \ul \lambda_{12}' h$ as $l \to \infty$. 
Therefore,
\[
 \liminf_{l \to \infty} \frac{\ul b_\delta(P_{n_l}) - \ul b_\delta(P_0) }{t_{n_l}} \geq  \max_{\ul \lambda_{12} \in \ul \Lambda_\delta(\ul \theta;P_0)} - \ul \lambda_{12}' h \geq \inf_{\theta \in \ul \Theta_\delta}  \max_{\ul \lambda_{12} \in \ul \Lambda_\delta( \theta;P_0)} - \ul \lambda_{12}' h \,.
\]
As the lower bound does not depend on the subsequence $\{\theta_{n_l}\}$, we have 
\begin{equation}\label{e:dd-lb}
 \liminf_{n \to \infty} \frac{\ul b_\delta(P_n) - \ul b_\delta(P_0) }{t_n} \geq \inf_{\theta \in \ul \Theta_\delta}  \max_{\ul \lambda_{12} \in \ul \Lambda_\delta( \theta;P_0)} - \ul \lambda_{12}' h \,,
\end{equation}
proving directional differentiability. 
Finally, Assumption~\ref{a:m}(vii) and Lemma~\ref{lem:mult:unif}(ii) imply $\theta \mapsto \ul \Lambda_\delta(\theta;P_0)$ is continuous at each $\theta \in \ul \Theta_\delta$. The set $\ul \Lambda_\delta(\theta;P_0)$ is also compact for each $\theta \in \ul \Theta_\delta$ by Proposition~\ref{prop:dual}(iv). 
It follows by the maximum theorem that the infima in (\ref{e:dd-ub}) and (\ref{e:dd-lb}) can be replaced by minima.

\underline{Step 3:} In view of step 2, the asymptotic distribution follows by Theorem 2.1 of \cite{Shapiro1991} and the fact that $\sqrt n (\hat P - P) \to_d N(0,\Sigma)$.
\end{proof}

\begin{proof}[Proof of Theorem~\ref{t:ci}]
We verify the conditions of Theorem 3.2 of \cite{FangSantos}. Their Assumptions 1 and 2 hold by Theorem~\ref{t:asydist} and because $\sqrt n (\hat P - P_0) \to_d N(0,\Sigma)$ with $\Sigma$ finite, respectively. Their Assumption 3 is assumed directly. Finally, Lemma~\ref{lem:deriv:est} shows that $\wh {d \ul b}_{\delta,P_0}$ and $\wh {d \ol b}_{\delta,P_0}$ satisfy the sufficient conditions for Assumption 4 of \cite{FangSantos}, which is presented in their Remark 3.4. This proves consistency. Coverage of $CS_{\delta,L}^{1-\alpha}$ and $CS_{\delta,U}^{1-\alpha}$ follows by continuity of the distribution functions. Coverage of $CS_\delta^{1-\alpha}$ follows by the Bonferroni inequality.
\end{proof}

\begin{proof}[Proof of Theorem~\ref{t:ci:2}]
We prove the result only for $CS_{\delta}^{1-\alpha}$; the result for the other CSs follow similarly. 
Say that $P_0 \in CS^{1-\alpha}_{P_0}$ if $P_{10} \leq \hat P_{1,U}^{1-\alpha}$ and $P_{20} \in [\hat P_{2,L}^{1-\alpha}, \hat P_{2,U}^{1-\alpha}]$ both hold. By Lemma~\ref{lem:endpoints}, for each $\varepsilon > 0$ we may choose $\ul \theta_\varepsilon, \ol \theta_\varepsilon \in \Theta_\delta(P_0)$ such that $\ul K_\delta(\ul \theta_\varepsilon;P_0) < \ul \kappa_\delta + \varepsilon$ and $\ol K_\delta(\ol \theta_\varepsilon;P_0) > \ol \kappa_\delta - \varepsilon$. Let $\ul F_{\ul \theta_\varepsilon}$ and $\ol F_{\ol \theta_\varepsilon}$ solve problem (\ref{e:p0_dist}) at $(\ul \theta_\varepsilon;P_0)$ and $(\ol \theta_\varepsilon;P_0)$, respectively. Whenever $P_0 \in CS^{1-\alpha}_{P_0}$ holds, $\ul F_{\ul \theta_\varepsilon}$ and $\ol F_{\ol \theta_\varepsilon}$ must also satisfy the ``relaxed'' moment conditions used for computing $\hat{\ul \kappa}_{\delta,1-\alpha}$ and $\hat{\ol \kappa}_{\delta,1-\alpha}$, so it follows that $\Delta_{cs}(\ul \theta_{\varepsilon};\hat P_{1-\alpha})  < \delta$ and $\Delta_{cs}(\ol \theta_{\varepsilon};\hat P_{1-\alpha})  < \delta$. Moreover, as the primal solutions for $\ul K_\delta(\ul \theta_\varepsilon;P_0)$ and $\ol K_\delta(\ol \theta_\varepsilon;P_0)$ are feasible for the relaxed problem whenever $P_0 \in CS^{1-\alpha}_{P_0}$, we have
\[
 \hat{\ul \kappa}_{\delta,1-\alpha} \leq  \ul K_{\delta,cs}(\ul \theta_\varepsilon;\hat P_{1-\alpha}) \leq \ul K_\delta(\ul \theta_\varepsilon; P_0) < \ul \kappa_\delta + \varepsilon \,,
\]
and similarly $\hat{\ol \kappa}_{\delta,1-\alpha} > \ol \kappa_\delta - \varepsilon$. As $\varepsilon$ is arbitrary, we have that $\ul \kappa_\delta \geq \hat{\ul \kappa}_{\delta,1-\alpha}$ and $\ol \kappa_\delta \leq \hat{\ol \kappa}_{\delta,1-\alpha}$ holds whenever $P_0 \in CS^{1-\alpha}_{P_0}$. The desired coverage now follows by (\ref{e:cvg:P}).
\end{proof}

\let\oldbibliography\thebibliography
\renewcommand{\thebibliography}[1]{\oldbibliography{#1}
\setlength{\itemsep}{0pt}}

{
\putbib
}
\end{bibunit}

\begin{bibunit}

\begingroup
\small

\newpage
\clearpage
\pagenumbering{arabic}\renewcommand{\thepage}{\arabic{page}}
\setcounter{equation}{0}
\renewcommand{\theequation}{A.\arabic{equation}}

\begin{center}
{\Large Secondary Online Appendix to ``Counterfactual Sensitivity and Robustness''}

\vskip 24pt

{\large Timothy Christensen \quad \quad Benjamin Connault}

\vskip 8pt

(This secondary online appendix is not intended for publication.)

\end{center}

\vskip 8pt

\section{Background Material on Orlicz Spaces}\label{ax:Orlicz}

In this appendix we briefly review some relevant aspects of the theory of paired Orlicz spaces. We refer the reader to \citeauthor{Kras} (\citeyear{Kras}; KR hereafter) for a textbook treatment. 

Let $L^1(F_*)$ denote the space of (equivalence classes of) all measurable $f : \mc U \to \mb R$ with finite first moment under $F_*$. Define
\[
\begin{aligned} 
 \mc L & = \{ f \in L^1(F_*) :  \E^{F_*}[\phi(1+c|f(U)|)] < \infty \mbox{ for some } c > 0 \} , \\
 \mc E & = \{ f \in L^1(F_*) :  \E^{F_*}[\psi(c|f(U)|)] < \infty \mbox{ for all } c > 0 \} ,
\end{aligned}
\]
where $\psi(x) = \phi^\star(x) - x$ with $\phi^\star$ denoting the convex conjugate of $\phi$. The space $\mc L$ corresponds to the space $L_M^\star$ in KR's notation with $M(x) = \phi(1 + x)$ and $\mc E$ corresponds to $E_N$ in KR's notation with $N(x) = \psi(x)$. The condition $\lim_{x \to \infty} x \phi'(x)/\phi(x) < \infty$ in Assumption~\ref{a:phi}(i) implies $\phi(1 + x)$ satisfies KR's $\Delta_2$-condition (KR, Theorem 4.1). As such,  $\mc L$ and $\mc E$ are separable Banach spaces when equipped with the Orlicz norms 
\[
  \|f\|_\phi = \inf_{c > 0} \frac{1}{c}\left( 1 + \E^{F_*}[\phi(1+c|f(U)|)]\right) , \quad \mbox{and} \quad 
  \|f\|_\psi = \inf_{c > 0} \frac{1}{c}\left( 1 + \E^{F_*}[\psi(c|f(U)|)]\right)  ,
\]
respectively (see KR, Section 10 for a discussion of separability and KR, Theorem 10.5 for the norm). Given $\phi_1,\phi_2$ satisfying Assumption~\ref{a:phi}(i), write $\phi_1 \prec \phi_2$ if there exist positive constants $c$ and $x_0$ such that $\phi_1(x) \leq \phi_2 (cx)$ for all $x \geq x_0$. If $\phi_1 \prec \phi_2$ and $\phi_2 \prec \phi_1$ then $\phi_1$ and $\phi_2$ are said to be \emph{equivalent}. Equivalent $\phi$ functions induce equivalent spaces $\mc L$ and $\mc E$ and equivalent Orlicz norms on these spaces (KR, Section 13). For example, the functions  inducing hybrid and $\chi^2$ divergence are equivalent, and their spaces $\mc L$ and $\mc E$ are equivalent to $L^2(F_*)$, and the Orlicz norms $\|\cdot\|_\phi$ and $\|\cdot\|_\psi$ are equivalent to the $L^2(F_*)$ norm. Similarly, any $\phi$ that is equivalent to $x^p$ ($p > 1$) induces a space $\mc L$ equivalent to $L^p(F_*)$, an Orlicz norm $\|\cdot\|_\phi$ equivalent to the $L^p(F_*)$ norm, a space $\mc E$ equivalent to $L^q(F_*)$ with $\frac{1}{p} + \frac{1}{q} = 1$, and a norm $\|\cdot\|_\psi$ equivalent to the $L^q(F_*)$ norm. As with $L^p$ spaces, there is a version of H\"older's inequality, namely
\begin{equation}\label{eq:holder}
 |\mb E^{F_*}[f(U) g(U)]| \leq \|f\|_\phi \|g\|_\psi,
\end{equation}
which holds for each $f \in \mc L$ and $g \in \mc E$ (KR, Theorem 9.3). 

The spaces $\mc L$ and $\mc E$ are paired spaces under the map $\langle \cdot , \cdot \rangle : \mc L \times \mc E \to \mb R$ given by
\[
 \langle f, g \rangle = \mb E^{F_*}[f(U) g(U)] \,.
\]
A sequence $\{f_n \} \subset \mc L$ is said to be \emph{$\mc E$-weakly convergent} if $\{ \langle f_n , g \rangle \}$ converges for each $g \in \mc E$. The space $\mc L$ is $\mc E$-weakly complete: any $\mc E$-weakly convergent sequence $\{f_n \} \subset \mc L$ has a unique limit $f_0 \in \mc L$ for which
\[
 \lim_{n \to \infty} \langle f_n , g \rangle = \langle f_0, g\rangle 
\]
for each $g \in \mc E$ (KR, Theorem 14.4). The space $\mc L$ is also \emph{$\mc E$-weakly compact}: every $\|\cdot\|_\phi$-norm bounded sequence  in $\mc L$ has an $\mc E$-weakly convergent subsequence (KR, Theorem 14.4). Any $\mc E$-weakly continuous linear functional $\ell$ on $\mc L$ is representable as $\ell(f) = \langle f , g \rangle$ for some $g \in \mc E$ (KR, Theorem 14.7). Similarly, $\{g_n\} \subset \mc E$ is $\mc L$-weakly convergent if $\{\langle f,g_n \rangle \}$ converges for each $f \in \mc L$. As $\phi(1 + x)$ satisfies the $\Delta_2$-condition, every $\mc L$-weakly continuous linear functional $\ell'$ on $\mc E$ is representable in the form $\ell' (g) = \langle f, g \rangle$ for $f \in \mc L$.\footnote{As $\phi(1+x)$ satisfies the $\Delta_2$-condition, in KR's notation $\mc L = L_M^\star = E_M$ with $M(x) = \phi(1+x)$. Therefore, $\mc L$-weak convergence corresponds to what KR calls $E_M$-weak convergence on $\mc E = E_N \subseteq L_N^\star$ with $N(x) = \psi(x)$.} 

We close this section by noting three useful results, the first of which is from \cite{KR}. Let $\mc L_+ = \{f \in \mc L : f\geq 0 \; F_*\mbox{-almost everywhere}\}$.

\begin{lemma}\label{lem:orlicz}
Suppose that Assumption~\ref{a:phi}(i) holds. Then: 
\begin{enumerate}[nosep]
\item[(i)] the functional $m \mapsto \mb E^{F_*}[\phi(m(U))]$ is l.s.c. on $\mc L$ in the $\mc E$-weak topology;
\item[(ii)] $\mb E^{F_*}[\phi(m(U))] \leq \delta$ implies $\|m\|_\phi \leq 2 + \phi(2) + \delta$;
\item[(iii)]  $\mb E^{F_*}[\phi(m(U))] < \infty$ if and only if $m \in \mc L_+$.
\end{enumerate}
\end{lemma}

\begin{proof}[Proof of Lemma~\ref{lem:orlicz}]
Part (i) is stated on p. 961 of \cite{KR}. Part (ii) follows by taking $c = \frac{1}{2}$ in the definition of $\|\cdot\|_\phi$. For part (iii), it suffices by part (ii) and the fact that $\phi(x) = +\infty$ for $x < 0$ to show $\mb E^{F_*}[\phi(m(U))] < \infty$ for all  $m \in \mc L_+$. As $\phi$ satisfies the $\Delta_2$-condition under Assumption~\ref{a:phi}(i), $m \in \mc L$ implies $\mb E^{F_*}[\phi(1+c|m(U)|)] < \infty$ for all $c > 0$. As $\mc L$ contains constant functions and is closed under addition, for any $m \in \mc L_+$ we have
\[
 \infty > \mb E^{F_*}[\phi(1+|m(U)-1|)] = \mb E^{F_*}[ \phi(m(U))\ind\{m(U) \geq 1\}] + \mb E^{F_*}[ \phi(2-m(U))\ind\{m(U) \leq 1\}]
\]
which, by non-negativity of $\phi$, implies that $\mb E^{F_*}[ \phi(m(U))\ind\{m(U) \geq 1\}]$ is finite. Finiteness of the remaining term $\mb E^{F_*}[ \phi(m(U))\ind\{m(U) \leq 1\}]$ follows because $\max_{x \in [0,1]} \phi(x) = \phi(0) < \infty$ under Assumption~\ref{a:phi}(i).
\end{proof}

\section{Supplementary Results and Proofs}

\subsection{Notation}

Throughout this Appendix, we let $\ul K_\delta(\theta;\gamma,P)$ and $\ol K_\delta(\theta;\gamma,P)$ denote the criterion functions (\ref{e:crit_l}) and (\ref{e:crit_u}). We use the notation $\ul K_\delta^\star(\theta;\gamma,P)$ and $\ol K_\delta^\star (\theta;\gamma,P)$ to denote their dual forms defined below in (\ref{eq:dual_l_proof}) and (\ref{eq:dual_u_proof}). Similarly, we let $\Delta(\theta;\gamma,P)$ to denote the primal form of the minimum divergence problem (\ref{e:p0_dist}) and $\Delta^\star(\theta;\gamma,P)$ to denote its dual form in (\ref{e:qual:dual}).

For $x \in \mb R^{n}$ and $A,B \subset \mb R^{n}$ we let $d(x,A) = \inf_{a \in A} \| x - a\|$ and $\vec d_H(A,B) = \sup_{a \in A} d(a,B)$. Let $B_\varepsilon$ denote a Euclidean ball centered at the origin with radius $\varepsilon$, where the dimension should be obvious from the context. Let $T \subseteq \mb R^n$ be a nonempty, closed convex cone with nonempty interior. Let $\partial A = \mr{cl}(A) \setminus \mr{int}(A)$ denote the boundary of $A \subset T$ (relative to $\mb R^n$) and $\partial_T A = \mr{cl}( \partial A \cap \mr{int}(T))$ denote the boundary of $A$ relative to $T$. For example, if $T = \mb R_+ \times \mb R$, and $A = \{(x,y) \in T : x^2 + y^2 \leq 1 \}$, then $\partial A = \{(x,y) \in T : x^2 + y^2 = 1 \} \cup \{0\} \times [-1,1]$ and $\partial_T A = \{(x,y) \in T : x^2 + y^2 = 1 \}$.

\subsection{Preliminary Results on the Dual Form of the Criterion Functions}\label{sec:dual_K}

In this section we will show that the dual problems of  (\ref{e:crit_l}) and (\ref{e:crit_u}) are
\begin{align}
 \ul K_\delta^\star(\theta;\gamma,P) & = \sup_{\eta \geq 0, \zeta \in \mb R, \lambda \in \Lambda} - \E^{F_*}\!\!\left[ (\eta \phi)^\star(- k(U,\theta,\gamma) - \zeta - \lambda' g(U,\theta,\gamma) ) \right] - \eta \delta - \zeta - \lambda_{12}'P , \label{eq:dual_l_proof} \\
 \ol K_\delta^\star(\theta;\gamma,P) & = \inf_{\eta \geq 0, \zeta \in \mb R, \lambda \in \Lambda}  \!\! \E^{F_*}\left[ (\eta \phi)^\star( k(U,\theta,\gamma) - \zeta- \lambda' g(U,\theta,\gamma) ) \right] + \eta \delta + \zeta + \lambda_{12}'P , \label{eq:dual_u_proof}
\end{align}
where $(\eta \phi)^\star$ is the convex conjugate of $x \mapsto \eta \cdot \phi(x)$. In particular, $(\eta \phi)^\star(x) = \eta \phi^\star(x/\eta)$ when $\eta > 0$.
Let $\ul \Xi_\delta(\theta;\gamma,P)$ and $\ol \Xi_\delta(\theta; \gamma,P)$ denote the (possibly empty) set of solutions to the dual problems (\ref{eq:dual_l_proof}) and (\ref{eq:dual_u_proof}), respectively. The main result we prove in this subsection is the following:

\begin{proposition}\label{prop:dual}
Suppose that Assumption~\ref{a:phi} holds. Then: 
\begin{enumerate}[nosep]
\item[(i)] $\ul K_\delta(\theta;\gamma,P) = \ul K_\delta^\star(\theta;\gamma,P)$ and $\ol K_\delta(\theta;\gamma,P) = \ol K_\delta^\star(\theta;\gamma,P)$ (i.e., strong duality holds);
\item[(ii)] Optimizing over $(\eta,\zeta,\lambda) \in (0,\infty) \times \mb R \times \Lambda$ yields the same value for $\ul K_\delta^\star(\theta;\gamma,P)$ and $\ol K_\delta^\star(\theta;\gamma,P)$ as optimizing over $(\eta,\zeta,\lambda)  \in \mb R_+ \times \mb R \times \Lambda$;
\item[(iii)] If Condition S holds at $(\theta,\gamma,P)$ and there is $F$ with $D_\phi(F\|F_*) < \delta$ such that (\ref{e:mod}) holds under $F$ at $(\theta,\gamma,P)$, then $\ul \Xi_\delta(\theta;\gamma,P)$ and $\ol \Xi_\delta(\theta; \gamma,P)$ are nonempty and convex;
\item[(iv)] If Condition S is replaced by Condition S' in (iii), then $\ul \Xi_\delta(\theta;\gamma,P)$ and $\ol \Xi_\delta(\theta; \gamma,P)$ are also compact.
\end{enumerate}
\end{proposition}

We first present some preliminary results used to derive the dual problems and verify the constraint qualification conditions. 
We derive the dual of (\ref{e:crit_l}); the derivation of the dual of (\ref{e:crit_u}) follows similarly, replacing $k$ with  $-k$. Fix any $\theta \in \Theta$ and $\gamma \in \Gamma$. We drop dependence of $k(u,\theta,\gamma)$ and $g(u,\theta,\gamma)$ on $(\theta,\gamma)$ to simplify notation. 

Consider the primal problem 
\begin{equation} \label{e:p0}
 \min_F \mb E^F[k(U)] \quad \mbox{subject to} \quad D_\phi(F\|F_*) \leq \delta \,,\, \mb E^F[g_1(U)] \leq P_1 \,,\, \ldots, \, \mb E^F[g_4(U)] = 0\,.
\end{equation}
The \emph{value} of problem (\ref{e:p0}) is obtained by replacing the $\min$ with an $\inf$ the above display. The value is $+\infty$ if (\ref{e:p0}) has no solution. The criterion function $\ul K_\delta(\theta;\gamma,P)$ in (\ref{e:crit_l}) is the value of problem (\ref{e:p0}). 

We apply duality theory as exposited in \citeauthor{BS} (\citeyear{BS}, Chapter 2.5). Identify each $F \in \mc N_\infty$ with $m = {\mr d F}/{\mr d F_*} \in \mc L$ (see Appendix~\ref{ax:Orlicz}). Pair $\mc L$ with $\mc E$ under $\langle \cdot, \cdot \rangle$, as described in Appendix~\ref{ax:Orlicz}.
Define $\varphi : \mc L \times \mb R^{d+2} \to \mb R \cup \{+\infty\}$ by
\begin{align*}
 \varphi(m,y) = \langle m, k \rangle + \mathbb I_{C}\left(Q_\phi(m) - \delta + y_1 , \langle m, 1 \rangle - 1 + y_2, \langle m, g \rangle - \vec P + y_3 \right) ,
\end{align*}
where $y= (y_1,y_2,y_3) \in \mb R \times \mb R \times \mb R^d$, and where $\vec P = (P,0_{d_3 + d_4})$, $ Q_\phi(m) = \mb E^{F_*}[\phi(m(U))] $, $\langle m, 1 \rangle  = \mb E^{F_*}[m(U)]$, $\langle m, k \rangle  = \mb E^{F_*}[m(U)k(U)]$, $\langle m, g \rangle = \mb E^{F_*}[m(U)g(U)]$,
and $\mb I_{C} : \mb R^{d+2} \to \mb R \cup \{+ \infty\}$ is given by
\[
 \mb I_{C}(a_1,a_2,a_3) = \left[ \begin{array}{cl} 0 & \mbox{if $a_1 \leq 0$, $a_2 = 0$, and $a_3 \in \mb R_{-}^{d_1} \times \{0_{d_2}\} \times \mb R_{-}^{d_3} \times \{0_{d_4}\}$}, \\
 +\infty & \mbox{otherwise},
 \end{array} \right.
\]
with $\mb R_- = (-\infty, 0]$. The \emph{primal} problem for $y \in \mb R^{d+2}$ is
\begin{equation} \label{e:py}
 \min_{m \in \mc L} \varphi(m,y) \tag{$\mr P_y$}
\end{equation}
and its value is
\[ 
 v(y) = \inf_{m \in \mc L} \varphi(m,y)\,.
\]
In particular, $v(0) = \ul K_\delta(\theta;\gamma,P)$.

We first establish some facts about $\varphi$ and $v$. 
A convex function $f : X \to \mb R \cup \{+\infty\}$ is \emph{proper} if  $f(x) > -\infty$ for all $x \in X$ and $f(x) < +\infty$ for some $x \in X$.

\begin{lemma}\label{lem:varphi_properties}
Suppose that Assumption~\ref{a:phi} holds. Then $\varphi$ is proper and convex.
\end{lemma}

\begin{proof}[Proof of Lemma~\ref{lem:varphi_properties}]
First note that $|\langle m, k \rangle | < + \infty$ for any $m \in \mc L$ by H\"older's inequality (see (\ref{eq:holder})) and Assumption~\ref{a:phi}(ii). It follows that $\varphi(m,y) > -\infty$ for all $m \in \mc L$ and $y \in \mb R^{d+2}$. Take any $m \in \mc L_+$. Then $Q_\phi(m) < +\infty$ by Lemma~\ref{lem:orlicz}(iii). Setting $y_1 = \delta - Q_\phi(m)$, $y_2 = 1- \langle m, 1\rangle$ and $y_3 = \vec P - \langle m, g\rangle$ ensures $\mathbb I_{C}(Q_\phi(m) - \delta + y_1 , \langle m, 1 \rangle - 1 + y_2, \langle m, g \rangle - \vec P + y_3 )  = 0$, hence $\varphi(m,y) < +\infty$. Therefore, $\varphi$ is proper.
Convexity of $\varphi$ follows from convexity of $m \mapsto Q_\phi(m)$ and convexity of $\mr{dom}\,\mb I_{C} \equiv \mb R_- \times \{0\} \times \mb R_{-}^{d_1} \times \{0_{d_2}\} \times \mb R_{-}^{d_3} \times \{0_{d_4}\}$.
\end{proof}

Recall $\mc C = \mb R_+^{d_1} \times \{0_{d_2}\} \times \mb R_+^{d_3} \times \{0_{d_4}\}$. For $A,B \subset \mb R^n$, let $A - B = \{a - b: a \in A, b \in B\}$. Let
\[
 \mc Y = \left\{   \left(  \delta - Q_\phi(m)  , 1 - \langle m,1 \rangle , \vec P - \langle m, g \rangle  \right) : m \in \mc L_+  \right\} - (\mb R_+ \times \{0\} \times  \mc C) \,.
\]
The \emph{effective domain} of $f$ is $\mr{dom}\,f = \{ x \in X : f(x) < +\infty\}$.

\begin{lemma}\label{lem:v_properties}
Suppose that Assumption~\ref{a:phi} holds. Then:
\begin{enumerate}[nosep]
\item[(i)] $v$ is proper, convex, and l.s.c. on $\mb R^{d+2}$ with $\mr{dom}\,v = \mc Y$;
\item[(ii)] A solution to the primal problem $($\ref{e:py}$)$ exists for each $y \in \mc Y$.
\end{enumerate}
\end{lemma}

\begin{proof}[Proof of Lemma~\ref{lem:v_properties}]
Convexity of $v$ follows from Lemma~\ref{lem:varphi_properties} and \citeauthor{BS} (\citeyear{BS}, Proposition 2.143).
Note that $\mc Y$ is the set of all $y = (y_1,y_2,y_3) \in \mb R \times \mb R \times \mb R^d$ for which there exists $m_y \in \mc L$ such that 
\[
 \mathbb I_{C}\left(Q_\phi(m_y) - \delta + y_1 , \langle m_y, 1 \rangle - 1 + y_2, \langle m_y, g \rangle - \vec P + y_3 \right) = 0\,.
\]
Take $y \in \mc Y$ and any such $m_y \in \mc L$. Then $\varphi(m_y,y) = \langle m_y, k \rangle < +\infty$ by the proof of Lemma~\ref{lem:varphi_properties}, and so $v(y) < +\infty$. Conversely, if $y \not \in \mc Y$ then $\varphi(m,y) = +\infty$ for all $m \in \mc L$ and so $v(y) = +\infty$. Therefore, $\mr{dom}\, v = \mc Y$.

To see $v$ is proper, take any $y \in \mr{dom}\,v$ and let $y_1$ denote its first element. Then 
\[
 |v(y)| \leq \sup \{ |\langle k,m \rangle| : m \in \mc L, Q_\phi(m) \leq \delta - y_1\} \leq \|k\|_\psi (2 + \phi(2) + \delta - y_1) < \infty \,,
\]
where the first inequality is by definition of $v(y)$ and the second is by inequality (\ref{eq:holder}) and Lemma~\ref{lem:orlicz}(ii).

Before proving l.s.c. we first prove assertion (ii). Take any $y \in \mc Y$. Choose $\{m_n\} \subset \mc L$ such that $\varphi(m_n,y) \downarrow v(y)$ as $n \to \infty$. As $Q_\phi(m_n) \leq \delta - y_1$ holds for each $n$, $\{m_n\}$ is $\|\cdot\|_\phi$-norm bounded (by Lemma~\ref{lem:orlicz}(ii)) and therefore has a $\mc E$-weakly convergent subsequence $\{m_{n_l}\}$ (see Appendix~\ref{ax:Orlicz}). Let $m_0 \in \mc L$ denote the $\mc E$-weak limit. Under Assumption~\ref{a:phi}, we have both  $\lim_{l \to \infty} \langle m_{n_l}, 1 \rangle = \langle m_0 , 1 \rangle$ and $\lim_{l \to \infty} \langle m_{n_l}, g \rangle = \langle m_0 , g \rangle$ by $\mc E$-weak convergence. We also have $\delta - y_1 \geq \liminf_{l \to \infty} Q(m_{n_l}) \geq Q(m_0)$ by Lemma~\ref{lem:orlicz}(i). It follows by closedness of $\mr{dom}\, \mb I_C$ that 
\begin{equation}\label{eq:indic_C}
 \mathbb I_{C}\left(Q_\phi(m_0) - \delta + y_1 , \langle m_0, 1 \rangle - 1 + y_2, \langle m_0, g \rangle - \vec P + y_3 \right) = 0 \,.
\end{equation}
Moreover, by $\mc E$-weak convergence we also have that $v(y) = \lim_{l \to \infty} \langle m_{n_l}, k \rangle = \langle m_0 , k \rangle$. Therefore, $m_0$ solves the primal problem (\ref{e:py}).

To prove l.s.c., we first show that $\mc Y$ is closed. Take $y \in \mr{cl}(\mc Y)$. Choose any sequence $\{y_n\} \subset \mc Y$ converging to $y$. Note $\{v(y_n)\}$ is bounded by the argument used above to establish properness. Choose a subsequence $\{n_l\}$ for which  $\lim_{l \to \infty} v(y_{n_l}) = \liminf_{n \to \infty} v(y_n)$. By part (ii), there exists a solution $m_{n_l}$ to the primal problem for each $y_{n_l}$. The sequence $\{m_{n_l}\}$ is $\|\cdot\|_\phi$-norm bounded (by Lemma~\ref{lem:orlicz}(ii)) and hence, taking a further subsequence if necessary, has an $\mc E$-weak limit $m_0 \in \mc L$. By similar arguments to the above, we may deduce that (\ref{eq:indic_C}) holds for this $m_0$  at $y$. Therefore, $y \in \mc Y$ (establishing closedness) and $v(y) < \infty$. 

To complete the proof of l.s.c., take any $y \in \mc Y$. Choose any sequence $\{y_n\} \subset \mc Y$ converging to $y$ (it is without loss of generality to consider sequences of elements of $\mc Y$, since $v(\tilde y) = +\infty$ for $\tilde y \not \in \mc Y$). By the argument just used to establish closedness, we may extract a subsequence $\{n_l\}$ for which  $\lim_{l \to \infty} v(y_{n_l}) = \liminf_{n \to \infty} v(y_n)$. Let $m_{n_l}$ solve the primal problem along this subsequence. Taking a further subsequence if necessary, $\{m_{n_l}\}$ has a $\mc E$-weak limit $m_0 \in \mc L$. By similar arguments to the above, we may deduce that (\ref{eq:indic_C}) holds for this $m_0$ at $y$. Therefore, by $\mc E$-weak convergence we have
\[
 \liminf_{n \to \infty} v(y_n) = \lim_{l \to \infty} v(y_{n_l}) = \lim_{l \to \infty} \langle m_{n_l}, k \rangle = \langle m_0 , k \rangle \geq v(y),
\]
establishing l.s.c. of $v$ at any $y \in \mr{cl}(\mc Y) \equiv \mc Y$. Finally, l.s.c. of $v$ on $\mb R^{d+2}$ now follows from l.s.c. of $v$ at any $y \in \mc Y$ and the fact that $\mc Y \equiv \mr{dom}\,v $ is closed.
\end{proof}

The \emph{dual} problem of (\ref{e:py}) is \cite[p. 96]{BS}
\begin{equation} \label{e:dy}
 \max_{y^\star \in \mb R^{d+2}} y' y^\star - \varphi^\star(0, y^\star) , \tag{$\mr D_y$}
\end{equation}
 where $y^\star = (y_1^\star,y_2^\star,y_3^\star) \in \mb R \times \mb R \times \mb R^d$ and  $\varphi^\star : \mc E \times \mb R^{d+2} \to \mb R \cup \{+\infty\}$ is the conjugate of $\varphi$:
\[
 \varphi^\star(m^\star,y^\star) = \sup_{(m,y) \in \mc L \times \mb R^{d+2}} \left( \langle m, m^\star \rangle + y' y^\star - \varphi(m, y) \right) .
\]
By direct calculation, 
\begin{align*}
 \varphi^\star(0,y^\star) & = \sup_{(m,y) \in \mc L \times \mb R^{d+2}} \left(  y' y^\star - \langle m, k \rangle - \mathbb I_{C}\left(Q_\phi(m) - \delta + y_1 , \langle m, 1 \rangle - 1 + y_2, \langle m, g \rangle - \vec P + y_3 \right) \right) \\
 & = \sup_{m \in \mc L } 
 \left( - y_1^\star ( Q_\phi(m) - \delta ) - y_2^\star ( \langle m,1 \rangle - 1 ) - y_3^{\star\prime} \left( \langle m, g \rangle - \vec P \right) - \langle m, k \rangle \right)  + \mathbb I_{C^o}(y^\star) 
\end{align*}
where $C^o = \mb R_+ \times \mb R \times \Lambda$ is the polar cone of $C$ and $\mb I_{C^o}(y^\star) = 0$ if $y^\star \in C^o$ and $+\infty$ otherwise.
Write any $y^\star \in C^o$ as $y^\star = (\eta,\zeta,\lambda) \in \mb R_+ \times \mb R \times \Lambda$. We then have
\begin{align}
 \varphi^\star(0,(\eta,\zeta,\lambda)) 
 & = \sup_{m \in \mc L } \left( - \eta  Q_\phi(m)  - \zeta \langle m,1 \rangle - \lambda' \langle m, g \rangle  - \langle m, k \rangle \right) + \eta \delta + \zeta + \lambda'\vec P \notag \\
 & = \sup_{m \in \mc L } \mb E^{F_*} \left[ m(U) (-k(U) - \zeta - \lambda' g(U)) - \eta \phi(m(U))  \right]  + \eta \delta + \zeta + \lambda'\vec P \label{e:dual_alternative_form}\,.
\end{align}
As $\mc L$ is decomposable \cite[Definition 14.59 and Theorem 14.60]{RW}, we may bring the supremum inside the expectation and optimize pointwise to obtain 
\[
 \varphi^\star(0,(\eta,\zeta,\lambda)) = \E^{F_*}\Big[ (\eta \phi)^\star(- k(U) - \zeta - \lambda' g(U)) \Big] + \eta \delta + \zeta + \lambda' \vec P\,,  \quad (\eta,\zeta,\lambda) \in C^o \,,
\]
where 
\[ 
 (\eta \phi)^\star(x) = \sup_{t \geq 0 : \eta \phi(t) < +\infty} (tx - \eta \phi(t)) = \left[ \begin{array}{ll}
  \eta \phi^\star(x/\eta) & \mbox{if $\eta > 0$}, \\
  0 & \mbox{if $\eta = 0$ and $x \leq 0$}, \\
  +\infty & \mbox{if $\eta = 0$ and $x > 0$}. \end{array} \right.
\]
As $\langle y, y^\star \rangle - \varphi^\star(0, y^\star) = -\infty$ whenever $y^\star \not \in C^o$, problem (\ref{e:dy}) can therefore be expressed as
\[
 \max_{ \eta \geq 0, \zeta \in \mb R, \lambda \in \Lambda} \eta y_1 + \zeta y_2 + \lambda' y_3 - \E^{F_*}\Big[ (\eta \phi)^\star(- k(U) - \zeta - \lambda' g(U)) \Big] - \eta \delta - \zeta - \lambda'\vec P \,.
\]
The \emph{value} of this dual problem is obtained by replacing the $\max$ with $\sup$ and corresponds to the biconjugate $v^{\star \star}(y)$ of $v(y)$. Hence, $v(0) = \ul K_\delta(\theta;\gamma,P)$ and $v^{\star \star}(0) = \ul K_\delta^\star(\theta;\gamma,P)$.

\begin{lemma}\label{lem:dual_0}
Suppose that Assumption~\ref{a:phi} holds. Then: 
\begin{enumerate}[nosep]
\item[(i)] $\ul K_\delta(\theta;\gamma,P) = \ul K_\delta^\star(\theta;\gamma,P)$; 
\item[(ii)] If $0 \in \mr{ri}(\mc Y)$, then the set of dual solutions $\ul \Xi_\delta(\theta;\gamma,P)$ is nonempty and convex; 
\item[(iii)] If $0 \in \mr{int}(\mc Y)$, then $\ul \Xi_\delta(\theta;\gamma,P)$ is also compact.
\end{enumerate}
\end{lemma}

\begin{proof}[Proof of Lemma~\ref{lem:dual_0}]
For part (i), we need to show that $v(0) = v^{\star \star}(0)$. As $\varphi$ is convex by Lemma~\ref{lem:varphi_properties} and $v$ is proper by Lemma~\ref{lem:v_properties}, it follows by \citeauthor{BS} (\citeyear{BS}, Theorem 2.144) that $v^{\star \star}(0) = \min\{v(0), \liminf_{y \to 0} v(y)\}$. But $\liminf_{y \to 0} v(y) \geq v(0)$ by Lemma~\ref{lem:v_properties}, and so $v(0) = v^{\star \star}(0)$, as required.

For part (ii), non-emptiness of the set of dual solutions follows by \citeauthor{BS} (\citeyear{BS}, Propositions 2.147 and 2.148(iii)), noting that $v$ is convex by Lemma~\ref{lem:v_properties} and $v(0)$ is finite because $v$ is proper by Lemma~\ref{lem:v_properties} and $0 \in \mr{dom}\,v \equiv \mc Y$ by assumption.
Convexity of the set of dual solutions follows by noting that, in view of (\ref{e:dy}) and (\ref{e:dual_alternative_form}), the dual objective is the pointwise infimum of affine functions of $(\eta,\zeta,\lambda)$, and is therefore concave and u.s.c.
Finally, part (iii) follows from \citeauthor{BS} (\citeyear{BS}, Theorem 2.151 and Proposition 2.152).
\end{proof}

Recall Condition S and the set $\mc C$ from Section~\ref{s:sharp} and Condition S' from Section~\ref{s:asymptotics}. Define
\begin{align}
 \mc Y_1 & = \left\{  \vec P - \langle m, g \rangle  : m \in \mc L_+, \langle m, 1 \rangle = 1 \right\} - \mc C , \notag \\ 
 \mc Y_2 & = \left\{ \left(   1 - \langle m,1 \rangle , \vec P - \langle m, g \rangle \right) : m \in \mc L_+ \right\} - (\{0\} \times \mc C) . \label{e:y2}
\end{align}
Let $0$ denote a vector of zeros whose dimension is determined by the context.

\begin{lemma}\label{lem:cq}
Suppose that Assumption~\ref{a:phi} holds and Condition S holds at $(\theta,\gamma,P)$. Then: 
\begin{enumerate}[nosep]
\item[(i)] $0 \in \mr{ri} (\mc Y_1)$; 
\item[(ii)] $0 \in \mr{ri} (\mc Y_2)$; 
\item[(iii)] if there exists $F$ with $D_\phi(F\|F_*) < \delta$ s.t. the conditions in (\ref{e:p0}) hold at $\theta$, then $0 \in \mr{ri} (\mc Y)$.
\end{enumerate}
Moreover, if Condition S' holds at $(\theta,\gamma,P)$ then ``relative interior'' can be replaced with ``interior'' in parts (i)--(iii).
\end{lemma}

\begin{proof}[Proof of Lemma~\ref{lem:cq}]
In view of Lemma~\ref{lem:orlicz}(iii), identify $F \in \mc N_\infty$ with its Radon--Nikodym derivative $m = \mr d F/\mr d F_* \in \mc L$. Let $\mc V_1 = \{ \langle m, g \rangle : m \in \mc L_+\}$. Then
\[
 \vec P \in \mr{ri}(\{ \E^{F} [ g(U) ] :  F \in \mc N_\infty \}  + \mc C) \iff \vec P \in \mr{ri}(\mc V_1  + \mc C)  \iff 0 \in \mr{ri}(\mc Y_1)\,,
\]
proving part (i).

For part (ii), note that $0 \in \mr{ri}(\mc Y_2)$ is equivalent to $(1,\vec P) \in \mr{ri}(\mc V_2)$, where 
\[
 \mc V_2 =  \left\{  \left(  \langle m,1 \rangle \,, \langle m, g \rangle  \right) : m \in \mc L_+ \right\} +  (\{0\} \times \mc C) = \mr{cone}\left(\{1\} \times (\mc V_1 + \mc C) \right),
\]
with $\mr{cone}(A) = \{t a : a \in A, t \geq 0\}$. By \citeauthor{Rockafellar} (\citeyear{Rockafellar}, Corollary 6.8.1), we have $\mr{ri}( \mc V_2) \supset \{1\} \times \mr{ri}(\mc V_1 + \mc C)$. The result follows because $\vec P \in \mr{ri}(\mc V_1 + \mc C)$ under Condition S. 

For part (iii), note that $0 \in \mr{ri}(\mc Y)$ is equivalent to $(\delta,1,\vec P) \in \mr{ri}(\mc V_3)$, where 
\[
 \mc V_3 =  \left\{ \footnotesize \left( \begin{array}{c} Q_\phi(m) \\ \langle m,1 \rangle \\ \langle m, g \rangle \end{array} \right) : m \in \mc L_+ \right\} + (\mb R_+ \times \{0\} \times \mc C) \,.
\]
It suffices to show that for every $v \in \mc V_3$ there exists $t > 1$ such that $t (\delta,1,\vec P) + (1-t) v \in \mc V_3$ \cite[Theorem 6.4]{Rockafellar}. Take any $v \in \mc V_3$. Write $v = (v_1,v_2) \in \mb R_+ \times \mc V_2$. By part (ii) and  \citeauthor{Rockafellar} (\citeyear{Rockafellar}, Theorem 6.4), there exists $s > 1$, $ m_v \in \mc L_+$, and $ c_3 \in \mc C$ such that 
\begin{equation} \label{e:cq0}
 s \left( \footnotesize \begin{array}{c} 1 \\ \vec P \end{array} \right) 
 + (1-s) v_2
 = \left( \footnotesize \begin{array}{c} \langle  m_v, 1 \rangle \\ \langle  m_v, g \rangle \end{array} \right) + \left( \footnotesize \begin{array}{c} 0 \\  c_3 \end{array} \right) \,.
\end{equation}
By assumption, there exists $F \in \mc N_\delta$ with $D_\phi(F\|F_*) < \delta$ such that the moment conditions in (\ref{e:p0}) hold at $\theta$. Let $\tilde m$ denote the Radon--Nikodym derivative of such an $F$. Then for any $\tau \in (0,1)$, setting $m_\tau = \tau  m_v + (1-\tau) \tilde m$, we have $\langle m_\tau, 1 \rangle = \tau \langle  m_v,1 \rangle + (1-\tau)$ and $\langle m_\tau, g \rangle = \tau \langle  m_v,g \rangle + (1-\tau) (\vec P - \tilde c)$ for some $\tilde c \in \mc C$. But then
\begin{equation} \label{e:cq1}
 \left( \footnotesize \begin{array}{c} \langle  m_v, 1 \rangle \\ \langle  m_v, g \rangle \end{array} \right)
 = \frac{1}{\tau} \left( \footnotesize \begin{array}{c} \langle  m_\tau, 1 \rangle \\ \langle m_\tau, g \rangle \end{array} \right) - \frac{1-\tau}{\tau} \left( \footnotesize \begin{array}{c} 1 \\ \vec P - \tilde c \end{array} \right) \,.
\end{equation}
Substituting (\ref{e:cq1}) into (\ref{e:cq0}) yields
\[
 (1 + \tau (s -1)) \left( \footnotesize \begin{array}{c} 1 \\ \vec P \end{array} \right) 
 - \tau(s - 1) v_2 
 = \left( \footnotesize \begin{array}{c} \langle m_\tau, 1 \rangle \\ \langle m_\tau, g \rangle \end{array} \right) + \left( \footnotesize \begin{array}{c} 0 \\ \tau  c_3 + (1-\tau) \tilde c \end{array} \right) \,.
\]
Note that $Q_\phi(m_\tau)$ can be made arbitrarily close to $Q_\phi(\tilde m) < \delta$ by choosing $\tau$ arbitrarily small. Setting $t = 1 + \tau (s -1)$ with $\tau$ sufficiently small that $t \delta + (1-t) v_1 \geq Q_\phi(m_\tau)$, we may write
\[
 t \left( \footnotesize \begin{array}{c} \delta \\ 1 \\ \vec P \end{array} \right) 
 + (1-t)v 
 = \left( \footnotesize \begin{array}{c} Q_\phi(m_\tau) \\ \langle m_\tau, 1 \rangle \\ \langle m_\tau, g \rangle \end{array} \right) + \left( \footnotesize \begin{array}{c}  c_1 \\ 0 \\ \tau  c_3 +  (1-\tau) \tilde c \end{array} \right) 
\]
for some $ c_1 \geq 0$. As the right-hand side belongs to $\mc V_3$, this completes the proof of part (iii).

Now suppose Condition S' holds. Part (i) holds with ``interior'' by definition of Condition S'. 
For part (ii) with ``interior'', it suffices to show that $\mc Y_2$ has positive volume, in which case its relative interior and interior coincide and the result follows by part (ii) above. A sufficient condition  is that the functions in $g$ and a function that is constant $F_*$-a.e. are not collinear $F_*$-a.e. We prove this by contradiction. Suppose Condition S' holds but that there exists $0 \neq \lambda \in \mb R^d$ and $\zeta \in \mb R$ such that $\lambda'(g(u) - \vec P) = \zeta$ $F_*$-a.e. Then by Condition S', we have $\{ \mb E^F[g(U)] - \vec P : D_\phi(F\|F_*) < \infty\}$ contains a $\varepsilon$-ball with center $c_0$ for some $c_0 \in \mc C$ and $\varepsilon > 0$. Then for any unit vector $x$, we have $\zeta = \lambda' c_0 + \varepsilon \lambda' x$, a contradiction. Thus, part (ii) must hold with ``interior'' when Condition S' holds. 
For part (iii), note that $\mc Y \supseteq (-\infty,\delta] \times \mc Y_2$. Therefore, $\mc Y$ has positive volume as $\mc Y_2$ as positive volume, so its relative interior and interior coincide and part (iii) with ``interior'' follows similarly.
\end{proof}

\begin{proof}[Proof of Proposition~\ref{prop:dual}]
We prove the result for $\ul K_\delta^\star$ and $\ul \Xi_\delta$; the result for  $\ol K_\delta^\star$ and $\ol \Xi_\delta$ follows similarly.

Part (i) follows by Lemma~\ref{lem:dual_0}(i).
For part (ii), in view of (\ref{e:dy}) and (\ref{e:dual_alternative_form}), the dual objective function, say $\ell(\eta,\zeta,\lambda)$, is the pointwise infimum of affine functions of $(\eta,\zeta,\lambda)$, and is therefore concave and u.s.c. If $\ell(0,\zeta,\lambda) = -\infty$ for all $\zeta \in \mb R$ and $\lambda \in \Lambda$, then restricting $(\eta,\zeta,\lambda)$ to $(0,\infty) \times \mb R \times \Lambda$ will not affect the dual value. Now suppose  $\ell(0,\zeta,\lambda) > -\infty$ for some  $(\zeta,\lambda) \in \mb R \times \Lambda$. By u.s.c. of $\ell(\cdot,\zeta,\lambda)$, we have $\ell(0,\zeta,\lambda) \geq \lim_{\eta \downarrow 0} \ell(\eta, \zeta,\lambda)$. To prove the reverse inequality, note by concavity of $\ell$ that for any $\tau \in [0,1]$ and $\bar \eta > 0$, we have
\[
 +\infty > \ell( \tau \bar \eta , \zeta, \lambda) \geq (1-\tau) \ell(0, \zeta,\lambda) + \tau \ell(\bar \eta, \zeta,\lambda) > -\infty \,,
\]
 because $|\ell(\eta,\zeta,\lambda) | < + \infty$ on $(0,\infty) \times \mb R \times \Lambda$. Therefore $\lim_{\eta_\downarrow 0} \ell(\eta,\zeta,\lambda) \geq \ell(0,\zeta,\lambda)$.

For part (iii), if Condition S holds and there exists $F$ with $D_\phi(F\|F_*) < \delta$ under which (\ref{e:mod}) holds, then $0 \in \mr{ri}(\mc Y)$ by Lemma~\ref{lem:cq}. Existence and convexity of the set of dual solutions follows by Lemma~\ref{lem:dual_0}(ii).
Finally, compactness of $\ul \Xi_\delta(\theta;\gamma,P)$ under Condition S' follows similarly by Lemma~\ref{lem:dual_0}(iii) and Lemma~\ref{lem:cq}, proving part (iv).
\end{proof}

\subsection{Additional Details on the Minimum-Divergence Problem} \label{ax:Delta}

Recall that $\Delta(\theta;\gamma,P)$ denotes the value of the primal form of the minimum divergence problem (\ref{e:p0_dist}) and $\Delta^\star(\theta;\gamma,P)$ denotes its dual value in (\ref{e:qual:dual}). The main result we prove in this subsection is the following:

\begin{proposition}\label{prop:dual_dist}
Suppose that Assumption~\ref{a:phi} holds. Then: 
\begin{enumerate}[nosep] 
\item[(i)] $\Delta(\theta;\gamma,P) = \Delta^\star(\theta;\gamma,P)$ (i.e., strong duality holds);
\item[(ii)] If Condition S holds at $(\theta,\gamma,P)$, then the set of dual solutions is nonempty and convex;
\item[(iii)] If Condition S' holds at $(\theta,\gamma,P)$, then the set of dual solutions is compact.
\end{enumerate}
\end{proposition}

We first present some preliminary results. We again drop dependence of $k$ and $g$ on $(\theta,\gamma)$ to simplify notation. Define $\varphi_2 : \mc L \times \mb R^{d+1} \to \mb R \cup \{+\infty\}$ by
\[
 \varphi_2(m,y) = Q_\phi(m) + \mathbb I_{C_2}\left( \langle m, 1 \rangle - 1 + y_1, \langle m, g \rangle - \vec P + y_2 \right) ,
\]
where $y = (y_1,y_2) \in \mb R \times \mb R^d$ and $\mb I_{C_2} : \mb R^{d+1} \to \mb R \cup \{+ \infty\}$ is given by
\[
 \mb I_{C_2}(y_1,y_2) = \left[ \begin{array}{cl} 0 & \mbox{if $y_1 = 0$, and $y_2 \in \mb R_{-}^{d_1} \times \{0_{d_2}\} \times \mb R_{-}^{d_3} \times \{0_{d_4}\}$}, \\
 +\infty & \mbox{otherwise}.
 \end{array} \right.
\]
The primal problem for $y \in \mb R^{d+1}$ is
\begin{equation} \label{e:py_dist}
 \min_{m \in \mc L} \varphi_2(m,y) \tag{$\mr P_y'$}
\end{equation}
and its value is 
\[ 
 v_2(y) = \inf_{m \in \mc L} \varphi_2(m,y)\,.
\]
In particular, $v_2(0) = \Delta(\theta,\gamma,P)$. Recall the set $\mc Y_2$ defined in (\ref{e:y2})

\begin{lemma} \label{lem:v_properties_dist}
Suppose that Assumption~\ref{a:phi} holds. Then:
\begin{enumerate}[nosep]
\item[(i)] $\varphi_2$ is proper and convex;
\item[(ii)] $v_2$ is proper, convex, and $\mr{dom}\,v_2 = \mc Y_2$;
\item[(iii)] A solution to the primal problem $($\ref{e:py_dist}$)$ exists and is unique for each $y \in \mc Y_2$.
\end{enumerate}
\end{lemma}

\begin{proof}[Proof of Lemma~\ref{lem:v_properties_dist}]
Properness and convexity of $\varphi_2$ and $v_2$ follow by similar arguments to Lemma~\ref{lem:varphi_properties} and \ref{lem:v_properties}, and $\mr{dom} \, v_2 = \mc Y_2$ follows by similar arguments to Lemma~\ref{lem:v_properties}.

For part (iii), take any $y \in \mc Y_2$. Choose $\{m_n\} \subset \mc L$ such that $\varphi_2(m_n, y) \downarrow v_2(y)$ as $n \to \infty$. Without loss of generality we may assume $\varphi_2(m_n, y) < \infty$ for each $n$. Then as $\varphi_2(m_n, y)  = Q_\phi(m_n)$, the sequence $\{m_n\}$ is $\|\cdot\|_\phi$-norm bounded (by Lemma~\ref{lem:orlicz}(ii)) and therefore has a $\mc E$-weakly convergent subsequence $\{m_{n_l}\}$ (see Appendix~\ref{ax:Orlicz}). Let $m_0 \in \mc L$ denote the $\mc E$-weak limit. We have both  $\lim_{l \to \infty} \langle m_{n_l}, 1 \rangle = \langle m_0 , 1 \rangle$ and $\lim_{l \to \infty} \langle m_{n_l}, g \rangle = \langle m_0 , g \rangle$ by $\mc E$-weak convergence. Therefore,
\begin{equation}\label{eq:indic_C2}
 \mathbb I_{C_2}\left( \langle m_0, 1 \rangle - 1 + y_1, \langle m_0, g \rangle - \vec P + y_2 \right)  = 0
\end{equation}
and so $\varphi_2(m_0,y) = Q_\phi(m_0)$. It follows by Lemma~\ref{lem:orlicz}(i) that $v_2(y) = \lim_{l \to \infty} Q_\phi(m_{n_l}) \geq Q_\phi(m_0) \geq v_2(y)$, where the final inequality is by definition of $v_2(y)$. Uniqueness of the primal solution follows by strict convexity of $\phi$.
\end{proof}

By similar arguments to Appendix~\ref{sec:dual_K}, the dual problem of (\ref{e:py_dist}) is
\begin{equation} \label{e:dy_dist}
 \max_{\zeta \in \mb R, \lambda \in \Lambda} \zeta y_1 + \lambda' y_2 - \E^{F_*}\Big[ \phi^\star( - \zeta - \lambda' g(U)) \Big] - \zeta - \lambda_{12}'P \,. \tag{$\mr D_y'$}
\end{equation}
The \emph{value} of (\ref{e:dy_dist}) is obtained by replacing the $\max$ with a $\sup$ and corresponds to the biconjugate $v^{\star \star}_2(y)$ of $v_2(y)$. Hence, $v_2(0) = \Delta (\theta; \gamma, P)$ and $v_2^{\star \star}(0) = \Delta^\star(\theta;\gamma,P)$.

\begin{proof}[Proof of Proposition~\ref{prop:dual_dist}]
For part (i), we need to show $ v_2(0) = v_2^{\star \star}(0)$. As $\varphi_2$ is convex by Lemma~\ref{lem:v_properties_dist} and $v_2(y) \geq 0$ for all $y \in \mb R^{d+1}$, it follows by \citeauthor{BS} (\citeyear{BS}, Theorem 2.144) that $v^{\star \star}_2(0) = \min\{v_2(0), \liminf_{y \to 0} v_2(y)\}$. Therefore, it suffices to show that $v_2(y)$ is l.s.c. at $y = 0$. 

To do so, let $\{y_n\}$ be any sequence converging to $0$. It is without loss of generality to assume $\liminf_{n \to \infty} v_2(y_n) < +\infty$ (otherwise there is nothing to prove). Extract a subsequence $\{n_l\}$ along which $v_2(y_{n_l}) < +\infty$ and $\lim_{l \to \infty} v_2(y_{n_l}) = \liminf_{n \to \infty} v_2(y_n)$. The primal problem at $y_{n_l}$ has a solution $m_{n_l}$ for each $l$ by Lemma~\ref{lem:v_properties_dist}. As $v_2(y_{n_l}) = Q_\phi(m_{n_l}) \leq C$ for all $l$ and a finite positive constant $C$, the sequence $\{m_{n_l}\}$ is $\|\cdot\|_\phi$-norm bounded (by Lemma~\ref{lem:orlicz}(ii)). Extracting a further subsequence if necessary, we may assume the sequence $\{m_{n_l}\}$ is $\mc E$-weakly convergent to some $m_0 \in \mc L$. By similar arguments to the proof of Lemma~\ref{lem:v_properties_dist}, we may deduce that (\ref{eq:indic_C2}) holds for this $m_0$ at $y = 0$. Then by Lemma~\ref{lem:orlicz}(i), we have 
\[
 \liminf_{n \to \infty} v_2(y_n) = \lim_{l \to \infty} v_2(y_{n_l}) = \lim_{l \to \infty} Q_\phi(m_{n_l}) \geq Q_\phi(m_0) \geq v_2(0)\,,
\]
as required. Note by way of contradiction that the preceding argument also implies that $\liminf_{n \to \infty} v_2(y_n) = +\infty$ for all sequences $y_n \to 0$ whenever $v_2(0) = +\infty$. 

For parts (ii) and (iii), Lemma~\ref{lem:cq} shows that Conditions S and S' are sufficient for $0 \in \mr{ri}(\mc Y_2)$ and $0 \in \mr{int}(\mc Y_2)$, respectively. Parts (ii) and (iii) now follow by similar arguments to Lemma~\ref{lem:dual_0}(ii)(iii).
\end{proof}

\subsection{Stability of Constraint Qualifications under Perturbations}

\begin{lemma}\label{lem:attain-1}
Suppose that Assumption~\ref{a:phi} holds and Condition S' holds at $(\theta,\gamma,P)$. Then there exists a neighborhood $N$ of $P$ such that Condition S' holds at $(\theta,\gamma,\tilde P)$ for each $\tilde P \in N$.
\end{lemma}

\begin{proof}[Proof of Lemma~\ref{lem:attain-1}]
Recall $\vec P = (P,0_{d_3+d_4})$ and let $\vec{\tilde P} = (\tilde P,0_{d_3 + d_4})$. By Condition S', there exists $\varepsilon > 0$ such that $B_{2 \varepsilon} \subseteq (\{ \E^{F} [ g(U,\theta,\gamma) ]  - \vec P : F \in \mc N_\infty \} + \mc C).$ Then for any $\tilde P$ for which $\|P - \tilde  P\| < \varepsilon$, we  have $\| (\E^F [ g(U,\theta,\gamma) ] - \vec P) - ( \E^F [g(U,\theta,\gamma) ] - \vec{\tilde P}) \| < \varepsilon$ for all $F \in \mc N_\infty$. Hence, $B_{\varepsilon} \subseteq (\{ \E^{F} [ g(U,\theta,\gamma) ] - \vec{\tilde P} : F \in \mc N_\infty \} + \mc C)$.
\end{proof}

\begin{lemma}\label{lem:attain-2}
Suppose that Assumption~\ref{a:phi} holds, each entry of $g$ is $\mc E$-continuous in $(\theta,\gamma)$, and Condition S' holds at $(\theta,\gamma,P)$. Then there exists a neighborhood $N$ of $(\theta,\gamma,P)$ such that Condition S' holds at $(\tilde \theta,\tilde \gamma,\tilde P)$ for each $(\tilde \theta,\tilde \gamma,\tilde P) \in N$.
\end{lemma}

\begin{proof}[Proof of Lemma~\ref{lem:attain-2}]
Recall $\vec P = (P,0_{d_3 + d_4})$. As Condition S' holds at $(\theta,\gamma,P)$, there exists sufficiently large $\delta$ such that $0 \in \mr{int}(\{ \E^F [ g(U,\theta,\gamma ) ] - \vec P: F \in \mc N_\delta  \}+ \mc C )$. (To see this, take a sufficiently small hypercube $\mc H$ with $0 \in \mr{int}(\mc H)$ and $\mc H \subseteq \mr{int}(\{ \E^F [ g(U,\theta,\gamma ) ] - \vec P: F \in \mc N_\infty  \}+ \mc C )$, identify a density $F \in \mc N_\infty$ with each vertex of $\mc H$, and take $\delta$ to be the largest $\phi$-divergence from $F_*$ of each of the densities at the vertex.) Therefore, we may choose $\varepsilon > 0$ such that $B_{4\varepsilon} \subseteq \mr{int}(\{ \E^F [ g(U,\theta,\gamma ) ] - \vec P: F \in \mc N_\delta \}+ \mc C )$. 

Let $\mc M_\delta = \{ \frac{\mr d F}{\mr d F_*}  : F \in \mc N_\delta\}$. Then $\mc M_\delta$ is a $\|\cdot\|_\phi$-bounded subset of $\mc L$ by Lemma~\ref{lem:orlicz}(ii). 
By $\mc E$-continuity, there exists a neighborhood $N_1$ of $(\theta,\gamma)$ such that for any $(\tilde \theta,\tilde \gamma) \in N_1$ and with $r$ denoting any entry of $g_1,\ldots,g_4$, we have
\[
 \| r(\cdot,\theta,\gamma) - r(\cdot,\tilde \theta,\tilde \gamma) \|_{\psi} < \frac{\varepsilon}{\sqrt{d} ( 2 + \phi(2) + \delta ) } \,.
\]
It follows by inequality (\ref{eq:holder}) and Lemma~\ref{lem:orlicz}(ii) that for any $(\tilde \theta,\tilde \gamma) \in N_1$, we have
\[
 \sup_{m \in \mc M_\delta} |\E^{F_*}[m(U) r(U,\theta,\gamma)] - \mb E^{F_*}[m(U) r(U,\tilde \theta,\tilde \gamma)] | \leq \frac{\varepsilon}{\sqrt{d}}.
\] 
Let $N_2$ be an $\varepsilon$-neighborhood of $P$. Then for any $F \in \mc N_\delta$ and any $(\tilde \theta,\tilde \gamma,\tilde P) \in N_1 \times N_2$, we have
\[
 \|  (\E^{F} [ m(U) g(U,\theta,\gamma ) ] - \vec P) - (\E^{F} [ m(U) g(U,\tilde \theta,\tilde \gamma ) ] - \vec{\tilde P})  \| < 2\varepsilon \,,
\]
with $\vec{\tilde P} = (\tilde P,0_{d_3 + d_4})$, and so 
\begin{align*}
 B_{2\varepsilon} & \subseteq \mr{int}(\{ \E^F [ g(U,\tilde \theta,\tilde \gamma ) ] - \vec{\tilde P}: F \in \mc N_\delta \}+ \mc C ) 
 \subseteq \mr{int}(\{ \E^F [ g(U,\tilde \theta,\tilde \gamma ) ] - \vec{\tilde P}: F \in \mc N_\infty \}+ \mc C )\,,
\end{align*}
as required.
\end{proof}

\subsection{Continuity of the Optimal Values}

\begin{lemma}\label{lem:phi-cts}
Suppose that Assumption~\ref{a:phi} holds and  $\mb E^{F_*}[\phi^\star(a_1 + a_2'g(U,\theta,\gamma))]$ is continuous in $(\theta,\gamma)$ for every $(a_1,a_2) \in \mb R \times \mb R^d$. Then $\Delta(\tilde \theta;\tilde \gamma,\tilde P)$ and  $\Delta^\star(\tilde \theta;\tilde \gamma,\tilde P)$ are continuous in $(\tilde \theta,\tilde \gamma,\tilde P)$ at any point $(\theta, \gamma, P)$ at which Condition S' holds.
\end{lemma}

\begin{proof}[Proof of Lemma~\ref{lem:phi-cts}]
In view of Proposition~\ref{prop:dual_dist}(i), it suffices to establish continuity of $\Delta^\star(\tilde \theta;\tilde \gamma,\tilde P)$.

Fix some $(\theta,\gamma,P)$ at which  Conditions S' holds. Then by definition of Condition S', we have $\Delta(\theta;\gamma,P) < \infty$ and hence $\Delta^\star(\theta;\gamma,P) < \infty$. To simplify notation, let $\xi = (\zeta,\lambda)$ and $\Xi = \mb R \times \Lambda$ for the remainder of this proof. Note
\[
 \xi \mapsto L(\xi;\theta,\gamma,P) := -\E^{F_*}[ \phi^\star(- \zeta - \lambda' g(U,\theta, \gamma))  ] - \zeta - \lambda_{12}' P
\]
is the pointwise infimum of affine functions of $\xi$ and is therefore concave and u.s.c. By Proposition~\ref{prop:dual_dist}(iii), $L(\,\cdot\,;\theta,\gamma,P)$ has a nonempty, convex, and compact set of maximizers $\Xi_0 \subset \Xi$. Fix $\varepsilon > 0$ and let $\Xi^\varepsilon_0 = \{ \xi \in \Xi : d(\xi,\Xi_0) \leq \varepsilon\}$. 

By continuity of $(\theta,\gamma) \mapsto \mb E^{F_*}[\phi^\star(\zeta + \lambda'g(U,\theta,\gamma))]$, we have $L(\xi;\tilde \theta, \tilde \gamma, \tilde P) \to L(\xi; \theta, \gamma, P)$  as $(\tilde \theta,\tilde \gamma,\tilde P) \to (\theta,\gamma,P)$.
By concavity of $L$ in $\xi$, convergence may be strengthened to hold uniformly over the compact set $\Xi^\varepsilon_0$ \cite[Theorem 10.8]{Rockafellar}, and so
\begin{equation} \label{e:phi-cts:1}
 \sup_{\xi \in \Xi^\varepsilon_0} L(\xi;\tilde \theta, \tilde \gamma, \tilde P) \to \Delta^\star(\theta;\gamma,P) \quad \mbox{as} \quad (\tilde \theta,\tilde \gamma,\tilde P) \to (\theta,\gamma,P)\,.
\end{equation}
By u.s.c. of $L(\,\cdot\,; \theta,\gamma,P)$, definition of $\Xi_0$, and the fact that $\Xi_0 \cap \partial_\Xi \Xi_0^\varepsilon = \emptyset$, we also have that
\[
 \Delta^\star(\theta;\gamma,P) > \sup_{\xi \in \partial_{\Xi} \Xi^\varepsilon_0}  L(\xi; \theta, \gamma, P)  \,.
\]
It follows that there exists a neighborhood $N$ of $(\theta,\gamma,P)$ such that for any $(\tilde \theta,\tilde \gamma,\tilde P) \in N$,
\[
  \sup_{\xi \in \Xi^\varepsilon_0} L(\xi;\tilde \theta, \tilde \gamma, \tilde P) > \sup_{\xi \in \partial_{\Xi} \Xi^\varepsilon_0} L(\xi;\tilde \theta, \tilde \gamma, \tilde P)  \,.
\]
By standard arguments for maximizers of concave functions (e.g., the proof of Theorem 2.7 of \cite{NeweyMcFadden}), whenever $(\tilde \theta,\tilde \gamma,\tilde P) \in N$ we have that 
\[
  \sup_{\xi \in \Xi^\varepsilon_0} L(\xi;\tilde \theta, \tilde \gamma, \tilde P) = \sup_{\xi \in \Xi} L(\xi;\tilde \theta, \tilde \gamma, \tilde P) \equiv \Delta^\star(\tilde \theta;\tilde \gamma,\tilde P) \,.
\]
The result now follows by (\ref{e:phi-cts:1}).
\end{proof}

We now establish a similar result for the objectives (\ref{e:crit_l}) and (\ref{e:crit_u}) and their duals. 
To simplify notation, in the remainder of this subsection fix $(\theta,\gamma,P)$ and let $\xi = (\eta,\zeta,\lambda)$, $\Xi = \mb R_{+} \times \mb R \times \Lambda$, $\ul \Xi_\delta = \ul \Xi_\delta(\theta;\gamma,P)$ and  $\ol \Xi_\delta = \ol \Xi_\delta(\theta;\gamma,P)$. If $\ul \Xi_\delta$ and  $\ol \Xi_\delta$ are compact, then for each $\varepsilon > 0$ we may cover them by compact sets $\ul \Xi_\delta^\varepsilon \subset \Xi$ and $\ol \Xi_\delta^\varepsilon \subset \Xi$ formed as the union of finitely many hypercubes with edges parallel to the coordinate axes, so that $d(\xi,\ul \Xi_\delta) \leq \varepsilon$ for all $\xi \in \ul \Xi_\delta^\varepsilon$, $d(\xi,\ol \Xi_\delta) \leq \varepsilon$ for all $\xi \in \ol \Xi_\delta^\varepsilon$, $(\partial_{\Xi} \ul \Xi_\delta^\varepsilon) \cap \ul \Xi_\delta = \emptyset$, and $(\partial_\Xi \ol \Xi_\delta^\varepsilon) \cap \ol \Xi_\delta = \emptyset$.

\begin{lemma}\label{lem:mult:unif}
Suppose that Assumptions~\ref{a:phi} and \ref{a:m}(i),(ii) hold, Condition S' holds at $(\theta,\gamma,P)$, and $\Delta(\theta;\gamma,P) < \delta$. Then:
\begin{enumerate}[nosep]
\item[(i)] $\ul K_\delta(\tilde \theta;\tilde \gamma,\tilde P)$, $\ol K_\delta(\tilde \theta;\tilde \gamma,\tilde P)$, $\ul K_\delta^\star(\tilde \theta;\tilde \gamma,\tilde P)$, and $\ol K_\delta^\star(\tilde \theta;\tilde \gamma,\tilde P)$ are continuous in $(\tilde \theta,\tilde \gamma,\tilde P)$ at $(\theta,\gamma,P)$;
\item[(ii)] For each $\varepsilon > 0$ there exists a neighborhood $N$ of $(\theta,\gamma,P)$ such that  $\ul \Xi_\delta(\tilde \theta;\tilde \gamma,\tilde P) \subseteq \ul \Xi_\delta^\varepsilon$ and $\ol \Xi_\delta(\tilde \theta;\tilde \gamma,\tilde P) \subseteq \ol \Xi_\delta^\varepsilon$ for each $(\tilde \theta,\tilde \gamma,\tilde P) \in N$;
\item[(iii)] $\vec d_H(\ul \Xi_\delta(\tilde \theta;\tilde \gamma,\tilde P),\ul \Xi_\delta)  \to 0$ and $\vec d_H(\ol \Xi_\delta(\tilde \theta;\tilde \gamma,\tilde P),\ol \Xi_\delta)  \to 0$ as $(\tilde \theta,\tilde \gamma,\tilde P) \to (\theta,\gamma,P)$.
\end{enumerate}
\end{lemma}

\begin{proof}[Proof of Lemma~\ref{lem:mult:unif}]
We prove the result for $\ul K_\delta^\star$ and $\ul \Xi_\delta$; the result for  $\ol K_\delta^\star$ and $\ol \Xi_\delta$ follows similarly.

\underline{Step 1:} Preliminaries. In view of Proposition~\ref{prop:dual}(i), it suffices to establish continuity of $\ul K_\delta^\star(\theta;\gamma,P)$.
Lemmas~\ref{lem:attain-2} and \ref{lem:phi-cts} imply there is a neighborhood $N'$ of $(\theta,\gamma,P)$ such that Condition S' holds at $(\tilde \theta,\tilde \gamma,\tilde P)$ and $\Delta (\tilde \theta,\tilde \gamma,\tilde P) < \delta$ for each $(\tilde \theta,\tilde \gamma,\tilde P) \in N'$. By Proposition~\ref{prop:dual}(iii)(iv), the multipliers $\ul \Xi_\delta(\tilde \theta;\tilde \gamma,\tilde P)$ solving the program $\ul K_\delta^\star(\tilde \theta;\tilde \gamma,\tilde P)$ are a nonempty, convex, compact subset of $\Xi$ for each $(\tilde \theta,\tilde \gamma,\tilde P) \in N'$. 

The dual objective function
\[
 \xi \mapsto L(\xi;\theta,\gamma,P) := -\E^{F_*}[ (\eta\phi)^\star(- k(U,\theta,\gamma) - \zeta - \lambda' g(U,\theta, \gamma))  ] -\eta \delta - \zeta - \lambda_{12}' P
\]
is the pointwise infimum of affine functions of $\xi = (\eta,\zeta,\lambda)$ and is therefore concave and u.s.c. Hence,
\begin{equation}\label{e:kappa:cts:5}
 4a := \ul K_\delta^\star(\theta;\gamma,P) - \sup_{\xi \in \partial_\Xi \ul \Xi_\delta^\varepsilon} L(\xi;\theta,\gamma,P) > 0  \,.
\end{equation}
The remaining steps of the proof depend on whether or not $\min\{ \eta : (\eta,\zeta,\lambda) \in \ul \Xi_\delta\} > 0$. 

\underline{Step 2:} Proof of parts (i) and (ii) when $\min\{ \eta : (\eta,\zeta,\lambda) \in \ul \Xi_\delta\} > 0$.
W.l.o.g. we may choose $\ul \Xi_\delta^\varepsilon$ so that $\min\{ \eta : (\eta,\zeta,\lambda) \in \ul \Xi_\delta^\varepsilon\} > 0$. For any $\xi = (\eta,\zeta,\lambda)$ with $\eta > 0$,
\[
 L(\xi;\tilde \theta,\tilde \gamma,\tilde P) =  -\eta \E^{F_*}\left[ \phi^\star \left({\textstyle \frac{k(U,\tilde \theta,\tilde \gamma) + \zeta + \lambda' g(U,\tilde \theta,\tilde \gamma) }{-\eta} }\right) \right] - \eta \delta - \zeta - \lambda_{12}'\tilde P \,.
\]
By Assumption~\ref{a:m}(ii), for any $\xi \in \Xi$ we have $L(\xi;\tilde \theta, \tilde \gamma, \tilde P) \to L(\xi; \theta, \gamma, P)$ as $(\tilde \theta,\tilde \gamma,\tilde P) \to (\theta,\gamma,P)$. By concavity of $L$ in $\xi$, convergence may be strengthened to hold uniformly over the compact set $\ul \Xi_\delta^\varepsilon$ \cite[Theorem 10.8]{Rockafellar} and so, in particular,
\begin{equation}\label{e:kappa:cts:8}
 \sup_{\xi \in \ul \Xi_\delta^\varepsilon} L(\xi;\tilde \theta, \tilde \gamma, \tilde P) \to \ul K_\delta^\star(\theta;\gamma,P) \mbox{ as $(\tilde \theta,\tilde \gamma,\tilde P) \to (\theta,\gamma,P)$.}
\end{equation}
It follows that there exists a neighborhood $N''$ of $(\theta,\gamma,P)$ such that for $(\tilde \theta,\tilde \gamma,\tilde P) \in N' \cap N''$, we have
\[
  \sup_{\xi \in \ul \Xi_\delta^\varepsilon} L(\xi;\tilde \theta, \tilde \gamma, \tilde P) > \sup_{\xi \in \partial_\Xi \ul \Xi_\delta^\varepsilon} L(\xi;\tilde \theta, \tilde \gamma, \tilde P) \,.
\]
By similar arguments to the proof of Lemma~\ref{lem:phi-cts} we may deduce that $\ul \Xi_\delta(\tilde \theta;\tilde \gamma,\tilde P) \subseteq \ul \Xi_\delta^\varepsilon$ holds on $N := N' \cap N''$. This proves part (ii). Continuity (part (i)) now follows by (\ref{e:kappa:cts:8}).

\underline{Step 3:} Proof of part (ii) when $\min\{ \eta : (\eta,\zeta,\lambda) \in \ul \Xi_\delta\} = 0$. Choose $\xi_0 := (0,\ul \zeta,\ul \lambda) \in \ul \Xi_\delta$. Let $\xi_\eta = (\eta, \ul \zeta, \ul \lambda)$ for $\eta > 0$. As in the proof of Proposition~\ref{prop:dual}, it follows by concavity and u.s.c. of $L$ in $\xi$ that
\[
 \lim_{\eta \downarrow 0} L(\xi_\eta;\theta,\gamma,P) = L(\xi_0;\theta,\gamma,P) = \ul K_\delta^\star(\theta;\gamma,P)\,.
\]
For any $\varepsilon_0  \in (0,a)$, choose $\bar \eta > 0$ such that $L(\xi_{\bar \eta};\theta,\gamma,P)  > \ul K_\delta^\star(\theta;\gamma,P) - \varepsilon_0$ and $\xi_{\bar \eta} \in \mr{int}(\ul \Xi_\delta^\varepsilon)$. By Assumption~\ref{a:m}(ii), there is a neighborhood $N''$ of $(\theta,\gamma,P)$ such that for all $(\tilde \theta,\tilde \gamma,\tilde P) \in N''$ we have
\begin{equation}\label{e:kappa:cts:3}
 L(\xi_{\bar \eta};\tilde \theta, \tilde \gamma, \tilde P) > \ul K_\delta^\star(\theta;\gamma,P)-2\varepsilon_0 .
\end{equation} 

We now argue by contradiction that the inequality 
\begin{equation}\label{e:kappa:cts:4}
 \sup_{\xi \in \partial_\Xi \ul \Xi_\delta^\varepsilon} L(\xi;\theta,\gamma,P) \geq \sup_{\xi \in \partial_\Xi \ul \Xi_\delta^\varepsilon} L(\xi;\tilde \theta, \tilde \gamma, \tilde P) - 2\varepsilon_0 
\end{equation}
holds for all $(\tilde \theta,\tilde \gamma,\tilde P)$ in a neighborhood $N'''$ of $(\theta,\gamma,P)$. Suppose that there is $\varepsilon_1 > 0$ and $(\theta_n,\gamma_n,P_n) \to (\theta,\gamma,P)$ along which 
\begin{equation} \label{e:kappa:cts:1}
 \sup_{\xi \in \partial_\Xi \ul \Xi_\delta^\varepsilon} L(\xi;\theta,\gamma,P) \leq \sup_{\xi \in \partial_\Xi \ul \Xi_\delta^\varepsilon} L(\xi;\theta_n,\gamma_n,P_n) - \varepsilon_1 \,.
\end{equation}
For each $n \geq 1$, choose $\xi_n \in {\textstyle \mr{arg}\sup_{\xi \in \partial_\Xi \ul \Xi_\delta^\varepsilon}} L(\xi;\theta_n,\gamma_n,P_n)$. 
As $\partial_\Xi \ul \Xi_\delta^\varepsilon$ is compact, take a convergent subsequence $\{\xi_{n_l}\}$ and let $\xi^* = (\eta^*, \zeta^*, \lambda^*) \in \partial_\Xi \ul \Xi_\delta^\varepsilon$ denote its limit point. Suppose $\eta^* > 0$. Then as $L(\,\cdot\,;\theta_n,\gamma_n,P_n)$ converges uniformly to $L(\,\cdot\,;\theta,\gamma,P)$ on compact subsets of $(0,\infty) \times \mb R \times \mb R^d$,  we obtain $\lim_{l \to \infty} L(\xi_{n_l};\theta_{n_l},\gamma_{n_l},P_{n_l})  \leq \sup_{\xi \in \partial_\Xi \ul \Xi_\delta^\varepsilon} L(\xi; \theta, \gamma, P)$, which contradicts (\ref{e:kappa:cts:1}). 

Conversely, if $\eta^* = 0$, define $\xi_\eta^* = (\eta,\zeta^*,\lambda^*)$. Fix any small $\varepsilon_2 > 0$ so that $\xi_{\varepsilon_2}^* \in \partial_\Xi \ul \Xi_\delta^\varepsilon$. By u.s.c. and concavity of $L(\,\cdot\,;\theta,\gamma,P)$, we may choose $\varepsilon_2$ sufficiently small that 
\begin{equation} \label{e:kappa:cts:13}
 L(\xi_{\varepsilon_2}^*;\theta,\gamma,P) - L(\xi_{2 \varepsilon_2}^*;\theta,\gamma,P) <  \frac{1}{2}\varepsilon_1\,.
\end{equation}
For all $l$ large enough we have $\eta_{n_l} < \varepsilon_2$ and hence $\tau_{n_l} := \frac{\varepsilon_2}{2\varepsilon_2 - \eta_{n_l}} \in (0,1)$. By concavity,
\[
 L(\varepsilon_2,\zeta_{n_l},\lambda_{n_l};\theta_{n_l},\gamma_{n_l},P_{n_l}) \geq \tau_{n_l} L(\xi_{n_l};\theta_{n_l},\gamma_{n_l},P_{n_l}) + (1-\tau_{n_l}) L(2 \varepsilon_2,\zeta_{n_l},\lambda_{n_l};\theta_{n_l},\gamma_{n_l},P_{n_l}) \,.
\]
which rearranges to yield
\[
 L(\xi_{n_l};\theta_{n_l},\gamma_{n_l},P_{n_l}) \leq \frac{1}{\tau_{n_l}} \left( L(\varepsilon_2,\zeta_{n_l},\lambda_{n_l};\theta_{n_l},\gamma_{n_l},P_{n_l})  - (1-\tau_{n_l}) L(2 \varepsilon_2,\zeta_{n_l},\lambda_{n_l};\theta_{n_l},\gamma_{n_l},P_{n_l})  \right) \,.
\]
As $L(\, \cdot \,;\theta_n,\gamma_n,P_n)$ converges uniformly to $L(\, \cdot \,;\theta,\gamma,P)$ on compact subsets of $(0,\infty) \times \mb R \times \mb R^d$, we obtain
\[
 \lim_{l \to \infty} L(\xi_{n_l};\theta_{n_l},\gamma_{n_l},P_{n_l})  \leq 2 L(\xi_{\varepsilon_2}^*;\theta,\gamma,P) - L(\xi_{2 \varepsilon_2}^*;\theta,\gamma,P) \,.
\]
It follows by (\ref{e:kappa:cts:13}) that for all $l$ sufficiently large we have
\[
 L(\xi_{n_l};\theta_{n_l},\gamma_{n_l},P_{n_l}) < \sup_{\xi \in \partial_\Xi \ul \Xi_\delta^\varepsilon} L(\xi;\theta,\gamma,P) + \varepsilon_1 \,,
\]
which contradicts (\ref{e:kappa:cts:1}). This completes the proof of inequality (\ref{e:kappa:cts:4}).

It now follows from displays (\ref{e:kappa:cts:5}),  (\ref{e:kappa:cts:3}), and (\ref{e:kappa:cts:4}) that on $N' \cap N'' \cap N'''$ we have
\begin{align*}
 L(\xi_{\bar \eta};\tilde \theta, \tilde \gamma, \tilde P) > \ul K_\delta^\star(\theta;\gamma,P)-2\varepsilon_0 
 & = \sup_{\xi \in \partial_\Xi \ul \Xi_\delta^\varepsilon} L(\xi;\theta,\gamma,P) + 4a - 2 \varepsilon_0 \\
 & \geq  \sup_{\xi \in \partial_\Xi \ul \Xi_\delta^\varepsilon} L(\xi;\tilde \theta, \tilde \gamma, \tilde P)  + 4(a - \varepsilon_0)  \,.
\end{align*}
As $a - \varepsilon_0 > 0$, the inequality $\sup_{\xi \in \ul \Xi_\delta^\varepsilon} L(\xi;\tilde \theta, \tilde \gamma, \tilde P) > \sup_{\xi \in \partial_\Xi \ul \Xi_\delta^\varepsilon} L(\xi;\tilde \theta, \tilde \gamma, \tilde P)$ holds on $N := N' \cap N'' \cap N'''$.
It now follows by standard arguments for maximizers of concave objective functions (e.g., the proof of Theorem 2.7 of \cite{NeweyMcFadden}) that $\ul \Xi_\delta(\tilde \theta;\tilde \gamma,\tilde P) \subseteq \ul \Xi_\delta^\varepsilon$ holds on $N$, proving part (ii).

\underline{Step 4:} Proof of part (i) when $\min\{ \eta : (\eta,\zeta,\lambda) \in \Xi\} = 0$. 
For any $(\tilde \theta,\tilde \gamma,\tilde P) \in N$, we have
\begin{align} \label{e:kappa:cts:6}
 \ul K_\delta^\star(\tilde \theta;\tilde \gamma,\tilde P)  = \sup_{\xi \in  \ul \Xi_\delta^\varepsilon} L(\xi ;\tilde \theta, \tilde \gamma, \tilde P) 
\end{align}
by Step 3, so by (\ref{e:kappa:cts:3}) we have $\ul K_\delta^\star(\tilde \theta;\tilde \gamma,\tilde P)  \geq L(\xi_{\bar \eta};\tilde \theta, \tilde \gamma, \tilde P) > \ul K_\delta^\star(\theta;\gamma,P)-2\varepsilon_0$ for $(\tilde \theta,\tilde \gamma,\tilde P) \in N$, proving l.s.c. 

To establish u.s.c., for any $\varepsilon_0 > 0$ one may deduce (by similar arguments used to establish (\ref{e:kappa:cts:4}) in Step 3) that there is a neighborhood $N''''$ of $(\theta,\gamma,P)$ such that for any $(\tilde \theta, \tilde \gamma, \tilde P) \in N''''$ we have
\begin{equation} \label{e:kappa:cts:7}
 \sup_{\xi \in \ul \Xi_\delta^\varepsilon} L(\xi;\theta,\gamma,P) \geq \sup_{\xi \in \ul \Xi_\delta^\varepsilon}L(\xi;\tilde \theta, \tilde \gamma, \tilde P) - \varepsilon_0
\end{equation}
holds. It follows by (\ref{e:kappa:cts:6}) and (\ref{e:kappa:cts:7}) that on $N \cap N''''$, we have
\[
 \ul K_\delta^\star(\tilde \theta;\tilde \gamma,\tilde P) = \sup_{\xi \in \ul \Xi_\delta^\varepsilon} L(\xi;\tilde \theta, \tilde \gamma, \tilde P) \leq \sup_{\xi \in \ul \Xi_\delta^\varepsilon} L(\xi;\theta,\gamma,P)  + \varepsilon_0 = \ul K_\delta^\star(\theta;\gamma,P) +  \varepsilon_0 \,.
\]

\underline{Step 5:} Proof of part (iii).
By part (ii), for each $\varepsilon > 0$ there exists a neighborhood $N$ such that $\ul \Xi_\delta(\tilde \theta;\tilde \gamma,\tilde P) \subseteq \ul \Xi_\delta^\varepsilon$ for all $(\tilde \theta,\tilde \gamma,\tilde P) \in N$. Therefore, $\vec d_H(\ul \Xi_\delta(\tilde \theta;\tilde \gamma,\tilde P),\ul \Xi_\delta) \leq \vec d_H(\ul \Xi_\delta^\varepsilon,\ul \Xi_\delta) \leq \varepsilon$ for all $(\tilde \theta,\tilde \gamma,\tilde P) \in N$.
\end{proof}

\subsection{Supplemental Results on Directional Derivatives}\label{ax:derivative_proofs}

This Appendix presents results to supplement the proof of Theorem~\ref{t:ci}.

\begin{lemma}\label{lem:deriv:set}
Let $\ul \Theta_{\delta,a} = \{ \theta \in \Theta_\delta(P_0) : \ul K_\delta(\theta;P_0) \leq \ul b_\delta(P_0) + a\}$.  Suppose that Assumptions~\ref{a:phi} and \ref{a:m}(i)-(iii),(v),(vi) hold. Then there is $a > 0$ and a neighborhood $N$ of $P_0$ such that Condition S' holds at $(\theta,P)$ and $\Delta(\theta;P) < \delta$ for all $(\theta, P) \in \ul \Theta_{\delta,a} \times N$.
\end{lemma}

\begin{proof}[Proof of Lemma~\ref{lem:deriv:set}]
Suppose the assertion is false. Then there is $a_n \downarrow 0$, $\theta_n \in \ul \Theta_{\delta,a_n}$, and $P_n \to P$ for which Condition S' does not hold at $(\theta_n;P_n)$ and/or $\Delta(\theta_n;P_n) \geq \delta$. By Assumption~\ref{a:m}(v), we can extract a convergent subsequence $\theta_{n_l} \to \ul \theta$. By similar arguments to Step 1 of the proof of Theorem~\ref{t:asydist}, we may deduce that $ \ul \theta \in \ul \Theta_\delta$. Then by Lemma~\ref{lem:dual_stable} and Assumption~\ref{a:m}(iii)(vi) we must have that Condition S' holds at $(\theta_{n_l},P_{b_l})$ and $\Delta(\theta_{n_l};P_{n_l}) < \delta$ for all $l$ sufficiently large, a contradiction. 
\end{proof}

\begin{lemma}\label{lem:deriv:est}
Suppose that the conditions of Theorem~\ref{t:ci} hold. Then:
\begin{enumerate}[topsep=-18pt,itemsep=0pt,parsep=0pt,partopsep=0pt]
\item[(i)] $|\wh{d\ul b}_{\delta,P_0}[h_1] - \wh{d\ul b}_{\delta,P_0}[h_2]| \leq \ul C_n \| h_1 - h_2 \|$ and $|\wh{d\ol b}_{\delta,P_0}[h_1] - \wh{d\ol b}_{\delta,P_0}[h_2]| \leq \ol C_n \| h_1 - h_2 \|$ for all $h_1,h_2 \in \mb R^{d_1+d_2}$ with $\ul C_n = O_p(1)$ and $\ol C_n = O_p(1)$;
\item[(ii)] $\wh{d\ul b}_{\delta,P_0}[h] \to_p d\ul b_{\delta,P_0}[h]$ and $\wh{d\ol b}_{\delta,P_0}[h] \to_p d\ol b_{\delta,P_0}[h]$ for all $h \in \mb R^{d_1+d_2}$.
\end{enumerate} 
\end{lemma}

\begin{proof}[Proof of Lemma~\ref{lem:deriv:est}]
We prove the results for the lower values; the results for the upper values follow similarly. 

\underline{Step 1:} Recall $\ul \Theta_{\delta,a}$ from Lemma~\ref{lem:deriv:set}. We show there exists a sequence $a_n \downarrow 0$ such that $\ul \Theta_\delta \subseteq \hat{\ul \Theta}_{\delta,n} \subseteq \ul \Theta_{\delta,a_n}$ holds with probability approaching one (wpa1). 

Fix any sufficiently small $a > 0$ and let $N$ be the neighborhood from Lemma~\ref{lem:deriv:set}. As $\ul \Theta_\delta \subset \ul \Theta_{\delta,a}$, for all $(\theta,P) \in \ul \Theta_\delta \times N$, we have that $\Delta(\theta;P) < \delta$ and Condition S' holds at $(\theta,P)$. Therefore, $ \ul \Theta_\delta \subseteq \Theta_\delta(\hat P)$ wpa1 by consistency of $\hat P$. Moreover,  $\ul \Lambda_\delta(\theta;P)$ is nonempty and compact for all $(\theta, P) \in \ul \Theta_\delta \times  N$ (cf. Proposition~\ref{prop:dual}(iv)). By similar arguments to the proof of Theorem~\ref{t:asydist}, we have
\[
 \ul K_\delta(\theta;P) = \ul K_\delta(\theta;P) - \ul K_\delta(\theta;P_0) + \ul b_\delta(P_0)  \leq \max_{\ul \lambda_{12} \in \Lambda_\delta(\theta;P)} - \ul \lambda_{12}'(P - P_0) + \ul b_\delta(P_0) \,,
\]
for any  $(\theta, P) \in \ul \Theta_\delta \times  N$, and so
\begin{equation} \label{eq:deriv:est-1}
 \sup_{\theta \in \ul \Theta_\delta} \ul K_\delta(\theta;P) - \ul b_\delta(P) \leq \sup_{\theta \in \ul \Theta_\delta} \max_{\ul \lambda_{12} \in \Lambda_\delta(\theta;P)} - \ul \lambda_{12}'(P - P_0) + \ul b_\delta(P_0) - \ul b_\delta(P).
\end{equation}
As $\ul \Theta_\delta$ is compact (see Step 1 of the proof of Theorem~\ref{t:asydist}) and $(\theta,P) \mapsto \max_{\ul \lambda_{12} \in \Lambda_\delta(\theta;P)} \| \ul \lambda_{12}\|$ is u.s.c. (by Lemma~\ref{lem:mult:unif}(ii)) on $\ul \Theta_\delta \times N$, we may choose a neighborhood $N'$ of $P_0$ upon which $\sup_{\theta \in \ul \Theta_\delta} \max_{\ul \lambda_{12} \in \Lambda_\delta(\theta;P)} \| \ul \lambda_{12}\| < \infty$. Setting $P = \hat P$ in (\ref{eq:deriv:est-1}) and noting $\ul b_\delta(P_0) - \ul b_\delta(\hat P) = O_p(n^{-1/2})$ by Theorem~\ref{t:asydist}, we may deduce that $\sup_{\theta \in \ul \Theta_\delta} \ul K_\delta(\theta;\hat P) - \ul b_\delta(\hat P) \leq O_p( n^{-1/2})$ and therefore that $\ul \Theta_\delta \subseteq \hat{\ul \Theta}_{\delta,n}$ wpa1.

By the almost sure representation theorem (see, e.g., \citeauthor{Shapiro1991} (\citeyear{Shapiro1991}, Theorem A1)), there exists a sequence of random vectors $\{(Z_n,\hat \nu_n)\}$ and a random vector $Z$ defined on a probability space $(\Omega,\mathscr F,\mb P)$ with $Z_n =_d \sqrt n(\hat P - P_0)$, $\hat \nu_n =_d \hat \nu$, $Z \sim N(0,\Sigma)$, and $(Z_n,\hat \nu_n) \to_{a.s.} (Z,\nu)$. Let $P_n = P_0 + n^{-1/2} Z_n$ so that $P_n =_d \hat P$. Fix any $\omega \in \Omega$ for which $(Z_n(\omega),\hat \nu_n(\omega)) \to (Z(\omega),\nu(\omega))$. 
Let $\hat{\ul \Theta}_{\delta,n}(\omega) = \{\theta \in \Theta_\delta(P_n(\omega)) : \ul K_\delta(\theta; P_n(\omega)) \leq \ul b_\delta(P_n(\omega)) + \nu(\omega)\sqrt{\log n/n}\}$. The set $\Theta_{\delta}(P_n(\omega))$, and therefore $\hat{\ul \Theta}_{\delta,n}(\omega)$, is nonempty for $n$ sufficiently large. 

Suppose there is $\{\theta_{n_l}(\omega)\}$ with $\theta_{n_l}(\omega) \in \hat{\ul \Theta}_{\delta,n}(\omega) \setminus \ul \Theta_{\delta, a}$ for all $l$. By Assumption~\ref{a:m}(v) (taking a further subsequence if necessary) we have $\theta_{n_l}(\omega) \to \ul \theta(\omega) \in \Theta$. We may deduce by similar arguments to Step 1 in the proof of Theorem~\ref{t:asydist} that $\ul \theta(\omega) \in \ul \Theta_\delta$. Then by Lemma~\ref{lem:mult:unif}(i) and Assumption~\ref{a:m}(iii)(vi), we have $\lim_{l \to \infty} \ul K_\delta(\theta_{n_l}(\omega);P_0)  = \ul K_\delta(\ul \theta(\omega);P_0) = \ul b_\delta(P_0)$. Moreover, $\lim_{l \to \infty} \Delta(\theta_{n_l}(\omega);P_0) = \Delta(\ul \theta(\omega);P_0) < \delta$ by Lemma~\ref{lem:phi-cts} and Assumption~\ref{a:m}(iii)(vi). Therefore, $\theta_{n_l}(\omega) \in \ul \Theta_{\delta,a}$ for $l$ sufficiently large, a contradiction. Therefore, $ \hat{\ul \Theta}_{\delta,n} \subseteq \ul \Theta_{\delta,a}$ wpa1 for each $a > 0$. 

In view of the above, we may choose $a_n \downarrow 0$ so that $\hat{\ul \Theta}_{\delta,n} \subseteq \ul \Theta_{\delta,a_n}$ wpa1.

\underline{Step 2:} We prove part (i). 
Let $a > 0$ and the neighborhood $N$ be as in Lemma~\ref{lem:deriv:set}. 
It follows by Proposition~\ref{prop:dual}(iv) that $\ul \Lambda_\delta(\theta;P)$ is compact and nonempty for all $(\theta,P) \in \ul \Theta_{\delta,a} \times N$. 
We may also deduce by similar arguments to Step 1 in the proof of Theorem~\ref{t:asydist} that $\ul \Theta_{\delta,a}$ is compact. Finally, as $(\theta,P) \mapsto \max_{\ul \lambda_{12} \in \ul \Lambda_\delta(\theta;P)} \| \ul \lambda_{12}\|$ is u.s.c. (by Lemma~\ref{lem:mult:unif}(ii)) on $\ul \Theta_{\delta,a} \times N$, we have $\sup_{\theta \in \ul \Theta_{\delta,a}} \max_{\ul \lambda_{12} \in \ul \Lambda_\delta(\theta;P)} \| \ul \lambda_{12}\| \leq C$ on a neighborhood $N'$ of $P_0$ for some $C < \infty$. Now, as $\hat P \in N \cap N'$ and $\hat{\ul \Theta}_{\delta,n} \subseteq \ul \Theta_{\delta,a_n}$ both hold wpa1, it follows from the fact that the $\max$ and $\min$ operations are Lipschitz and the Cauchy-Schwarz inequality that
\[
 \left| \wh{d\ul b}_{\delta,P_0} [h_1] - \wh{d\ul b}_{\delta,P_0}[ h_2] \right| \leq  C \| h_1 - h_2 \| \quad \mbox{for all $h_1,h_2 \in \mb R^{d_1 + d_2}$}
\]
wpa1, proving part (i).

\underline{Step 3:} We prove part (ii).
As $\ul \Theta_\delta \subseteq \hat{\ul \Theta}_{\delta,n} \subseteq \ul \Theta_{\delta,a_n}$ wpa1 by step 1, it suffices to show
\begin{align*}
 \inf_{\theta \in \ul \Theta_\delta} \max_{\ul \lambda_{12} \in \ul \Lambda_\delta(\theta, \hat P)} - \ul \lambda_{12}'h & \to_p  d\ul b_{\delta,P_0}[h] \,, &
 \inf_{\theta \in \ul \Theta_{\delta,a_n}} \max_{\ul \lambda_{12} \in \ul \Lambda_\delta(\theta, \hat P)} - \ul \lambda_{12}'h & \to_p  d\ul b_{\delta,P_0}[h] 
\end{align*}
for all $h \in \mb R^{d_1 + d_2}$. Using the almost sure representation in Step 1, fix any $\omega \in \Omega$ for which $Z_n(\omega) \to Z(\omega)$. Let $\ul \theta_h$ solve $\min_{\theta \in \ul \Theta_\delta} \max_{\ul \lambda_{12} \in \ul \Lambda_\delta(\theta;P_0)} - \ul \lambda_{12}h$. By Lemma~\ref{lem:mult:unif}(ii), we have
\begin{align}
 & \limsup_{n \to \infty} \inf_{\theta \in \ul \Theta_\delta} \max_{\ul \lambda_{12} \in \ul \Lambda_\delta(\theta, P_n(\omega))} - \ul \lambda_{12}'h \notag \\
 & \leq \limsup_{n \to \infty} \max_{\ul \lambda_{12} \in \ul \Lambda_\delta(\ul \theta_h, P_n(\omega))} - \ul \lambda_{12}'h
 \leq  \max_{\ul \lambda_{12} \in \ul \Lambda_\delta(\ul \theta_h, P_0)} - \ul \lambda_{12}'h = d\ul b_{\delta,P_0}[h] \,. \label{e:dd:ub}
\end{align}
Now choose $\ul \theta_n(\omega) \in \ul \Theta_{\delta,a_n}$ for which
\[
 \inf_{\theta \in \ul \Theta_{\delta,a_n}} \max_{\ul \lambda_{12} \in \ul \Lambda_\delta(\theta, P_n(\omega))} - \ul \lambda_{12}'h \geq \max_{\ul \lambda_{12} \in \ul \Lambda_\delta(\theta_n(\omega), P_n(\omega))} - \ul \lambda_{12}'h - n^{-1} \,.
\]
Take any subsequence $\{\theta_{n_l}(\omega), P_{n_l}(\omega)\}$. By Assumption~\ref{a:m}(v) (taking a further subsequence if necessary) we have $\theta_{n_l}(\omega) \to \ul \theta(\omega) \in \Theta$. By similar arguments to Step 1 of the proof of Theorem~\ref{t:asydist}, we may deduce $ \ul \theta(\omega) \in \ul \Theta_\delta$. Lemma~\ref{lem:mult:unif}(ii) and Assumption~\ref{a:m}(vii') together imply the correspondence $(\theta,P) \mapsto \ul \Lambda_\delta(\theta,P)$ is continuous at $(\ul \theta, P_0)$. Therefore,
\begin{align*}
 & \liminf_{l \to \infty} \inf_{\theta \in \ul \Theta_{\delta,a_{n_l}}} \, \max_{\ul \lambda_{12} \in \ul \Lambda_\delta(\theta, P_{n_l}(\omega))} - \ul \lambda_{12}'h \\
 & \geq \liminf_{l \to \infty} \max_{\ul \lambda_{12} \in \ul \Lambda_\delta(\theta_{n_l}(\omega), P_{n_l}(\omega))} - \ul \lambda_{12}'h 
 \geq \max_{\ul \lambda_{12} \in \ul \Lambda_\delta(\ul \theta(\omega), P_0)} - \ul \lambda_{12}'h
 \geq  d\ul b_{\delta,P_0}[h] \,.
\end{align*}
As the lower bound on the right-hand side does not depend on the subsequence, we therefore have $\liminf_{n \to \infty} \inf_{\theta \in \ul \Theta_{\delta,a_n}} \max_{\ul \lambda_{12} \in \ul \Lambda_\delta(\theta, P_n(\omega))} - \ul \lambda_{12}'h \geq d\ul b_{\delta,P_0}[h]$. This, in conjunction with the upper bound (\ref{e:dd:ub}), completes the proof.
\end{proof}

\subsection{Proofs for Appendix~\ref{ax:extension}}\label{ax:extension_proofs}

\begin{proof}[Proof of Proposition~\ref{prop:criterion-exch}]
We prove only the result for $\ul K_\delta^\Pi$; the result for $\ol K_\delta^\Pi$ follows similarly.

Dropping dependence of $k$ and $g$ on $(\theta,\gamma)$ to simplify notation, we have 
\begin{align}
 \ul K_\delta^\Pi(\theta;\gamma,P) & = \inf_{F \in \mc N_\delta^\Pi} \mb E^F[k^\Pi(U)] \quad \mbox{subject to} \quad \mb E^F[g_1^\Pi(U)] \leq P_1 \,,\, \ldots, \, \mb E^F[g_4^\Pi(U)] = 0 \notag \\
 & \geq  \inf_{F \in \mc N_\delta^{\phantom{\Pi}}} \mb E^F[k^\Pi(U)] \quad \mbox{subject to} \quad \mb E^F[g_1^\Pi(U)] \leq P_1 \,,\, \ldots, \, \mb E^F[g_4^\Pi(U)] = 0  \,, \label{prop:criterion-exch-1}
\end{align}
where the first line uses $\Pi$-invariance of $F \in \mc N_\delta^\Pi$ and the second is because $\mc N_\delta^\Pi \subseteq \mc N_\delta$. Problem (\ref{prop:criterion-exch-1}) is an optimization over $\mc N_\delta$ may therefore be restated as problem (\ref{e:dual:1_exch}) by virtue of Proposition~\ref{prop:criterion}. The right-hand side of (\ref{e:dual:1_exch}) is therefore a lower bound for $\ul K_\delta^\Pi(\theta;\gamma,P)$. 

To establish equality, first suppose  problem (\ref{prop:criterion-exch-1}) is infeasible. Then problem (\ref{e:crit_l_ex}) must also be infeasible because $\mc N_\delta \supseteq \mc N_\delta^\Pi$, so the value of problem (\ref{e:crit_l_ex})  is $+\infty$. By Proposition~\ref{prop:criterion}, the value of the dual program to (\ref{prop:criterion-exch-1}), and hence the right-hand side of (\ref{e:dual:1_exch}), must also be $+\infty$. 

Now suppose that problem (\ref{prop:criterion-exch-1}) is feasible. We claim that there exists a minimizing $F \in \mc N_\delta$ that is $\Pi$-invariant, so it is without loss of generality to optimize over $\mc N_\delta$ rather than $\mc N_\delta^\Pi$. First note by Lemma~\ref{lem:v_properties} there exists a (not necessarily $\Pi$-invariant) minimizing $F \in \mc N_\delta$, say $F_0$.  Let $m_0 = \frac{\mr d F_0}{\mr d F_*}$ and let $\kappa_0 = \E^{F_0}[k^\Pi(U)]$ denote the minimizing value of (\ref{prop:criterion-exch-1}). Note that $\kappa_0$ must also be the value of the dual program (\ref{e:dual:1_exch}) by virtue of Proposition~\ref{prop:criterion}.  Define $m_\Pi(u) = \frac{1}{|\Pi|} \sum_{\varpi \in \Pi} m_0(\varpi u)$ and $F_\Pi$ by $\mr d F_\Pi = m_\Pi \mr d F_*$. 

We have $\E^{F_*}[m_\Pi(U)] = \frac{1}{|\Pi|} \sum_{\varpi \in \Pi} \E^{F_*}[ m_0(\varpi U)] = \frac{1}{|\Pi|} \sum_{\varpi \in \Pi} \E^{F_*}[ m_0( U)] = 1$ by $\Pi$-invariance of $F_*$, so $F_\Pi$ is a probability measure. By convexity of $\phi$ we also have 
\[
 D_\phi(F_\Pi\|F_*) = \E[\phi(m_\Pi(U))] \leq \frac{1}{|\Pi|} \sum_{\varpi \in \Pi} \E^{F_*}[\phi(m_0(\varpi U))] = \frac{1}{|\Pi|} \sum_{\varpi \in \Pi} \E^{F_*}[\phi(m_0( U))] \leq \delta
\]
because $F_0 \in \mc N_\delta$. This proves $F_\Pi \in \mc N_\delta$. Finally, as $\Pi$ is a group, for any function $h$ and any $\varpi \in \Pi$, we have 
\[
 \sum_{\varpi' \in \Pi} h(\varpi' u) 
 = \sum_{\varpi' \in \Pi} h(\varpi' \varpi^{-1} \varpi u)
 = \sum_{\varpi' \in \Pi} h(\varpi' \varpi u) \,.
\]
Therefore, $m_\Pi(u) = m_\Pi(\varpi u)$ for each $\varpi \in \Pi$. It follows that for any $\varpi \in \Pi$ and $A \subset \mc U$, 
\[
\begin{aligned}
 \E^{F_\Pi}[ \ind\{ \varpi U \in A\}] 
 = \E^{F_*}[m_\Pi( U) \ind\{ \varpi U \in A\}] 
 & = \E^{F_*}[m_\Pi(\varpi U) \ind\{ \varpi U \in A\}] \\
 & = \E^{F_*}[m_\Pi(U) \ind\{ U \in A\}] = \E^{F_\Pi}[ \ind\{U \in A\}] \,,
\end{aligned}
\]
by $\Pi$-invariance of $F_*$. Hence, $F_\Pi$ is $\Pi$-invariant and so $F_\Pi \in \mc N_\delta^\Pi$. 

It remains to show that $F_\Pi$ yields the optimal value $\kappa_0$  satisfies (\ref{e:mod}). Note by $\Pi$-invariance of $F_*$ that 
\[
\begin{aligned}
 \E^{F_\Pi}[k(U)] & = \frac{1}{|\Pi|} \sum_{\varpi \in \Pi} \E^{F_*}[ m_0(\varpi U) k(\varpi \varpi^{-1} U) ] \\
 & = \frac{1}{|\Pi|} \sum_{\varpi \in \Pi} \E^{F_*}[ m_0(U) k( \varpi^{-1} U) ] \\
 & = \frac{1}{|\Pi|} \sum_{\varpi \in \Pi} \E^{F_*}[ m_0(U) k( \varpi U) ]
 = \E^{F_0}[k^\Pi(U)] = \kappa_0 \,,
\end{aligned}
\]
where the first equality on the final line is because $\{\varpi^{-1} : \varpi \in \Pi\} = \Pi$. An identical argument  shows that $\E^{F_\Pi}[g_1(U)] = \mb E^{F_0}[g_1^\Pi(U)] \leq P_1$, \ldots, $\E^{F_\Pi}[g_4(U)] = \mb E^{F_0}[ g_4^\Pi(U)] = 0$.
\end{proof}

\begin{proof}[Proof of Proposition~\ref{prop:criterion-conditional}]
The minimization problem is additively separable across each $x \in \mc X$. The result follows by applying Proposition~\ref{prop:criterion} for each $x$.
\end{proof}

\begin{proof}[Proof of Proposition~\ref{prop:criterion-nonsep}]
Follows by similar arguments to the proof of Proposition~\ref{prop:criterion}.
\end{proof}

\subsection{Proofs for Appendix~\ref{ax:sharp}}\label{ax:sharp_proofs}

\begin{proof}[Proof of Lemma~\ref{lem:limits}]
We prove the result only for $\ul K_\infty$; the result for $\ol K_\infty$ follows similarly.
 First suppose that $\inf \mc K_\infty$ is finite. Fix any $\varepsilon > 0$. Then there is $F_\varepsilon \in \mc N_\infty$ and $\theta_\varepsilon \in \Theta$ such that (\ref{e:mod}) holds at $(\theta_\varepsilon,\gamma_0,P_0)$ under $F_\varepsilon$ and $\mb E^{F_\varepsilon}[k(U,\theta_\varepsilon,\gamma_0)] < \inf \mc K_\infty + \varepsilon$. Then for any $\delta \geq D_\phi(F_\varepsilon\|F_0)$ we have $\ul \kappa_\delta < \inf \mc K_\infty + \varepsilon$. Conversely, $\inf \mc K_\infty = -\infty$, then for each $n \in \mb N$ there exists $F_n \in \mc N_\infty$ and $\theta_n \in \Theta$ such that (\ref{e:mod}) holds at $(\theta_n,\gamma_0,P_0)$ under $F_n$ and $\mb E^{F_n}[k(U,\theta_n,\gamma_0)] < -n$. But then for any $\delta \geq D_\phi(F_n\|F_0)$ we necessarily have $\ul \kappa_\delta < -n$. 
\end{proof}

\begin{proof}[Proof of Lemma~\ref{lem:dual:infty}]
We prove the result only for $\ul K_\infty$; the result for $\ol K_\infty$ follows similarly.

We follow similar arguments to Appendix~\ref{sec:dual_K}. Dropping dependence of $k$ and $g$ on $(\theta,\gamma)$, consider 
\begin{equation} \label{e:p0_infty}
 \inf_F \mb E^F[k(U)] \quad \mbox{subject to} \quad \mb E^F[g_1(U)] \leq P_1 \,,\, \ldots, \, \mb E^F[g_4(U)] = 0.
\end{equation}
Identify each $F \in \mc N_\infty$ with its Radon--Nikodym derivative $m = \mr d F/\mr d F_* \in \mc L$ (see Appendix~\ref{ax:Orlicz}). Define $\varphi_\infty : \mc L \times \mb R^{d+1} \to \mb R \cup \{+\infty\}$ by
\begin{align*}
 \varphi_\infty(m,y) = \langle m, k \rangle + \mathbb I_{C_+}(m) + \mathbb I_{C_2}\left(\langle m, 1 \rangle - 1 + y_1, \langle m, g \rangle - \vec P + y_2 \right) ,
\end{align*}
where $y_1 \in \mb R$, $y_2 \in \mb R^d$, $\mb I_{C} : \mb R^{d+1} \to \mb R \cup \{+ \infty\}$ is given by
\[
 \mb I_{C_2}(y_1,y_2) = \left[ \begin{array}{cl} 0 & \mbox{if $y_1 = 0$, and $y_2 \in \mb R_{-}^{d_1} \times \{0_{d_2}\} \times \mb R_{-}^{d_3} \times \{0_{d_4}\}$}, \\
 +\infty & \mbox{otherwise},
 \end{array} \right.
\]
and $\mb I_{C_+} : \mc L \to \mb R \cup \{+ \infty\}$ is given by
\[
 \mb I_{C_+}(m) = \left[ \begin{array}{cl} 0 & \mbox{if $m \in \mc L_+$}, \\
 +\infty & \mbox{otherwise}.
 \end{array} \right.
\]
The \emph{primal} problem  for $y \in \mb R^{d+1}$ is $\min_{m \in \mc L} \varphi_\infty(m, y)$ and its \emph{value} is $v_\infty(y) = \inf_{m \in \mc L} \varphi_\infty(m, y)$.
By similar arguments to Lemmas~\ref{lem:varphi_properties} and \ref{lem:v_properties}, one may deduce that $\varphi_\infty$ is proper and convex and that $v_\infty : \mb R^{d+1} \to \mb R \cup \{+\infty\}$ is convex and its effective domain is $\mc Y_2$ (see display (\ref{e:y2})).

Now consider the \emph{dual} problem $\max_{y^\star \in \mb R^{d+1}} \left( y' y^\star - \varphi^\star_\infty(0, y^\star)\right)$ at $y = 0$, where $\varphi^\star_\infty : \mc E \times \mb R^{d+1} \to \mb R \cup \{+\infty\}$ is the convex conjugate of $\varphi_\infty$.
By direct calculation, with $y^\star = (y_1^\star,y_2^\star) \in \mb R \times \mb R^d$ we have
\begin{align*}
 \varphi^\star_\infty(0,y^\star) & = \sup_{(m,y) \in \mc L_+ \times \mb R^{d+1}} \left(  y' y^\star - \langle m, k \rangle  - \mb I_{C_2}\left( \langle m, 1 \rangle - 1 + y_1, \langle m, g \rangle - \vec P + y_2 \right) \right) \\
 & = \sup_{m \in \mc L_+ } \left( - y_1^\star ( \langle m,1 \rangle - 1 ) - y_2^{\star\prime} \left( \langle m, g \rangle - \vec P \right) - \langle m, k \rangle \right)  + \mb I_{C_2^o}(y^\star) ,
\end{align*}
where $C_2^o = \mb R \times \Lambda$, $\mb I_{C_2^o}(y^\star) = 0$ if $y^\star \in C_2^o$ and $+\infty$ otherwise, and it suffices to optimize over $\mc L_+$ because $\mb I_{C_+}(m) = +\infty$ for $m \in \mc L \setminus \mc L_+$. Write $y^\star \in C^o_2$ as $y^\star = (\zeta,\lambda)$, where $\zeta \in \mb R$ and $\lambda \in \Lambda$. By \citeauthor{RW} (\citeyear{RW}, Definition 14.59 and Theorem 14.60), we have
\begin{align*}
 \varphi^\star_\infty(0,(\zeta,\lambda)) 
 & = \sup_{m \in \mc L_+ } \mb E^{F_*} \left[ m(U) (-k(U) - \zeta - \lambda' g(U))   \right] + \zeta + \lambda' \vec P \\ 
 & = \mb E^{F_*} \left[ \sup_{x \geq 0} x (-k(U) - \zeta - \lambda' g(U))   \right] + \zeta + \lambda' \vec P \\
 & = \left[ \begin{array}{cl}
 \zeta + \lambda' \vec P & \mbox{if $\zeta + F_*\text{-}\mr{ess}\inf (  k + \lambda'g ) \geq 0$}, \\
 +\infty & \mbox{otherwise}, \end{array} \right.
\end{align*}
provided $(\zeta,\lambda) \in \mb R \times \Lambda$. The dual value at $y = 0$ is therefore
\begin{align}
 & \sup_{\zeta \in \mb R, \lambda \in \Lambda} - \zeta - \lambda' \vec P \quad \mbox{subject to } \zeta + F_*\text{-}\mr{ess}\inf (  k + \lambda'g ) \geq 0 \notag \\
 & = \sup_{\lambda \in \Lambda : F_*\text{-}\mr{ess}\inf (  k + \lambda'g ) > - \infty} \left( F_*\text{-}\mr{ess}\inf (  k + \lambda'g ) - \lambda' \vec P \right) \label{e:p0_dual_infty}.
\end{align}

Lemma~\ref{lem:cq}(ii) implies that $0 \in \mr{ri}(\mc Y_2)$ under Condition S. It then follows by \citeauthor{BS} (\citeyear{BS}, Propositions 2.147 and 2.148(iii)) that the primal and dual values (\ref{e:p0_infty})  and (\ref{e:p0_dual_infty}) are equal and the set of dual solutions is nonempty. 
\end{proof}

Lemma~\ref{lem:dual:sharp} is proved by modifying results of \cite{CM} from a setting with moment equality restrictions to one with moment inequality restrictions.

\begin{proof}[Proof of Lemma~\ref{lem:dual:sharp}]
We prove the result only for $\ul K_{np}$; the result for $\ol K_{np}$ follows similarly.

We drop dependence of $g$ and $k$ on $(\theta,\gamma)$ in what follows to simplify notation. Define $\mc M = \{m \in L^1(\mu) : \int m g \, \mr d \mu$ is finite$\}$ and $\mc M_+ = \{m \in \mc M: m \geq 0$ $\mu$-a.e.$\}$. Note $\mr d F/\mr d \mu \in \mc M_+$ if and only if $F \in \mc F_\theta$.  For $y = (y_1,y_2) \in \mb R \times \mb R^d$, let
\[
 \mc M[y] = \left\{m \in \mc M_+ : \int m \, \mr d \mu = 1 + y_1\,, \, \int m g \, \mr d \mu = \vec P + y_2 \right\} ,
\]
where the second integration is element-wise. 
Define $\varphi_a, \varphi_b : \mb R^{d+1} \to \mb R \cup \{+\infty\}$  by 
\[
\begin{aligned}
 \varphi_a(y) & = \inf_{ m \in \mc M[y]} \int mk \, \mr d \mu  \,, & & & 
 \varphi_b(y) = \left[ \begin{array}{cl}
 0 & \mbox{if } y \in \{0\} \times \mb R_-^{d_1} \times \{0_{d_2}\} \times \mb R_-^{d_3} \times \{0_{d_4}\} , \\
 +\infty & \mbox{otherwise,}
 \end{array} \right.
\end{aligned}
\]
with the understanding that $\varphi_a(y) = +\infty$ if the infimum runs over an empty set. Both $\varphi_a$ and $\varphi_b$ are proper and convex (note that $\mu\text{-}\mr{ess}\inf |k| < \infty$ ensures that $\varphi_a(y) > -\infty$ for all $y$ and $|\varphi_a(y)| < \infty$ whenever $\mc M[y] \not = \emptyset$).

Let $\mc V = \{ \int (1,g) \,m\, \mr d \mu : m \in \mc M_+\}$. Note that $\mr{dom}\,\varphi_a = \mc V - (1,\vec P )$. By Condition S$_{np}$ and similar arguments to the proof of Lemma~\ref{lem:cq}(ii), we have that $(1, \vec P) \in \mr{ri}( \mc V + (\{0 \} \times \mc C))$. By \citeauthor{Rockafellar} (\citeyear{Rockafellar}, Corollary 6.6.2), we also have $\mr{ri}( \mc V + (\{0 \} \times \mc C)) = \mr{ri}( \mc V ) + \mr{ri}( \{0 \} \times \mc C)$ and therefore $0 \in \mr{ri}(\mr{dom}\,\varphi_a) + \mr{ri}(\{0\} \times \mc C)$. As such, there exists $v \in \mb R^{d+1}$ such that $v \in \mr{ri}(\mr{dom}\,\varphi_a)$ and $-v \in \mr{ri}(\{0\} \times \mc C) \equiv - \mr{ri}(\mr{dom}\,\varphi_b)$, so $\mr{ri}(\mr{dom}\, \varphi_a)  \cap \mr{ri} (\mr{dom}\, \varphi_b)$ is nonempty. Then by Fenchel's Duality Theorem \cite[Theorem 31.1]{Rockafellar}, 
\begin{equation}\label{e:sharp-1}
 \ul K_{np}(\theta;\gamma,P) = \inf_{y \in \mb R^{d+1}} \left( \varphi_a(y) + \varphi_b(y) \right) = \sup_{y^\star \in \mb R^{d+1}} (-\varphi_a^\star(y^\star) - \varphi_b^\star(-y^\star)) ,
\end{equation}
where $\varphi_a^\star$ and $\varphi_b^\star$ are the convex conjugates of $\varphi_a$ and $\varphi_b$. Write $y^\star = (\zeta,\lambda)$. Then
\begin{equation}\label{e:sharp-2}
 \varphi_b^\star(-(\zeta,\lambda)) = \left[ \begin{array}{cl}
 0 & \mbox{if } -\lambda \in \Lambda ,\\[-4pt]
 +\infty & \mbox{otherwise,}
 \end{array} \right.
\end{equation}
and with $y = (y_1,y_2) \in \mb R \times \mb R_d$,
\[
\begin{aligned}
 - \varphi_a^\star((\zeta,\lambda)) 
 & = \inf_{y \in \mb R^{d+1}} \inf_{m \in \mc M[y]} \left( - \zeta y_1 - \lambda' y_2 + \int mk \, \mr d \mu \right) \\
 & = \inf_{y \in \mb R^{d+1}} \inf_{m \in \mc M[y]} \left( \zeta + \lambda' \vec P + \int  \left( k - \zeta - \lambda'g \right) m \, \mr d \mu \right) \\
 & = \inf_{m \in \mc M_+} \left( \zeta + \lambda' \vec P  +  \int  \left( k - \zeta - \lambda'g \right) m \, \mr d \mu  \right) .
\end{aligned}
\]
Let 
\[ 
 Q(u,m(u)) = \left[ \begin{array}{ll}
 k(u) m(u) & \mbox{if $m(u) \geq 0$} ,\\
 +\infty & \mbox{otherwise} ,
 \end{array} \right.
\]
so that
\[
 - \varphi_a^\star((\zeta,\lambda)) = \inf_{m \in \mc M} \left( \zeta + \lambda' \vec P  +  \int  Q(u,m(u)) -  \left( \zeta + \lambda' g(u) \right) m(u) \, \mr d \mu(u) \right) .
\]
By Remark A.3 and Theorem A.4 of \cite{CM}, we have
\begin{align}
 - \varphi_a^\star((\zeta,\lambda)) 
 & =  \zeta + \lambda' \vec P  + \int \inf_{x \geq 0} \left( k(u) - \zeta - \lambda' g(u)   \right) x \, \mr d \mu(u) \notag \\
 & = \left[ \begin{array}{cl}
 - \infty & \mbox{if $\mu\text{-}\mr{ess}\inf (k - \zeta - \lambda'g) < 0$} , \\
 \zeta + \lambda' \vec P & \mbox{otherwise.} \label{e:sharp-3}
 \end{array} \right.
\end{align}
It now follows from (\ref{e:sharp-1}), (\ref{e:sharp-2}), and (\ref{e:sharp-3}) that
\begin{align*}
 \ul K_{np}(\theta;\gamma,P) & = \sup_{\zeta \in \mb R, \lambda \in \Lambda : \mu\text{-}\mr{ess}\inf (k - \zeta + \lambda' g ) \geq 0} \zeta  - \lambda' \vec P \\
 & = \sup_{\lambda \in \Lambda : \mu\text{-}\mr{ess}\inf (k + \lambda'g ) > -\infty}  \mu\text{-}\mr{ess}\inf (k + \lambda'g ) - \lambda' \vec P \,,
\end{align*} 
as required.
\end{proof}

\subsection{Proofs for Appendix~\ref{ax:local}}\label{ax:local_proofs}

\begin{proof}[Proof of Theorem~\ref{t:local}]
To simplify notation, we drop dependence of $g(u,\theta,\gamma_0,P_{20})$ on $(\gamma_0,P_{20})$ and $k(u,\theta,\gamma_0)$ on $\gamma_0$. Under the stated regularity conditions, $k$ and each entry of $g$ belong to $L^2(F_*)$ for all $\theta$ in a neighborhood of $\theta_*$.

\underline{Step 1:} We first show $s \geq  2 \mb E^{F_*}[\iota(U)^2]$. Let $L^2_0(F_*) = \{b \in L^2(F_*) : \E^{F_*}[b(U)] = 0\}$. Define $\mb M : L^2_0(F_*) \to L^2_0(F_*)$ by
\[
 \mb M b = b - \mb E^{F_*}[b(U)g_*(U)'](V^{-1} - V^{-1}G(G'V^{-1}G)^{-1}G'V^{-1})g_* \,.
\]
Fix $b \in L^2_0(F_*)$. By Example 3.2.1 of \cite{BKRW}, define a smooth parametric family $\{F_t : t \in (-1,1)\}$ passing through $F_*$ at $t = 0$ via 
\[
 \frac{\mr dF_t}{\mr dF_*}  = \frac{\upsilon(t \mb Mb)}{\mb E^{F_*}[\upsilon(t \mb Mb(U))] } \,, \quad \mbox{where }  
 \upsilon(x)  = \frac{2}{1+e^{-2x}} \,.
\]
Fix any $d_\theta \times (d_2 + d_4)$ matrix $A$ of full rank. By the implicit function theorem and invertibility of $AG$, there exists $\varepsilon > 0$ such that $\mb E^{F_t}[Ag(U,\theta)] = 0$ has a unique solution $\theta(F_t) \in \Theta$ for all $t \in (-\varepsilon,\varepsilon)$, and
\[
 \left. \frac{d\theta(F_t)}{dt} \right|_{t = 0} = - (AG)^{-1} A\mb E^{F_*}[g_*(U)\mb Mb(U)] \,.
\]
Writing $\kappa(F_t) = \mb E^{F_t}[k(U,\theta(F_t))]$ and $\tilde \iota (u) = k_*(u) - \kappa_* - J' (AG)^{-1} Ag_*(u)$, we have
\begin{align*}
 \left. \frac{d\kappa(F_t)}{dt} \right|_{t = 0} & = \mb E^{F_*}[k_*(U)\mb Mb(U)] - J' (AG)^{-1} A\mb E^{F_*}[g_*(U)\mb Mb(U)] \\
 & = \mb E^{F_*} [\tilde \iota(U) \mb Mb(U)] \\
 & =  \mb E^{F_*} [\mb M\tilde \iota(U) \mb Mb(U)] \,,
\end{align*}
where the final line is because $\mb M$ is an orthogonal projection. 
For any $A$, we have
\begin{align*}
 \mb M\tilde \iota & = \mb Mk_* - J' (AG)^{-1} A ( g_* -\mb E^{F_*}[g_*(U)g_*(U)'](V^{-1} - V^{-1}G(G'V^{-1}G)^{-1}G'V^{-1})g_* ) \\
 & = \mb Mk_* - J' (G'V^{-1}G)^{-1}G'V^{-1} g_*  = \iota\,.
\end{align*}
Hence,
\[
 \frac{d\kappa(F_t)}{dt} |_{t = 0} = \mb E^{F_*} [ \iota(U) \mb Mb(U)].
\]
 
As $\phi(x) = \frac{1}{2}(x-1)^2$ for $x \geq 0$, a Taylor series expansion of $\upsilon(x)$ around $x = 0$ yields
\[
 D_\phi(F_t\|F_*) = \frac{t^2}{2}\mb E^{F_*}[(\mb Mb(U))^2] + o(t^2)\,.
\]
Therefore, whenever $\mb E^{F_*}[(\mb Mb(U))^2] \neq 0$, it follows from the preceding two displays that we have
\[
 \frac{(\kappa(F_t) - \kappa(F_{-t}))^2}{4D_\phi(F_t\|F_*)} = \frac{\mb E^{F_*}[\iota(U) \mb Mb(U)]^2 + o(1)}{\frac{1}{2}\mb E^{F_*}[(\mb Mb(U))^2] + o(1)}.
\]
Hence,
\[
 s \geq  \frac{\mb E^{F_*}[ \iota(U) \mb Mb(U)]^2}{\frac{1}{2}\mb E^{F_*}[(\mb Mb(U))^2]} \,.
\]
If $\iota(u) = 0$ ($F_*$-almost everywhere) then we trivially have $s \geq  2 \mb E^{F_*}[\iota(U)^2]$. Otherwise, choosing $b = \iota$ (which is valid because $\E^{F_*}[\iota(U)] = 0$ by construction) yields $s \geq  2 \mb E^{F_*}[\iota(U)^2]$.

\underline{Step 2:} We prove $s \leq  2 \mb E^{F_*}[\iota(U)^2]$ by contradiction. 
Suppose there exists a sequence $\delta_n \downarrow 0$ and $\varepsilon > 0$ such that
\[
 \frac{( \ol \kappa_{\delta_n} - \ul \kappa_{\delta_n})^2}{4\delta_n} \geq 2 \mb E^{F_*}[\iota(U)^2] + 2\varepsilon \,.
\]
for each $n$. We may then choose $\ul \theta_n, \ol \theta_n \in \Theta$ and $\ul F_n, \ol F_n \in \mc N_{\delta_n}$ such that $\ul F_n$ and $\ol F_n$ satisfy $\mb E^{\ol F_n}[g(U,\ol \theta_n)] = 0$ and $\mb E^{\ul F_n}[k(U,\ul \theta_n)] = 0$, and
\begin{equation}\label{e:local:ub1}
 \frac{(\mb E^{\ol F_n}[k(U,\ol \theta_n)] - \mb E^{\ul F_n}[ k(U,\ul \theta_n)])^2}{4\delta_n} \geq 2\mb E^{F_*}[\iota(U)^2] + \varepsilon \,.
\end{equation}
As $\Theta$ is compact, (taking a subsequence if necessary) we have $\ul \theta_n \to \ul \theta^* \in \Theta$ and $\ol \theta_n \to \ol \theta^* \in \Theta$. 

The spaces $\mc L$ and $\mc E$ with $\phi(x) = \frac{1}{2}(x-1)^2$ are equivalent to $L^2(F_*)$. Let $\|\cdot\|_2$ denote the $L^2(F_*)$ norm. Note $\mb E^{F_*}[\phi(m(U))] = \frac{1}{2}\|m-1\|_2^2$ where $m - 1$ is the function $u \mapsto m(u) - 1$. Let $\ul m_n = \mr d \ul F_n/\mr d F_*$ and $\ol m_n = \mr d \ol F_n/\mr d F_*$. As $\ul F_n, \ol F_n \in \mc N_{\delta_n}$, \begin{equation}\label{e:mbdd:l2}
  \| \ul m_n-1\|_2^2 ,  \| \ol m_n-1\|_2^2 \leq 2 \delta_n \downarrow 0 \quad \mbox{as $n \to \infty$. }
\end{equation}
By similar arguments to Lemma~\ref{lem:cgt}, we may deduce $\mb E^{F_*}[g(U,\ul \theta^*)] = \mb E^{F_*}[g(U,\ol \theta^*)] = 0$. It then follows by identifiability of $\theta_*$ that $\ul \theta^* = \ol \theta^* = \theta_*$.

By differentiability of $\theta \mapsto \mb E^{F_*}[g(u,\theta)]$ at $\theta_*$, we have 
\[
 -G(\ul \theta_n - \theta_*) + o( \| \ul \theta_n - \theta_*\|) = \mb E^{F_*}[ (\ul m_n(U)-1)g(U,\ul \theta_n)]  \quad \mbox{as $\ul \theta_n \to \theta_*$.}
\]
It follows by Cauchy--Schwarz and the fact that $G$ has full rank that $\|\ul \theta_n - \theta_*\| = O( \|\ul m_n-1\|_2 )$. Therefore, by (\ref{e:mbdd:l2}), Cauchy--Schwarz, and $L^2(F_*)$ continuity of $\theta \mapsto g(\cdot,\theta,\gamma_0,P_{20})$ at $\theta_*$,
\begin{equation} \label{e:local:ub4}
 -G(\ul \theta_n - \theta_*) = \mb E^{F_*}[ (\ul m_n(U)-1)g_*(U)] + o(  \delta_n^{1/2} )
\end{equation}
and so $\ul \theta_n - \theta_* = -(G'V^{-1}G )^{-1} G'V^{-1}\mb E^{F_*}[ (\ul m_n(U)-1)g_*(U)] + o( \delta_n^{1/2} )$. 
By similar arguments,
\begin{align*}
 & \mb E^{F_*}[\ul m_n(U) k(U,\ul \theta_n)] - \kappa_* \\
 & = \mb E^{F_*}[(\ul m_n(U)-1) (k_*(U)-\kappa_*)] + J'(\ul\theta_n - \theta_*) + o( \delta_n^{1/2} )\\ 
 & = \mb E^{F_*}[(\ul m_n(U)-1) (k_*(U)-\kappa_*-J'(G'V^{-1}G )^{-1} G'V^{-1}g_*(U))]  + o( \delta_n^{1/2} ) \,.
\end{align*}
However, by (\ref{e:local:ub4}) and definition of $\mb M$ we also have
\[
 \mb E^{F_*}[(\ul m_n(U)-1) (k_*(U)-\kappa_*-\mb M (k_*-\kappa_*)(U))] = o( \delta_n^{1/2} )  
\]
hence 
\[
 \mb E^{F_*}[\ul m_n(U) k(U,\ul \theta_n)] - \kappa_* = \mb E^{F_*}[(\ul m_n(U)-1) \iota(U)] + o(\delta_n^{1/2})\,,
\] 
with an analogous result holding for $\ol m_n$ and $\ol \theta_n$. It follows from the preceding display and its counterpart for $\ol m_n$ that
\begin{equation} \label{e:local:ub2}
 \frac{(\mb E^{\ol F_n}[k(U,\ol \theta_n)] - \mb E^{\ul F_n}[ k(U,\ul \theta_n)])^2}{4\delta_n} =  \frac{(\mb E^{F_*}[(\ol m_n(U) - \ul m_n(U)) \iota(U)])^2}{4\delta_n} + o(1) \,.
\end{equation}
Note that $\ol m_n \neq \ul m_n$ must hold for $n$ sufficiently large; otherwise, substituting (\ref{e:local:ub2}) into (\ref{e:local:ub1}) yields $o(1) \geq 2\mb E^{F_*}[\iota(U)^2] + \varepsilon$, a contradiction. By the triangle inequality and (\ref{e:mbdd:l2}) we have
\begin{equation} \label{e:local:ub3}
 \|\ol m_n - \ul m_n\|_2^2 \leq 2 \| \ol m_n - 1\|_2^2 + 2 \| \ul m_n - 1\|_2^2 \leq 8 \delta_n .
\end{equation}
It follows by substituting (\ref{e:local:ub2}) and (\ref{e:local:ub3}) into (\ref{e:local:ub1}) that
\[
 2\mb E^{F_*}[\iota(U)^2] + \varepsilon  \leq \frac{2(\mb E^{F_*}[(\ol m_n(U) - \ul m_n(U)) \iota(U)])^2}{\|\ol m_n - \ul m_n\|_2^2} + o(1) \leq 2 \mb E^{F_*}[\iota(U)^2] + o(1) ,
\]
where the second inequality is by Cauchy--Schwarz. As $n \to \infty$, $\varepsilon$ dominates the $o(1)$ term and we obtain a contradiction.
\end{proof}

\begin{lemma}\label{lem:local}
Suppose that the conditions of Theorem~\ref{t:local} hold, $(\hat \theta,\hat \gamma,\hat P_2) \to_p (\theta_*,\gamma_0,P_{20})$, and $\mb E^{F_*}[g(U, \theta, \gamma, P_2)g(U, \theta, \gamma, P_2)']$, $\mb E^{F_*}[g(U, \theta, \gamma, P_2)k(U, \theta, \gamma)]$, $\mb E^{F_*}[g(U, \theta, \gamma, P_2)]$, $\mb E^{F_*}[k(U, \theta, \gamma)]$, $\frac{\partial}{\partial \theta'} \mb E^{F_*}[g(U, \theta, \gamma, P_2)]$, $\frac{\partial}{\partial \theta'} \mb E^{F_*}[k(U, \theta, \gamma)]$, and $\mb E^{F_*}[k(U, \theta, \gamma)^2]$ are all continuous in $(\theta,\gamma,P_2)$ at $(\theta_*,\gamma_0,P_{20})$.
Then: $\hat s \to_p s$.
\end{lemma}

\begin{proof}[Proof of Lemma~\ref{lem:local}]
Immediate by consistency of $(\hat \theta,\hat \gamma,\hat P)$ and Slutsky's theorem.
\end{proof}

\section{Additional Details for Example~\ref{ex:simple}}\label{ax:simple}

We first check Condition S'. Note $D_\phi(F\|F^*) = \frac{1}{2}\theta^2$ with $F$ the $N(\theta,1)$ distribution. Therefore, $|\theta| < \sqrt{2 \delta}$ implies that Condition S' holds. In what follows we implicitly assume $|\theta| < \sqrt{ 2 \delta}$. 

By Proposition~\ref{prop:criterion} and the discussion thereafter for KL neighborhoods, 
\begin{align}
 \ol K_\delta(\theta) & = \inf_{\eta > 0, \lambda \in \mb R} \eta \log \E^{F_*} \left[ e^{ (\ind\{U \leq \theta\} - \lambda (U - \theta))/\eta} \right] + \eta \delta \notag \\
 & = \inf_{\eta > 0, \lambda \in \mb R} \eta \log \left( e^{1/\eta} \int_{-\infty}^\theta e^{-\lambda (u - \theta)/\eta} f_N(u) \, du + \int_\theta^\infty e^{-\lambda ( u - \theta)/\eta} f_N(u) \, du \right) + \eta \delta \notag \\
 & = \inf_{\eta > 0, \lambda \in \mb R} \lambda \theta + \frac{\lambda^2}{2\eta} + \eta \log \left\{ 1 + (e^{1/\eta} - 1) F_N\left( \theta + \frac{\lambda}{\eta} \right) \right\}  + \eta \delta  , \label{obj3}
\end{align}
where $f_N$ and $F_N$ denote the standard normal PDF and CDF, respectively, and the second line follows from the functional form of the moment generating function of the truncated normal distribution. As Condition S' holds at $\theta$, it follows from Proposition~\ref{prop:dual} that minimizing values of $\eta$ and $\lambda$ exist. The first-order conditions (FOCs)  are
\begin{align*}
 \eta: \quad 0 & = -\frac{\lambda^2}{2\eta^2} + \log \left\{ 1 + (e^{1/\eta} - 1) F_N\left( \theta + \frac{\lambda}{\eta} \right) \right\} + \delta  - \frac{1}{\eta} \frac{e^{1/\eta} F_N( \theta + \frac{\lambda}{\eta} ) + \lambda (e^{1/\eta} - 1) f_N( \theta + \frac{\lambda}{\eta} )}{1 + (e^{1/\eta} - 1) F_N( \theta + \frac{\lambda}{\eta} )} \\
 \lambda: \quad 0 & = \theta + \frac{\lambda}{\eta} + \frac{(e^{1/\eta} - 1)f_N( \theta + \frac{\lambda}{\eta} ) }{1 + (e^{1/\eta} - 1) F_N( \theta + \frac{\lambda}{\eta} )}.
\end{align*}

The multiplier $\lambda$ enters both FOCs through the ratio $r := \lambda/\eta$. Rearranging the FOC for $\lambda$ yields
\begin{equation}
 0 = (\theta + r) + (e^{1/\eta} - 1) ((\theta + r) F_N(\theta + r) + f_N(\theta + r)) \,, \label{obj5}
\end{equation}
which implicitly defines a function $r(\eta)$ on $(0,\infty)$. One may verify that $r(\eta)$ is strictly increasing, with $r(\eta) \to -\infty$ as $\eta \downarrow 0$. An asymptotic expansion of the error function \cite[formula 8.254]{GradshteynRyzhik} yields
\begin{equation}\label{eq:erf-expansion-0}
 F_N(x) = \frac{f_N(x)}{-x} \left( 1 - \frac{1}{x^2} + O\left( \frac{1}{x^4} \right) \right), \quad \mbox{as $x \to -\infty$,}
\end{equation}
 and hence
\begin{equation}\label{eq:erf-expansion}
 x F_N(x) + f_N(x) = \frac{f_N(x)}{x^2} \left( 1 + O \left( \frac{1}{x^2} \right) \right) , \quad \mbox{as $x \to -\infty$.} 
\end{equation}
Substituting into (\ref{obj5}) and taking logs, we obtain for small $\eta$ (and hence large negative $r$) that
\[
  \log(e^{1/\eta} - 1) =  \frac{1}{2} (- \theta - r)^2 + \frac{1}{2} \log(2 \pi) + 3 \log(-\theta - r) + O \left( \frac{1}{(\theta + r)^2} \right) \,.
\]
As $\eta \log(e^{1/\eta}-1) \to 1$ as $\eta \downarrow 0$ and $\log(-x)/x^2 \to 0$ as $x \to -\infty$, we obtain  
\begin{equation} \label{eq:asy-eta-r}
 -\theta-r = \sqrt{\frac{2}{\eta}}(1+o(1))  \quad \mbox{as $\eta \downarrow 0$.}
\end{equation}
Also note by (\ref{obj5}), (\ref{eq:erf-expansion-0}), and (\ref{eq:erf-expansion}) that
\begin{equation} \label{obj8}
 (e^{1/\eta} - 1)F_N(\theta + r) = (\theta + r)^2\left(1 + O\left( \frac{1}{(\theta + r)^2} \right) \right).
\end{equation}

Substituting the FOC for $\lambda$ into the FOC for $\eta$ and rearranging, yields
\begin{align*}
 0 & = \frac{\theta \lambda}{\eta} + \frac{\lambda^2}{2\eta^2} + \log \left\{ 1 + (e^{1/\eta} - 1) F_N( \theta + \frac{\lambda}{\eta} ) \right\} + \delta  - \frac{1}{\eta} \frac{e^{1/\eta} F_N\left( \theta + \frac{\lambda}{\eta} \right) }{1 + (e^{1/\eta} - 1) F_N\left( \theta + \frac{\lambda}{\eta} \right)}  \\
 & = - \frac{1}{2}\theta^2 + \log  (e^{1/\eta} - 1) - \frac{1}{2} \log (2 \pi) - \log(-\theta - r) + \delta  - \frac{1}{\eta} \frac{F_N( \theta + r ) + (e^{1/\eta} - 1) F_N( \theta + r ) }{1 + (e^{1/\eta} - 1) F_N( \theta + r )} \notag .
\end{align*} 
Now substituting the approximation (\ref{obj8}) into the previous display and using the fact that $\eta \log(e^{1/\eta}-1) \to 1$ as $\eta \downarrow 0$, we obtain
\begin{align*}
 0 & =-\frac{1}{2}\theta^2 + \frac{1}{\eta} + o(1) -\frac{1}{2}\log(2 \pi) -  \log(-\theta - r)  + \delta - \frac{1}{\eta} \left( 1 - \frac{1}{(\theta + r)^2} + o \left(  \frac{1}{(\theta + r)^2} \right) \right) \\
 & =-\frac{1}{2}\theta^2 + o(1) -\frac{1}{2}\log(2 \pi) -  \log(-\theta - r)  + \delta + \frac{1}{\eta (\theta + r)^2} + o \left(  \frac{1}{\eta(\theta + r)^2} \right) .
\end{align*}
Note by (\ref{eq:asy-eta-r}) that $\eta (\theta + r)^2 \to 2$ as $\eta \downarrow 0$. It follows that as $\delta \to \infty$, all terms in the above display remain bounded aside from $\log(-\theta - r)$. Therefore, the optimal $r = \lambda/\eta$ behaves like
\begin{equation} \label{eq:asy-eta-r-1}
 \log(-\theta - r) = \frac{1}{2}(1 - \theta^2) - \frac{1}{2} \log (2 \pi) + \delta + o(1) .
\end{equation}
Equivalently, by (\ref{eq:asy-eta-r}) we have that the optimal $\eta$ behaves like
\begin{equation} \label{eq:asy-eta-delta}
 \log \eta = \log (4 \pi) - (1 - \theta^2)  - 2 \delta + o(1). 
\end{equation}
Note, in particular, that this implies that the optimal $\eta$ is always positive but converges to zero at an exponential rate in $\delta$. This approximation is verified numerically in Figure~\ref{fig:toy_eta}.

\begin{figure}[t]
\makebox[\textwidth]{
\begin{subfigure}{.5\textwidth}
  \centering
  \includegraphics[trim = 0.3cm 0.3cm 0.3cm 0cm, clip, width=\linewidth]{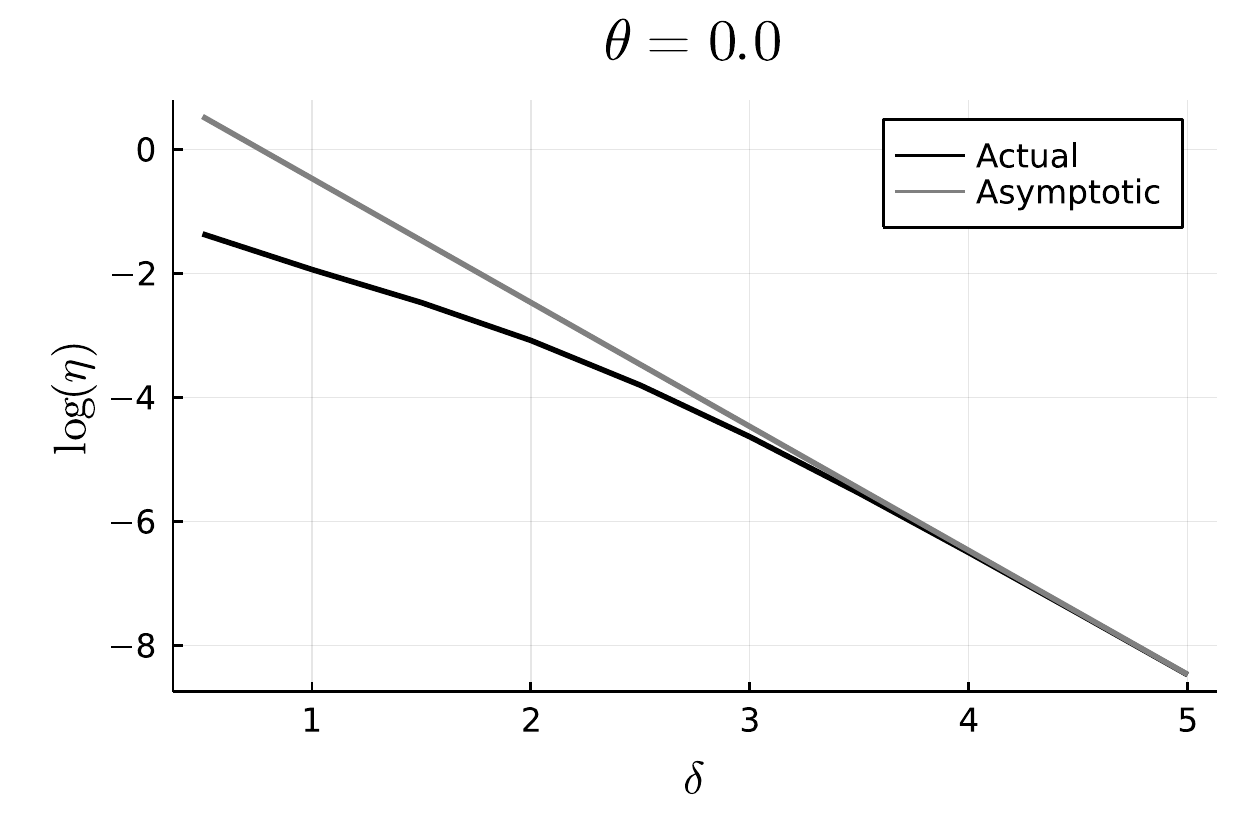}
  \end{subfigure}%
\begin{subfigure}{.5\textwidth}
  \centering
  \includegraphics[trim = 0.3cm 0.3cm 0.3cm 0cm, clip, width=\linewidth]{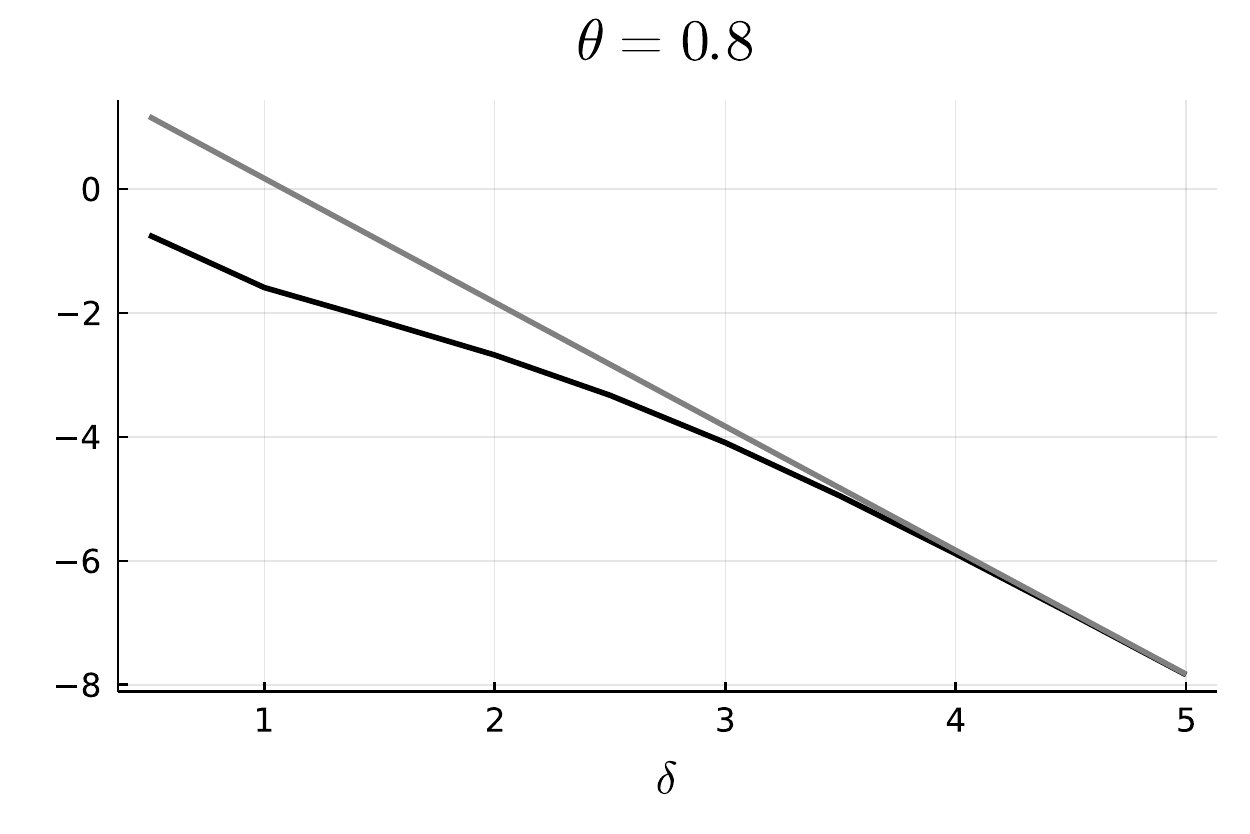}
\end{subfigure}}
\centering

\caption{\label{fig:toy_eta}Optimal value of $\log(\eta)$ and its large-$\delta$ approximation.}

\end{figure}

Finally, substituting the FOC for $\lambda$ into (\ref{obj3}) and then using (\ref{eq:asy-eta-r-1}), we obtain
\begin{align*}
 \ol K_\delta (\theta) & = \eta \left( -\frac{1}{2} \theta^2  +  \log (e^{1/\eta} - 1) - \frac{1}{2} \log(2 \pi) - \log(- \theta - r)  +  \delta \right) \\
 & = \eta \left( -\frac{1}{2} \theta^2  +  \log (e^{1/\eta} - 1) - \frac{1}{2} \log(2 \pi) - \frac{1}{2} + \frac{1}{2} \theta^2 + \frac{1}{2} \log(2 \pi) - \delta + o(1) +  \delta \right) \\
 & = 1 - \frac{\eta}{2} + o(\eta)\,,
\end{align*}
It now follows by substituting (\ref{eq:asy-eta-delta}) into the final line of the above display that
\[
 \ol K_\delta(\theta) = 1 - 2 \pi e^{-2 \delta - (1 - \theta^2)}(1+o(1)).
\]
This approximation is verified numerically in Figure~\ref{fig:toy_K}. 

\begin{figure}[t]
\makebox[\textwidth]{
\begin{subfigure}{.5\textwidth}
  \centering
  \includegraphics[trim = 0.3cm 0.3cm 0.3cm 0cm, clip, width=\linewidth]{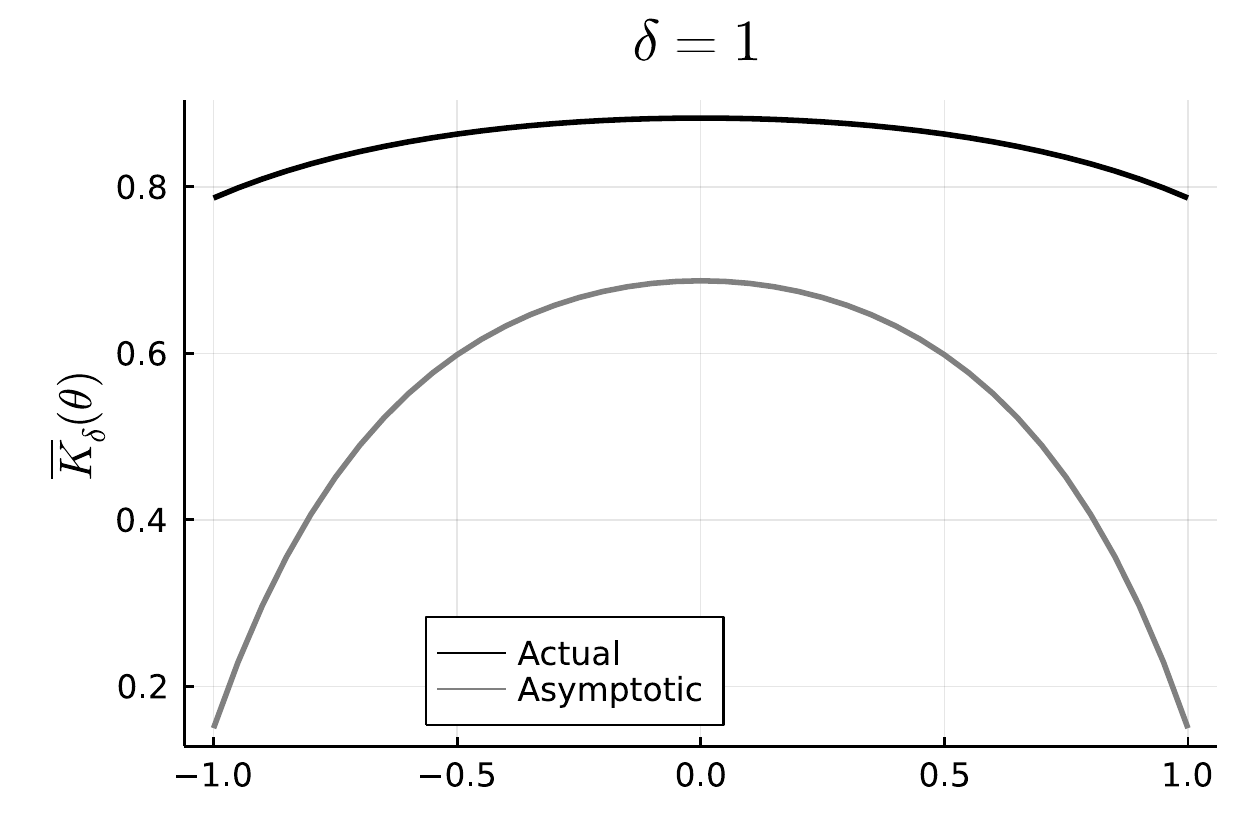}
  \end{subfigure}%
\begin{subfigure}{.5\textwidth}
  \centering
  \includegraphics[trim = 0.3cm 0.3cm 0.3cm 0cm, clip, width=\linewidth]{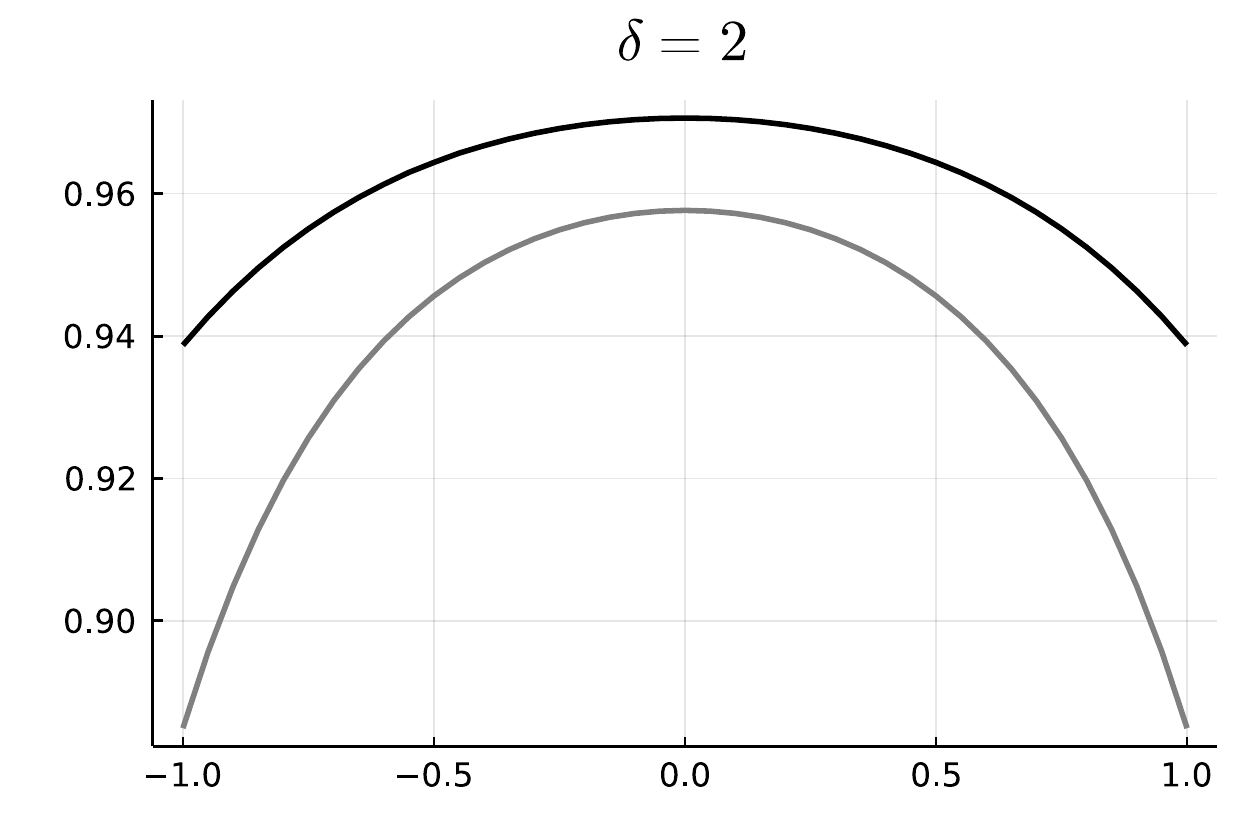}
\end{subfigure}}

\vskip 4pt

\makebox[\textwidth]{
\begin{subfigure}{.5\textwidth}
  \centering
  \includegraphics[trim = 0.3cm 0.3cm 0.3cm 0cm, clip, width=\linewidth]{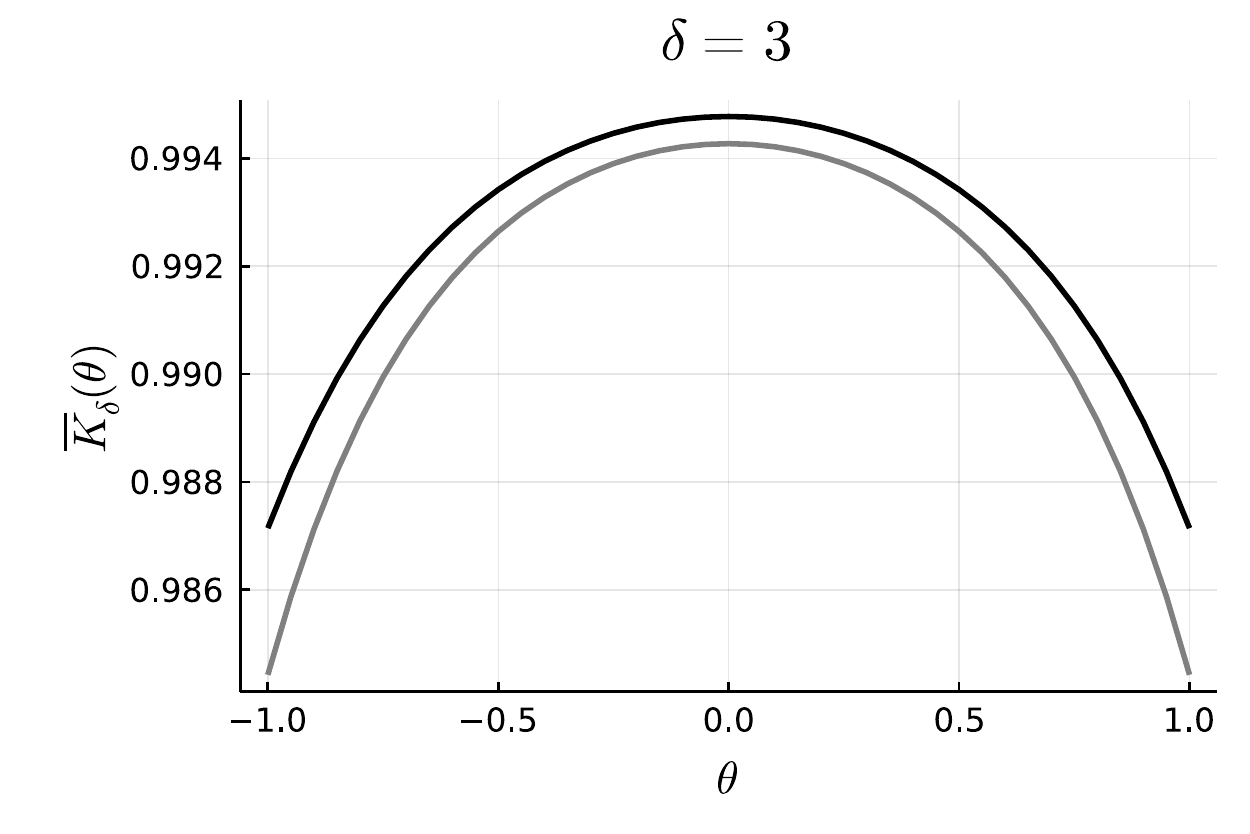}
  \end{subfigure}%
\begin{subfigure}{.5\textwidth}
  \centering
  \includegraphics[trim = 0.3cm 0.3cm 0.3cm 0cm, clip, width=\linewidth]{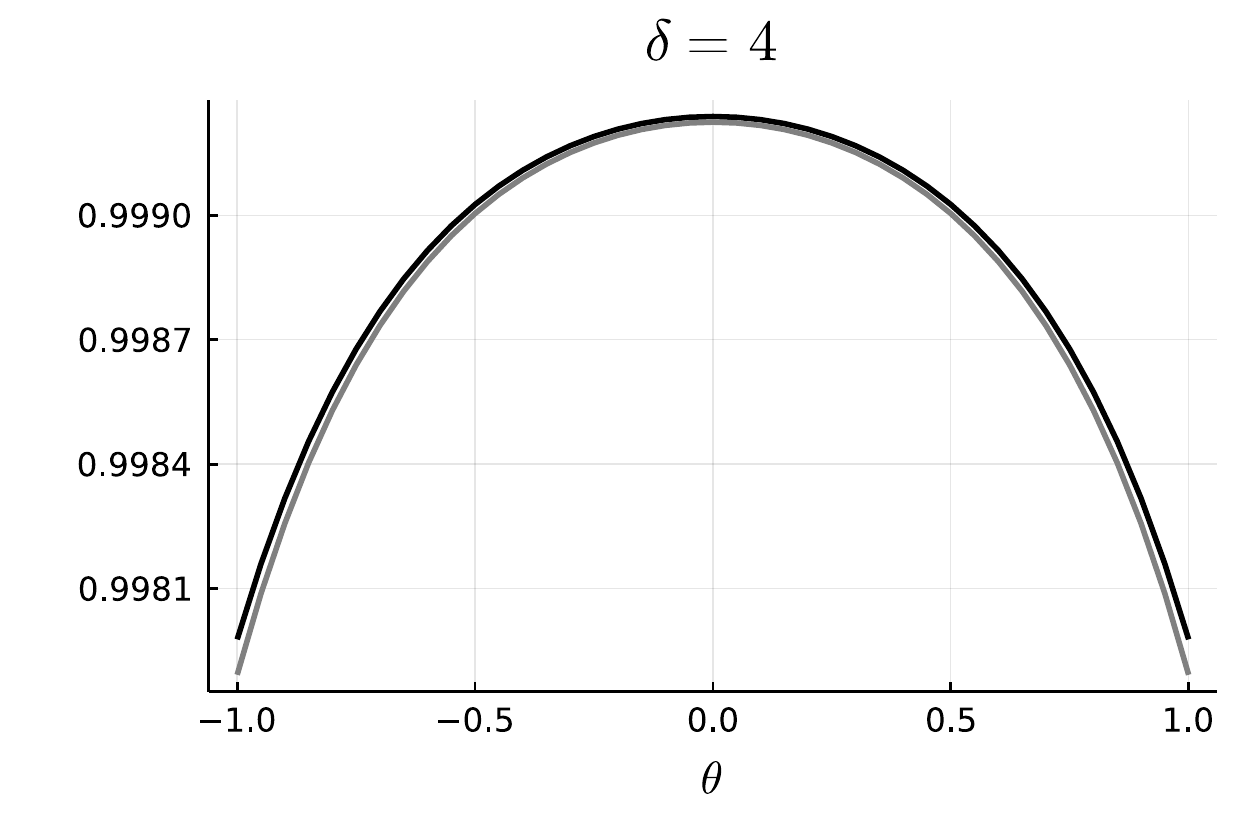}
\end{subfigure}}
\centering

\caption{\label{fig:toy_K} Objective function $\ol K_\delta(\theta)$ and its large-$\delta$ approximation.}

\end{figure}

It remains to derive an asymptotic expansion for the optimal value $\ol \kappa_\delta$. To this end, by the envelope theorem and the FOC for $\lambda$, we have
\[
 \frac{\partial \ol K_\delta(\theta)}{\partial \theta} = - \eta \theta,
\]
where $\eta > 0$ solves the inner problem at $\theta$. 
It follows that $\ol K_\delta(\theta)$ is maximized at $\theta = 0$ for all $\delta > 0$, and therefore that $\ol \kappa_\delta = 1 - 2 \pi e^{-2 \delta - 1}(1 + o(1))$.

\let\oldbibliography\thebibliography
\renewcommand{\thebibliography}[1]{\oldbibliography{#1}
\setlength{\itemsep}{0pt}}

{
\putbib
}

\endgroup

\end{bibunit}


\begin{thebibliography}{}

\bibitem[\protect\citeauthoryear{Aguirregabiria and Mira}{Aguirregabiria and
  Mira}{2007}]{AM2007}
Aguirregabiria, V. and P.~Mira (2007).
\newblock Sequential estimation of dynamic discrete games.
\newblock {\em Econometrica\/}~{\em 75\/}(1), 1--53.

\bibitem[\protect\citeauthoryear{Allen and Rehbeck}{Allen and
  Rehbeck}{2019}]{AllenRehbeck}
Allen, R. and J.~Rehbeck (2019).
\newblock Identification with additively separable heterogeneity.
\newblock {\em Econometrica\/}~{\em 87\/}(3), 1021--1054.

\bibitem[\protect\citeauthoryear{Andrews, Gentzkow, and Shapiro}{Andrews
  et~al.}{2017}]{AGS2017}
Andrews, I., M.~Gentzkow, and J.~M. Shapiro (2017).
\newblock Measuring the sensitivity of parameter estimates to estimation
  moments.
\newblock {\em The Quarterly Journal of Economics\/}~{\em 132\/}(4),
  1553--1592.

\bibitem[\protect\citeauthoryear{Andrews, Gentzkow, and Shapiro}{Andrews
  et~al.}{2020}]{AGS2018}
Andrews, I., M.~Gentzkow, and J.~M. Shapiro (2020).
\newblock On the informativeness of descriptive statistics for structural
  estimates.
\newblock {\em Econometrica\/}~(6), 2231--2258.

\bibitem[\protect\citeauthoryear{Armstrong and Koles\'{a}r}{Armstrong and
  Koles\'{a}r}{2021}]{AK}
Armstrong, T.~B. and M.~Koles\'{a}r (2021).
\newblock Sensitivity analysis using approximate moment condition models.
\newblock {\em Quantitative Economics\/}~{\em 12\/}(1), 77--108.

\bibitem[\protect\citeauthoryear{Bajari, Benkard, and Levin}{Bajari
  et~al.}{2007}]{BBL}
Bajari, P., C.~L. Benkard, and J.~Levin (2007).
\newblock Estimating dynamic models of imperfect competition.
\newblock {\em Econometrica\/}~{\em 75\/}(5), 1331--1370.

\bibitem[\protect\citeauthoryear{Beresteanu, Molchanov, and
  Molinari}{Beresteanu et~al.}{2011}]{BMM}
Beresteanu, A., I.~Molchanov, and F.~Molinari (2011).
\newblock Sharp identification regions in models with convex moment
  predictions.
\newblock {\em Econometrica\/}~{\em 79\/}(6), 1785--1821.

\bibitem[\protect\citeauthoryear{Berger}{Berger}{1984}]{Berger1984}
Berger, J.~O. (1984).
\newblock The robust {B}ayesian viewpoint (with discussion).
\newblock In J.~B. Kadane (Ed.), {\em Robustness of Bayesian analyses}.
  North-Holland.

\bibitem[\protect\citeauthoryear{Berry}{Berry}{1992}]{Berry1992}
Berry, S.~T. (1992).
\newblock Estimation of a model of entry in the airline industry.
\newblock {\em Econometrica\/}~{\em 60\/}(4), 889--917.

\bibitem[\protect\citeauthoryear{Berry and Haile}{Berry and
  Haile}{2010}]{BerryHaile2010}
Berry, S.~T. and P.~A. Haile (2010).
\newblock Nonparametric identification of multinomial choice demand models with
  heterogeneous consumers.
\newblock {Cowles Foundation Discussion Paper No. 1718}.

\bibitem[\protect\citeauthoryear{Berry and Haile}{Berry and
  Haile}{2014}]{BerryHaile2014}
Berry, S.~T. and P.~A. Haile (2014).
\newblock Identification in differentiated products markets using market level
  data.
\newblock {\em Econometrica\/}~{\em 82\/}(5), 1749--1797.

\bibitem[\protect\citeauthoryear{Bonhomme and Weidner}{Bonhomme and
  Weidner}{2021}]{BW}
Bonhomme, S. and M.~Weidner (2021).
\newblock Minimizing sensitivity to model misspecification.
\newblock {\em Quantitative Economics\/}, forthcoming.

\bibitem[\protect\citeauthoryear{Bresnahan and Reiss}{Bresnahan and
  Reiss}{1990}]{BresnahanReiss}
Bresnahan, T.~F. and P.~C. Reiss (1990).
\newblock Entry in monopoly market.
\newblock {\em The Review of Economic Studies\/}~{\em 57\/}(4), 531--553.

\bibitem[\protect\citeauthoryear{Bresnahan and Reiss}{Bresnahan and
  Reiss}{1991}]{BresnahanReiss1991}
Bresnahan, T.~F. and P.~C. Reiss (1991).
\newblock Empirical models of discrete games.
\newblock {\em Journal of Econometrics\/}~{\em 48\/}(1), 57--81.

\bibitem[\protect\citeauthoryear{Chamberlain and Leamer}{Chamberlain and
  Leamer}{1976}]{ChamberlainLeamer}
Chamberlain, G. and E.~E. Leamer (1976).
\newblock Matrix weighted averages and posterior bounds.
\newblock {\em Journal of the Royal Statistical Society. Series B
  (Methodological)\/}~{\em 38\/}(1), 73--84.

\bibitem[\protect\citeauthoryear{Chen, Christensen, and Tamer}{Chen
  et~al.}{2018}]{CCT}
Chen, X., T.~M. Christensen, and E.~Tamer (2018).
\newblock Monte carlo confidence sets for identified sets.
\newblock {\em Econometrica\/}~{\em 86\/}(6), 1965--2018.

\bibitem[\protect\citeauthoryear{Chen, Tamer, and Torgovitsky}{Chen
  et~al.}{2011}]{CTT}
Chen, X., E.~Tamer, and A.~Torgovitsky (2011).
\newblock Sensitivity analysis in semiparametric likelihood models.
\newblock {Cowles Foundation Discussion Paper No. 1836}.

\bibitem[\protect\citeauthoryear{Chesher and Rosen}{Chesher and
  Rosen}{2017}]{ChesherRosen2017}
Chesher, A. and A.~M. Rosen (2017).
\newblock Generalized instrumental variable models.
\newblock {\em Econometrica\/}~{\em 85\/}(3), 959--989.

\bibitem[\protect\citeauthoryear{Chiappori, Salani\'{e}, and Weiss}{Chiappori
  et~al.}{2017}]{CSW}
Chiappori, P.-A., B.~Salani\'{e}, and Y.~Weiss (2017).
\newblock Partner choice, investment in children, and the marital college
  premium.
\newblock {\em American Economic Review\/}~{\em 107\/}(8), 2109--2167.

\bibitem[\protect\citeauthoryear{Choo and Siow}{Choo and Siow}{2006}]{ChooSiow}
Choo, E. and A.~Siow (2006).
\newblock Who marries whom and why.
\newblock {\em Journal of Political Economy\/}~{\em 114\/}(1), 175--201.

\bibitem[\protect\citeauthoryear{Ciliberto and Tamer}{Ciliberto and
  Tamer}{2009}]{CilibertoTamer}
Ciliberto, F. and E.~Tamer (2009).
\newblock Market structure and multiple equilibria in airline markets.
\newblock {\em Econometrica\/}~{\em 77\/}(6), 1791--1828.

\bibitem[\protect\citeauthoryear{Csisz{\'a}r}{Csisz{\'a}r}{1975}]{Csiszar1975}
Csisz{\'a}r, I. (1975).
\newblock ${I}$-divergence geometry of probability distributions and
  minimization problems.
\newblock {\em The Annals of Probability\/}~{\em 3\/}(1), 146--158.

\bibitem[\protect\citeauthoryear{Csisz\'{a}r and Shields}{Csisz\'{a}r and
  Shields}{2004}]{CsiszarShields}
Csisz\'{a}r, I. and P.~C. Shields (2004).
\newblock Information theory and statistics: A tutorial.
\newblock {\em Communications and Information Theory\/}~{\em 1\/}(4), 417--528.

\bibitem[\protect\citeauthoryear{Cuturi}{Cuturi}{2013}]{Cuturi2013}
Cuturi, M. (2013).
\newblock Sinkhorn distances: Lightspeed computation of optimal transport.
\newblock {\em Advances in Neural Information Processing Systems\/}~{\em 26},
  2292--2300.

\bibitem[\protect\citeauthoryear{Dagsvik}{Dagsvik}{2000}]{Dagsvik}
Dagsvik, J.~K. (2000).
\newblock Aggregation in matching markets.
\newblock {\em International Economic Review\/}~{\em 41\/}(1), 27--58.

\bibitem[\protect\citeauthoryear{Duchi and Namkoong}{Duchi and
  Namkoong}{2021}]{DuchiNamkoong}
Duchi, J.~C. and H.~Namkoong (2021).
\newblock {Learning models with uniform performance via distributionally robust
  optimization}.
\newblock {\em The Annals of Statistics\/}~{\em 49\/}(3), 1378--1406.

\bibitem[\protect\citeauthoryear{Ekeland, Galichon, and Henry}{Ekeland
  et~al.}{2010}]{EGH}
Ekeland, I., A.~Galichon, and M.~Henry (2010).
\newblock Optimal transportation and the falsifiability of incompletely
  specified economic models.
\newblock {\em Economic Theory\/}~{\em 42\/}(2), 355--374.

\bibitem[\protect\citeauthoryear{Fang and Santos}{Fang and
  Santos}{2019}]{FangSantos}
Fang, Z. and A.~Santos (2019).
\newblock Inference on directionally differentiable functions.
\newblock {\em Review of Economic Studies\/}~{\em 86\/}(1), 377--412.

\bibitem[\protect\citeauthoryear{Galichon and Henry}{Galichon and
  Henry}{2011}]{GalichonHenry2011}
Galichon, A. and M.~Henry (2011).
\newblock {Set Identification in Models with Multiple Equilibria}.
\newblock {\em The Review of Economic Studies\/}~{\em 78\/}(4), 1264--1298.

\bibitem[\protect\citeauthoryear{Galichon and Salani\'{e}}{Galichon and
  Salani\'{e}}{2021}]{GalichonSalanie}
Galichon, A. and B.~Salani\'{e} (2021).
\newblock Cupid's invisible hand: Social surplus and identification in matching
  models.
\newblock {\em The Review of Economic Studies\/}, forthcoming.

\bibitem[\protect\citeauthoryear{Giacomini, Kitagawa, and Uhlig}{Giacomini
  et~al.}{2016}]{GKU}
Giacomini, R., T.~Kitagawa, and H.~Uhlig (2016).
\newblock Estimation under ambiguity.
\newblock Technical report, UCL and Chicago.

\bibitem[\protect\citeauthoryear{Grieco}{Grieco}{2014}]{Grieco}
Grieco, P. L.~E. (2014).
\newblock Discrete games with flexible information structures: an application
  to local grocery markets.
\newblock {\em The RAND Journal of Economics\/}~{\em 45\/}(2), 303--340.

\bibitem[\protect\citeauthoryear{Gualdani and Sinha}{Gualdani and
  Sinha}{2020}]{GualdaniSinha}
Gualdani, C. and S.~Sinha (2020).
\newblock Partial identification in matching models for the marriage market.
\newblock arxiv:1902.05610 [econ.em].

\bibitem[\protect\citeauthoryear{Hansen and Sargent}{Hansen and
  Sargent}{2001}]{HS2001}
Hansen, L.~P. and T.~J. Sargent (2001).
\newblock Robust control and model uncertainty.
\newblock {\em The American Economic Review\/}~{\em 91\/}(2), 60--66.

\bibitem[\protect\citeauthoryear{Hansen and Sargent}{Hansen and
  Sargent}{2008}]{HS2008}
Hansen, L.~P. and T.~J. Sargent (2008).
\newblock {\em Robustness}.
\newblock Princeton.

\bibitem[\protect\citeauthoryear{Ho}{Ho}{2018}]{Ho}
Ho, P. (2018).
\newblock Global robust bayesian analysis in large models.
\newblock Technical report, Princeton.

\bibitem[\protect\citeauthoryear{Kalouptsidi, Kitamura, Lima, and
  Souza-Rodrigues}{Kalouptsidi et~al.}{2020}]{KKLS}
Kalouptsidi, M., Y.~Kitamura, L.~Lima, and E.~Souza-Rodrigues (2020).
\newblock Partial identification and inference for dynamic models and
  counterfactuals.
\newblock {\em NBER Working Paper No. 26761\/}.

\bibitem[\protect\citeauthoryear{Kalouptsidi, Scott, and
  Souza-Rodrigues}{Kalouptsidi et~al.}{2021}]{KSS}
Kalouptsidi, M., P.~T. Scott, and E.~Souza-Rodrigues (2021).
\newblock Identification of counterfactuals in dynamic discrete choice models.
\newblock {\em Quantitative Economics\/}~{\em 12\/}(2), 351--403.

\bibitem[\protect\citeauthoryear{Kitamura, Otsu, and Evdokimov}{Kitamura
  et~al.}{2013}]{KOE}
Kitamura, Y., T.~Otsu, and K.~Evdokimov (2013).
\newblock Robustness, infinitesimal neighborhoods, and moment restrictions.
\newblock {\em Econometrica\/}~{\em 81\/}(3), 1185--1201.

\bibitem[\protect\citeauthoryear{Kline and Tamer}{Kline and
  Tamer}{2016}]{KlineTamer2016}
Kline, B. and E.~Tamer (2016).
\newblock Bayesian inference in a class of partially identified models.
\newblock {\em Quantitative Economics\/}~{\em 7\/}(2), 329--366.

\bibitem[\protect\citeauthoryear{Komunjer and Ragusa}{Komunjer and
  Ragusa}{2016}]{KR}
Komunjer, I. and G.~Ragusa (2016).
\newblock Existence and characterization of conditional density projections.
\newblock {\em Econometric Theory\/}~{\em 32\/}(4), 947--987.

\bibitem[\protect\citeauthoryear{Laff\'{e}rs}{Laff\'{e}rs}{2019}]{Laffers}
Laff\'{e}rs, L. (2019).
\newblock Bounding average treatment effects using linear programming.
\newblock {\em Empirical Economics\/}~{\em 57\/}(3), 727--767.

\bibitem[\protect\citeauthoryear{Leamer}{Leamer}{1982}]{Leamer1982}
Leamer, E.~E. (1982).
\newblock Sets of posterior means with bounded variance priors.
\newblock {\em Econometrica\/}~{\em 50\/}(3), 725--736.

\bibitem[\protect\citeauthoryear{Leamer}{Leamer}{1985}]{Leamer1985}
Leamer, E.~E. (1985).
\newblock Sensitivity analyses would help.
\newblock {\em The American Economic Review\/}~{\em 75\/}(3), 308--313.

\bibitem[\protect\citeauthoryear{Li}{Li}{2018}]{Li}
Li, L. (2018).
\newblock Identification of structural and counterfactual parameters in a large
  class of structural econometric models.
\newblock Technical report, Pennsylvania State University.

\bibitem[\protect\citeauthoryear{Liese and Vajda}{Liese and
  Vajda}{1987}]{LieseVajda}
Liese, F. and I.~Vajda (1987).
\newblock {\em Convex statistical distances}.
\newblock Teubner, Leipzig.

\bibitem[\protect\citeauthoryear{Maccheroni, Marinacci, and
  Rustichini}{Maccheroni et~al.}{2006}]{MMR}
Maccheroni, F., M.~Marinacci, and A.~Rustichini (2006).
\newblock Ambiguity aversion, robustness, and the variational representation of
  preferences.
\newblock {\em Econometrica\/}~{\em 74\/}(6), 1447--1498.

\bibitem[\protect\citeauthoryear{Manski}{Manski}{2007}]{Manski2007}
Manski, C.~F. (2007).
\newblock Partial identification of counterfactual choice probabilities.
\newblock {\em International Economic Review\/}~{\em 48\/}(4), 1393--1410.

\bibitem[\protect\citeauthoryear{Manski}{Manski}{2014}]{Manski2014}
Manski, C.~F. (2014).
\newblock Identification of income--leisure preferences and evaluation of
  income tax policy.
\newblock {\em Quantitative Economics\/}~{\em 5\/}(1), 145--174.

\bibitem[\protect\citeauthoryear{McFadden}{McFadden}{1974}]{McFadden1974}
McFadden, D.~L. (1974).
\newblock {\em Conditional Logit Analysis of Qualitative Choice Behavior}, pp.\
   105--142.
\newblock Frontiers in Econometrics. Academic Press.

\bibitem[\protect\citeauthoryear{McFadden}{McFadden}{1978}]{McFadden1978}
McFadden, D.~L. (1978).
\newblock {\em Modelling the Choice of Residential Location}, pp.\  75--96.
\newblock Spatial Interaction Theory and Planning Models. North Holland:
  Amsterdam.

\bibitem[\protect\citeauthoryear{Milgrom and Segal}{Milgrom and
  Segal}{2002}]{MilgromSegal}
Milgrom, P. and I.~Segal (2002).
\newblock Envelope theorems for arbitrary choice sets.
\newblock {\em Econometrica\/}~{\em 70\/}(2), 583--601.

\bibitem[\protect\citeauthoryear{Mukhin}{Mukhin}{2018}]{Mukhin}
Mukhin, Y. (2018).
\newblock Sensitivity of regular estimators.
\newblock arxiv:1805.08883 [econ.em].

\bibitem[\protect\citeauthoryear{Norets and Tang}{Norets and
  Tang}{2014}]{NoretsTang}
Norets, A. and X.~Tang (2014).
\newblock Semiparametric inference in dynamic binary choice models.
\newblock {\em The Review of Economic Studies\/}~{\em 81\/}(3), 1229--1262.

\bibitem[\protect\citeauthoryear{Owen}{Owen}{2017}]{Owen2017}
Owen, A.~B. (2017).
\newblock A randomized {H}alton algorithm in {R}.
\newblock arxiv:1706.02808 [stat.co].

\bibitem[\protect\citeauthoryear{{Qiao} and {Minematsu}}{{Qiao} and
  {Minematsu}}{2010}]{QiaoMinematsu}
{Qiao}, Y. and N.~{Minematsu} (2010).
\newblock A study on invariance of $f$-divergence and its application to speech
  recognition.
\newblock {\em IEEE Transactions on Signal Processing\/}~{\em 58\/}(7),
  3884--3890.

\bibitem[\protect\citeauthoryear{Rockafellar}{Rockafellar}{1970}]{Rockafellar}
Rockafellar, R.~T. (1970).
\newblock {\em Convex Analysis}.
\newblock Princeton University Press.

\bibitem[\protect\citeauthoryear{Rust}{Rust}{1987}]{Rust}
Rust, J. (1987).
\newblock Optimal replacement of {GMC} bus engines: An empirical model of
  {H}arold {Z}urcher.
\newblock {\em Econometrica\/}~{\em 55\/}(5), 999--1033.

\bibitem[\protect\citeauthoryear{Schennach}{Schennach}{2014}]{Schennach}
Schennach, S.~M. (2014).
\newblock Entropic latent variable integration via simulation.
\newblock {\em Econometrica\/}~{\em 82\/}(1), 345--385.

\bibitem[\protect\citeauthoryear{Shapiro}{Shapiro}{1990}]{Shapiro1990}
Shapiro, A. (1990).
\newblock On concepts of directional differentiability.
\newblock {\em Journal of Optimization Theory and Applications\/}~{\em
  66\/}(3), 477--487.

\bibitem[\protect\citeauthoryear{Shapiro}{Shapiro}{2017}]{Shapiro2017}
Shapiro, A. (2017).
\newblock Distributionally robust stochastic programming.
\newblock {\em SIAM Journal on Optimization\/}~{\em 27\/}(4), 2258--2275.

\bibitem[\protect\citeauthoryear{Su and Judd}{Su and Judd}{2012}]{SuJudd}
Su, C.-L. and K.~L. Judd (2012).
\newblock Constrained optimization approaches to estimation of structural
  models.
\newblock {\em Econometrica\/}~{\em 80\/}(5), 2213--2230.

\bibitem[\protect\citeauthoryear{Tamer}{Tamer}{2003}]{Tamer2003}
Tamer, E. (2003).
\newblock Incomplete simultaneous discrete response model with multiple
  equilibria.
\newblock {\em The Review of Economic Studies\/}~{\em 70\/}(1), 147--165.

\bibitem[\protect\citeauthoryear{Tamer}{Tamer}{2015}]{Tamer2015}
Tamer, E. (2015).
\newblock Sensitivity analysis in some econometric models.
\newblock Cowles Lecture, Econometric Society World Congress, August 17-21,
  2015, Montr\'{e}al.

\bibitem[\protect\citeauthoryear{Tebaldi, Torgovitsky, and Yang}{Tebaldi
  et~al.}{2019}]{TTY}
Tebaldi, P., A.~Torgovitsky, and H.~Yang (2019).
\newblock Nonparametric estimates of demand in the california health insurance
  exchange.
\newblock {NBER} working paper.

\bibitem[\protect\citeauthoryear{Torgovitsky}{Torgovitsky}{2019a}]{Torgovitsky}
Torgovitsky, A. (2019a).
\newblock Nonparametric inference on state dependence in unemployment.
\newblock {\em Econometrica\/}~{\em 87\/}(5), 1475--1505.

\bibitem[\protect\citeauthoryear{Torgovitsky}{Torgovitsky}{2019b}]{Torgovitsky2019QE}
Torgovitsky, A. (2019b).
\newblock Partial identification by extending subdistributions.
\newblock {\em Quantitative Economics\/}~{\em 10\/}(1), 105--144.

\end{thebibliography}


\begin{thebibliography}{}

\bibitem[\protect\citeauthoryear{Andrews, Gentzkow, and Shapiro}{Andrews
  et~al.}{2017}]{AGS2017}
Andrews, I., M.~Gentzkow, and J.~M. Shapiro (2017).
\newblock Measuring the sensitivity of parameter estimates to estimation
  moments.
\newblock {\em The Quarterly Journal of Economics\/}~{\em 132\/}(4),
  1553--1592.

\bibitem[\protect\citeauthoryear{Andrews, Gentzkow, and Shapiro}{Andrews
  et~al.}{2020}]{AGS2018}
Andrews, I., M.~Gentzkow, and J.~M. Shapiro (2020).
\newblock On the informativeness of descriptive statistics for structural
  estimates.
\newblock {\em Econometrica\/}~(6), 2231--2258.

\bibitem[\protect\citeauthoryear{Bonhomme and Weidner}{Bonhomme and
  Weidner}{2021}]{BW}
Bonhomme, S. and M.~Weidner (2021).
\newblock Minimizing sensitivity to model misspecification.
\newblock {\em Quantitative Economics\/}, forthcoming.

\bibitem[\protect\citeauthoryear{Csisz\'{a}r and Shields}{Csisz\'{a}r and
  Shields}{2004}]{CsiszarShields}
Csisz\'{a}r, I. and P.~C. Shields (2004).
\newblock Information theory and statistics: A tutorial.
\newblock {\em Communications and Information Theory\/}~{\em 1\/}(4), 417--528.

\bibitem[\protect\citeauthoryear{Fang and Santos}{Fang and
  Santos}{2019}]{FangSantos}
Fang, Z. and A.~Santos (2019).
\newblock Inference on directionally differentiable functions.
\newblock {\em Review of Economic Studies\/}~{\em 86\/}(1), 377--412.

\bibitem[\protect\citeauthoryear{Hiriart-Urruty and
  Lemar\'echal}{Hiriart-Urruty and Lemar\'echal}{2001}]{HU-L}
Hiriart-Urruty, J.-B. and C.~Lemar\'echal (2001).
\newblock {\em Fundamentals of Convex Analysis}.
\newblock Springer.

\bibitem[\protect\citeauthoryear{Lehmann and Casella}{Lehmann and
  Casella}{1998}]{LehmannCasella}
Lehmann, E.~L. and G.~Casella (1998).
\newblock {\em Theory of Point Estimation\/} (2 ed.).
\newblock Springer-Verlag New York, Inc.

\bibitem[\protect\citeauthoryear{Milgrom and Segal}{Milgrom and
  Segal}{2002}]{MilgromSegal}
Milgrom, P. and I.~Segal (2002).
\newblock Envelope theorems for arbitrary choice sets.
\newblock {\em Econometrica\/}~{\em 70\/}(2), 583--601.

\bibitem[\protect\citeauthoryear{Rockafellar}{Rockafellar}{1970}]{Rockafellar}
Rockafellar, R.~T. (1970).
\newblock {\em Convex Analysis}.
\newblock Princeton University Press.

\bibitem[\protect\citeauthoryear{Shapiro}{Shapiro}{1991}]{Shapiro1991}
Shapiro, A. (1991).
\newblock Asymptotic analysis of stochastic programs.
\newblock {\em Annals of Operations Research\/}~{\em 30\/}(1), 169--186.

\end{thebibliography}


\begin{thebibliography}{}

\bibitem[\protect\citeauthoryear{Bickel, Klaassen, Ritov, and Wellner}{Bickel
  et~al.}{1993}]{BKRW}
Bickel, P.~J., C.~A. Klaassen, Y.~Ritov, and J.~A. Wellner (1993).
\newblock {\em Efficient and Adaptive Estimation for Semiparametric Models}.
\newblock Springer-Verlag, New York.

\bibitem[\protect\citeauthoryear{Bonnans and Shapiro}{Bonnans and
  Shapiro}{2000}]{BS}
Bonnans, J. and A.~Shapiro (2000).
\newblock {\em Perturbation Analysis of Optimization Problems}.
\newblock Springer.

\bibitem[\protect\citeauthoryear{Csisz\'{a}r and Mat\'{u}\v{s}}{Csisz\'{a}r and
  Mat\'{u}\v{s}}{2012}]{CM}
Csisz\'{a}r, I. and F.~Mat\'{u}\v{s} (2012).
\newblock Generalized minimizers of convex integral functionals, {B}regman
  distance, {P}ythagorean identities.
\newblock {\em Kybernetika\/}~{\em 48\/}(4), 637--689.

\bibitem[\protect\citeauthoryear{Gradshteyn and Ryzhik}{Gradshteyn and
  Ryzhik}{2014}]{GradshteynRyzhik}
Gradshteyn, I.~S. and I.~M. Ryzhik (2014).
\newblock {\em Table of Integrals, Series, and Products\/} (8 ed.).
\newblock Academic Press.

\bibitem[\protect\citeauthoryear{Komunjer and Ragusa}{Komunjer and
  Ragusa}{2016}]{KR}
Komunjer, I. and G.~Ragusa (2016).
\newblock Existence and characterization of conditional density projections.
\newblock {\em Econometric Theory\/}~{\em 32\/}(4), 947--987.

\bibitem[\protect\citeauthoryear{Krasnosel'skii and Rutickii}{Krasnosel'skii
  and Rutickii}{1961}]{Kras}
Krasnosel'skii, M.~A. and {\relax Ya}.~B. Rutickii (1961).
\newblock {\em Convex Functions and Orlicz Spaces}.
\newblock P. Noordhoff Ltd., Groningen.

\bibitem[\protect\citeauthoryear{Newey and McFadden}{Newey and
  McFadden}{1994}]{NeweyMcFadden}
Newey, W.~K. and D.~McFadden (1994).
\newblock Chapter 36 large sample estimation and hypothesis testing.
\newblock Volume~4 of {\em Handbook of Econometrics}, pp.\  2111--2245.
  Elsevier.

\bibitem[\protect\citeauthoryear{Rockafellar}{Rockafellar}{1970}]{Rockafellar}
Rockafellar, R.~T. (1970).
\newblock {\em Convex Analysis}.
\newblock Princeton University Press.

\bibitem[\protect\citeauthoryear{Rockafellar and Wets}{Rockafellar and
  Wets}{1998}]{RW}
Rockafellar, R.~T. and R.~J.~B. Wets (1998).
\newblock {\em Variational Analysis}.
\newblock Springer.

\bibitem[\protect\citeauthoryear{Shapiro}{Shapiro}{1991}]{Shapiro1991}
Shapiro, A. (1991).
\newblock Asymptotic analysis of stochastic programs.
\newblock {\em Annals of Operations Research\/}~{\em 30\/}(1), 169--186.

\end{thebibliography}
\end{document}